\documentclass[11pt,oneside]{book}

\usepackage{graphicx}
\usepackage{amsmath, amssymb}
\usepackage{bussproofs}
\usepackage{url}
\usepackage{enumerate}
\usepackage{verbatim}
\usepackage{lscape}
\usepackage{multirow}
\usepackage{wrapfig}

\usepackage{amsthm, thmtools}
\usepackage{nameref}
\usepackage[breaklinks=true]{hyperref}
\usepackage{cleveref}

\usepackage[usenames,dvipsnames]{pstricks}
\usepackage{pst-node, pst-tree}
\usepackage{epsfig}
\usepackage{pst-grad} % For gradients
\usepackage{pst-plot} % For axes

\makeatletter
\renewcommand\thmt@listnumwidth{4.3em}
\makeatother

\newenvironment{dedication}
{
   \cleardoublepage
   \thispagestyle{empty}
   \vspace*{\stretch{1}}
   \hfill\begin{minipage}[t]{0.66\textwidth}
   \raggedright
}%
{
   \end{minipage}
   \vspace*{\stretch{3}}
   \clearpage
}

\declaretheorem
[
    numberwithin=section,
    name=Theorem,
    refname={theorem,theorems},
    Refname={Theorem,Theorems}
] {theorem}

\declaretheorem
[
    numberwithin=section,
    name=Lemma,
    refname={lemma,lemmas},
    Refname={Lemma,Lemmas}
] {lemma}

\declaretheorem
[
    numberwithin=section,
    name=Corollary,
    refname={corollary,corollaries},
    Refname={Corollary,Corollaries}
] {corollary}

\declaretheorem
[
    style=definition,
    numberwithin=section,
    name=Definition,
    refname={definition,definitions},
    Refname={Definition,Definitions}
] {definition}

\declaretheorem
[
    style=remark,
    numbered=no,
    name=Remark,
    refname={remark,remarks},
    Refname={Remark,Remarks}
]{remark}

\declaretheorem
[
    style=example,
    numbered=no,
    name=Example,
    refname={example,examples},
    Refname={Example,Examples}
]{example}
\newcommand{\turnaround}[1]{%
  \rotatebox[origin=c]{180}{\ensuremath#1}}

\newcommand{\logicpar}{\turnaround{\&}}

\newcommand{\bra}[1]{\ensuremath{\langle#1|}}
\newcommand{\ket}[1]{\ensuremath{|#1\rangle}}
\begin{document}

\begin{titlepage}
\begin{center}
    {\LARGE Higher-order semantics for quantum programming languages with classical control}\\[2cm]
    {\Large George Philip Atzemoglou}\\
    {\large Wolfson College, Oxford}\\[3cm]
    \includegraphics[width=5cm]{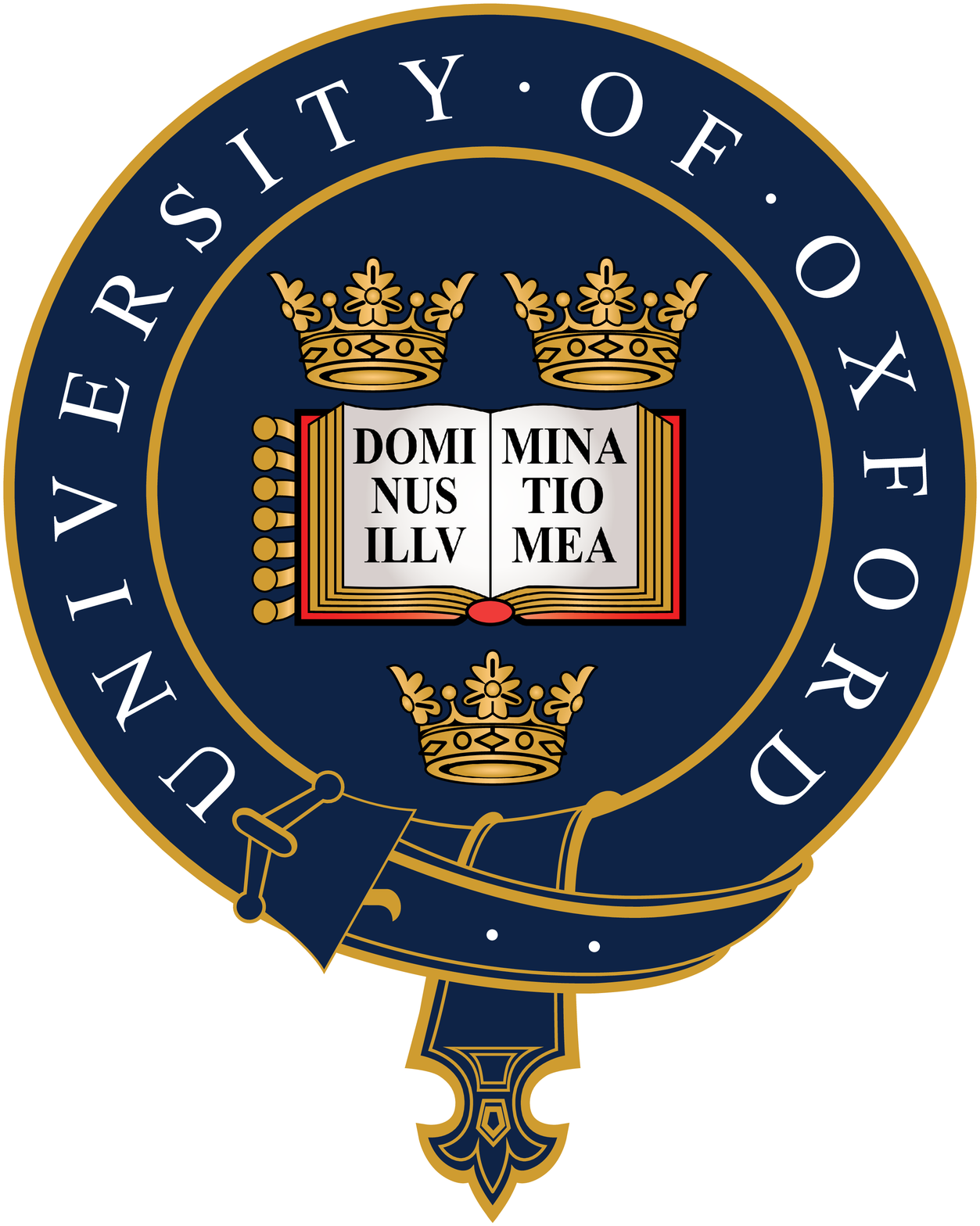}\\[3cm]
    {\large Department of Computer Science, University of Oxford}\\[1cm]
    \textit{Submitted for the degree of Doctor of Philosophy}\\
    \textit{Michaelmas Term 2012}
\end{center}
\end{titlepage}

\pagenumbering{gobble}

\begin{dedication}
To my grandfathers,\\
George and Ermis.
\end{dedication}

\pagenumbering{roman}

\chapter*{\centering \begin{LARGE}Abstract\end{LARGE}}
\begin{quotation}
    \noindent This thesis studies the categorical formalisation of quantum computing, through the prism of type theory, in a three-tier process. The first stage of our investigation involves the creation of the dagger lambda calculus; a lambda calculus for dagger compact categories \cite{AC04}. Our second contribution lifts the expressive power of the dagger lambda calculus to that of a quantum programming language, by adding classical control in the form of complementary classical structures \cite{CPP10, CD11} and dualisers \cite{CPP08}. Finally, our third contribution demonstrates how our lambda calculus can be applied to various well known problems in quantum computation.

    Our construction of the dagger lambda calculus extends the linear typed lambda calculus, by defining a higher-order language for quantum protocols which is an internal language for dagger compact categories. The resulting language includes a linear negation operator and redefines the notion of binding as a symmetric relation whose scope spans the entire sequent. Reduction works by means of an explicit substitution, in the spirit of the operational semantics of the linear chemical abstract machine. The rules for explicit substitution act globally on the entire typing judgement, instead of limiting their scope to a specific subterm. This particular implementation of reduction enables us to enrich our typing dynamics by allowing the binding not just of variables, but of arbitrary terms. An elimination procedure allows us to reconstruct application using Cut, hence removing it from our primitive rule set. The new rules allow for a fully symmetric language, where inputs and outputs are treated as elements of a symmetric relation, and give rise to a new structural rule called the \textit{dagger-flip}. The resulting set of rules is minimal and simple to use, which allows us to prove with ease properties like subject reduction, confluence, strong normalisation and consistency. Our analysis of the language's semantics is completed by a proof that the dagger lambda calculus is an internal language for dagger compact categories.

    In our second contribution, we provide a controlled way of breaching the linearity constraints of the dagger lambda calculus. Instead of the exponential connective of linear logic, we enrich our language with a view of axiomatising the basis structures of Hilbert spaces. We achieve this by providing the type-theoretic equivalent of complementary classical structures, which allows us to implement a controlled form of copying and deleting of terms in our classical basis. In order to retain sensible types, our language is also equipped with dualisers in a non self-dual setting, allowing us to factor the familiar notion of Currying by making its rule admissible in the dagger lambda calculus.

    Our final contribution is in the study of three well known applications of quantum computation. We first demonstrate the expressiveness of the dagger lambda calculus by using it to represent the formalism of \textit{Quantum Key Distribution}; we then put it to use, by using the language to verify the protocol's procedures. Our second application is in the \textit{quantum Fourier transform}; we demonstrate how our language can represent a controlled phase gate, we use that construction to build a \textit{quantum Fourier transform}, and we then use the dagger lambda calculus to "run" a sample input through the \textit{quantum Fourier transform}. Lastly, our third application provides an examination of the \textit{teleportation protocol}, explaining how it is represented in our language, demonstrating corrections and the flow of measurement outcomes, and showing how a quantum term actually ends up teleporting from one party to the other.
\end{quotation}
\clearpage

\clearpage
\setcounter{tocdepth}{2}
\tableofcontents

\clearpage
\pagenumbering{arabic}
\part{Introduction}
\label{Part:Introduction}
\chapter{Motivation}
\label{Chapter:Motivation}
Quantum mechanics was initially developed in the first third of the previous century, with quantum computation being explicitly studied since the 1980's. Though a lot of work has been done on quantum programming languages, it is still a nascent branch of science, where most of the languages have not yet been developed to higher levels of abstraction. As a result, while we do have a handful of quantum algorithms which provide promising results in the area of computational efficiency, much of the work involved in designing such an algorithm seems to be still largely based on guesswork; there is no clear set of rules, or unifying principle, that would easily allow us to combine computational primitives into building a new and efficient quantum algorithm. Furthermore, even though quantum computation is probably a more powerful means of computation than classical computing, this remains an unproven conjecture to this day. Unless we develop the means of abstracting our computational primitives to higher notions, thereby setting the foundations for a solid theory of algorithms and figuring out what gives quantum computers their extra power, we will probably never be able to prove a computational speedup or write efficient algorithms en masse.

Since the turn of the century, we have seen increasing interest in the development of a quantum programming language. One very actively pursued approach has been from a programming languages perspective \cite{Sel04a, Sel04b, vTD03, vT04, SV06, SV08, SV10}. The researchers in this field, whose work has been seminal in establishing a semantic approach to quantum programming language design, have focused in designing a higher order lambda calculus for quantum computation with classical control. More specifically, in \cite{SV10}, a quantum lambda calculus with a complicated set of rules is presented, whose structural equations nevertheless allow for higher-order structures. The rest of the work towards constructing a concrete model for the language's semantics remains an open problem.

Another very successful approach has come from a category theoretic perspective, where researchers have proposed a variety of diagrammatic calculi: Starting with the work of \cite{AC04} and \cite{Sel05}; progressing into the classical structures of \cite{CP06}, \cite{CP07} and \cite{CPP10}; the dualisers of \cite{CPP08}; up until the full axiomatisation of bases through complementary observables in \cite{CD08} and \cite{CD11}. Because of their design, these languages are capable of expressing both quantum and classical processes, which makes them ideal for representing measurement based quantum computation, the computational paradigm that is closest to being practically implemented. Furthermore, these languages possess an obvious visual appeal and are, as a result, easy to understand and work with. Despite all that, the diagrammatic calculus does not readily lend itself to higher order operations. Moreover, despite some recent results by \cite{DP10} and \cite{Kis11}, which deal with the issue of automated rewriting in the diagrammatic calculus, the rewrite steps are not always immediately apparent or easy to follow.

The purpose of this dissertation is to bridge these two approaches, hence bringing the programming languages approach closer to the categorical approach of complementary observables, by casting the diagrammatic formalism into the rich and well established tradition of type theory. The language presented in part \ref{Part:Quantum programming and classical control} of the dissertation attempts this bridge by providing a higher-order computational interpretation for the categorical semantics of \cite{AC04}, \cite{CD11}, \cite{CPP08} and \cite{CPP10}. The dagger lambda calculus is expressive enough to perform operations on arbitrary (black box) functions and, as such, it is capable of encoding many well known quantum algorithms.

The natural way in which the dagger lambda calculus represents higher-order operations, together with the simplicity of its sequent rewriting, make it an ideal tool to be used side by side with the diagrammatic calculus. The two calculi can then complement each other, by providing the best of both worlds; visual clarity, intuitive rewrite rules and higher-order expressibility.

\chapter{Outline of the Dissertation}
\label{Chapter:Outline of the Dissertation}
Part \ref{Part:Introduction} serves as an introduction to the semantics of quantum programming languages. It prepares the ground for the rest of the thesis by presenting the main motivation for research in this area. The current section will explain the structure of the remaining parts of the dissertation, by providing an outline of the various sections and subsections.

In order to make this dissertation more self-contained, part \ref{Part:Background material} covers all of the background material that will be used in the constructive portions of this work. This presentation starts with a quick overview of quantum computing in chapter \ref{Chapter:Quantum computing}. Chapter \ref{Chapter:Category theory} provides some background on category theory, including definitions for a selection of topics from \cite{Mac98}, as well as the definitions for many of the categories that will be used in later chapters. Chapter \ref{Chapter:Categorical model} explains how the categorical structures of \cite{AC04}, together with the Frobenius algebras of \cite{CD08}, can be used to model quantum computation and classical control structures. Chapter \ref{Chapter:Linear logic} contains a brief presentation of linear logic, as well as a variant by \cite{AD06} that is better suited for quantum computation. Finally, chapter \ref{Chapter:The linear typed lambda calculus}, the last chapter in this part of the dissertation, prepares the ground for the later parts by presenting the linear typed lambda calculus of \cite{AT10}.

Part \ref{Part:Quantum programming and classical control} of the dissertation forms the main constructive portion of this work. It begins by introducing in section \ref{Section:Language construction} the dagger lambda calculus, a language for quantum protocols that corresponds to dagger compact categories. The next section, \ref{Section:Proofs of properties}, provides proofs of the language's most important properties; namely of subject reduction, strong normalisation, confluence and consistency. Section \ref{Section:Correspondence to dagger compact categories} presents a proof of the language's correspondence to dagger compact categories by showing that the syntactic category is indeed a free dagger compact category. In section \ref{Section:Classical structures}, the dagger lambda calculus is enriched with the classical structures of \cite{CP06}, \cite{CP07} and \cite{CPP10}. The language is then further enriched by introducing the notion of a dualiser \cite{CPP08} in section \ref{Section:Dualisers}, which can be used as a primitive to factor the notion of Currying, hence making the language's Curry rule admissible. Sections \ref{Section:Monoidal product of terms and phase shifts} and \ref{Section:Unbiased and classical constants} define the monoidal product and phase shift operations in the calculus and outline the requirements for the properties of unbiasedness and classical constants. Section \ref{Section:Complementary observables} enriches the dagger lambda calculus with the complementary classical structures of \cite{CD11}, which lift its expressive power from quantum protocols to that of a quantum programming language.

Part \ref{Part:Applications} focuses on the applications of the language designed in part \ref{Part:Quantum programming and classical control}. Chapter \ref{Chapter:Quantum Key Distribution} shows how the dagger lambda calculus can be used to perform the Quantum Key Distribution of \cite{Eke91}, using the formalism of \cite{CWWWZ11}. Chapter \ref{Chapter:Quantum Fourier Transform} then demonstrates how the language can represent and run a quantum Fourier transform, the most essential part of Shor's factoring algorithm \cite{Sho97}. Completing the applications' part, chapter \ref{Chapter:Teleportation Protocol} uses the language to perform the teleportation protocol.

Finally, the dissertation comes to a close in part \ref{Part:Conclusion}. This part highlights the importance of the higher-order computational interpretation provided in the dissertation, by outlining the structural insights that were gained through the study of its semantics. A list of possible directions for future work is provided, along with some concluding remarks.

\clearpage
\part{Background material}
\label{Part:Background material}
\chapter{Quantum computing}
\label{Chapter:Quantum computing}
Quantum computing is a radically different paradigm for computation which relies on the laws of quantum mechanics, in the hopes of achieving a higher computational efficiency than its currently used classical counterpart. This section will cover some fundamental concepts of quantum computing \cite{Mer07,NC00}, reviewing all the material that is necessary for understanding this dissertation.

Classical computers operate on regular bits, whose value is either $0$ or $1$. When studying quantum computers, we view $\ket{0}$ and $\ket{1}$ as orthonormal vectors, using them as a basis to span a complex Hilbert space. We are free to pick a different set of orthonormal vectors as the basis of our Hilbert space, however, the one mentioned earlier is usually referred to as the \textit{standard basis}. The length of a vector in this Hilbert space does not really matter; we therefore only keep its direction and group all like vectors up to a complex multiple into equivalence classes called rays. Qubits, the quantum analogue of a bit, can have any of these rays as their value. This means that the value or \textit{state} of any qubit can be written as
\[ \ket{\psi} = \alpha \ket{0} + \beta \ket{1} \]
where $\alpha$ and $\beta$ are complex coefficients. The fact that a quantum state can be "a little bit of" $\ket{0}$ and "a little bit of" $\ket{1}$ at the same time, is called \textit{superposition}. Transformations of a quantum system's state are described by unitary operations acting on the system's Hilbert space. Since the underlying field for our Hilbert spaces is the field of complex numbers, state vectors in $\mathcal{H}$ can be trivially shown to be isomorphic to the linear maps in $\mathbb{C} \rightarrow \mathcal{H}$. The isomorphism maps every state $\ket{\psi}$ to the linear map spanned by $1 \mapsto \psi$. By a slight abuse of notation, we sometimes use $\ket{\psi}$ to refer to the linear map as such:
\[ \ket{\psi} : \mathbb{C} \rightarrow \mathcal{H} :: c \mapsto c\psi \]

The state of composite quantum systems is represented by the tensor product of the Hilbert spaces that describe their constituent parts. This behaves like a regular Kronecker product; for a system composed of $A$ and $B$, we write $A \otimes B$. Similarly, for two linear maps $f$ and $g$ running parallel to each other, each acting on a different state of a composite system, we would write $f \otimes g$. It seems natural that we could use $\ket{\psi} \otimes \ket{\phi}$ to describe the state of two qubits. We tend to write $\ket{00}$ for $\ket{0} \otimes \ket{0}$ and $\ket{11}$ for $\ket{1} \otimes \ket{1}$. Because of superposition, however, there are some cases where a state cannot be written as the tensor product of two or more states. Typical examples of this are the Bell states: $\ket{00} + \ket{11}$, $\ket{00} - \ket{11}$, $\ket{01} + \ket{10}$ and $\ket{01} - \ket{10}$.

Hilbert spaces come equipped with an inner-product which, as a convention in quantum computation, is usually defined as linear in the second argument:
\[ \langle - | - \rangle : \mathcal{H} \times \mathcal{H} \rightarrow \mathbb{C} \]
In order to formally introduce Dirac notation in our work, we want to further refine the definition of an inner-product by breaking it down to a composition of a \textit{bra} $\bra{\phi}$ and a \textit{ket} $\ket{\psi}$, yielding $\langle\phi|\psi\rangle = \bra{\phi} \circ \ket{\psi}$. A \textit{bra} $\bra{\phi}$ stands for the dual vector of $\ket{\phi}$. At this point it is useful to introduce the notion of an \textit{adjoint} or, more precisely, a Hermitian adjoint for linear maps:
\[ (f : \mathcal{H}_1 \rightarrow \mathcal{H}_2) \mapsto (f^\dag : \mathcal{H}_2 \rightarrow \mathcal{H}_1) \]
Building up on our convention to use $\ket{\phi}$ to represent linear maps, we will define the adjoint of that map to be equivalent to the dual of the corresponding state vector \cite{NC00}:
\[ \left(\ket{\phi}\right)^\dag \equiv \bra{\phi} : \mathcal{H} \rightarrow \mathbb{C} \]

Performing a measurement on a quantum system against some orthonormal basis destroys its state, by making it collapse into one of the basis vectors. To define this more formally, a measurement against some orthonormal basis consists of a set of projectors ${P_i}$, each sending the measured state to one of the basis vectors. At the time of measurement, the system's wave function collapses, causing a non-deterministic jump in the system's state by stochastically applying one of the projectors to the state vector. If we represent the state as a linear combination of basis vectors, the square of the complex coefficient of any basis vector gives us the probability of collapsing to that outcome during a measurement.

\chapter{Category theory}
\label{Chapter:Category theory}
Category theory is an area of mathematics that provides us with a means of reasoning about common properties of abstract structures; it allows the use of diagrammatic reasoning while, at the same time, extending connections to mathematical logic. The rest of this dissertation makes extensive use of category theory \cite{Mac98}. Therefore, for the sake of completeness, we will provide definitions for all of the notions used.

\section{Basic notions}
\label{Section:Basic notions}
Since this chapter is about category theory, we will begin our exposition appropriately by providing a definition for categories. Following that, we will be able to define functors and natural transformations, all of which will be used extensively later on.

\begin{definition}[Category]
A \textit{category} is a collection of objects and arrows between objects such that the following conditions hold:
\begin{itemize}
\item There is a composition operator $\circ$, that can take any two arrows of the form $f : A \rightarrow B$ and $g : B \rightarrow C$ and produce a new arrow $g \circ f : A \rightarrow C$
\item Composition is associative, so $h \circ (g \circ f) = (h \circ g) \circ f$
\item For every object $A$ in our category, there is an identity arrow $id_A : A \rightarrow A$
\item The identity arrows satisfy the \textit{unit law}, whereby, for any arrow $f : A \rightarrow B$ in our category, $id_B \circ f = f = f \circ id_A$
\end{itemize}
\end{definition}

\begin{definition}[Functor]
A \textit{functor} is a morphism between categories, mapping objects to objects and arrows to arrows, in a way that preserves identities and composition.
\end{definition}
\begin{proof}[Example]
Consider two categories $\mathcal{C}$ and $\mathcal{D}$. A functor $F : \mathcal{C} \rightarrow \mathcal{D}$, assigns to each object $A \in \mathcal{C}$ an object $FA \in \mathcal{D}$ and to each arrow $f : A \rightarrow B$ of $\mathcal{C}$ an arrow $Ff : FA \rightarrow FB$ in $\mathcal{D}$. Since it preserves identities and composition, it will have to be the case that $F(id_A) = id_{FA}$ and $F(g \circ f) = Fg \circ Ff$.
\renewcommand{\qedsymbol}{}
\end{proof}

\begin{definition}[Natural transformation]
When $F$ and $G$ are functors such that $F,G : \mathcal{C} \rightarrow \mathcal{D}$, a \textit{natural transformation} $\eta : F \Rightarrow G$ is a collection of arrows $\eta_A : FA \rightarrow GA$ in $\mathcal{D}$ for every object $A$ in $\mathcal{C}$. These arrows have to be such that for any $f : A \rightarrow B$ in $\mathcal{C}$, the following diagram commutes:
\[
\setlength{\arraycolsep}{1.1cm}
\begin{array}{cc}
\Rnode{FA}{FA} & \Rnode{GA}{GA} \\[1.5cm]
\Rnode{FB}{FB} & \Rnode{GB}{GB}
\end{array}
\psset{nodesep=5pt,arrows=->} %
\everypsbox{\scriptstyle} %
\ncline{FA}{GA}\taput{\eta_A} %
\ncline{FB}{GB}\tbput{\eta_B} %
\ncline{FA}{FB}\tlput{Ff} %
\ncline{GA}{GB}\trput{Gf}
\]
\end{definition}
\begin{proof}[Note]
When all the $\eta_A$ of a natural transformation are isomorphisms, we call $\eta$ a natural isomorphism.
\renewcommand{\qedsymbol}{}
\end{proof}

\section{More advanced notions}
\label{Section:More advanced notions}
We now wish to define a way of describing how different categories are related to each other. In order to do that, consider categories $\mathcal{C}$ and $\mathcal{D}$, with functors $F$ and $G$ between them such that:
\[
    \Large
    \Rnode{C}{\mathcal{C}} \hskip 1.5cm \Rnode{D}{\mathcal{D}}
    \normalsize
    \psset{nodesep=5pt,offset=4pt,arrows=->} %
    \everypsbox{\scriptstyle} %
    \ncline{C}{D}\taput{F} %
    \ncline{D}{C}\tbput{G} %
\]

The strictest and perhaps most obvious type of relation occurs when both of these functors are identity endofunctors $F = G = Id$, which would mean that the categories are \textit{equal}. An \textit{isomorphism} is a weaker kind of relation where $G \circ F = Id_\mathcal{C}$ and $F \circ G = Id_\mathcal{D}$. An even weaker kind of relation is an \textit{equivalence}, where instead of requiring the composition of functors to be equal to the identity functor, we ask that there be a natural isomorphism between them. In other words $G \circ F \cong Id_\mathcal{C}$ and $F \circ G \cong Id_\mathcal{D}$.

Continuing down that path leads us to a very important notion called an adjunction, which is another yet weaker kind of relation between categories \cite{Che07}. To define this more rigorously:

\begin{definition}[Adjunction]
Let $\mathcal{C}$ and $\mathcal{D}$ be categories, with functors $F : \mathcal{C} \rightarrow \mathcal{D}$ and $G : \mathcal{D} \rightarrow \mathcal{C}$. We say there is an \textit{adjunction} $\langle F,G,\eta,\varepsilon\rangle$ when there exist two natural transformations $\eta : Id_\mathcal{C} \Rightarrow GF$ and $\varepsilon : FG \Rightarrow Id_\mathcal{D}$, respectively called the \textit{unit} and \textit{counit} of the adjunction, such that the following diagrams commute:
\[
\psmatrix[colsep=1.5cm,rowsep=1cm]
G   &   GFG &   &   F   &   FGF \\
    &   G   &   &       &   F
\endpsmatrix
\everypsbox{\scriptstyle} %
\psset{nodesep=5pt,arrows=->} %
% First triangle diagram
\ncline{1,1}{1,2} \taput{\eta G} %
\ncline{1,2}{2,2} \trput{G \varepsilon} %
\ncline{1,1}{2,2} \tlput{1_G} %
% Second triangle diagram
\ncline{1,4}{1,5} \taput{F \eta} %
\ncline{1,5}{2,5} \trput{\varepsilon F} %
\ncline{1,4}{2,5} \tlput{1_F} %
\]

\end{definition}
\begin{proof}[Note]
We say that $F$ is \textit{left adjoint} to $G$, writing this as $F \dashv G$. Similarly, $G$ is right adjoint to $F$.
\renewcommand{\qedsymbol}{}
\end{proof}

\begin{definition}[Monad]
A \textit{monad} $\langle T,\eta,\mu\rangle$ in a category $\mathcal{C}$ consists of an endofunctor $T : \mathcal{C} \rightarrow \mathcal{C}$, together with two natural transformations $\eta : Id_\mathcal{C} \Rightarrow T$ and $\mu : T^2 \Rightarrow T$, such that the following diagrams commute:
\[
\psmatrix[colsep=1.5cm,rowsep=1cm]
T^3 &   T^2 &   &   T   &   T^2 &   T \\
T^2 &   T   &   &       &   T
\endpsmatrix
\everypsbox{\scriptstyle} %
\psset{nodesep=5pt,arrows=->} %
% Square diagram
\ncline{1,1}{1,2} \taput{T \mu} %
\ncline{1,2}{2,2} \trput{\mu} %
\ncline{1,1}{2,1} \tlput{\mu T} %
\ncline{2,1}{2,2} \tbput{\mu} %
% Triangle diagram
\ncline{1,4}{1,5} \taput{\eta T} %
\ncline{1,6}{1,5} \taput{T \eta} %
\ncline{1,5}{2,5} \trput{\mu} %
\ncline{1,4}{2,5} \tlput{1_T} %
\ncline{1,6}{2,5} \trput{1_T} %
\]

\end{definition}

\begin{corollary}[Adjunctions define a monad]
Any adjunction $\langle F,G,\eta,\varepsilon\rangle$ defines a monad. This can be done by setting $T = GF$, which would make the unit of our adjunction a natural transformation of the form $\eta : Id_\mathcal{C} \Rightarrow T$ and our $\mu = G \varepsilon F : GFGF \Rightarrow GF$. The resulting monad would be of the form $\langle FG,\eta,G \varepsilon F\rangle$. This can be easily verified by checking that the following diagrams do indeed commute:
\[
\psmatrix[colsep=2cm,rowsep=1.3cm]
GFGFGF  &   GFGF    &   GF  &   GFGF    &   GF \\
GFGF    &   GF      &       &   GF
\endpsmatrix
\everypsbox{\scriptstyle} %
\psset{nodesep=5pt,arrows=->} %
% Square diagram
\ncline{1,1}{1,2} \taput{GFG \varepsilon F} %
\ncline{1,2}{2,2} \trput{G \varepsilon F} %
\ncline{1,1}{2,1} \tlput{G \varepsilon FGF} %
\ncline{2,1}{2,2} \tbput{G \varepsilon F} %
% Triangle diagram
\ncline{1,3}{1,4} \taput{\eta GF} %
\ncline{1,5}{1,4} \taput{GF \eta} %
\ncline{1,4}{2,4} \trput{G \varepsilon F} %
\ncline{1,3}{2,4} \tlput{1_T} %
\ncline{1,5}{2,4} \trput{1_T} %
\]

\end{corollary}

\section{Specific constructions}
\label{Section:Specific constructions}
\begin{definition}[Monoidal category]
A \textit{monoidal category} is a category that has been equipped with an bifunctor called tensor $\otimes : \mathcal{C} \times \mathcal{C} \longrightarrow \mathcal{C}$. Up to appropriate natural isomorphisms, the tensor is associative and features a special object $I$ that acts as a left and right identity:

\[ \alpha_{A,B,C} : A \otimes (B \otimes C) \stackrel{\cong}{\longrightarrow} (A \otimes B) \otimes C \]
\begin{center}
$\lambda_A : I \otimes A \stackrel{\cong}{\longrightarrow} A$ \hskip 2cm $\rho_A : A \otimes I \stackrel{\cong}{\longrightarrow} A$
\end{center}

The tensor product thus forms a monoid, with $I$ acting as the unit. These isomorphisms have to further satisfy some conditions called coherence conditions. These can be summarily represented by requiring that the following diagrams commute for all $A$, $B$, $C$ and $D$:
\[
% Generated with LaTeXDraw 2.0.8
% Wed Oct 09 11:23:31 BST 2013
% \usepackage[usenames,dvipsnames]{pstricks}
% \usepackage{epsfig}
% \usepackage{pst-grad} % For gradients
% \usepackage{pst-plot} % For axes
\scalebox{1} % Change this value to rescale the drawing.
{
\begin{pspicture}(0,-4.706704)(14.542812,4.7130847)
\usefont{T1}{ptm}{m}{n}
\rput(7.3323436,4.4543896){$(A \otimes B) \otimes (C \otimes D)$}
\usefont{T1}{ptm}{m}{n}
\rput(2.9323437,2.2543893){$A \otimes (B \otimes (C \otimes D))$}
\usefont{T1}{ptm}{m}{n}
\rput(11.732344,2.2543893){$((A \otimes B) \otimes C) \otimes D$}
\usefont{T1}{ptm}{m}{n}
\rput(2.9323437,-0.5456106){$A \otimes ((B \otimes C) \otimes D)$}
\usefont{T1}{ptm}{m}{n}
\rput(11.732344,-0.5456106){$(A \otimes (B \otimes C)) \otimes D$}
\psline[linewidth=0.04cm,arrowsize=0.05291667cm 2.0,arrowlength=1.4,arrowinset=0.4]{->}(2.9809375,1.9443893)(2.9809375,-0.25561062)
\psline[linewidth=0.04cm,arrowsize=0.05291667cm 2.0,arrowlength=1.4,arrowinset=0.4]{<-}(11.780937,1.9443893)(11.780937,-0.25561062)
\psline[linewidth=0.04cm,arrowsize=0.05291667cm 2.0,arrowlength=1.4,arrowinset=0.4]{->}(4.7809377,-0.5556106)(9.980938,-0.5556106)
\usefont{T1}{ptm}{m}{n}
\rput{23.539944}(1.7754987,-1.4517372){\rput(4.3523436,3.5543895){$\alpha_{A, B, C \otimes D}$}}
\usefont{T1}{ptm}{m}{n}
\rput{-23.857168}(-0.60100454,4.424468){\rput(10.152344,3.6543894){$\alpha_{A \otimes B, C, D}$}}
\usefont{T1}{ptm}{m}{n}
\rput(7.3523436,-0.34561062){$\alpha_{A, B \otimes C, D}$}
\usefont{T1}{ptm}{m}{n}
\rput{-270.0}(3.5261083,-1.8567044){\rput(2.6723437,0.85438937){$id_A \otimes \alpha_{B,C,D}$}}
\usefont{T1}{ptm}{m}{n}
\rput{-90.0}(11.256704,12.926107){\rput(12.072344,0.85438937){$\alpha_{A,B,C} \otimes id_D$}}
\psline[linewidth=0.04cm,arrowsize=0.05291667cm 2.0,arrowlength=1.4,arrowinset=0.4]{->}(2.9809375,2.5443895)(6.5809374,4.144389)
\psline[linewidth=0.04cm,arrowsize=0.05291667cm 2.0,arrowlength=1.4,arrowinset=0.4]{<-}(11.780937,2.5443895)(8.180938,4.144389)
\usefont{T1}{ptm}{m}{n}
\rput(7.6023436,-1.7456106){$A \otimes (I \otimes B)$}
\usefont{T1}{ptm}{m}{n}
\rput(12.002344,-1.7456106){$(A \otimes I) \otimes B$}
\usefont{T1}{ptm}{m}{n}
\rput(7.6023436,-3.9456105){$A \otimes B$}
\psline[linewidth=0.04cm,arrowsize=0.05291667cm 2.0,arrowlength=1.4,arrowinset=0.4]{->}(8.780937,-1.7556106)(10.780937,-1.7556106)
\psline[linewidth=0.04cm,arrowsize=0.05291667cm 2.0,arrowlength=1.4,arrowinset=0.4]{->}(7.5809374,-2.0556107)(7.5809374,-3.6556106)
\psline[linewidth=0.04cm,arrowsize=0.05291667cm 2.0,arrowlength=1.4,arrowinset=0.4]{->}(11.980938,-2.0556107)(7.9809375,-3.6556106)
\usefont{T1}{ptm}{m}{n}
\rput(9.742344,-1.4456106){$\alpha_{A, I, B}$}
\usefont{T1}{ptm}{m}{n}
\rput{21.993166}(-0.4415648,-4.058219){\rput(10.202344,-3.1456106){$\rho_A \otimes id_B$}}
\usefont{T1}{ptm}{m}{n}
\rput{-270.0}(4.5361085,-10.066705){\rput(7.282344,-2.7456107){$id_A \otimes \lambda_B$}}
\usefont{T1}{ptm}{m}{n}
\rput(3.6323438,-1.7456106){$\lambda_I = \rho_I : I \otimes I \longrightarrow I$}
\end{pspicture}
}
\]
\end{definition}

\begin{definition}[Symmetric monoidal category]
A \textit{symmetric} monoidal category is a monoidal category with an additional natural isomorphism called symmetry, $\sigma_{A,B} : A \otimes B \stackrel{\cong}{\longrightarrow} B \otimes A$, such that the following diagrams commute:
\[
% Generated with LaTeXDraw 2.0.8
% Sun Oct 13 07:14:32 BST 2013
% \usepackage[usenames,dvipsnames]{pstricks}
% \usepackage{epsfig}
% \usepackage{pst-grad} % For gradients
% \usepackage{pst-plot} % For axes
\scalebox{1} % Change this value to rescale the drawing.
{
\begin{pspicture}(0,-3.8677344)(14.002812,3.8677344)
\usefont{T1}{ptm}{m}{n}
\rput(1.8623438,0.3060156){$(A \otimes B) \otimes C$}
\usefont{T1}{ptm}{m}{n}
\rput(12.062344,0.3060156){$(C \otimes A) \otimes B$}
\psline[linewidth=0.04cm,arrowsize=0.05291667cm 2.0,arrowlength=1.4,arrowinset=0.4]{->}(3.1809375,0.29601562)(5.9809375,0.29601562)
\usefont{T1}{ptm}{m}{n}
\rput{-90.57707}(0.05250093,3.424631){\rput(1.7023437,1.7060156){$\alpha_{A,B,C}$}}
\psline[linewidth=0.04cm,arrowsize=0.05291667cm 2.0,arrowlength=1.4,arrowinset=0.4]{->}(1.9809375,2.7960157)(1.9809375,0.59601563)
\psline[linewidth=0.04cm,arrowsize=0.05291667cm 2.0,arrowlength=1.4,arrowinset=0.4]{->}(12.180938,2.7960157)(12.180938,0.59601563)
\usefont{T1}{ptm}{m}{n}
\rput(1.8623438,3.1060157){$A \otimes (B \otimes C)$}
\usefont{T1}{ptm}{m}{n}
\rput(12.062344,3.1060157){$(A \otimes C) \otimes B$}
\usefont{T1}{ptm}{m}{n}
\rput{-90.0}(10.755078,14.127734){\rput(12.422344,1.7060156){$\sigma_{A,C} \otimes id_B$}}
\usefont{T1}{ptm}{m}{n}
\rput(7.262344,3.1060157){$A \otimes (C \otimes B)$}
\usefont{T1}{ptm}{m}{n}
\rput(7.262344,0.3060156){$C \otimes (A \otimes B)$}
\usefont{T1}{ptm}{m}{n}
\rput(9.702344,0.10601562){$\alpha_{C,A,B}$}
\usefont{T1}{ptm}{m}{n}
\rput(9.702344,3.4060156){$\alpha_{A,C,B}$}
\usefont{T1}{ptm}{m}{n}
\rput(4.552344,0.10601562){$\sigma_{A \otimes B, C}$}
\usefont{T1}{ptm}{m}{n}
\rput(4.5223436,3.4060156){$id_A \otimes \sigma_{B,C}$}
\psline[linewidth=0.04cm,arrowsize=0.05291667cm 2.0,arrowlength=1.4,arrowinset=0.4]{->}(3.1809375,3.0960157)(5.9809375,3.0960157)
\psline[linewidth=0.04cm,arrowsize=0.05291667cm 2.0,arrowlength=1.4,arrowinset=0.4]{->}(8.580937,3.0960157)(10.780937,3.0960157)
\psline[linewidth=0.04cm,arrowsize=0.05291667cm 2.0,arrowlength=1.4,arrowinset=0.4]{->}(8.580937,0.29601562)(10.780937,0.29601562)
\usefont{T1}{ptm}{m}{it}
\rput(1.9014063,-1.0939844){$A \otimes B$}
\usefont{T1}{ptm}{m}{n}
\rput(5.4023438,-1.0939844){$B \otimes A$}
\usefont{T1}{ptm}{m}{it}
\rput(5.4014063,-3.6939843){$A \otimes B$}
\usefont{T1}{ptm}{m}{it}
\rput(12.151406,-1.0939844){$I \otimes A$}
\usefont{T1}{ptm}{m}{it}
\rput(9.051406,-1.0939844){$A \otimes I$}
\usefont{T1}{ptm}{m}{it}
\rput(9.061406,-3.6939843){$A$}
\psline[linewidth=0.04cm,arrowsize=0.05291667cm 2.0,arrowlength=1.4,arrowinset=0.4]{->}(9.080937,-1.4039844)(9.080937,-3.4039843)
\psline[linewidth=0.04cm,arrowsize=0.05291667cm 2.0,arrowlength=1.4,arrowinset=0.4]{->}(9.680938,-1.1039844)(11.580937,-1.1039844)
\psline[linewidth=0.04cm,arrowsize=0.05291667cm 2.0,arrowlength=1.4,arrowinset=0.4]{->}(11.980938,-1.4039844)(9.280937,-3.4039843)
\usefont{T1}{ptm}{m}{n}
\rput(10.602344,-0.8939844){$\sigma_{A,I}$}
\usefont{T1}{ptm}{m}{n}
\rput{36.866886}(0.5998461,-7.0270395){\rput(10.822344,-2.5939844){$\lambda_A$}}
\usefont{T1}{ptm}{m}{n}
\rput{90.46118}(6.6191354,-11.193414){\rput(8.842343,-2.2939844){$\rho_A$}}
\psline[linewidth=0.04cm,arrowsize=0.05291667cm 2.0,arrowlength=1.4,arrowinset=0.4]{->}(2.5809374,-1.1039844)(4.7809377,-1.1039844)
\psline[linewidth=0.04cm,arrowsize=0.05291667cm 2.0,arrowlength=1.4,arrowinset=0.4]{->}(5.3809376,-1.4039844)(5.3809376,-3.4039843)
\psline[linewidth=0.04cm,arrowsize=0.05291667cm 2.0,arrowlength=1.4,arrowinset=0.4]{->}(1.9809375,-1.4039844)(5.1809373,-3.4039843)
\usefont{T1}{ptm}{m}{n}
\rput(3.6623437,-0.8939844){$\sigma_{A,B}$}
\usefont{T1}{ptm}{m}{n}
\rput{-89.14804}(7.911822,3.2030058){\rput(5.5623436,-2.3939843){$\sigma_{B,A}$}}
\usefont{T1}{ptm}{m}{n}
\rput{-32.523285}(1.9308116,1.3919636){\rput(3.3323438,-2.5939844){$id_{A,B}$}}
\end{pspicture}
}
\]
\end{definition}

\begin{definition}[Symmetric monoidal closed category]
A symmetric monoidal \textit{closed} category is a symmetric monoidal category where, for any two objects $A$ and $B$, there is an \textit{exponential object}\footnote[1]{Note that these are different from the exponential connectives of linear logic.} $A \multimap B$, together with an evaluation morphism $ev_{A,B} : (A \multimap B) \otimes A \longrightarrow B$. In addition to that, on any arrow of the form $f : C \otimes A \longrightarrow B$, a process called \textit{Currying} yields a unique morphism $\Lambda(f) : C \longrightarrow (A \multimap B)$ such that:
\[ ev_{A,B} \circ \left(\Lambda(f) \otimes id_A\right) = f \]
\[
% Generated with LaTeXDraw 2.0.2
% Thu Jan 15 21:55:38 GMT 2009
% \usepackage[usenames,dvipsnames]{pstricks}
% \usepackage{epsfig}
% \usepackage{pst-grad} % For gradients
% \usepackage{pst-plot} % For axes
\scalebox{1} % Change this value to rescale the drawing.
{
\begin{pspicture}(0,-1.9947656)(5.5028124,1.9947656)
\usefont{T1}{ptm}{m}{n}
\rput(2.0023437,1.5330468){$C \otimes A$}
\usefont{T1}{ptm}{m}{n}
\rput(2.0623438,-1.4669532){$(A \multimap B) \otimes A$}
\usefont{T1}{ptm}{m}{n}
\rput(5.162344,-1.4669532){$B$}
\psline[linewidth=0.04cm,linestyle=dashed,dash=0.16cm 0.16cm,arrowsize=0.05291667cm 2.0,arrowlength=1.4,arrowinset=0.4]{<-}(1.9809375,-1.2769531)(1.9809375,1.3230469)
\psline[linewidth=0.04cm,arrowsize=0.05291667cm 2.0,arrowlength=1.4,arrowinset=0.4]{->}(3.3809376,-1.4769531)(4.9809375,-1.4769531)
\psline[linewidth=0.04cm,arrowsize=0.05291667cm 2.0,arrowlength=1.4,arrowinset=0.4]{->}(2.1809375,1.3230469)(4.9809375,-1.2769531)
\usefont{T1}{ptm}{m}{n}
\rput(3.6023438,0.33304688){$f$}
\usefont{T1}{ptm}{m}{n}
\rput{-270.0}(1.7547656,-1.7280469){\rput(1.7223438,0.033046875){$\Lambda(f) \otimes id_A$}}
\usefont{T1}{ptm}{m}{n}
\rput(4.1723437,-1.7669531){$ev_{A,B}$}
\end{pspicture}
}
\]
\end{definition}

\begin{definition}[Compact closed category]
A \textit{compact closed category} is a symmetric monoidal category where, for every object $A$, there is a dual object $A^*$ along with two morphisms $\eta_A : I \longrightarrow A^* \otimes A$ and $\varepsilon_A : A \otimes A^* \longrightarrow I$ such that:
\begin{center}
$(\varepsilon_A \otimes id_A) \circ (id_A \otimes \eta_A) = id_A$ \hskip 0.5cm and \hskip 0.5cm $(id_{A^*} \otimes \eta_A) \circ (\varepsilon_A \otimes id_{A^*}) = id_{A^*}$
\end{center}
The dual object is unique up to canonical isomorphism.
\end{definition}

\begin{corollary}[Closure of Compact Closed Categories]
Every compact closed category is closed.
\end{corollary}
\begin{proof}
All of the exponential structure can be recreated by setting $A \multimap B = A^* \otimes B$. The evaluation function can be simulated by:
\[ ev_{A,B} = \rho_A \circ (id_B \otimes \varepsilon_A) \circ (id_B \otimes \sigma) \circ \alpha^{-1} \circ (\sigma \otimes id_A) \]
\end{proof}

\begin{definition}[Dagger compact category]
A \textit{$\dag$-compact category} is a compact closed category that is equipped with an involutive, contravariant, identity-on-objects endofunctor. That functor, called dagger, reverses all arrows, leaves objects unchanged, and preserves the tensor structure. For any $f : A \rightarrow B$, it will be the case that $f^\dag : B \rightarrow A$ and $f^{\dag\dag} = f$. Moreover, for any object $A$ in our category, it must be the case that $\sigma_{A, A^*} \circ \varepsilon_A^\dag = \eta_A$
\end{definition}

\chapter{Frobenius algebras}
\label{Chapter:Frobenius algebras}
This chapter presents a quick overview of \textit{Frobenius algebras}, how they generalise over a category's monoidal structure, as well as some related algebraic properties that will be used extensively in the parts of this dissertation that deal with classical structures. We will provide diagrammatic representations to accompany many of the textual definitions for the required conditions and properties. It should be noted that, as a notational convention, all the diagrams should be read from bottom to top.

Frobenius algebras were studied since the 1930's but have become exceedingly popular, in the past three decades, in the mathematical representation of quantum physics \cite{JS91, KL01, Koc03}. More recent developments in quantum information theory \cite{CP07, CD08, CPP10, CD11} have used these algebras to axiomatise the notion of a classical basis in quantum computation.

We will start by presenting some of the most common definitions of a Frobenius algebra, adapted from \cite{Str04}:
\begin{definition}[Frobenius algebra]
A \text{Frobenius algebra} $A$ is a finite dimensional, unital and associative algebra over a field $k$, that is equipped with a nondegenerate bilinear pairing $\sigma : A \otimes A \rightarrow k$. The bilinear form must be such that the following condition holds $\sigma((ab) \otimes c) = \sigma(a \otimes (bc))$, for $a,b,c \in A$.

Alternatively, we could say that the algebra is Frobenius if it is equipped with a linear function $\varepsilon : A \rightarrow k$, such that:
\begin{center}
    $\varepsilon(ab) = 0$ for all $a \in A$ implies $b=0$.
\end{center}
\end{definition}

In order to generalise this to the categorical setting, we will look at monoidal categories, and will define what it means for an object in that category to have a Frobenius structure associated with it.
\begin{definition}[Frobenius structure]
We say that an object $A$ in a monoidal category $\mathcal{C}$ has a \textit{Frobenius structure} attached to it when it is equipped with four morphisms, $\mu : A \otimes A \rightarrow A$, $\eta : I \rightarrow A$, $\delta : A \rightarrow A \otimes A$ and $\varepsilon : A \rightarrow I$, such that the following conditions hold:
\begin{itemize}
  \item $(A,\mu,\eta)$ forms a monoid,
  \item $(A,\delta,\varepsilon)$ forms a comonoid,
  \item and $(id_A \otimes \mu) \circ (\delta \otimes id_A) = \delta \circ \mu = (\mu \otimes id_A) \circ (id_A \otimes \delta)$
\end{itemize}
\end{definition}

We denote the Frobenius structure as $(A,\mu,\eta,\delta,\varepsilon)$. In the case of dagger monoidal categories, the dagger functor can give us $\mu = \delta^\dag$ and $\eta = \varepsilon^\dag$, simplifying the Frobenius structure into $(A,\delta,\varepsilon)$. The conditions required by our definition are easier to visualise in diagrammatic form. In order to achieve this, we will represent our fundamental morphisms, $\delta$ and $\varepsilon$ as follows:
\begin{center}
$\delta : A \longrightarrow A \otimes A :: a_i \mapsto a_i \otimes a_i$
\hskip 2cm
$\epsilon : A \longrightarrow I :: a_i \mapsto 1$
\end{center}
\[
% Generated with LaTeXDraw 2.0.2
% Fri Jan 16 17:41:15 GMT 2009
% \usepackage[usenames,dvipsnames]{pstricks}
% \usepackage{epsfig}
% \usepackage{pst-grad} % For gradients
% \usepackage{pst-plot} % For axes
\scalebox{1} % Change this value to rescale the drawing.
{
\begin{pspicture}(0,-0.72)(6.52,0.72)
\psline[linewidth=0.04cm](6.22,-0.1)(6.22,-0.7)
\pscircle[linewidth=0.04,dimen=outer,fillstyle=solid,fillcolor=green](6.22,0.0){0.3}
\psline[linewidth=0.04cm](0.62,-0.1)(0.62,-0.7)
\psarc[linewidth=0.04](0.62,0.7){0.6}{-180.0}{0.0}
\pscircle[linewidth=0.04,dimen=outer,fillstyle=solid,fillcolor=green](0.62,0.0){0.3}
\end{pspicture}
}
\]

\noindent This allows us to represent the monoidal and comonoidal conditions diagrammatically as such:
\[
% Generated with LaTeXDraw 2.0.8
% Sat Aug 24 14:18:45 BST 2013
% \usepackage[usenames,dvipsnames]{pstricks}
% \usepackage{epsfig}
% \usepackage{pst-grad} % For gradients
% \usepackage{pst-plot} % For axes
\scalebox{1} % Change this value to rescale the drawing.
{
\begin{pspicture}(1,-3.1128125)(14.34,3.1528125)
\usefont{T1}{ptm}{m}{n}
\rput(2.3045313,1.1121875){\huge =}
\psline[linewidth=0.04cm](3.32,-3.0928125)(3.32,-0.7928125)
\psline[linewidth=0.04cm](0.62,1.2071875)(0.62,0.6071875)
\psarc[linewidth=0.04](0.62,-0.1928125){0.6}{0.0}{180.0}
\pscircle[linewidth=0.04,dimen=outer,fillstyle=solid,fillcolor=green](0.62,0.5071875){0.3}
\psline[linewidth=0.04cm](1.22,2.5071876)(1.22,1.9071875)
\psarc[linewidth=0.04](1.22,1.1071875){0.6}{0.0}{180.0}
\pscircle[linewidth=0.04,dimen=outer,fillstyle=solid,fillcolor=green](1.22,1.8071876){0.3}
\psline[linewidth=0.04cm](1.82,1.2071875)(1.82,-0.1928125)
\psline[linewidth=0.04cm](3.32,2.5071876)(3.32,-0.1928125)
\psarc[linewidth=0.04](3.32,0.4071875){0.6}{0.0}{180.0}
\pscircle[linewidth=0.04,dimen=outer,fillstyle=solid,fillcolor=green](3.32,1.1071875){0.3}
\psline[linewidth=0.04cm](2.72,0.5071875)(2.72,-0.1928125)
\psline[linewidth=0.04cm](3.92,0.5071875)(3.92,-0.1928125)
\psline[linewidth=0.04cm](6.02,1.2071875)(6.02,0.6071875)
\psarc[linewidth=0.04](6.02,-0.1928125){0.6}{0.0}{180.0}
\pscircle[linewidth=0.04,dimen=outer,fillstyle=solid,fillcolor=green](6.02,0.5071875){0.3}
\psline[linewidth=0.04cm](5.42,2.5071876)(5.42,1.9071875)
\psarc[linewidth=0.04](5.42,1.1071875){0.6}{0.0}{180.0}
\pscircle[linewidth=0.04,dimen=outer,fillstyle=solid,fillcolor=green](5.42,1.8071876){0.3}
\psline[linewidth=0.04cm](4.82,1.2071875)(4.82,-0.1928125)
\usefont{T1}{ptm}{m}{n}
\rput(4.4045315,1.1121875){\huge =}
\psline[linewidth=0.04cm](1.72,-0.7928125)(1.72,-1.3928125)
\psarc[linewidth=0.04](1.72,-2.1928124){0.6}{0.0}{180.0}
\pscircle[linewidth=0.04,dimen=outer,fillstyle=solid,fillcolor=green](1.72,-1.4928125){0.3}
\psline[linewidth=0.04cm](4.92,-0.7928125)(4.92,-1.3928125)
\psarc[linewidth=0.04](4.92,-2.1928124){0.6}{0.0}{180.0}
\pscircle[linewidth=0.04,dimen=outer,fillstyle=solid,fillcolor=green](4.92,-1.4928125){0.3}
\psline[linewidth=0.04cm](5.52,-2.0928125)(5.52,-2.6928124)
\pscircle[linewidth=0.04,dimen=outer,fillstyle=solid,fillcolor=green](5.52,-2.7928126){0.3}
\psline[linewidth=0.04cm](4.32,-3.0928125)(4.32,-2.0928125)
\psline[linewidth=0.04cm](1.12,-2.0928125)(1.12,-2.6928124)
\pscircle[linewidth=0.04,dimen=outer,fillstyle=solid,fillcolor=green](1.12,-2.7928126){0.3}
\psline[linewidth=0.04cm](2.32,-3.0928125)(2.32,-2.0928125)
\usefont{T1}{ptm}{m}{n}
\rput(2.8045313,-2.0878124){\huge =}
\usefont{T1}{ptm}{m}{n}
\rput(3.8045313,-2.0878124){\huge =}
\usefont{T1}{ptm}{m}{n}
\rput(10.004531,1.2121875){\huge =}
\psline[linewidth=0.04cm](8.32,1.7071875)(8.32,1.1071875)
\psarc[linewidth=0.04](8.32,2.5071876){0.6}{180.0}{0.0}
\pscircle[linewidth=0.04,dimen=outer,fillstyle=solid,fillcolor=green](8.32,1.8071876){0.3}
\psline[linewidth=0.04cm](9.52,2.5071876)(9.52,1.1071875)
\psline[linewidth=0.04cm](8.92,0.4071875)(8.92,-0.1928125)
\psarc[linewidth=0.04](8.92,1.2071875){0.6}{-180.0}{0.0}
\pscircle[linewidth=0.04,dimen=outer,fillstyle=solid,fillcolor=green](8.92,0.5071875){0.3}
\psline[linewidth=0.04cm](10.42,2.5071876)(10.42,1.8071876)
\psline[linewidth=0.04cm](11.62,2.5071876)(11.62,1.8071876)
\psline[linewidth=0.04cm](12.52,2.5071876)(12.52,1.1071875)
\usefont{T1}{ptm}{m}{n}
\rput(12.104531,1.2121875){\huge =}
\psline[linewidth=0.04cm](11.02,2.5071876)(11.02,-0.1928125)
\psarc[linewidth=0.04](11.02,1.9071875){0.6}{-180.0}{0.0}
\pscircle[linewidth=0.04,dimen=outer,fillstyle=solid,fillcolor=green](11.02,1.2071875){0.3}
\psline[linewidth=0.04cm](13.72,1.7071875)(13.72,1.1071875)
\psarc[linewidth=0.04](13.72,2.5071876){0.6}{180.0}{0.0}
\pscircle[linewidth=0.04,dimen=outer,fillstyle=solid,fillcolor=green](13.72,1.8071876){0.3}
\psline[linewidth=0.04cm](13.12,0.4071875)(13.12,-0.1928125)
\psarc[linewidth=0.04](13.12,1.2071875){0.6}{-180.0}{0.0}
\pscircle[linewidth=0.04,dimen=outer,fillstyle=solid,fillcolor=green](13.12,0.5071875){0.3}
\psline[linewidth=0.04cm](11.02,-3.0928125)(11.02,-0.7928125)
\psline[linewidth=0.04cm](13.22,-1.1928124)(13.22,-1.7928125)
\psline[linewidth=0.04cm](12.02,-1.7928125)(12.02,-0.7928125)
\psline[linewidth=0.04cm](8.82,-1.1928124)(8.82,-1.7928125)
\psline[linewidth=0.04cm](9.42,-2.4928124)(9.42,-3.0928125)
\psarc[linewidth=0.04](9.42,-1.6928124){0.6}{-180.0}{0.0}
\pscircle[linewidth=0.04,dimen=outer,fillstyle=solid,fillcolor=green](9.42,-2.3928125){0.3}
\psline[linewidth=0.04cm](10.02,-1.7928125)(10.02,-0.7928125)
\pscircle[linewidth=0.04,dimen=outer,fillstyle=solid,fillcolor=green](8.82,-1.0928125){0.3}
\usefont{T1}{ptm}{m}{n}
\rput(10.504531,-1.8878125){\huge =}
\usefont{T1}{ptm}{m}{n}
\rput(11.504531,-1.8878125){\huge =}
\psline[linewidth=0.04cm](12.62,-2.4928124)(12.62,-3.0928125)
\psarc[linewidth=0.04](12.62,-1.6928124){0.6}{-180.0}{0.0}
\pscircle[linewidth=0.04,dimen=outer,fillstyle=solid,fillcolor=green](12.62,-2.3928125){0.3}
\pscircle[linewidth=0.04,dimen=outer,fillstyle=solid,fillcolor=green](13.22,-1.0928125){0.3}
\usefont{T1}{ptm}{m}{n}
\rput(3.2442188,2.9421875){\large Monoid conditions}
\usefont{T1}{ptm}{m}{n}
\rput(10.8698435,2.9421875){\large Comonoid conditions}
\end{pspicture}
}
\]

\noindent The last condition in the definition, also known as the \textit{Frobenius condition}, thus becomes particularly easy to visualise:
\[
% Generated with LaTeXDraw 2.0.2
% Fri Jan 16 17:33:34 GMT 2009
% \usepackage[usenames,dvipsnames]{pstricks}
% \usepackage{epsfig}
% \usepackage{pst-grad} % For gradients
% \usepackage{pst-plot} % For axes
\scalebox{1} % Change this value to rescale the drawing.
{
\begin{pspicture}(4.5,0)(11.225312,3.07)
\psline[linewidth=0.04cm](3.6053126,0.95)(3.6053126,0.35)
\psarc[linewidth=0.04](3.6053126,1.75){0.6}{-180.0}{0.0}
\pscircle[linewidth=0.04,dimen=outer,fillstyle=solid,fillcolor=green](3.6053126,1.05){0.3}
\psline[linewidth=0.04cm](4.8053126,3.05)(4.8053126,2.45)
\psarc[linewidth=0.04](4.8053126,1.65){0.6}{-0.0}{180.0}
\pscircle[linewidth=0.04,dimen=outer,fillstyle=solid,fillcolor=green](4.8053126,2.35){0.3}
\psline[linewidth=0.04cm](3.0053124,1.63)(3.0053124,3.03)
\psline[linewidth=0.04cm](5.4053125,1.73)(5.4053125,0.33)
\psline[linewidth=0.04cm](10.605312,0.95)(10.605312,0.35)
\psarc[linewidth=0.04](10.605312,1.75){0.6}{-180.0}{0.0}
\pscircle[linewidth=0.04,dimen=outer,fillstyle=solid,fillcolor=green](10.605312,1.05){0.3}
\psline[linewidth=0.04cm](9.405313,3.05)(9.405313,2.45)
\psarc[linewidth=0.04](9.405313,1.65){0.6}{-0.0}{180.0}
\pscircle[linewidth=0.04,dimen=outer,fillstyle=solid,fillcolor=green](9.405313,2.35){0.3}
\psline[linewidth=0.04cm](8.805312,1.73)(8.805312,0.33)
\psline[linewidth=0.04cm](11.205313,1.63)(11.205313,3.03)
\psline[linewidth=0.04cm](7.1053123,2.25)(7.1053123,1.65)
\psarc[linewidth=0.04](7.1053123,3.05){0.6}{-180.0}{0.0}
\pscircle[linewidth=0.04,dimen=outer,fillstyle=solid,fillcolor=green](7.1053123,2.35){0.3}
\psline[linewidth=0.04cm](7.1053123,1.75)(7.1053123,1.15)
\psarc[linewidth=0.04](7.1053123,0.35){0.6}{-0.0}{180.0}
\pscircle[linewidth=0.04,dimen=outer,fillstyle=solid,fillcolor=green](7.1053123,1.05){0.3}
\usefont{T1}{ptm}{m}{n}
\rput(5.989844,1.535){\huge =}
\usefont{T1}{ptm}{m}{n}
\rput(8.189844,1.535){\huge =}
\end{pspicture}
}
\]

\noindent We will now define two important properties, \textit{symmetry} and \textit{isometry}, that will be used in later parts of our exposition. We will start with the definition of symmetry:
\begin{definition}[Symmetric algebra]
Consider the symmetry isomorphism $\sigma_{A,B} : A \otimes B \rightarrow B \otimes A$ that is part of the definition of a symmetric monoidal category. A Frobenius algebra over an object of such a category is \textit{symmetric} if its underlying monoid and comonoid are commutative and cocommutative, respectively. This means that $\delta$ must be $\delta = \sigma_{A,A} \circ \delta$ and $\mu$ must be $\mu = \mu \circ \sigma_{A,A}$.
\[
% Generated with LaTeXDraw 2.0.8
% Sat Aug 24 14:11:22 BST 2013
% \usepackage[usenames,dvipsnames]{pstricks}
% \usepackage{epsfig}
% \usepackage{pst-grad} % For gradients
% \usepackage{pst-plot} % For axes
\scalebox{1} % Change this value to rescale the drawing.
{
\begin{pspicture}(0,-1.32)(7.54,1.32)
\psline[linewidth=0.04cm](2.62,1.3)(2.62,0.7)
\psarc[linewidth=0.04](2.62,-0.1){0.6}{0.0}{180.0}
\pscircle[linewidth=0.04,dimen=outer,fillstyle=solid,fillcolor=green](2.62,0.6){0.3}
\psline[linewidth=0.04](2.02,0.0)(2.02,-0.4)(3.22,-1.0)(3.22,-1.3)
\psline[linewidth=0.04](3.22,0.0)(3.22,-0.4)(2.02,-1.0)(2.02,-1.3)
\psline[linewidth=0.04cm](0.62,1.3)(0.62,0.7)
\psarc[linewidth=0.04](0.62,-0.1){0.6}{0.0}{180.0}
\pscircle[linewidth=0.04,dimen=outer,fillstyle=solid,fillcolor=green](0.62,0.6){0.3}
\psline[linewidth=0.04cm](0.02,0.0)(0.02,-1.3)
\psline[linewidth=0.04cm](1.22,0.0)(1.22,-1.3)
\usefont{T1}{ptm}{m}{n}
\rput(1.6045313,0.005){\huge =}
\psline[linewidth=0.04cm](4.32,1.3)(4.32,0.0)
\psline[linewidth=0.04cm](5.52,1.3)(5.52,0.0)
\usefont{T1}{ptm}{m}{n}
\rput(5.9045315,-0.095){\huge =}
\psline[linewidth=0.04cm](4.92,-0.7)(4.92,-1.3)
\psarc[linewidth=0.04](4.92,0.1){0.6}{-180.0}{0.0}
\pscircle[linewidth=0.04,dimen=outer,fillstyle=solid,fillcolor=green](4.92,-0.6){0.3}
\psline[linewidth=0.04cm](6.92,-0.7)(6.92,-1.3)
\psarc[linewidth=0.04](6.92,0.1){0.6}{-180.0}{0.0}
\pscircle[linewidth=0.04,dimen=outer,fillstyle=solid,fillcolor=green](6.92,-0.6){0.3}
\psline[linewidth=0.04](6.32,1.3)(6.32,0.9)(7.52,0.3)(7.52,0.0)
\psline[linewidth=0.04](7.52,1.3)(7.52,0.9)(6.32,0.3)(6.32,0.0)
\end{pspicture}
}
\]
\end{definition}

\noindent The last property that we will define in this chapter is isometry:
\begin{definition}[Isometric or special algebra]
A Frobenius algebra $(A,\mu,\eta,\delta,\varepsilon)$ over an object $A$ is said to be \textit{isometric} or \textit{special} if $\mu \circ \delta = id_A$.
\[
% Generated with LaTeXDraw 2.0.8
% Sat Aug 24 14:14:20 BST 2013
% \usepackage[usenames,dvipsnames]{pstricks}
% \usepackage{epsfig}
% \usepackage{pst-grad} % For gradients
% \usepackage{pst-plot} % For axes
\scalebox{1} % Change this value to rescale the drawing.
{
\begin{pspicture}(0,-1.39)(2.44,1.37)
\psline[linewidth=0.04cm](0.62,-0.75)(0.62,-1.35)
\psarc[linewidth=0.04](0.62,0.05){0.6}{-180.0}{0.0}
\pscircle[linewidth=0.04,dimen=outer,fillstyle=solid,fillcolor=green](0.62,-0.65){0.3}
\psline[linewidth=0.04cm](0.62,1.35)(0.62,0.75)
\psarc[linewidth=0.04](0.62,-0.05){0.6}{0.0}{180.0}
\pscircle[linewidth=0.04,dimen=outer,fillstyle=solid,fillcolor=green](0.62,0.65){0.3}
\psline[linewidth=0.04cm](2.42,-1.37)(2.42,1.33)
\usefont{T1}{ptm}{m}{n}
\rput(1.8045312,-0.065){\huge =}
\end{pspicture}
}
\]
\end{definition}

\chapter{Categorical model}
\label{Chapter:Categorical model}
This section presents the categories used to model quantum computation and classical operations in the rest of the dissertation. It begins by explaining how the Hilbert space formalism can be recast into the language of $\dag$-compact categories \cite{AC04}. One of the biggest practical advantages of monoidal categories is that, on many occasions, they ``formally justify their absence'' \cite{Coe06}, meaning that they can be represented using a graphical calculus that greatly simplifies categorical reasoning. This section also demonstrates how every element of the initial quantum structure can be represented graphically in what resembles a two dimensional Dirac notation. Following our notational convention for the direction of the compositional flow of time, all the diagrams should be read from bottom to top. The next part of this section deals with classical operations, which are modelled in terms of internal spider monoids, as well as with the computational interplay inherent in introducing complementarity.

\section{Categorical quantum computation}
\label{Section:Categorical quantum computation}
The category we will be using to model quantum computation is called \textit{FDHilb} and is the category of finite dimensional complex Hilbert spaces. Its objects are finite dimensional Hilbert spaces and its arrows are linear maps. Monoidal multiplication is represented by the Kronecker tensor product, while the monoidal unit object corresponds to the set of complex numbers $I = \mathbb{C}$. Associativity of the tensor and tensor identities are up to equality, so $\alpha_{A,B,C}$, $\lambda_A$ and $\rho_A$ are reduced to identity arrows.

The adjoint is modelled using the dagger functor. In terms of the picture calculus, the dagger denotes flipping a picture upside down, while the arrows continue pointing the same way they were before (i.e. upwards).

For any three arrows $f:A \rightarrow B$, $g:C \rightarrow D$ and $h:E \rightarrow F$, associativity allows us to write their tensor product as:
\[
% Generated with LaTeXDraw 2.0.2
% Wed Jan 14 09:48:19 GMT 2009
% \usepackage[usenames,dvipsnames]{pstricks}
% \usepackage{epsfig}
% \usepackage{pst-grad} % For gradients
% \usepackage{pst-plot} % For axes
\scalebox{1} % Change this value to rescale the drawing.
{
\begin{pspicture}(0,-1.0376563)(3.2515626,1.0376563)
\definecolor{color868b}{rgb}{0.8,1.0,1.0}
\definecolor{color886b}{rgb}{1.0,1.0,0.8}
\definecolor{color898b}{rgb}{0.8,0.8,1.0}
\usefont{T1}{ptm}{m}{n}
\rput(0.7071875,0.71765625){B}
\usefont{T1}{ptm}{m}{n}
\rput(0.7265625,-0.88234377){A}
\psline[linewidth=0.06cm,arrowsize=0.05291667cm 2.0,arrowlength=1.4,arrowinset=0.4]{->}(0.4,-0.9923437)(0.4,1.0076562)
\psframe[linewidth=0.04,dimen=outer,fillstyle=solid,fillcolor=color868b](0.76,0.40765625)(0.0,-0.35234374)
\usefont{T1}{ptm}{m}{n}
\rput(0.36734375,0.00265625){\large $f$}
\psline[linewidth=0.06cm,arrowsize=0.05291667cm 2.0,arrowlength=1.4,arrowinset=0.4]{->}(1.6,-0.9923437)(1.6,1.0076562)
\usefont{T1}{ptm}{m}{n}
\rput(1.9071875,-0.88234377){C}
\usefont{T1}{ptm}{m}{n}
\rput(1.92,0.71765625){D}
\psframe[linewidth=0.04,dimen=outer,fillstyle=solid,fillcolor=color886b](2.0,0.40765625)(1.2,-0.39234376)
\usefont{T1}{ptm}{m}{n}
\rput(1.5973438,0.02265625){\large $g$}
\psline[linewidth=0.06cm,arrowsize=0.05291667cm 2.0,arrowlength=1.4,arrowinset=0.4]{->}(2.8,-0.9923437)(2.8,1.0076562)
\usefont{T1}{ptm}{m}{n}
\rput(3.1020312,-0.88234377){E}
\usefont{T1}{ptm}{m}{n}
\rput(3.0896876,0.71765625){F}
\psframe[linewidth=0.04,dimen=outer,fillstyle=solid,fillcolor=color898b](3.2,0.40765625)(2.4,-0.39234376)
\usefont{T1}{ptm}{m}{n}
\rput(2.7773438,0.04265625){\large $h$}
\end{pspicture}
}
\]

By bifunctoriality of the tensor, we know that it preserves composition. For any arrows of the form $f:A \rightarrow B$, $g:B \rightarrow C$, $h:D \rightarrow E$ and $k:E \rightarrow F$, once we add composition to our diagrams, the following property should become more evident:
\[ (g \circ f) \otimes (k \circ h) = (g \otimes k) \circ (f \otimes h) \]

\[
% Generated with LaTeXDraw 2.0.2
% Wed Jan 14 10:17:03 GMT 2009
% \usepackage[usenames,dvipsnames]{pstricks}
% \usepackage{epsfig}
% \usepackage{pst-grad} % For gradients
% \usepackage{pst-plot} % For axes
\scalebox{1} % Change this value to rescale the drawing.
{
\begin{pspicture}(0,-1.73)(5.5234375,1.73)
\definecolor{color1395b}{rgb}{1.0,1.0,0.8}
\definecolor{color1398b}{rgb}{0.8,1.0,1.0}
\usefont{T1}{ptm}{m}{n}
\rput(3.56375,-1.39){A}
\usefont{T1}{ptm}{m}{n}
\rput(3.544375,0.01){B}
\usefont{T1}{ptm}{m}{n}
\rput(3.544375,1.41){C}
\psline[linewidth=0.06cm,arrowsize=0.05291667cm 2.0,arrowlength=1.4,arrowinset=0.4]{->}(3.8371875,-0.1)(3.8371875,1.7)
\psline[linewidth=0.06cm,arrowsize=0.05291667cm 2.0,arrowlength=1.4,arrowinset=0.4]{->}(3.8371875,-1.7)(3.8371875,0.1)
\psframe[linewidth=0.04,dimen=outer,fillstyle=solid,fillcolor=color1395b](4.2371874,-0.3)(3.4371874,-1.1)
\usefont{T1}{ptm}{m}{n}
\rput(3.8045313,-0.705){\large $f$}
\psframe[linewidth=0.04,dimen=outer,fillstyle=solid,fillcolor=color1398b](4.2371874,1.1)(3.4371874,0.3)
\usefont{T1}{ptm}{m}{n}
\rput(3.8345313,0.715){\large $g$}
\usefont{T1}{ptm}{m}{n}
\rput(5.3571873,-1.39){D}
\usefont{T1}{ptm}{m}{n}
\rput(5.326875,1.41){F}
\usefont{T1}{ptm}{m}{n}
\rput(5.3392186,0.01){E}
\psline[linewidth=0.06cm,arrowsize=0.05291667cm 2.0,arrowlength=1.4,arrowinset=0.4]{->}(5.0371876,-0.1)(5.0371876,1.7)
\psline[linewidth=0.06cm,arrowsize=0.05291667cm 2.0,arrowlength=1.4,arrowinset=0.4]{->}(5.0371876,-1.7)(5.0371876,0.1)
\psframe[linewidth=0.04,dimen=outer,fillstyle=solid,fillcolor=color1395b](5.4371877,-0.3)(4.6371875,-1.1)
\usefont{T1}{ptm}{m}{n}
\rput(5.014531,-0.665){\large $h$}
\psframe[linewidth=0.04,dimen=outer,fillstyle=solid,fillcolor=color1398b](5.4371877,1.1)(4.6371875,0.3)
\usefont{T1}{ptm}{m}{n}
\rput(5.0245314,0.715){\large $k$}
\usefont{T1}{ptm}{m}{n}
\rput(2.6690626,-0.04){\Huge =}
\usefont{T1}{ptm}{m}{n}
\rput(0.16375,-1.39){A}
\usefont{T1}{ptm}{m}{n}
\rput(0.144375,0.01){B}
\usefont{T1}{ptm}{m}{n}
\rput(0.144375,1.41){C}
\psline[linewidth=0.06cm,arrowsize=0.05291667cm 2.0,arrowlength=1.4,arrowinset=0.4]{->}(0.4371875,-0.1)(0.4371875,1.7)
\psline[linewidth=0.06cm,arrowsize=0.05291667cm 2.0,arrowlength=1.4,arrowinset=0.4]{->}(0.4371875,-1.7)(0.4371875,0.1)
\psframe[linewidth=0.04,dimen=outer,fillstyle=solid,fillcolor=color1398b](0.7971875,-0.3)(0.0371875,-1.06)
\usefont{T1}{ptm}{m}{n}
\rput(0.40453124,-0.705){\large $f$}
\psframe[linewidth=0.04,dimen=outer,fillstyle=solid,fillcolor=color1398b](0.8371875,1.1)(0.0371875,0.3)
\usefont{T1}{ptm}{m}{n}
\rput(0.43453124,0.715){\large $g$}
\usefont{T1}{ptm}{m}{n}
\rput(1.9571875,-1.39){D}
\usefont{T1}{ptm}{m}{n}
\rput(1.926875,1.41){F}
\usefont{T1}{ptm}{m}{n}
\rput(1.9392188,0.01){E}
\psline[linewidth=0.06cm,arrowsize=0.05291667cm 2.0,arrowlength=1.4,arrowinset=0.4]{->}(1.6371875,-0.1)(1.6371875,1.7)
\psline[linewidth=0.06cm,arrowsize=0.05291667cm 2.0,arrowlength=1.4,arrowinset=0.4]{->}(1.6371875,-1.7)(1.6371875,0.1)
\psframe[linewidth=0.04,dimen=outer,fillstyle=solid,fillcolor=color1395b](2.0371876,-0.3)(1.2371875,-1.1)
\usefont{T1}{ptm}{m}{n}
\rput(1.6145313,-0.665){\large $h$}
\psframe[linewidth=0.04,dimen=outer,fillstyle=solid,fillcolor=color1395b](2.0371876,1.1)(1.2371875,0.3)
\usefont{T1}{ptm}{m}{n}
\rput(1.6245313,0.715){\large $k$}
\end{pspicture}
}
\]

Moreover, it should also be evident that the tensor preserves identities:
\[ id_{A \otimes B} = id_A \otimes id_B \]

\[
% Generated with LaTeXDraw 2.0.2
% Wed Jan 14 10:25:13 GMT 2009
% \usepackage[usenames,dvipsnames]{pstricks}
% \usepackage{epsfig}
% \usepackage{pst-grad} % For gradients
% \usepackage{pst-plot} % For axes
\scalebox{1} % Change this value to rescale the drawing.
{
\begin{pspicture}(0,-0.83)(1.4575,0.83)
\psline[linewidth=0.06cm,arrowsize=0.05291667cm 2.0,arrowlength=1.4,arrowinset=0.4]{->}(0.3971875,-0.8)(0.3971875,0.8)
\usefont{T1}{ptm}{m}{n}
\rput(0.12375,-0.09){A}
\psline[linewidth=0.06cm,arrowsize=0.05291667cm 2.0,arrowlength=1.4,arrowinset=0.4]{->}(0.9971875,-0.8)(0.9971875,0.8)
\usefont{T1}{ptm}{m}{n}
\rput(1.304375,-0.09){B}
\end{pspicture}
}
\]

The proper graphical representation for the unit $I$ is "no line", while arrows containing $I$ as their domain or codomain are represented as follows:

\begin{center}
$\psi : I \rightarrow A$ \hskip 1.2cm $\phi^\dag : A \rightarrow I$ \hskip 1.2cm $\phi^\dag \circ \psi : I \rightarrow I$
\end{center}
\[
% Generated with LaTeXDraw 2.0.2
% Sun Jan 18 19:48:30 GMT 2009
% \usepackage[usenames,dvipsnames]{pstricks}
% \usepackage{epsfig}
% \usepackage{pst-grad} % For gradients
% \usepackage{pst-plot} % For axes
\scalebox{1} % Change this value to rescale the drawing.
{
\begin{pspicture}(0,-0.945)(7.3098593,0.915)
\usefont{T1}{ptm}{m}{n}
\rput(0.8145469,0.495){A}
\psline[linewidth=0.06cm,arrowsize=0.05291667cm 2.0,arrowlength=1.4,arrowinset=0.4]{->}(0.4879844,-0.315)(0.4879844,0.885)
\rput{-180.0}(0.9759688,-1.03){\pstriangle[linewidth=0.04,dimen=outer,fillstyle=solid,fillcolor=green](0.4879844,-0.915)(1.0,0.8)}
\usefont{T1}{ptm}{m}{n}
\rput(0.47939062,-0.405){$\psi$}
\usefont{T1}{ptm}{m}{n}
\rput(3.8145468,-0.505){A}
\psline[linewidth=0.06cm,arrowsize=0.05291667cm 2.0,arrowlength=1.4,arrowinset=0.4]{>-}(3.4879844,-0.915)(3.4879844,0.485)
\pstriangle[linewidth=0.04,dimen=outer,fillstyle=solid,fillcolor=red](3.4879844,0.085)(1.0,0.8)
\usefont{T1}{ptm}{m}{n}
\rput(3.4393907,0.375){$\phi^\dag$}
\psline[linewidth=0.06cm](6.487984,-0.315)(6.487984,0.485)
\pstriangle[linewidth=0.04,dimen=outer,fillstyle=solid,fillcolor=red](6.4879847,0.085)(1.0,0.8)
\rput{-180.0}(12.975968,-1.03){\pstriangle[linewidth=0.04,dimen=outer,fillstyle=solid,fillcolor=green](6.4879847,-0.915)(1.0,0.8)}
\usefont{T1}{ptm}{m}{n}
\rput(6.4793906,-0.405){$\psi$}
\usefont{T1}{ptm}{m}{n}
\rput(6.4393907,0.375){$\phi^\dag$}
\end{pspicture}
}
\]

A special case of arrows called scalars consists of all arrows of the form $c : I \rightarrow I$. In these specific cases, our categorical structure collapses to the point where tensor is equal to composition. In other words $c_1 \otimes c_2 = c_1 \circ c_2 = c_2 \otimes c_1$. Scalars can be moved freely around in the category's graphical representation.

Symmetry corresponds to a well known quantum operation called \textit{swap}; it is graphically represented by a pair of crossing lines. The following properties are more easily understood graphically:
\[
% Generated with LaTeXDraw 2.0.2
% Fri Jan 16 00:48:32 GMT 2009
% \usepackage[usenames,dvipsnames]{pstricks}
% \usepackage{epsfig}
% \usepackage{pst-grad} % For gradients
% \usepackage{pst-plot} % For axes
\scalebox{1} % Change this value to rescale the drawing.
{
\begin{pspicture}(0,-2.4453125)(13.2925,2.43)
\definecolor{color121b}{rgb}{0.8,1.0,1.0}
\definecolor{color131b}{rgb}{1.0,1.0,0.8}
\psline[linewidth=0.06cm](0.3965625,-2.4)(0.3965625,-1.4)
\psline[linewidth=0.06cm](0.3965625,-1.4)(1.3965625,-0.4)
\psline[linewidth=0.06cm,arrowsize=0.05291667cm 2.0,arrowlength=1.4,arrowinset=0.4]{->}(1.3965625,-0.4)(1.3965625,0.2)
\psline[linewidth=0.06cm](1.3965625,-2.4)(1.3965625,-1.4)
\psline[linewidth=0.06cm](1.3965625,-1.4)(0.3965625,-0.4)
\usefont{T1}{ptm}{m}{n}
\rput(0.15125,-1.865){\large A}
\usefont{T1}{ptm}{m}{n}
\rput(1.7276562,-1.865){\large B}
\psline[linewidth=0.06cm](0.3965625,0.4)(1.3965625,1.4)
\psline[linewidth=0.06cm,arrowsize=0.05291667cm 2.0,arrowlength=1.4,arrowinset=0.4]{->}(1.3965625,1.4)(1.3965625,2.4)
\psline[linewidth=0.06cm](1.3965625,-0.2)(1.3965625,0.4)
\psline[linewidth=0.06cm](1.3965625,0.4)(0.3965625,1.4)
\psline[linewidth=0.06cm,arrowsize=0.05291667cm 2.0,arrowlength=1.4,arrowinset=0.4]{->}(0.3965625,1.4)(0.3965625,2.4)
\usefont{T1}{ptm}{m}{n}
\rput(2.8284376,-0.14){\Huge =}
\psline[linewidth=0.06cm,arrowsize=0.05291667cm 2.0,arrowlength=1.4,arrowinset=0.4]{->}(4.1965623,-2.4)(4.1965623,2.4)
\psline[linewidth=0.06cm,arrowsize=0.05291667cm 2.0,arrowlength=1.4,arrowinset=0.4]{->}(5.1965623,-2.4)(5.1965623,2.4)
\usefont{T1}{ptm}{m}{n}
\rput(0.12765625,-0.065){\large B}
\usefont{T1}{ptm}{m}{n}
\rput(1.7276562,1.935){\large B}
\usefont{T1}{ptm}{m}{n}
\rput(1.75125,-0.065){\large A}
\usefont{T1}{ptm}{m}{n}
\rput(0.15125,1.935){\large A}
\usefont{T1}{ptm}{m}{n}
\rput(3.75125,1.935){\large A}
\usefont{T1}{ptm}{m}{n}
\rput(5.527656,1.935){\large B}
\usefont{T1}{ptm}{m}{n}
\rput(7.90375,-0.69){B}
\usefont{T1}{ptm}{m}{n}
\rput(7.923125,-2.29){A}
\psline[linewidth=0.06cm,arrowsize=0.05291667cm 2.0,arrowlength=1.4,arrowinset=0.4]{->}(8.196563,-2.4)(8.196563,-0.4)
\psframe[linewidth=0.04,dimen=outer,fillstyle=solid,fillcolor=color121b](8.556562,-1.0)(7.7965627,-1.76)
\usefont{T1}{ptm}{m}{n}
\rput(8.163906,-1.405){\large $f$}
\psline[linewidth=0.06cm,arrowsize=0.05291667cm 2.0,arrowlength=1.4,arrowinset=0.4]{->}(9.196563,-2.4)(9.196563,-0.4)
\usefont{T1}{ptm}{m}{n}
\rput(9.50375,-2.29){C}
\usefont{T1}{ptm}{m}{n}
\rput(9.516562,-0.69){D}
\psframe[linewidth=0.04,dimen=outer,fillstyle=solid,fillcolor=color131b](9.556562,-1.0)(8.796562,-1.76)
\usefont{T1}{ptm}{m}{n}
\rput(9.193906,-1.385){\large $g$}
\usefont{T1}{ptm}{m}{n}
\rput(5.527656,-1.865){\large B}
\usefont{T1}{ptm}{m}{n}
\rput(3.75125,-1.865){\large A}
\psline[linewidth=0.06cm](8.196563,-0.6)(8.196563,0.4)
\psline[linewidth=0.06cm](8.196563,0.4)(9.196563,1.4)
\psline[linewidth=0.06cm,arrowsize=0.05291667cm 2.0,arrowlength=1.4,arrowinset=0.4]{->}(9.196563,1.4)(9.196563,2.4)
\psline[linewidth=0.06cm](9.196563,-0.6)(9.196563,0.4)
\psline[linewidth=0.06cm](9.196563,0.4)(8.196563,1.4)
\psline[linewidth=0.06cm,arrowsize=0.05291667cm 2.0,arrowlength=1.4,arrowinset=0.4]{->}(8.196563,1.4)(8.196563,2.4)
\psline[linewidth=0.06cm](11.796562,-2.4)(11.796562,-1.4)
\psline[linewidth=0.06cm](11.796562,-1.4)(12.796562,-0.4)
\psline[linewidth=0.06cm,arrowsize=0.05291667cm 2.0,arrowlength=1.4,arrowinset=0.4]{->}(12.796562,-0.4)(12.796562,0.6)
\psline[linewidth=0.06cm](12.796562,-2.4)(12.796562,-1.4)
\psline[linewidth=0.06cm](12.796562,-1.4)(11.796562,-0.4)
\psline[linewidth=0.06cm,arrowsize=0.05291667cm 2.0,arrowlength=1.4,arrowinset=0.4]{->}(11.796562,-0.4)(11.796562,0.6)
\usefont{T1}{ptm}{m}{n}
\rput(13.10375,2.11){B}
\usefont{T1}{ptm}{m}{n}
\rput(13.123125,0.51){A}
\psline[linewidth=0.06cm,arrowsize=0.05291667cm 2.0,arrowlength=1.4,arrowinset=0.4]{->}(12.796562,0.4)(12.796562,2.4)
\psframe[linewidth=0.04,dimen=outer,fillstyle=solid,fillcolor=color121b](13.156563,1.8)(12.396563,1.04)
\usefont{T1}{ptm}{m}{n}
\rput(12.7639065,1.395){\large $f$}
\psline[linewidth=0.06cm,arrowsize=0.05291667cm 2.0,arrowlength=1.4,arrowinset=0.4]{->}(11.796562,0.4)(11.796562,2.4)
\usefont{T1}{ptm}{m}{n}
\rput(11.50375,0.51){C}
\usefont{T1}{ptm}{m}{n}
\rput(11.516562,2.11){D}
\psframe[linewidth=0.04,dimen=outer,fillstyle=solid,fillcolor=color131b](12.156563,1.8)(11.396563,1.04)
\usefont{T1}{ptm}{m}{n}
\rput(11.793906,1.415){\large $g$}
\psline[linewidth=0.06cm,arrowsize=0.05291667cm 2.0,arrowlength=1.4,arrowinset=0.4]{->}(0.3965625,-0.4)(0.3965625,0.2)
\psline[linewidth=0.06cm](0.3965625,-0.2)(0.3965625,0.4)
\usefont{T1}{ptm}{m}{n}
\rput(7.9165626,1.91){D}
\usefont{T1}{ptm}{m}{n}
\rput(9.50375,1.91){B}
\usefont{T1}{ptm}{m}{n}
\rput(13.10375,-2.09){C}
\usefont{T1}{ptm}{m}{n}
\rput(11.523125,-2.09){A}
\usefont{T1}{ptm}{m}{n}
\rput(10.428437,-0.14){\Huge =}
\end{pspicture}
}
\]
\begin{center}
$\sigma_{A,B} \circ \sigma_{A,B} = id_A \otimes id_B$ \hskip 3cm $\sigma_{B,D} \circ (f \otimes g) = (g \otimes f) \circ \sigma_{A,C}$
\end{center}

Compact closure is used to model entangled states. These are the only cases where we see arrows pointing downwards, as the $*$ in $A^*$ reverses the arrow's direction. The graphical representation looks like this:
\[
% Generated with LaTeXDraw 2.0.2
% Sun Jan 18 19:49:45 GMT 2009
% \usepackage[usenames,dvipsnames]{pstricks}
% \usepackage{epsfig}
% \usepackage{pst-grad} % For gradients
% \usepackage{pst-plot} % For axes
\scalebox{1} % Change this value to rescale the drawing.
{
\begin{pspicture}(0,-1.5178125)(11.3828125,0.83)
\rput{-180.0}(5.561875,1.6){\psarc[linewidth=0.06,arrowsize=0.05291667cm 2.0,arrowlength=1.4,arrowinset=0.4]{>->}(2.7809374,0.8){1.6}{-0.0}{180.0}}
\usefont{T1}{ptm}{m}{n}
\rput(2.8623438,-1.29){$d_A : I \longrightarrow A^* \otimes A$}
\usefont{T1}{ptm}{m}{n}
\rput(8.452344,-1.29){$e_A : A \otimes A^* \longrightarrow I$}
\psarc[linewidth=0.06,arrowsize=0.05291667cm 2.0,arrowlength=1.4,arrowinset=0.4]{<-<}(8.380938,-0.8){1.6}{-0.0}{180.0}
\end{pspicture}
}
\]

These have to adhere to a property, fundamental in proving teleportation, whose graphical representation is reminiscent of \textit{yanking} a wire:
\[
% Generated with LaTeXDraw 2.0.2
% Sun Jan 18 18:50:49 GMT 2009
% \usepackage[usenames,dvipsnames]{pstricks}
% \usepackage{epsfig}
% \usepackage{pst-grad} % For gradients
% \usepackage{pst-plot} % For axes
\scalebox{1} % Change this value to rescale the drawing.
{
\begin{pspicture}(0,-1.12)(12.82,1.12)
\rput{-180.0}(6.0,0.0){\psarc[linewidth=0.04](3.0,0.0){1.0}{-0.0}{180.0}}
\psarc[linewidth=0.04](1.0,-0.1){1.0}{-0.0}{180.0}
\psline[linewidth=0.04cm,arrowsize=0.05291667cm 2.0,arrowlength=1.4,arrowinset=0.4]{>-}(0.0,-1.1)(0.0,0.0)
\psline[linewidth=0.04cm,arrowsize=0.05291667cm 2.0,arrowlength=1.4,arrowinset=0.4]{->}(4.0,-0.1)(4.0,1.0)
\usefont{T1}{ptm}{m}{n}
\rput(4.631875,-0.04){\Huge =}
\psline[linewidth=0.04cm,arrowsize=0.05291667cm 2.0,arrowlength=1.4,arrowinset=0.4]{->}(5.4,-1.1)(5.4,1.1)
\rput{-180.0}(16.8,0.0){\psarc[linewidth=0.04](8.4,0.0){1.0}{-0.0}{180.0}}
\psarc[linewidth=0.04](10.4,-0.1){1.0}{-0.0}{180.0}
\psline[linewidth=0.04cm,arrowsize=0.05291667cm 2.0,arrowlength=1.4,arrowinset=0.4]{<-}(11.4,-1.1)(11.4,0.0)
\psline[linewidth=0.04cm,arrowsize=0.05291667cm 2.0,arrowlength=1.4,arrowinset=0.4]{-<}(7.4,-0.1)(7.4,1.0)
\usefont{T1}{ptm}{m}{n}
\rput(12.031875,-0.04){\Huge =}
\psline[linewidth=0.04cm,arrowsize=0.05291667cm 2.0,arrowlength=1.4,arrowinset=0.4]{<-}(12.8,-1.1)(12.8,1.1)
\end{pspicture}
}
\]

\section{Representation of classical structures}
\label{Section:Representation of classical structures}
One of the fundamental known distinctions between quantum and classical computation is derived from no-go theorems. Classical computers routinely copy and delete data; it is such a commonplace thing to do that we hardly ever notice how entwined it is to the classical computational paradigm itself. Quantum computers, on the other hand, cannot perform either of these operations. In this part of this section, we will see how to turn this problem into a very important feature, which will in turn enable us to account for classical operations within the, already defined, quantum categorical framework. The following result \cite{WZ82} is referred to as the \textit{no-cloning} theorem:

\begin{theorem}[No-cloning theorem]
There is no quantum operation $D$, such that
\[ D : \ket{0} \otimes \ket{\psi} \mapsto \ket{\psi} \otimes \ket{\psi} \]
\[ D : \ket{0} \otimes \ket{\phi} \mapsto \ket{\phi} \otimes \ket{\phi} \]
unless $\ket{\psi}$ and $\ket{\phi}$ are orthogonal.
\end{theorem}

The ancilla qubit $\ket{0}$ is sometimes not included in information-theoretic notation \cite{CP07,Abr10}, reducing $D$ to $D : \ket{\psi} \mapsto \ket{\psi} \otimes \ket{\psi}$. Another result \cite{PB00}, complementary to the no-cloning theorem, is commonly referred to as the \textit{no-deleting} theorem. In the original formulation of this result, Pati and Braunstein used an ancilla qubit $\ket{A}$ and included a qubit state in a standard state $\ket{\Sigma}$, making $E$ look more like $E : \ket{\psi} \otimes \ket{\psi} \otimes \ket{A} \mapsto \ket{\psi} \otimes \ket{\Sigma} \otimes \ket{A_\psi}$. The formulation we will use is more common in the information-theoretic literature \cite{CP07,Abr10}:

\begin{theorem}[No-deleting theorem]
There is no quantum operation $E$, such that
\[ E : \ket{\psi} \mapsto 1 \]
\[ E : \ket{\phi} \mapsto 1 \]
unless $\ket{\psi}$ and $\ket{\phi}$ are orthogonal.
\end{theorem}

From these two theorems, we can deduce that the only cases of quantum states that we could treat as classical are those pertaining to orthogonal vectors. This seems to reinforce the notion of a classical basis spanning a complex Hilbert space.

\begin{definition}[Classical structure (categorical)]
A \textit{classical structure}\footnote[2]{N.B. Thanks to \cite{CPV08}, we now know that classical structures are in bijective correspondence to bases.} \cite{CP07, CD08, CPP10} is defined in terms of special $\dag$-Frobenius cocommutative comonoids, also referred to as spider monoids. These are represented by a triplet $(A,\delta,\epsilon)$, where $A$ is an object and $\delta$ and $\epsilon$ are two morphisms of the form:
\begin{center}
$\delta : A \longrightarrow A \otimes A :: a_i \mapsto a_i \otimes a_i$
\hskip 2cm
$\epsilon : A \longrightarrow I :: a_i \mapsto 1$
\end{center}
\[
% Generated with LaTeXDraw 2.0.2
% Fri Jan 16 17:41:15 GMT 2009
% \usepackage[usenames,dvipsnames]{pstricks}
% \usepackage{epsfig}
% \usepackage{pst-grad} % For gradients
% \usepackage{pst-plot} % For axes
\scalebox{1} % Change this value to rescale the drawing.
{
\begin{pspicture}(0,-0.72)(6.52,0.72)
\psline[linewidth=0.04cm](6.22,-0.1)(6.22,-0.7)
\pscircle[linewidth=0.04,dimen=outer,fillstyle=solid,fillcolor=green](6.22,0.0){0.3}
\psline[linewidth=0.04cm](0.62,-0.1)(0.62,-0.7)
\psarc[linewidth=0.04](0.62,0.7){0.6}{-180.0}{0.0}
\pscircle[linewidth=0.04,dimen=outer,fillstyle=solid,fillcolor=green](0.62,0.0){0.3}
\end{pspicture}
}
\]

where the $a_i$ are the orthogonal vectors to which copying and deleting take place. These two maps, also called \textit{copying} and \textit{deleting} maps, have to satisfy a number of conditions to ensure that they are well behaved and that their addition does not cause a collapse of the compact structure. More specifically, they must satisfy all of the conditions of a cocommutative comonoid, as well as the isometry and Frobenius conditions.
\end{definition}

\begin{example}
To help the reader become more accustomed to our new concepts, we will present a linear algebraic example with matrices. Suppose that our chosen basis consisted of $\ket{0}$ and $\ket{1}$, the computational basis. The appropriate $\delta$ to copy the basis vectors would be:
\begin{center}
    $\delta = \left(\begin{array}{cc}
        1 & 0 \\
        0 & 0 \\
        0 & 0 \\
        0 & 1
    \end{array}\right)$ and $\varepsilon = \left(\begin{array}{cc}
        1 & 1
    \end{array}\right)$
\end{center}
\noindent It should be easy to verify that $\delta \ket{0} = \ket{00}$, $\delta \ket{1} = \ket{11}$, $\varepsilon \ket{0} = 1$ and $\varepsilon \ket{1} = 1$. The Hermitian adjoints of those states, $\delta^\dag$ and $\varepsilon^\dag$, would be:
\begin{center}
    $\delta^\dag = \left(\begin{array}{cccc}
        1 & 0 & 0 & 0 \\
        0 & 0 & 0 & 1
    \end{array}\right)$ and $\varepsilon^\dag = \left(\begin{array}{c}
        1 \\
        1
    \end{array}\right)$
\end{center}
\noindent Once again we can verify that $\delta^\dag$, also known as \textit{fusion}, merges two basis states when they are the same: $\delta^\dag \ket{00} = \ket{0}$ and $\delta^\dag \ket{11} = \ket{1}$. The $\varepsilon^\dag$ operation ``creates'' a state by being itself equivalent up to a scalar multiple to the constant state $\ket{+}$.
\end{example}

When expressed more rigorously, the comonoid conditions require that the copying and deleting maps form an internal cocommutative comonoid:
\begin{align*}
(\delta \otimes id_A) \circ \delta   & = (id_A \otimes \delta) \circ \delta \\
(\epsilon \otimes id_A) \circ \delta & = (id_A \otimes \epsilon) \circ \delta  = id_A \\
\delta                               & = \sigma_{A,A} \circ \delta
\end{align*}

Because of the dagger functor, all of the comonoidal conditions can be reversed to give us the conditions required of a commutative monoid:
\begin{align*}
\delta^\dag \circ (\delta^\dag \otimes id_A)    & = \delta^\dag \circ (id_A \otimes \delta^\dag) \\
\delta^\dag \circ (\epsilon^\dag \otimes id_A)  & = \delta^\dag \circ (id_A \otimes \epsilon^\dag) = id_A \\
\delta^\dag                                     & = \delta^\dag \circ \sigma_{A,A}
\end{align*}

The graphical representation for all of these conditions corresponds to connected graphs. Furthermore, the inputs and outputs on all of the required equations match, so, as proved in \cite{CP06,CPP10,CD08}, we can equivalently define classical structures using the \textit{spider theorem}. This theorem states that if a graph generated by $\delta$ and $\epsilon$ is connected, then it is completely characterized by its domain and codomain. If the domain is $\underbrace{A \otimes \ldots \otimes A}_n$ and the codomain is $\underbrace{A \otimes \ldots \otimes A}_m$, then it can be reduced to a \textit{"spider"} with $n$ input and $m$ output wires.
\[
% Generated with LaTeXDraw 2.0.2
% Fri Jan 16 10:04:58 GMT 2009
% \usepackage[usenames,dvipsnames]{pstricks}
% \usepackage{epsfig}
% \usepackage{pst-grad} % For gradients
% \usepackage{pst-plot} % For axes
\scalebox{1} % Change this value to rescale the drawing.
{
\begin{pspicture}(0,-1.9076562)(3.06,1.9076562)
\rput{-180.0}(3.06,2.2184374){\psarc[linewidth=0.06](1.53,1.1092187){1.0}{0.0}{180.0}}
\rput{-180.0}(3.06,3.2184374){\psarc[linewidth=0.06](1.53,1.6092187){1.5}{0.0}{180.0}}
\psdots[dotsize=0.12](0.33,1.4092188)
\psdots[dotsize=0.12](0.63,1.4092188)
\psdots[dotsize=0.12](0.93,1.4092188)
\psdots[dotsize=0.12](1.23,1.4092188)
\psdots[dotsize=0.12](1.53,1.4092188)
\psdots[dotsize=0.12](1.83,1.4092188)
\psdots[dotsize=0.12](2.13,1.4092188)
\psdots[dotsize=0.12](2.43,1.4092188)
\psdots[dotsize=0.12](2.73,1.4092188)
\psdots[dotsize=0.12](1.23,0.9092187)
\psdots[dotsize=0.12](1.53,0.9092187)
\psdots[dotsize=0.12](1.83,0.9092187)
\psarc[linewidth=0.06](1.53,-1.0907812){1.0}{0.0}{180.0}
\psarc[linewidth=0.06](1.53,-1.5907812){1.5}{0.0}{180.0}
\psdots[dotsize=0.12,dotangle=-180.0](2.73,-1.3907813)
\psdots[dotsize=0.12,dotangle=-180.0](2.43,-1.3907813)
\psdots[dotsize=0.12,dotangle=-180.0](2.13,-1.3907813)
\psdots[dotsize=0.12,dotangle=-180.0](1.83,-1.3907813)
\psdots[dotsize=0.12,dotangle=-180.0](1.53,-1.3907813)
\psdots[dotsize=0.12,dotangle=-180.0](1.23,-1.3907813)
\psdots[dotsize=0.12,dotangle=-180.0](0.93,-1.3907813)
\psdots[dotsize=0.12,dotangle=-180.0](0.63,-1.3907813)
\psdots[dotsize=0.12,dotangle=-180.0](0.33,-1.3907813)
\psdots[dotsize=0.12,dotangle=-180.0](1.83,-0.8907812)
\psdots[dotsize=0.12,dotangle=-180.0](1.53,-0.8907812)
\psdots[dotsize=0.12,dotangle=-180.0](1.23,-0.8907812)
\pscircle[linewidth=0.04,dimen=outer,fillstyle=solid,fillcolor=green](1.53,0.00921875){0.5}
\usefont{T1}{ptm}{m}{n}
\rput(1.5142188,1.7192187){$m$ outputs}
\usefont{T1}{ptm}{m}{n}
\rput(1.4642187,-1.6807812){$n$ inputs}
\end{pspicture}
}
\]

Arbitrary states that are points of $A$ (i.e. of the form $|\psi\rangle : I \longrightarrow A$) are denoted by "black dots":
%{\textcolor{white}{\boldmath$\psi$\unboldmath}}
%{\textcolor{white}{\boldmath$\psi^\dag$\unboldmath}}
\[
% Generated with LaTeXDraw 2.0.2
% Sun Jan 18 19:57:27 GMT 2009
% \usepackage[usenames,dvipsnames]{pstricks}
% \usepackage{epsfig}
% \usepackage{pst-grad} % For gradients
% \usepackage{pst-plot} % For axes
\scalebox{1} % Change this value to rescale the drawing.
{
\begin{pspicture}(0,-0.945)(9.182813,0.915)
\psline[linewidth=0.04cm](0.4009375,0.485)(0.4009375,-0.115)
\pscircle[linewidth=0.04,dimen=outer,fillstyle=solid,fillcolor=black](0.4009375,-0.215){0.3}
\usefont{T1}{ptm}{m}{n}
\rput(0.39234376,-0.245){\textcolor{white}{\boldmath$\psi$\unboldmath}}
\usefont{T1}{ptm}{m}{n}
\rput(1.62,-0.135){\LARGE $:=$}
\usefont{T1}{ptm}{m}{n}
\rput(3.1275,0.495){A}
\psline[linewidth=0.06cm,arrowsize=0.05291667cm 2.0,arrowlength=1.4,arrowinset=0.4]{->}(2.8009374,-0.315)(2.8009374,0.885)
\rput{-180.0}(5.601875,-1.03){\pstriangle[linewidth=0.04,dimen=outer,fillstyle=solid,fillcolor=green](2.8009374,-0.915)(1.0,0.8)}
\usefont{T1}{ptm}{m}{n}
\rput(2.7923439,-0.405){$\psi$}
\usefont{T1}{ptm}{m}{n}
\rput(8.7275,-0.505){A}
\psline[linewidth=0.06cm,arrowsize=0.05291667cm 2.0,arrowlength=1.4,arrowinset=0.4]{<-}(8.400937,-0.915)(8.400937,0.485)
\pstriangle[linewidth=0.04,dimen=outer,fillstyle=solid,fillcolor=red](8.400937,0.085)(1.0,0.8)
\usefont{T1}{ptm}{m}{n}
\rput(8.332344,0.375){$\psi^\dag$}
\psline[linewidth=0.04cm](6.0009375,0.085)(6.0009375,-0.515)
\pscircle[linewidth=0.04,dimen=outer,fillstyle=solid,fillcolor=black](6.0009375,0.185){0.3}
\usefont{T1}{ptm}{m}{n}
\rput(6.012344,0.155){\textcolor{white}{\boldmath$\psi^\dag$\unboldmath}}
\usefont{T1}{ptm}{m}{n}
\rput(7.22,-0.135){\LARGE $:=$}
\usefont{T1}{ptm}{m}{n}
\rput(4.4035935,-0.18){\large and}
\end{pspicture}
}
\]

On any classical structure $(A,\delta,\epsilon)$, we define a map $K$ that lifts\footnote[3]{The notation used for this lifting by \cite{CD08} is $\Lambda^\mu(\psi)$. In our case, we avoid using the letter $\Lambda$, to prevent any confusion with the letter used for Currying.} any state $|\psi\rangle$ of $A$ to the endomorphism $K^\delta(\psi) := \delta^\dag \circ (|\psi\rangle \otimes id_A) : A \longrightarrow A$. Similarly, we can lift any bra $\langle\psi|$ by setting $K^\delta(\psi)^\dag = K^\delta(\psi_*) = (\psi^\dag \otimes id_A) \circ \delta : A \longrightarrow A$. We denote these graphically as:
%{\textcolor{white}{\boldmath$\psi$\unboldmath}}
%{\textcolor{white}{\boldmath$\psi^\dag$\unboldmath}}
%{\textcolor{white}{\boldmath$\psi_*$\unboldmath}}
\[
% Generated with LaTeXDraw 2.0.2
% Sun Jan 18 20:33:51 GMT 2009
% \usepackage[usenames,dvipsnames]{pstricks}
% \usepackage{epsfig}
% \usepackage{pst-grad} % For gradients
% \usepackage{pst-plot} % For axes
\scalebox{1} % Change this value to rescale the drawing.
{
\begin{pspicture}(0,-2.67)(10.360937,2.67)
\psline[linewidth=0.04cm](3.2409375,2.65)(3.2409375,2.05)
\psarc[linewidth=0.04](3.2409375,1.23){0.6}{-0.0}{180.0}
\pscircle[linewidth=0.04,dimen=outer,fillstyle=solid,fillcolor=green](3.2409375,1.93){0.3}
\psline[linewidth=0.04cm](2.6409376,1.29)(2.6409376,0.75)
\pscircle[linewidth=0.04,dimen=outer,fillstyle=solid,fillcolor=black](2.6409376,0.65){0.3}
\usefont{T1}{ptm}{m}{n}
\rput(2.6323438,0.64){\textcolor{white}{\boldmath$\psi$\unboldmath}}
\psline[linewidth=0.04cm](3.8409376,1.27)(3.8409376,0.35)
\usefont{T1}{ptm}{m}{n}
\rput(1.7254688,1.455){\huge =}
\psline[linewidth=0.04cm](0.6409375,2.65)(0.6409375,0.35)
\pscircle[linewidth=0.04,dimen=outer,fillstyle=solid,fillcolor=green](0.6409375,1.51){0.3}
\usefont{T1}{ptm}{m}{n}
\rput(0.63234377,1.5){$\psi$}
\psline[linewidth=0.04cm](3.2409375,-2.05)(3.2409375,-2.65)
\psarc[linewidth=0.04](3.2409375,-1.25){0.6}{-180.0}{0.0}
\pscircle[linewidth=0.04,dimen=outer,fillstyle=solid,fillcolor=green](3.2409375,-1.95){0.3}
\psline[linewidth=0.04cm](2.6409376,-0.75)(2.6409376,-1.35)
\pscircle[linewidth=0.04,dimen=outer,fillstyle=solid,fillcolor=black](2.6409376,-0.65){0.3}
\usefont{T1}{ptm}{m}{n}
\rput(2.6523438,-0.68){\textcolor{white}{\boldmath$\psi^\dag$\unboldmath}}
\psline[linewidth=0.04cm](3.8409376,-0.35)(3.8409376,-1.35)
\psline[linewidth=0.04cm](7.1409373,-0.35)(7.1409373,-0.95)
\psarc[linewidth=0.04](7.1409373,-1.77){0.6}{-0.0}{180.0}
\pscircle[linewidth=0.04,dimen=outer,fillstyle=solid,fillcolor=green](7.1409373,-1.07){0.3}
\psline[linewidth=0.04cm](7.7409377,-1.73)(7.7409377,-2.65)
\psarc[linewidth=0.04](5.9409375,-1.73){0.6}{-180.0}{0.0}
\psline[linewidth=0.04cm](5.3409376,-1.15)(5.3409376,-1.75)
\pscircle[linewidth=0.04,dimen=outer,fillstyle=solid,fillcolor=black](5.3409376,-1.05){0.3}
\usefont{T1}{ptm}{m}{n}
\rput(5.3523436,-1.08){\textcolor{white}{\boldmath$\psi^\dag$\unboldmath}}
\psline[linewidth=0.04cm](9.740937,-0.35)(9.740937,-0.95)
\psarc[linewidth=0.04](9.740937,-1.77){0.6}{-0.0}{180.0}
\pscircle[linewidth=0.04,dimen=outer,fillstyle=solid,fillcolor=green](9.740937,-1.07){0.3}
\psline[linewidth=0.04cm](9.140938,-1.71)(9.140938,-2.25)
\pscircle[linewidth=0.04,dimen=outer,fillstyle=solid,fillcolor=black](9.140938,-2.35){0.3}
\usefont{T1}{ptm}{m}{n}
\rput(9.112344,-2.36){\textcolor{white}{\boldmath$\psi_*$\unboldmath}}
\psline[linewidth=0.04cm](10.340938,-1.73)(10.340938,-2.65)
\usefont{T1}{ptm}{m}{n}
\rput(4.525469,-1.545){\huge =}
\usefont{T1}{ptm}{m}{n}
\rput(8.325469,-1.545){\huge =}
\psline[linewidth=0.04cm](0.6409375,-0.35)(0.6409375,-2.65)
\pscircle[linewidth=0.04,dimen=outer,fillstyle=solid,fillcolor=green](0.6409375,-1.49){0.3}
\usefont{T1}{ptm}{m}{n}
\rput(0.6123437,-1.5){$\psi_*$}
\usefont{T1}{ptm}{m}{n}
\rput(1.7254688,-1.545){\huge =}
\end{pspicture}
}
\]

The monoid operation $\delta^\dag : A \otimes A \longrightarrow A$, also known as \textit{fusion}, can be used to combine arbitrary states. For two states $|\psi\rangle$ and $|\phi\rangle$, their fusion is written as $\psi \odot \phi := \delta^\dag \circ (\psi \otimes \phi)$. The same operation can be used to merge pairs of lifted states $K^\delta(\psi \odot \phi) := \delta^\dag (K^\delta(\psi) \otimes K^\delta(\phi)) = \delta^\dag (K^\delta(\phi) \otimes K^\delta(\psi))$. As per the definition of $\delta^\dag$, the fusion operation is associative and commutative. This concept of fusion can be represented diagrammatically as follows:
\[
% Generated with LaTeXDraw 2.0.2
% Sat Jan 17 15:08:32 GMT 2009
% \usepackage[usenames,dvipsnames]{pstricks}
% \usepackage{epsfig}
% \usepackage{pst-grad} % For gradients
% \usepackage{pst-plot} % For axes
\scalebox{1} % Change this value to rescale the drawing.
{
\begin{pspicture}(0,-1.19)(5.0028124,1.19)
\psline[linewidth=0.04cm](1.0009375,1.17)(1.0009375,0.57)
\psarc[linewidth=0.04](1.0009375,-0.23){0.6}{-0.0}{180.0}
\pscircle[linewidth=0.04,dimen=outer,fillstyle=solid,fillcolor=green](1.0009375,0.47){0.3}
\psline[linewidth=0.04cm](0.4009375,-0.19)(0.4009375,-0.79)
\pscircle[linewidth=0.04,dimen=outer,fillstyle=solid,fillcolor=black](0.4009375,-0.89){0.3}
\usefont{T1}{ptm}{m}{n}
\rput(0.39234376,-0.9){\textcolor{white}{\boldmath$\psi$\unboldmath}}
\psline[linewidth=0.04cm](1.6009375,-0.19)(1.6009375,-0.79)
\pscircle[linewidth=0.04,dimen=outer,fillstyle=solid,fillcolor=black](1.6009375,-0.89){0.3}
\usefont{T1}{ptm}{m}{n}
\rput(1.5923438,-0.9){\textcolor{white}{\boldmath$\phi$\unboldmath}}
\usefont{T1}{ptm}{m}{n}
\rput(2.4854689,-0.045){\huge =}
\psline[linewidth=0.04cm](3.7009375,0.15)(3.7009375,0.95)
\psellipse[linewidth=0.04,dimen=outer,fillstyle=solid,fillcolor=black](3.7009375,0.05)(0.6,0.3)
\usefont{T1}{ptm}{m}{n}
\rput(3.6723437,0.06){\textcolor{white}{\boldmath$\psi\!\odot\!\phi$\unboldmath}}
\end{pspicture}
}
\]

\begin{example}
Let $\ket{+_\alpha}$ and $\ket{+_\beta}$ be states such that:
\begin{center}
    $\ket{+_\alpha} = \ket{0} + e^{i\alpha}\ket{1} = \left(\begin{array}{c}
            1 \\
            e^{i\alpha}
    \end{array}\right)$ and $\ket{+_\beta} = \ket{0} + e^{i\beta}\ket{1} = \left(\begin{array}{c}
            1 \\
            e^{i\beta}
    \end{array}\right)$
\end{center}
\noindent The fusion of those two states will be:
\[
\delta^\dag \circ (\ket{+_\alpha} \otimes \ket{+_\beta}) = \left(\begin{array}{cccc}
    1 & 0 & 0 & 0 \\
    0 & 0 & 0 & 1
\end{array}\right) \left(\begin{array}{c}
    1 \\
    e^{i\beta} \\
    e^{i\alpha} \\
    e^{i(\alpha+\beta)}
\end{array}\right) = \left(\begin{array}{c}
    1 \\
    e^{i(\alpha+\beta)}
\end{array}\right) = \ket{+_{\alpha+\beta}}
\]
\end{example}

After combining all of these definitions, we can prove \cite{CD08} that classical structures follow what is known as the generalized spider theorem:

\begin{theorem}[Spider theorem]
Any connected graph generated by the operations of the classical structure $(A,\delta,\epsilon)$, states $|\psi_i\rangle : I \longrightarrow A$ and the $\dag$-compact structure, is completely characterized by its domain, codomain and $\displaystyle K^\delta(\bigodot_i \psi_i)$. The graphical representation is that of a "decorated spider"
\[
% Generated with LaTeXDraw 2.0.2
% Sat Jan 17 18:57:01 GMT 2009
% \usepackage[usenames,dvipsnames]{pstricks}
% \usepackage{epsfig}
% \usepackage{pst-grad} % For gradients
% \usepackage{pst-plot} % For axes
\scalebox{1} % Change this value to rescale the drawing.
{
\begin{pspicture}(0,-1.9076562)(3.47125,1.9076562)
\rput{-180.0}(3.370625,2.2184374){\psarc[linewidth=0.06](1.6853125,1.1092187){1.0}{-0.0}{180.0}}
\rput{-180.0}(3.370625,3.2184374){\psarc[linewidth=0.06](1.6853125,1.6092187){1.5}{-0.0}{180.0}}
\psdots[dotsize=0.12](0.4853125,1.4092188)
\psdots[dotsize=0.12](0.7853125,1.4092188)
\psdots[dotsize=0.12](1.0853125,1.4092188)
\psdots[dotsize=0.12](1.3853126,1.4092188)
\psdots[dotsize=0.12](1.6853125,1.4092188)
\psdots[dotsize=0.12](1.9853125,1.4092188)
\psdots[dotsize=0.12](2.2853124,1.4092188)
\psdots[dotsize=0.12](2.5853126,1.4092188)
\psdots[dotsize=0.12](2.8853126,1.4092188)
\psdots[dotsize=0.12](1.3853126,0.9092187)
\psdots[dotsize=0.12](1.6853125,0.9092187)
\psdots[dotsize=0.12](1.9853125,0.9092187)
\psarc[linewidth=0.06](1.6853125,-1.0907812){1.0}{-0.0}{180.0}
\psarc[linewidth=0.06](1.6853125,-1.5907812){1.5}{-0.0}{180.0}
\psdots[dotsize=0.12,dotangle=-180.0](2.8853126,-1.3907813)
\psdots[dotsize=0.12,dotangle=-180.0](2.5853126,-1.3907813)
\psdots[dotsize=0.12,dotangle=-180.0](2.2853124,-1.3907813)
\psdots[dotsize=0.12,dotangle=-180.0](1.9853125,-1.3907813)
\psdots[dotsize=0.12,dotangle=-180.0](1.6853125,-1.3907813)
\psdots[dotsize=0.12,dotangle=-180.0](1.3853126,-1.3907813)
\psdots[dotsize=0.12,dotangle=-180.0](1.0853125,-1.3907813)
\psdots[dotsize=0.12,dotangle=-180.0](0.7853125,-1.3907813)
\psdots[dotsize=0.12,dotangle=-180.0](0.4853125,-1.3907813)
\psdots[dotsize=0.12,dotangle=-180.0](1.9853125,-0.8907812)
\psdots[dotsize=0.12,dotangle=-180.0](1.6853125,-0.8907812)
\psdots[dotsize=0.12,dotangle=-180.0](1.3853126,-0.8907812)
\pscircle[linewidth=0.04,dimen=outer,fillstyle=solid,fillcolor=green](1.6853125,0.00921875){0.5}
\usefont{T1}{ptm}{m}{n}
\rput(1.6695312,1.7192187){$m$ outputs}
\usefont{T1}{ptm}{m}{n}
\rput(1.6195313,-1.6807812){$n$ inputs}
\usefont{T1}{ptm}{m}{n}
\rput(1.7009375,-0.04578125){\scriptsize $\displaystyle\bigodot_i \psi_i$}
\end{pspicture}
}
\]
\end{theorem}

\begin{definition}[Unbiasedness (categorical)]
A point $\alpha : I \longrightarrow A$ is \textit{unbiased} relative to $(A,\delta,\epsilon)$ iff $K^\delta(\alpha)$ is unitary. In other words, there needs to be a scalar $s : I \longrightarrow I$ such that $s \cdot \alpha \odot \alpha^\dag = \epsilon^\dag$, or graphically:
%{\textcolor{white}{\boldmath$\alpha$\unboldmath}}
%{\textcolor{white}{\boldmath$\alpha^\dag$\unboldmath}}
\[
% Generated with LaTeXDraw 2.0.2
% Sun Jan 18 21:20:22 GMT 2009
% \usepackage[usenames,dvipsnames]{pstricks}
% \usepackage{epsfig}
% \usepackage{pst-grad} % For gradients
% \usepackage{pst-plot} % For axes
\scalebox{1} % Change this value to rescale the drawing.
{
\begin{pspicture}(0,-2.2)(8.640938,2.22)
\psline[linewidth=0.04cm](2.1409376,2.2)(2.1409376,1.3)
\psline[linewidth=0.04cm](6.2409377,-0.4)(6.2409377,1.1)
\pscircle[linewidth=0.04,dimen=outer,fillstyle=solid,fillcolor=green](6.2409377,-0.5){0.5}
\usefont{T1}{ptm}{m}{n}
\rput(6.2765627,-0.555){\scriptsize $\alpha\!\odot\!\alpha_*$}
\psline[linewidth=0.04cm](1.5409375,0.1)(1.5409375,-0.5)
\psarc[linewidth=0.04](1.5409375,-1.32){0.6}{-0.0}{180.0}
\pscircle[linewidth=0.04,dimen=outer,fillstyle=solid,fillcolor=green](1.5409375,-0.62){0.3}
\psline[linewidth=0.04cm](0.9409375,-1.26)(0.9409375,-1.8)
\pscircle[linewidth=0.04,dimen=outer,fillstyle=solid,fillcolor=black](0.9409375,1.9){0.3}
\usefont{T1}{ptm}{m}{n}
\rput(0.95234376,1.89){\textcolor{white}{\boldmath$\alpha^\dag$\unboldmath}}
\psline[linewidth=0.04cm](2.1409376,-1.28)(2.1409376,-1.8)
\usefont{T1}{ptm}{m}{n}
\rput(3.0254688,0.0050){\huge =}
\psline[linewidth=0.04cm](1.5409375,0.6)(1.5409375,0.0)
\psarc[linewidth=0.04](1.5409375,1.4){0.6}{-180.0}{0.0}
\pscircle[linewidth=0.04,dimen=outer,fillstyle=solid,fillcolor=green](1.5409375,0.7){0.3}
\psline[linewidth=0.04cm](0.9409375,1.9)(0.9409375,1.3)
\pscircle[linewidth=0.04,dimen=outer,fillstyle=solid,fillcolor=black](0.9409375,-1.9){0.3}
\usefont{T1}{ptm}{m}{n}
\rput(0.95234376,-1.93){\textcolor{white}{\boldmath$\alpha$\unboldmath}}
\psline[linewidth=0.04cm](3.9409375,2.2)(3.9409375,-0.46)
\pscircle[linewidth=0.04,dimen=outer,fillstyle=solid,fillcolor=green](3.9409375,0.7){0.3}
\usefont{T1}{ptm}{m}{n}
\rput(3.9323437,0.69){$\alpha_*$}
\psline[linewidth=0.04cm](3.9409375,-0.16)(3.9409375,-2.0)
\pscircle[linewidth=0.04,dimen=outer,fillstyle=solid,fillcolor=green](3.9409375,-0.6){0.3}
\usefont{T1}{ptm}{m}{n}
\rput(3.9523437,-0.61){$\alpha$}
\pscircle[linewidth=0.04,dimen=outer,fillstyle=solid,fillcolor=green](3.9409375,-1.9){0.3}
\usefont{T1}{ptm}{m}{n}
\rput(4.8254685,0.0050){\huge =}
\psline[linewidth=0.04cm](8.340938,0.6)(8.340938,0.0)
\pscircle[linewidth=0.04,dimen=outer,fillstyle=solid,fillcolor=green](8.340938,-0.1){0.3}
\usefont{T1}{ptm}{m}{n}
\rput(7.425469,0.0050){\huge =}
\pscircle[linewidth=0.04,dimen=outer,fillstyle=solid,fillcolor=green](2.1409376,-1.9){0.3}
\end{pspicture}
}
\]
\end{definition}

\begin{example}
Recall the state $\ket{+_\alpha}$ from our previous example. States of the form $\ket{+_\theta}$ will always be unbiased with respect to the classical basis structure we defined for copying $\ket{0}$ and $\ket{1}$. That is because $\delta^\dag \circ (\ket{+_\theta} \otimes \ket{+_{-\theta}})$ will always be equal to $\ket{+_0} = \varepsilon^\dag$.
\end{example}

\begin{definition}[Complementarity (categorical)]
Two classical structures $(A,\delta_G,\epsilon_G)$ and $(A,\delta_R,\epsilon_R)$ in a $\dag$-compact category are called \textit{complementary} if the points that are classical for one are unbiased for the other and vice versa. $\delta_G$ and $\epsilon_G$ are depicted using green dots; the points they copy and delete are drawn in red and are unbiased for the second classical structure. $\delta_R$ and $\epsilon_R$ are depicted using red dots; they copy and delete green points, which are unbiased for the green classical structure. Graphically this condition is depicted as:
\[
% Generated with LaTeXDraw 2.0.2
% Sun Jan 18 22:08:04 GMT 2009
% \usepackage[usenames,dvipsnames]{pstricks}
% \usepackage{epsfig}
% \usepackage{pst-grad} % For gradients
% \usepackage{pst-plot} % For axes
\scalebox{1} % Change this value to rescale the drawing.
{
\begin{pspicture}(0,-2.7)(9.742812,2.7)
\usefont{T1}{ptm}{m}{n}
\rput(2.6854687,-1.715){\huge =}
\psarc[linewidth=0.04](1.2009375,-1.84){0.6}{-0.0}{180.0}
\psline[linewidth=0.04cm](1.2009375,-0.22)(1.2009375,-1.08)
\pscircle[linewidth=0.04,dimen=outer,fillstyle=solid,fillcolor=red](1.2009375,-1.14){0.3}
\psline[linewidth=0.04cm](0.6009375,-1.78)(0.6009375,-2.32)
\pscircle[linewidth=0.04,dimen=outer,fillstyle=solid,fillcolor=red](0.6009375,-2.4){0.3}
\usefont{T1}{ptm}{m}{n}
\rput(0.6123437,-2.41){$\alpha$}
\psline[linewidth=0.04cm](1.8009375,-1.8)(1.8009375,-2.32)
\pscircle[linewidth=0.04,dimen=outer,fillstyle=solid,fillcolor=red](1.8009375,-2.4){0.3}
\usefont{T1}{ptm}{m}{n}
\rput(1.8323437,-2.41){$\alpha_*$}
\psline[linewidth=0.04cm](3.6009376,-0.62)(3.6009376,-1.48)
\pscircle[linewidth=0.04,dimen=outer,fillstyle=solid,fillcolor=red](3.6009376,-1.62){0.3}
\usefont{T1}{ptm}{m}{n}
\rput(7.885469,-1.715){\huge =}
\psarc[linewidth=0.04](6.4009376,-1.84){0.6}{-0.0}{180.0}
\psline[linewidth=0.04cm](6.4009376,-0.22)(6.4009376,-1.08)
\pscircle[linewidth=0.04,dimen=outer,fillstyle=solid,fillcolor=green](6.4009376,-1.14){0.3}
\psline[linewidth=0.04cm](5.8009377,-1.78)(5.8009377,-2.32)
\pscircle[linewidth=0.04,dimen=outer,fillstyle=solid,fillcolor=green](5.8009377,-2.4){0.3}
\usefont{T1}{ptm}{m}{n}
\rput(5.7723436,-2.41){$\beta$}
\psline[linewidth=0.04cm](7.0009375,-1.8)(7.0009375,-2.32)
\pscircle[linewidth=0.04,dimen=outer,fillstyle=solid,fillcolor=green](7.0009375,-2.4){0.3}
\usefont{T1}{ptm}{m}{n}
\rput(6.992344,-2.41){$\beta_*$}
\psline[linewidth=0.04cm](8.800938,-0.62)(8.800938,-1.48)
\pscircle[linewidth=0.04,dimen=outer,fillstyle=solid,fillcolor=green](8.800938,-1.62){0.3}
\psline[linewidth=0.04cm](1.2009375,1.88)(1.2009375,1.28)
\psarc[linewidth=0.04](1.2009375,2.68){0.6}{-180.0}{0.0}
\pscircle[linewidth=0.04,dimen=outer,fillstyle=solid,fillcolor=green](1.2009375,1.98){0.3}
\psline[linewidth=0.04cm](3.2009375,2.02)(3.2009375,1.48)
\pscircle[linewidth=0.04,dimen=outer,fillstyle=solid,fillcolor=red](3.2009375,1.4){0.3}
\usefont{T1}{ptm}{m}{n}
\rput(3.2123437,1.39){$\alpha$}
\psline[linewidth=0.04cm](4.0009375,2.02)(4.0009375,1.48)
\pscircle[linewidth=0.04,dimen=outer,fillstyle=solid,fillcolor=red](4.0009375,1.4){0.3}
\usefont{T1}{ptm}{m}{n}
\rput(4.012344,1.39){$\alpha$}
\psline[linewidth=0.04cm](1.2009375,1.42)(1.2009375,0.88)
\pscircle[linewidth=0.04,dimen=outer,fillstyle=solid,fillcolor=red](1.2009375,0.8){0.3}
\usefont{T1}{ptm}{m}{n}
\rput(1.2123437,0.79){$\alpha$}
\usefont{T1}{ptm}{m}{n}
\rput(2.2854688,1.485){\huge =}
\psline[linewidth=0.04cm](6.4009376,1.42)(6.4009376,0.88)
\pscircle[linewidth=0.04,dimen=outer,fillstyle=solid,fillcolor=green](6.4009376,0.8){0.3}
\usefont{T1}{ptm}{m}{n}
\rput(6.3723435,0.79){$\beta$}
\psline[linewidth=0.04cm](9.200937,2.02)(9.200937,1.48)
\pscircle[linewidth=0.04,dimen=outer,fillstyle=solid,fillcolor=green](9.200937,1.4){0.3}
\usefont{T1}{ptm}{m}{n}
\rput(9.172344,1.39){$\beta$}
\psline[linewidth=0.04cm](8.400937,2.02)(8.400937,1.48)
\pscircle[linewidth=0.04,dimen=outer,fillstyle=solid,fillcolor=green](8.400937,1.4){0.3}
\usefont{T1}{ptm}{m}{n}
\rput(8.372344,1.39){$\beta$}
\psline[linewidth=0.04cm](6.4009376,1.88)(6.4009376,1.28)
\psarc[linewidth=0.04](6.4009376,2.68){0.6}{-180.0}{0.0}
\pscircle[linewidth=0.04,dimen=outer,fillstyle=solid,fillcolor=red](6.4009376,1.98){0.3}
\usefont{T1}{ptm}{m}{n}
\rput(7.485469,1.485){\huge =}
\end{pspicture}
}
\]
\end{definition}

\begin{example}
Let $(A,\delta_G,\epsilon_G)$ be the classical structure that we described earlier for copying and deleting $\ket{0}$ and $\ket{1}$. Let $(A,\delta_R,\epsilon_R)$ be a similarly constructed classical structure that can copy and delete $\ket{+}$ and $\ket{-}$ (i.e. state vectors orthogonal to the computational basis), such as:
\begin{center}
    $\delta_R = \frac{1}{\sqrt{2}}\left(\begin{array}{cc}
        1 & 0 \\
        0 & 1 \\
        0 & 1 \\
        1 & 0
    \end{array}\right)$ and $\varepsilon_R = \left(\begin{array}{cc}
        1 & 0
    \end{array}\right)$
\end{center}
\noindent As one can easily verify ,the two classical structures are an example of structures that are complementary to one another.
\end{example}

\chapter{Linear logic}
\label{Chapter:Linear logic}
This section provides an overview of the structures found in some flavours of linear logic, building up to the corresponding logic for compact closed categories, while illustrating how some of these concepts relate to category theoretic notions and properties.

Linear logic is a resource sensitive logic, first introduced in \cite{Gir87}. Whereas other logics provide the structural rules of weakening and contraction in order to facilitate predicate re-use or non-use in proving theorems, linear logic drops the indiscriminate use of these rules and treats predicates as resources that need to be expended in order to produce proofs. In terms of the Gentzen sequent calculus, the rules of weakening and contraction would be represented as:
\begin{center}
Weakening $\displaystyle\frac{\vdash \Gamma}{\vdash \Gamma, A}$
\hskip 1cm
Contraction $\displaystyle\frac{\vdash \Gamma, A, A}{\vdash \Gamma, A}$
\end{center}
Due to its resource sensitivity, linear logic finds many applications in computer science, such as in type theory, the semantics of programming languages, and the study of concurrency. The definitions and presentation of this section are largely based on \cite{Abr93} and \cite{AT10}. The sequent rules for linear logic are as follows:
\begin{center}
Axiom $\displaystyle\frac{}{\vdash A^\perp, A}$ \hskip 1cm Exchange $\displaystyle\frac{\vdash \Gamma, A, B, \Delta}{\vdash \Gamma, B, A, \Delta}$ \hskip 1cm Cut $\displaystyle\frac{\vdash \Gamma, A \;\;\;\; \vdash \Delta, A^\perp}{\vdash \Gamma, \Delta}$\\
\vskip 0.5cm
Unit $\displaystyle\frac{}{\vdash 1}$ \hskip 1cm Perp $\displaystyle\frac{\vdash \Gamma}{\vdash \Gamma, -}$\\
\vskip 0.5cm
Times $\displaystyle\frac{\vdash \Gamma, A \;\;\;\; \vdash \Delta, B}{\vdash \Gamma, \Delta, A \otimes B}$ \hskip 1cm Par $\displaystyle\frac{\vdash \Gamma, A, B}{\vdash \Gamma, A \logicpar B}$\\
\vskip 0.5cm
With $\displaystyle\frac{\vdash \Gamma, A \;\;\;\; \vdash \Gamma, B}{\vdash \Gamma, A \& B}$ \hskip 1cm {Plus (i)} $\displaystyle\frac{\vdash \Gamma, A}{\vdash \Gamma, A \oplus B}$ \hskip 1cm {Plus (ii)} $\displaystyle\frac{\vdash \Gamma, B}{\vdash \Gamma, A \oplus B}$
\end{center}

In order to recover the structural rules of weakening and contraction, we can introduce an exponential operator called \textit{bang}, which is denoted by $!$
\begin{center}
Dereliction $\displaystyle\frac{\vdash \Gamma, A}{\vdash \Gamma, !A}$ \hskip 1cm Of Course $\displaystyle\frac{\vdash !\Gamma, A}{\vdash !\Gamma, !A}$\\
\vskip 0.5cm
Weakening $\displaystyle\frac{\vdash \Gamma}{\vdash \Gamma, !A}$ \hskip 1cm Contraction $\displaystyle\frac{\vdash \Gamma, !A, !A}{\vdash \Gamma, !A}$
\end{center}

\section{Multiplicative quantum logic}
\label{Section:Multiplicative quantum logic}
We will now focus on and extend the multiplicative fragment of linear logic by choosing to ignore the rules for \textit{additives} and \textit{exponentials}. The reader is referred to \cite{AT10} for a comprehensive, yet still accessible, overview of multiplicative linear logic. That logic was extended by \cite{AD06}, to yield a multiplicative quantum logic that simulates the compact structure found in compact closed categories. The key idea lies in the definition of linear negation, whereby multiplicative conjunction (tensor) is equated with multiplicative disjunction (par), by trivializing the notion of De Morgan duality. This is the most natural way of introducing compactness to linear logic and uses $A^\perp$ to represent the dual object $A^*$ in our category. Linear negation in this case is characterized by the following laws:
\begin{eqnarray}
A^{\perp\perp} & = & A \nonumber \\
1^\perp & = & 1 \nonumber \\
(A \otimes B)^\perp & = & A^\perp \otimes B^\perp \nonumber \\
A \multimap B & = & A^\perp \otimes B \nonumber
\end{eqnarray}

A categorical interpretation will be given for each of the rules presented, in order to better illustrate the Curry-Howard parallelism. The proof rules for this logic are as follows:
\begin{center}
\renewcommand{\arraystretch}{3}
\begin{tabular}{|c|c|c|}
  \hline
    & Logic & Categories \\ \hline

  Id &
  $\displaystyle\frac{}{A \vdash A}$ &
  $\displaystyle\frac{}{id_A: A \longrightarrow A}$ \\ \hline

  $\bigotimes R$ &
  $\displaystyle\frac{\Gamma \vdash A \;\;\;\; \Delta \vdash B}{\Gamma, \Delta \vdash A \otimes B}$ &
  $\displaystyle\frac{f: \Gamma \longrightarrow A \;\;\;\; g: \Delta \longrightarrow B}{f \otimes g: \Gamma \otimes \Delta \longrightarrow A \otimes B}$ \\ \hline

  $\bigotimes L$ &
  $\displaystyle\frac{\Gamma, A, B \vdash C}{\Gamma, A \otimes B \vdash C }$ &
  $\displaystyle\frac{f: (\Gamma \otimes A) \otimes B \longrightarrow C}{f \circ a_{A,B,\Gamma}: \Gamma \otimes (A \otimes B) \longrightarrow C }$ \\ \hline

  Cut &
  $\displaystyle\frac{\Gamma \vdash A \;\;\;\; A, \Delta \vdash B}{\Gamma, \Delta \vdash B}$ &
  $\displaystyle\frac{f: \Gamma \longrightarrow A \;\;\;\; g: A \otimes \Delta \longrightarrow B}{g \circ (f \otimes id_\Delta): \Gamma \otimes \Delta \longrightarrow B}$ \\ \hline

  $\multimap E$ &
  $\displaystyle\frac{\Gamma \vdash A \multimap B \;\;\;\; \Delta \vdash A}{\Gamma, \Delta \vdash B}$ &
  $\displaystyle\frac{f: \Gamma \longrightarrow (A \multimap B) \;\;\;\; g: \Delta \longrightarrow A}{ev_{A,B} \circ (f \otimes g): \Gamma \otimes \Delta \longrightarrow B}$ \\ \hline

  $\multimap R$ &
  $\displaystyle\frac{\Gamma, A \vdash B}{\Gamma \vdash A \multimap B}$ &
  $\displaystyle\frac{f: \Gamma \otimes A \longrightarrow B}{\Lambda(f): \Gamma \longrightarrow (A \multimap B)}$ \\ \hline
\end{tabular}
\end{center}

At this point, it is interesting to compare the two worlds and see how some notions translate from one to the other. The identity rule corresponds to identity arrows in our categories. The Cut rule defines function composition. The right tensor rule ($\bigotimes R$) defines tensoring, while the left tensor rule ($\bigotimes L$) defines the associativity of the tensor. Linear implication ($\multimap$) is a notion equivalent to a category's exponential objects, so naturally, implication elimination ($\multimap E$) and the right implication rule ($\multimap R$) respectively define the category's Evaluation and Currying functions. The only structural rule in this logic is the rule known as \textit{exchange}:
\begin{center}
\small
$\displaystyle\frac{\Gamma, A, B, \Delta \vdash C}{\Gamma, B, A, \Delta \vdash C}$
\hskip 1cm
$\displaystyle\frac{f: \Gamma \otimes A \otimes B \otimes \Delta \longrightarrow C}{f \circ (id_\Gamma \otimes s_{A,B} \otimes id_\Delta): \Gamma \otimes B \otimes A \otimes \Delta \longrightarrow C}$
\normalsize
\end{center}

This rule corresponds to the symmetry isomorphism for the tensor. All we need in order to have a fully fledged representation of symmetric monoidal closed categories is a monoidal unit. That is denoted by $1$ and comes with the simple rule of $\overline{\vdash 1}$. We use the right implication rule to transform our identity rule to $\vdash A \multimap A$ which, when translated via the linear negation laws, becomes $\vdash A^\perp \otimes A$, or the equivalent of $d_A: I \longrightarrow A^* \otimes A$ in our category. Thus, we can now represent compact structure in this flavour of linear logic.

\chapter{The linear typed lambda calculus}
\label{Chapter:The linear typed lambda calculus}
This section will provide an overview of the linear typed lambda calculus of \cite{AT10}, to prepare the ground for our extension to \textit{dagger compact categories} in the next section. We will define well formed formulas for terms, types and sequents and provide Gentzen-style inference rules for deriving these formulas. This lambda calculus provides a computational interpretation for symmetric monoidal closed categories and is in direct correspondence with the multiplicative fragment of intuitionistic linear logic.

\begin{definition}[Variables and terms in the lambda calculus]
The fundamental building blocks of our language are \textit{variables}. They are denoted by single letters and are traditionally represented using the later letters of the alphabet (i.e. $x,y,z$). These variables can then be combined with each other to form composite \textit{terms}, denoted by different combinations of the following forms:
\[ t \;\;\;\; ::= \;\;\;\; x \;\; | \;\; t_1 \otimes t_2 \;\; | \;\; \text{ let } \psi \text{ be } x \otimes y \text{ in } c \;\; | \;\; \lambda x.t  \;\; | \;\; fa \]
\end{definition}

\vskip 0.1cm

\begin{definition}[Types in the lambda calculus]
Every term in our language, regardless of whether it is a variable or composite, has a \textit{type}. Types can appear as any combination of the following forms:
\[ \text{type} \;\;\;\; ::= \;\;\;\; A \;\;\; | \;\;\; A \otimes B \;\;\; | \;\;\; A \multimap B \]
\end{definition}

\vskip 0.1cm

\begin{definition}[Typing judgements in the lambda calculus]
The typing judgements, or sequents, of our language are composed of terms and their respective types. They are always of the form:
\[ x_1:A_1, \; x_2:A_2, \; \ldots, \; x_n:A_n \vdash t:B \]
\end{definition}

A set of typing rules is used to produce typing judgements. Now that we know which formulas are well formed, we can continue our language exposition by providing these rules in the form of a Gentzen-style Sequent Calculus. The set of inference rules corresponds to the rules that were presented in the section on linear logic. The rules and their respective correspondences in category theory are as follows:

\begin{center}
\renewcommand{\arraystretch}{3}
\begin{tabular}{|c|c|c|}
  \hline
    & Lambda Calculus & Categories \\ \hline

  Id &
  $\displaystyle\frac{}{x:A \vdash x:A}$ &
  $\displaystyle\frac{}{id_A: A \longrightarrow A}$ \\ \hline

  $\bigotimes R$ &
  $\displaystyle\frac{\Gamma \vdash a:A \;\;\;\; \Delta \vdash b:B}{\Gamma, \Delta \vdash a \otimes b : A \otimes B}$ &
  $\displaystyle\frac{f: \Gamma \longrightarrow A \;\;\;\; g: \Delta \longrightarrow B}{f \otimes g: \Gamma \otimes \Delta \longrightarrow A \otimes B}$ \\ \hline

  $\bigotimes L$ &
  $\displaystyle\frac{\Gamma, x:A, y:B \vdash c:C}{\Gamma, \psi : A \otimes B \vdash \textrm{ let } \psi \textrm{ be } x \otimes y \textrm{ in } c:C}$ &
  $\displaystyle\frac{f: (\Gamma \otimes A) \otimes B \longrightarrow C}{f \circ a_{A,B,\Gamma}: \Gamma \otimes (A \otimes B) \longrightarrow C }$ \\ \hline

  Cut &
  $\displaystyle\frac{\Gamma \vdash c:A \;\;\;\; x:A, \Delta \vdash b:B}{\Gamma, \Delta \vdash b[c/x]:B}$ &
  $\displaystyle\frac{f: \Gamma \longrightarrow A \;\;\;\; g: A \otimes \Delta \longrightarrow B}{g \circ (f \otimes id_\Delta): \Gamma \otimes \Delta \longrightarrow B}$ \\ \hline

  $\multimap R$ &
  $\displaystyle\frac{\Gamma, x:A \vdash b:B}{\Gamma \vdash \lambda x.b : A \multimap B}$ &
  $\displaystyle\frac{f: \Gamma \otimes A \longrightarrow B}{\Lambda(f): \Gamma \longrightarrow (A \multimap B)}$ \\ \hline

  $\multimap E$ &
  $\displaystyle\frac{\Gamma \vdash g: A \multimap B \;\;\;\; \Delta \vdash a:A}{\Gamma, \Delta \vdash ga:B}$ &
  $\displaystyle\frac{f: \Gamma \longrightarrow (A \multimap B) \;\;\;\; g: \Delta \longrightarrow A}{ev_{A,B} \circ (f \otimes g): \Gamma \otimes \Delta \longrightarrow B}$ \\ \hline
\end{tabular}
\end{center}

Applying these rules a number of times can lead us to composite but reducible forms. These are reduced using a process called $\beta$-reduction as follows:
\begin{eqnarray}
(\lambda x.b)a & \stackrel{\beta}{\longrightarrow} & b[a/x] \nonumber \\
\textrm{ let } a \otimes b \textrm{ be } x \otimes y \textrm{ in } c & \stackrel{\beta}{\longrightarrow} & c[a/x, b/y] \nonumber
\end{eqnarray}

The notation $b[c/x]$ is referred to as substitution and means "take $b$ and replace all free occurrences of $x$ in it with $c$". Substitution is a meta operation that takes place outside of the language. In defining an operational semantics for substitution, care has to be taken to prevent us from violating the capture of free variables\footnote[4]{In the cases where $b$ is a $\lambda$-abstraction, certain provisions have to be in place to prevent us from violating the capture of free variables. These provisions are that $x \neq z$ and $z \not\in FV(c)$.} of terms. A more rigorous way of defining substitution is given by an induction on the structure of $b$:
\begin{align*}
    y\,[c/x]            & :=
    \begin{cases}
        \;c & \text{for $x = y$,} \\
        \;y & \text{for $x \neq y$} \\
    \end{cases} \\
    (pq)[c/x]           & := \; (p[c/x])(q[c/x]) \\
    (\lambda z.q)[c/x]  & := \; \lambda z.(q[c/x])
\end{align*}

A number of interesting translations take place between category theoretic and lambda calculus notions. One of them occurs in $\bigotimes L$, the rule that defines tensor associativity in our category, whereby we gain a rule on how to type a term when, instead of using two variables for its derivation, we plug in a composite term in their place. Another interesting case is Cut, the rule that defines composition of functions, which in the lambda calculus corresponds to a type rule for substitution. Implication elimination ($\multimap E$), the rule responsible for the evaluation operation in closed categories, turns out to be function application. Consequently, the related notion of Currying that comes with the right implication rule ($\multimap R$), corresponds to lambda abstraction.

Similarly to the way we defined rules in the section on Logic, we need to add the structural rule of Exchange, in order to account for the symmetry isomorphism of the tensor:
\begin{center}
\small
$\displaystyle\frac{\Gamma, x:A, y:B, \Delta \vdash c:C}{\Gamma, y:B, x:A, \Delta \vdash c:C}$
\hskip 1cm
$\displaystyle\frac{f: \Gamma \otimes A \otimes B \otimes \Delta \longrightarrow C}{f \circ (id_\Gamma \otimes s_{A,B} \otimes id_\Delta): \Gamma \otimes B \otimes A \otimes \Delta \longrightarrow C}$
\normalsize
\end{center}

\clearpage
\part{Quantum programming and classical control}
\label{Part:Quantum programming and classical control}
\chapter{The dagger lambda calculus}
\label{Chapter:The dagger lambda calculus}
Dagger compact categories were first introduced in \cite{ABP99}, albeit under a different name, using some of the terminology of \cite{DR89}. They were later proposed by \cite{AC04} and \cite{Sel05} as an axiomatic framework for the study of quantum protocols. Though a lot of work has been done on categorically driven quantum programming languages \cite{SV06}, \cite{SV08} and \cite{SV10}, these lambda calculi did not provide a way of modelling the dagger functor of dagger compact categories. The work of \cite{BS10} highlighted the importance of dagger compact categories for the semantics of quantum computation; it presented a rough correspondence between quantum computation, logic and the lambda calculus, yet its type theory fell short of providing a correspondence to the entire structure of dagger compact categories. This section fills this gap by presenting the \textit{dagger lambda calculus}: a computational interpretation for dagger compact categories.

\section{Language construction}
\label{Section:Language construction}
We will now construct a language for \textit{dagger compact categories} by defining well formed formulas for terms, types and sequents. The rules for deriving these formulas will be given in the form of Gentzen-style inference rules. In order to give computational meaning to our language, we will begin our presentation of the typing dynamics by reformalising the linear typed lambda calculus of \cite{AT10} with the explicit substitution used by the linear chemical abstract machine of \cite{Abr93}. The linear negation we will be using causes a significant collapse between conjunction and disjunction, extends tensor to a (potentially) binding operator, and provides us with a semantics similar to that of the proof nets in \cite{AD06}. The set of rules we use to define this language is kept at a minimum, allowing for clean proofs of the various desired properties. Many familiar computational notions do not appear as primitives, but they do arise as constructed notions in good time.

\begin{definition}[Variables, constants and terms in the dagger lambda calculus]
The fundamental building blocks of our language are \textit{variables}; they are denoted by single letters and are traditionally represented using the later letters of the alphabet (i.e. $x,y,z$). We also allow for the use of \textit{constant terms} (i.e. $c_1,c_2,c_3$); these are terms with an inherent value, that cannot serve as placeholders for substitution. These primitives can then be combined with each other to form composite \textit{terms}, denoted by different combinations of the following forms:
\[ \langle term \rangle \;\;\;\; ::= \;\;\;\; variable \;\; | \;\; \langle term \rangle_* \;\; | \;\; \langle term \rangle \otimes \langle term \rangle \;\; | \;\; constant \]
\end{definition}

\vskip 0.1cm

\begin{definition}[Types in the dagger lambda calculus]
Every term in our language, regardless of whether it is a variable, a constant or composite, has a \textit{type}. We will first start by defining a set of \textit{atomic types}; these are traditionally represented using capital letters (i.e. $A,B,C$). Atomic types can then be combined to give us types of the following forms:
\[ \langle type \rangle \;\;\;\; ::= \;\;\;\; atomic \;\;\; | \;\;\; \langle type \rangle^* \;\;\; | \;\;\; \langle type \rangle \otimes \langle type \rangle \]
\end{definition}

\vskip 0.1cm

The star operator that we use is not a repetition operator; instead, it corresponds to a particular form of \textit{linear negation}. As one would expect from a negation operation, the star operator is involutive $(a_*)_* \equiv a$ and $(A^*)^* \equiv A$. Abramsky \cite{Abr93} proposed using linear negation as the passageway between Intuitionistic Linear Logic and Classical Linear Logic. The linear negation used in \cite{AD06} "trivialized" the notion of De Morgan duality of \cite{Abr93} by setting $(A \otimes B)^* := A^* \otimes B^*$. The linear negation that we use is similar to the one used in \cite{CPP08}; it distributes differently over tensor by performing a swap of the terms/types at hand and allows for a more "planar" representation.

\begin{definition}[Linear negation]
The star operator is a form of linear negation whose De Morgan duality is defined by:
\begin{align*}
(a \otimes b)_* := & b_* \otimes a_* && \text{on terms and}\\
(A \otimes B)^* := & B^* \otimes A^* && \text{on types.}
\end{align*}
\end{definition}

\begin{definition}[Scalars]
One of the language's atomic types, denoted by $I$, acts as the tensor unit. One of the very important properties of the type $I$ is \textit{negation invariance}, whereby $I \equiv I^*$. We say that a term $i$ is a \textit{scalar} iff it is of type $I$.
\end{definition}

\begin{definition}[Dimensions]
For every type $A$, we will define a scalar constant $D_A : I$, referring to it as the \textit{dimension} of type $A$. The dimension of $I$ is defined to be $D_I = 1 : I$, where $1 = 1_* : I \equiv I^*$.
\end{definition}

\begin{definition}[Soup]
All of the computation in our language is performed inside a relational \textit{soup}. The soup  is a set, consisting of pairs of equityped terms, connecting them to each other in a form of explicit substitution. A \textit{soup connection} between two terms of type $A$ is written as $t_1 :_A t_2$ or $t_1 : t_2$; to simplify our notation in the soup, we omit writing the type whenever there is no ambiguity about the type of the connected terms. The resulting soup is of the form $S = \{v_1:v_2, \ldots, v_{m-1}:v_m\}$. We use the following property $\{a_1:a_2\} \equiv \{a_{2*}:a_{1*}\}$ to equate some soup terms by collapsing them into the same congruence class. We also define \textit{soup negation} as $(S \cup S^\prime)_* := S_* \cup S^\prime_*$, where $\{t:u\}_* := \{t_*:u_*\}$.
\end{definition}

\begin{definition}[Typing judgements in the dagger lambda calculus]
The \textit{typing judgements}, or \textit{sequents}, of our language are composed of terms, their respective types and a relational soup. A typing judgement is thus represented by:
\[ t_1:A_1, \; t_2:A_2, \; \ldots, \; t_n:A_n \vdash_S t:B \]
\end{definition}

\begin{example}
In the following typing judgement, the types of $t_1$ and $t_2$ are both known to be $A$. Similarly, we know that both $D_C$ and $1$ are scalars, so their type is $I$. We omit writing the types for soup connections $t_1:t_2$ and $D_C:1$ but, to prevent ambiguity, we have to write it for $x :_B x$, because we have no other way of deducing it from the sequent:
\[ t_1:A \vdash_{\{ t_1 : t_2, x :_B x, D_C : 1 \}} t_2 : A \]
\end{example}

Now that we know which formulas are well formed in our language, we can proceed by defining a notion of binding. Contrary to what we are used to from the lambda calculus, where the notion of binding is restricted in scope to the confines of a single term, the dagger lambda calculus supports a binding that is global and whose scope spans the entire typing judgement. The computational interpretation of classical linear logic, which was provided by \cite{Abr93} in his linear chemical abstract machine, views two occurrences of the same variable as two ends of a communication channel. Adhering to the spirit of that definition, we define binding as follows:

\begin{definition}[Bound variables and terms in the dagger lambda calculus]
For any variable $x$, we say that it is a \textit{bound variable} when it appears twice within a given sequent. As such, variable capture is not limited to the scope of a single term but spans the entire sequent. For any term $t$ that does not contain any occurrences of constants, we say that term is captured when it consists entirely of variables that are captured within the scope of the current sequent. We use the phrases \textit{bound term} and \textit{bundle of bound variables} interchangeably when referring to captured terms. Trivially, a bound variable is also a bound term.
\end{definition}
\begin{example}
In the following sequent, $x_1$, $x_2$, $y_1$, $y_2$ and $f$ are all bound variables. The individual variables may be free when looking at subterms $x_{1*} \otimes y_1$ and $x_{2*} \otimes y_2$ but, when considering the scope of the entire sequent, they are captured by other occurrences of themselves in the soup. Moreover, the terms $x_{1*} \otimes y_1$ and $x_{2*} \otimes y_2$ are both bound terms because they contain no constants and they consist solely of variables that are captured by variables in the soup:
\[ x_{1*} \otimes y_1 : A^* \otimes B \vdash_{\{x_{1*} \otimes y_1:f, f:x_{2*} \otimes y_2\}} x_{2*} \otimes y_2 : A^* \otimes B \]

In the following sequent, $f$, $y$, $x_1$ and $x_2$ are bound variables; they can also be viewed as bound terms since a single variable is a term and since they do not contain any constants. The term $x_{1}* \otimes x_1$ is a bundle of bound variables because it contains no constants and consists solely of bound variables. The term $c_* \otimes x_2$, however, is not a bundle of bound variables because it contains a constant called $c$:
\[ f : A^* \otimes B \vdash_{\{ x_{1}* \otimes x_1:c_* \otimes x_2, f:x_{2*} \otimes y \}} y : B \]
\end{example}
\begin{remark}
As will become obvious from our language's sequent rules, which will impose linearity constraints on the introduction of variables, the nature of linearity in our language mandates that all of the variables within a given sequent occur exactly twice. This means that all of the free variables in a given term will occur once more in the sequent within which they reside, hence becoming captured in the scope of that sequent. Within that scope, all terms will essentially consist of captured variables and constants.
\end{remark}

\begin{definition}[$\alpha$-renaming on variables in the dagger lambda calculus]
A bound variable $x$ can be $\alpha$-renamed by replacing all of its instances, in a given sequent, with a bundle of bound variables $t$. The term $t$ has to be of the same type as $x$, must not contain any constants (since it will be a bundle of bound variables), and it must consist of variables that do not already appear in the sequent.
\end{definition}

\noindent We can now extend the operation of $\alpha$-renaming to operate on captured terms:

\begin{definition}[$\alpha$-renaming on terms in the dagger lambda calculus]
A bound term $t$ can be $\alpha$-renamed by either $\alpha$-renaming its constituent variables or, in cases where $t$ appears twice in a given sequent, by replacing all of its instances with a variable $x$. The variable $x$ has to be of the same type as $t$ and it must not already appear in the sequent.
\end{definition}

\begin{definition}[$\alpha$-equivalence in the dagger lambda calculus]
We define a notion of \textit{$\alpha$-equivalence} as the reflexive, symmetric and transitive closure of $\alpha$-renaming. In other words, we say that two sequents are \textit{$\alpha$-equivalent}, or \textit{equivalent up to $\alpha$-renaming}, when one can be transformed to the other by $\alpha$-renaming zero or more terms.
\end{definition}

\begin{example}
Going back to the examples we used earlier, the sequent
\[ x_{1*} \otimes y_1 : A^* \otimes B \vdash_{\{x_{1*} \otimes y_1:f, f:x_{2*} \otimes y_2\}} x_{2*} \otimes y_2 : A^* \otimes B \]
\noindent is $\alpha$-equivalent to
\[ g : A^* \otimes B \vdash_{\{g:f, f:x_{2*} \otimes y_2\}} x_{2*} \otimes y_2 : A^* \otimes B \]
\noindent because we can $\alpha$-rename the bound term $x_{1*} \otimes y_1$ into the variable $g$. Similarly, the sequent
\[ f : A^* \otimes B \vdash_{\{ x_{1}* \otimes x_1:c_* \otimes x_2, f:x_{2*} \otimes y \}} y : B \]
\noindent is $\alpha$-equivalent to
\[ x_{3*} \otimes y_2 : A^* \otimes B \vdash_{\{ x_{1}* \otimes x_1:c_* \otimes x_2, x_{3*} \otimes y_2:x_{2*} \otimes y_1 \}} y_1 : B \]
\noindent because we can $\alpha$-rename the bound variable $y$ into $y_1$ and also $\alpha$-rename the bound variable $f$ into the term $x_{3*} \otimes y_2$.
\end{example}

\begin{definition}[Typing contexts in the dagger lambda calculus]
The left-hand-side of a typing judgement is actually a list of typed terms. We use the letters $\Gamma$ and $\Delta$ as shorthand for arbitrary (possibly empty) lists of such terms. Let $\Delta$ be the list $t_1:T_1, t_2:T_2, \ldots, t_n:T_n$. We define $\bigotimes\Delta$ to be the term $\left(\left(\left(t_1 \otimes t_2\right) \otimes \ldots\right) \otimes t_n\right) : \left(\left(\left(T_1 \otimes T_2\right) \otimes \ldots\right) \otimes T_n\right)$, referring to it as \textit{$\Delta$ in tensor form}.
\end{definition}

Our language exposition features a Gentzen-style Sequent Calculus, which provides us with the inference rules used to produce judgements. Rules with a double line are bidirectional; sequents matching the top of the rule can be used to derive sequents matching the bottom and vice versa. The rules are formed in a way that allows composite terms to appear to the left of the turnstile. The set of sequent rules is:

\begin{prooftree}
\AxiomC{}
\RightLabel{\small Id,}
\UnaryInfC{$x:A \vdash x:A$}
\end{prooftree}

\begin{prooftree}
\AxiomC{$a:A \vdash_S b:B$}
\RightLabel{\small Negation,}
\UnaryInfC{$a_*:A^* \vdash_{S_*} b_*:B^*$}
\end{prooftree}

\begin{prooftree}
\AxiomC{$\Gamma \vdash_{S_1} a:A$}
\AxiomC{$\Delta \vdash_{S_2} b:B$}
\RightLabel{\small $\otimes R$,}
\BinaryInfC{$\Gamma, \bigotimes\Delta \vdash_{S_1 \cup S_2} a \otimes b : A \otimes B$}
\end{prooftree}

\begin{prooftree}
\AxiomC{$\Gamma, a:A, b:B \vdash_S c:C$}
\doubleLine\RightLabel{\small $\otimes L$,}
\UnaryInfC{$\Gamma, a \otimes b : A \otimes B \vdash_S c:C$}
\end{prooftree}

\begin{prooftree}
\AxiomC{$\Gamma \vdash_{S_1} a:A$}
\AxiomC{$a^\prime : A, \Delta \vdash_{S_2} b:B$}
\RightLabel{\small Cut,}
\BinaryInfC{$\Gamma, \Delta \vdash_{S_1 \cup S_2 \cup \{a:a^\prime\}} b:B$}
\end{prooftree}

\begin{prooftree}
\AxiomC{$a:A, \Gamma \vdash_S b:B$}
\doubleLine\RightLabel{\small Curry.}
\UnaryInfC{$\Gamma \vdash_S a_* \otimes b : A^* \otimes B$}
\end{prooftree}

\vskip 0.2cm

\noindent\textit{Linearity constraints:} The identity axiom (Id) is the only inference rule we have for introducing variables into our expressions. Consequently, variables are always introduced as bound pairs. We should note at this point that $\otimes R$ and Cut are the only two rules that can be used to merge two typing judgements. In order to preserve linearity, we have to impose a very important condition on such a merge; in doing so, we will prevent the appearance of more than two instances of a variable in a given sequent. The condition required, in order to merge two sequents, is that they do not share any common variables. Whenever we want to use the $\otimes R$ and Cut rules to merge two sequents whose variables overlap, we have to $\alpha$-rename them first to ensure that the linearity condition is satisfied.

\vskip 0.25cm

We sometimes use sequents with an empty right-hand-side, like $a:A, \Gamma \vdash$, as shorthand for $a:A, \Gamma \vdash 1:I$. Such sequents are easy to produce by using \textit{Uncurrying}, the inverse of the \textit{Curry} rule, together with the constant $1:I$:
\begin{prooftree}
\AxiomC{$\Gamma \vdash a_*:A^*$}
\AxiomC{$\vdash 1:I$}
\RightLabel{\small $\otimes R$}
\BinaryInfC{$\Gamma \vdash a_* \otimes 1 : A^* \otimes I$}
\RightLabel{\small Uncurry}
\UnaryInfC{$a:A, \Gamma \vdash 1:I$}
\end{prooftree}

The language has a structural exchange rule that can be used to swap terms on the left hand side of a sequent. When navigating through a proof tree, instances of the exchange rule can be used to keep track of which terms were swapped and at which points during a derivation:
\begin{prooftree}
\AxiomC{$\Gamma, a:A, b:B, \Delta \vdash c:C$}
\doubleLine\RightLabel{\small Exchange.}
\UnaryInfC{$\Gamma, b:B, a:A, \Delta \vdash c:C$}
\end{prooftree}

Our language also has two unit rules, $\lambda_\Gamma$ and $\rho_\Gamma$, that are used to more accurately represent scalars:
\begin{prooftree}
\AxiomC{$\Gamma \vdash_{S \cup \{i_*:1\}} b:B$}
\doubleLine\RightLabel{\small $\lambda_\Gamma$,}
\UnaryInfC{$i:I, \Gamma \vdash_S b:B$}
\end{prooftree}

\begin{prooftree}
\AxiomC{$\Gamma \vdash_{S \cup \{i_*:1\}} b:B$}
\doubleLine\RightLabel{\small $\rho_\Gamma$.}
\UnaryInfC{$\Gamma, i:I \vdash_S b:B$}
\end{prooftree}

Our language dynamics are defined through soup rules. These rules explain how the relational connections propagate within the soup, giving rise to an operational semantics for a form of "global substitution" that resembles pattern matching on terms. The soup propagation rules, called \textit{bifunctoriality}, \textit{trace} and \textit{cancellation} respectively, are:
\begin{align*}
    S \cup \{a \otimes b : c \otimes d\}    \longrightarrow & \; S \cup \{a:c, b:d\} \\
    S \cup \{x :_A x\}                      \longrightarrow & \; S \cup \{ D_A : 1 \} \\
    S \cup \{1 : 1\}                        \longrightarrow & \; S
\end{align*}

\noindent where $\psi$ is a constant and $x$ is a variable. Our soup rules also contain a \textit{consumption rule}. This rule uses up a relational connection between $\{t : u\}$ to perform a substitution in the typing judgement. Note, however, that the term we are substituting for has to be one that was captured in the scope of the sequent:
\begin{align*}
    \Gamma \vdash_{S \cup \{t : u\}} b:B  \longrightarrow \bigg(\Gamma \vdash_S b:B\bigg) \bigg[ t / u \bigg] && \text{if $u$ does not contain constants,}\\
    \Gamma \vdash_{S \cup \{t : u\}} b:B  \longrightarrow \bigg(\Gamma \vdash_S b:B\bigg) \bigg[ u / t \bigg] && \text{if $t$ does not contain constants.}
\end{align*}

If $t$ and $u$ are both without constants, linearity implies that their constituent variables were all captured in the scope of the original sequent. In such a case, we can choose the way in which we want to substitute. This gives us a symmetric notion of substitution, where our choice of substitution does not affect the typing judgement, as the sequents will be equivalent up to alpha renaming.

\begin{example}
\noindent Consider the following sequent:
\[ f : A^* \otimes B \vdash_{\{ f:c_* \otimes y \}} y : B \]
\noindent The variable $f$ is captured within the scope of the sequent. As such, we can use the \textit{consumption rule} to consume the connection in the soup and substitute $c_* \otimes y$ for $f$ in the rest of the sequent. This will change the sequent to:
\[ c_* \otimes y : A^* \otimes B \vdash y : B \]

\noindent Alternatively, if we had $\alpha$-renamed the original sequent to:
\[ x_{1*} \otimes y_1 : A^* \otimes B \vdash_{\{ x_{1*} \otimes y_1:c_* \otimes y_2 \}} y_2 : B \]
\noindent we could have then used the \textit{bifunctoriality} rule to split the soup connection:
\[ x_{1*} \otimes y_1 : A^* \otimes B \vdash_{\{ x_{1*}:c_*, y_1:y_2 \}} y_2 : B \]
\noindent The first connection of the resulting soup is only consumable in one way, since $c$ is a constant, by substituting $c_*$ for $x_{1*}$. The second soup connection, however, presents us with a choice, since both $y_1$ and $y_2$ are captured in the sequent. One choice will give us
\[ c_* \otimes y_2 : A^* \otimes B \vdash y_2 : B \]
\noindent while the other choice will give us
\[ c_* \otimes y_1 : A^* \otimes B \vdash y_1 : B \]

\noindent Upon closer inspection, one will notice that all three of the resulting sequents are $\alpha$-equivalent.
\end{example}

\begin{definition}[Soup reduction]
We use the term \textit{soup reduction} to refer to the binary relation that extends $\alpha$-equivalence with the sequent transformations that are caused by applying one of the soup rules. Thus, for two sequents $\Gamma \vdash_{S_1} t : T$ and $\Gamma \vdash_{S_2} t : T$, if the soup $S_1$ is transformed into $S_2$ through the application of one of the soup propagation rules, $S_1 \rightarrow S_2$, then we say that one sequent reduces to the other via \textit{soup reduction}. Similarly, if a sequent $J_1$ is transformed into $J_2$ by using the consumption rule to perform a substitution, we say that $J_1$ reduces to $J_2$ via \textit{soup reduction}.
\end{definition}

\begin{definition}[Soup equivalence]
We define a notion of \textit{soup equivalence} as the reflexive, symmetric and transitive closure of soup reduction. In other words, we say that two sequents $J_1$ and $J_2$ are \textit{soup-equivalent}, or \textit{equivalent up to soup-reduction}, when we can convert one to the other by using zero or more instances of $\alpha$-renaming and soup reduction.
\end{definition}

\noindent We can now use the rules that we have defined so far in order to express the computational notion of application:
\begin{definition}[Application in the dagger lambda calculus]
Let $t$ and $f$ be terms such that $t:A$ and $f:A^* \otimes B$ for some types $A$ and $B$. We define the \textit{application} $ft$ as representing a variable $x:B$, along with a connection in our soup. The origins of the application affect the structure of its corresponding soup connection:
\begin{align*}
ft:B, \Gamma \vdash c:C \;\;    & := \;\; x:B, \Gamma \vdash_{\{f:t_* \otimes x\}_*} c:C \\
\Gamma \vdash ft:B \;\;         & := \;\; \Gamma \vdash_{\{f:t_* \otimes x\}} x:B
\end{align*}
For an application originating inside our soup, we have:
\begin{align*}
\{ ft:c \} & := \{ x:c \} \cup \{ f:t_* \otimes x \} \\
\{ c:ft \} & := \{ c:x \} \cup \{ f:t_* \otimes x \}_*
\end{align*}
\end{definition}

\begin{corollary}[Beta reduction]
This immediately allows us to represent a form of \textit{beta reduction}. Instead of relying on an implicit meta concept of substitution, our beta reduction is going to express the binding and reduction of terms by connecting them in the context soup by setting $(a_* \otimes b)t \stackrel{\beta}{\longrightarrow} b$ while causing $\{t:a\}$ or $\{t:a\}_*$ to be added to the relational soup.
\end{corollary}
\begin{proof}
This is derived from our definition of application because $(a_* \otimes b)t$ represents a variable $x$ along with one of two possible connections in our soup. The soup connection can be manipulated into:
\[ \{a_* \otimes b : t_* \otimes x\}   \rightarrow \{a_*:t_*, b:x\}    \rightarrow \{t:a\} \cup \{b:x\} \]
\[ \{a_* \otimes b : t_* \otimes x\}_* \rightarrow \{a_*:t_*, b:x\}_*  \rightarrow \{t:a\}_* \cup \{x:b\} \]
The connection between $b$ and $x$ can then be consumed to change the variable $x$ into a $b$. All that remains is $\{t:a\}$ or $\{t:a\}_*$.
\end{proof}

Now that all of the language's rules are in place, we can make up for its apparent lack of a $\lambda$ operator by defining it to be a notational shorthand:

\begin{definition}[Lambda abstraction in the dagger lambda calculus]
$\lambda a.b := a_* \otimes b$ and $A \multimap B := A^* \otimes B$
\end{definition}

\noindent The following combinators are used in the rest of this dissertation:
\begin{align*}
    id_A    & := \lambda a.a \textrm{ (where $a:A$)}\\
    \bar{b} & := \lambda g.\lambda f.\lambda a.g(fa)\\
    \bar{s} & := \lambda (a \otimes b).(b \otimes a)\\
    \bar{t} & := \lambda f.\lambda g.\lambda(x_1 \otimes x_2).(fx_1 \otimes gx_2)
\end{align*}

\begin{theorem}[Admissibility of $\multimap E$]
We can also use the definition of application to demonstrate that an implication elimination rule ($\multimap E$) is admissible within our set of rules:
\end{theorem}
\begin{prooftree}
    \AxiomC{$\Gamma \vdash_{S_1} t:A$}
    \AxiomC{$\Delta \vdash_{S_2} f:A^* \otimes B$}
    \AxiomC{}
    \UnaryInfC{$a:A \vdash a : A$}
    \UnaryInfC{$a_*:A^* \vdash a_*:A^*$}
    \AxiomC{}
    \UnaryInfC{$b:B \vdash b : B$}
    \BinaryInfC{$a_*:A^*, b:B \vdash a_* \otimes b:A^* \otimes B$}
    \UnaryInfC{$a_* \otimes b:A^* \otimes B \vdash a_* \otimes b:A^* \otimes B$}
    \RightLabel{\footnotesize Cut}
    \BinaryInfC{$\Delta \vdash_{S_2 \cup \{f:a_* \otimes b\}} a_* \otimes b:A^* \otimes B$}
    \RightLabel{\footnotesize Uncurry}
    \UnaryInfC{$a:A, \Delta \vdash_{S_2 \cup \{f:a_* \otimes b\}} b:B$}
    \RightLabel{\footnotesize Cut}
    \BinaryInfC{$\Gamma, \Delta \vdash_{S_1 \cup S_2 \cup \{t:a, f:a_* \otimes b\}} b:B$}
    \UnaryInfC{$\Gamma, \Delta \vdash_{S_1 \cup S_2 \cup \{f:t_* \otimes b\}} b:B$}
    \UnaryInfC{$\Gamma, \Delta \vdash_{S_1 \cup S_2} ft:B$}
\end{prooftree}

We define some additional notational conventions, so that we can more easily describe the reversal in the causal order of computation:
\begin{definition}[Complex conjugation]
Let $f: A^* \otimes B$ be an arbitrary function. As a notational convention, we set $f^* := \bar{s} f : B \otimes A^*$.
\end{definition}

\begin{theorem}[Admissibility of $\dag$-flip]
We can use the language's rules and definitions in order to admit a new structural rule called the $\dag$-flip. This rule contains all the computational symmetry that we will later need in order to model the \textit{dagger functor}:
\end{theorem}
\begin{prooftree}
    \AxiomC{$a:A \vdash_{S} b:B$}
    \RightLabel{\footnotesize Negation}
    \UnaryInfC{$a_*:A^* \vdash_{S_*} b_*:B^*$}
    \RightLabel{\footnotesize Uncurry}
    \UnaryInfC{$b:B, a_*:A^* \vdash_{S_*}$}
    \RightLabel{\footnotesize Exchange}
    \UnaryInfC{$a_*:A^*, b:B \vdash_{S_*}$}
    \RightLabel{\footnotesize Curry}
    \UnaryInfC{$b:B \vdash_{S_*} a:A$}
\end{prooftree}

\begin{theorem}[Interchangeability of $\dag$-flip and Negation]
Alternatively, we could have defined the language by including $\dag$-flip in our initial set of sequent rules. That would have allowed us to admit the Negation rule as a derived rule:
\end{theorem}
\begin{prooftree}
    \AxiomC{$a:A \vdash_{S} b:B$}
    \RightLabel{\footnotesize $\dag$-flip}
    \UnaryInfC{$b:B \vdash_{S_*} a:A$}
    \RightLabel{\footnotesize Uncurry}
    \UnaryInfC{$a_*:A^*, b:B \vdash_{S_*}$}
    \RightLabel{\footnotesize Exchange}
    \UnaryInfC{$b:B, a_*:A^* \vdash_{S_*}$}
    \RightLabel{\footnotesize Curry}
    \UnaryInfC{$a_*:A^* \vdash_{S_*} b_*:B^*$}
\end{prooftree}

\subsection{Scalars}
\label{Subsection:Scalars}
Similarly to the attachable monoid that is described in \cite{Abr05} for multiplying scalars, we can optionally define a multiplication operation for the scalars in the dagger lambda calculus. This is not part of the structure that is necessary to model dagger compact categories computationally, hence the designation \textit{optional}, but it does provide a good example of how connections propagate in the soup:
\begin{definition}[Scalar multiplication]
For any two scalars $m:I$ and $n:I$, we define a multiplication operation $m \cdot n : I$ such that:
\[ m \cdot 1 = 1 \cdot m = m \]
\noindent and
\[ \{ m \cdot p : n \cdot q \} := \{ m:n, p:q \} \]
\end{definition}

\noindent The operation features a number of properties. To help the reader get more accustomed to the way things propagate in the soup, we will demonstrate some of them as an example. First of all, scalar multiplication is \textit{associative}:
\begin{lemma}[Associativity of multiplication]
$ (a \cdot b) \cdot c = a \cdot (b \cdot c) $
\end{lemma}
\begin{proof}
\begin{align*}
    \{ (a \cdot b) \cdot c : 1 \}   & = \{ (a \cdot b) \cdot c : (1 \cdot 1) \cdot 1 \} \\
                                    & = \{ a:1, b:1, c:1 \} \\
                                    & = \{ a \cdot (b \cdot c) : 1 \cdot (1 \cdot 1) \} \\
                                    & = \{ a \cdot (b \cdot c) : 1 \}
\end{align*}
\end{proof}

\noindent The multiplication operation is also \textit{commutative}:
\begin{lemma}[Commutativity of multiplication]
$m \cdot n = n \cdot m$
\end{lemma}
\begin{proof}
\begin{align*}
    \{ m \cdot n : 1 \} & = \{ m \cdot n : 1 \cdot 1 \} \\
                        & = \{ m:1, n:1 \} \\
                        & = \{ n:1, m:1 \} \\
                        & = \{ n \cdot m : 1 \cdot 1 \} \\
                        & = \{ n \cdot m : 1 \}
\end{align*}
\end{proof}

\noindent It is \textit{sesquilinear}:
\begin{lemma}[Sesquilinearity of scalar connections]
$ \{ m : n \} = \{ m \cdot n_* : 1 \} $
\end{lemma}
\begin{proof}
\begin{align*}
    \{ m : n \} & = \{ m \cdot 1 : 1 \cdot n \} \\
                & = \{ m:1, 1:n \} \\
                & = \{ m:1, n_*:1 \} \\
                & = \{ m \cdot n_* : 1 \cdot 1 \} \\
                & = \{ m \cdot n_* : 1 \} \\
\end{align*}
\end{proof}

\noindent Finally, it is easy to deduce that the dimension of a tensor of types distributes into a product of dimensions:
\begin{corollary}[Dimension multiplication]
$ \{ D_A \cdot D_B : 1 \} = \{ D_{A \otimes B} : 1 \} $
\end{corollary}
\begin{proof}
\begin{align*}
    \{ D_A \cdot D_B : 1 \} & = \{ D_A:1, D_B:1 \}                              & = \{ a :_A a, b :_B b \} \\
                            & = \{ a \otimes b :_{A \otimes B} a \otimes b \}   & = \{ D_{A \otimes B} : 1 \}
\end{align*}
\end{proof}

\section{Proofs of properties}
\label{Section:Proofs of properties}
Many lambda calculi suffer from being complicated, which makes it hard to prepare and follow proofs about their properties. Our language is tractable and consists of a minimal set of rules. As a result, most of the language's properties are easy to prove by structural inductions. Throughout the rest of this section, we prove that our lambda calculus satisfies most of the really important properties a calculus can have, namely subject reduction, confluence, strong normalisation and consistency.

\subsection{Subject reduction}
\label{Subsection:Subject reduction}
The first thing we have to prove, in order to demonstrate that our typing system is well defined, is the consistency of our typing dynamics. In other words, we have to verify that the way in which relational connections propagate through our soup preserves type assignments. This is easy to observe because our soup only connects \textit{equityped} terms. Pair consumption substitutes a term for another of the same type, thus preserving types.

\begin{theorem}[Subject reduction]
Let $J_1$ and $J_2$ be two typing judgements such that $J_1 = \Gamma \vdash_S t_1:A_1$ and $J_2 = \Delta \vdash_{S^\prime} t_2:A_2$. Suppose that these two judgements are such that we can use a soup reduction rule $S \longrightarrow S^\prime$ to reduce one to the other: $J_1 \longrightarrow J_2$. Then, the reduction will not alter type assignments in any way: $types(\Gamma) = types(\Delta)$ and $A_1 \equiv A_2$.
\end{theorem}
\begin{proof}
Due to the way the dagger lambda calculus was designed, the proof of subject reduction will be trivial. We will prove this by induction on the rules of the soup reduction. There are four different rules that could be used when performing a soup reduction:
\begin{itemize}
  \item $S \cup \{a \otimes b : c \otimes d\} \longrightarrow \; S \cup \{a:c, b:d\}$ (the \textit{bifunctoriality} rule);
  \item $S \cup \{x :_A x\} \longrightarrow \; S \cup \{D_A : 1\}$ (the \textit{trace} rule);
  \item $S \cup \{1 : 1\} \longrightarrow \; S$ (the \textit{cancellation} rule);
  \item The \textit{consumption} rule.
\end{itemize}
If the reduction is an instance of one of the first three rules, then the theorem holds trivially; $types(\Gamma) = types(\Delta)$ and $A_1 \equiv A_2$ since the bifunctoriality, trace and cancellation rules do not alter anything outside of the soup. If the soup reduction is an instance of the consumption rule, then a soup connection will be consumed to substitute a term at the other end of a bounded pair of variables. The substitution may be global in scope, but it does not affect the sequent's typing, since it is substituting one term for another one of the same type. The act of consumption itself does not affect the typing of the sequent either, since it removes a connection from the soup without affecting the terms outside. Therefore, regardless of the soup rule used, soup reduction has no effect on the typing of terms outside of the soup. This ensures that the typing dynamics of the dagger lambda calculus will be consistent.
\end{proof}

\subsection{Normalisation}
\label{Subsection:Normalisation}
Strong normalisation is a highly sought after property for lambda calculi, primarily because of the implications it has on the practical implementation of the language. A reduction that is strongly normalising implies that every sequent has a normal form. Furthermore, it requires that the normal form is attained after a finite number of steps, without any chance of running into an infinite reduction loop. We now prove that the dagger lambda calculus has this property:

\begin{theorem}[Strong normalisation]
Every sequence of soup reduction steps is finite and ends with a typing judgement that is in normal form.
\end{theorem}
\begin{proof}
We begin by proving a simplified form of the theorem, where the terms connected in our soup are all of atomic type. Obviously, in a setting like this, the first reduction rule for our soup would never be used since there would be no tensored terms. The simplified theorem can be proved by induction on the length of the soup. In judgements where the soup is empty, or only contains pairs of constants that are not usable, it is obvious that reduction cannot proceed any further and that we have already reached a normal form in a finite (actually zero) number of steps. For a reduction soup with at least one usable connection, there are three possible ways this could go:
\begin{enumerate}[(a)]
  \item A pair of the form $\{x :_A x\}$ can be transformed into $\{ D_A : 1 \}$, a scalar reference of the dimension of $A$. This will either give us an unusable soup connection or, if the type $A$ is $I$, it will give us $\{1:1\}$, which can later be thrown away by using the cancellation rule. Either way, since our soup has a finite length, we will be left with a smaller usable soup. Hence, by the induction hypothesis, a normal form is attainable after finitely many steps.
  \item A pair of the form $\{1:1\}$ can be thrown out of the soup as it does not contribute anything to our sequent. In doing so, since our soup has a finite length, we are left with a smaller soup. Hence, by the induction hypothesis, a normal form is attainable after finitely many steps.
  \item A pair of the form $\{t:u\}$, where $t$ or $u$ is bound, can be consumed to perform a substitution. In doing so, since our soup has a finite length, we are again left with a smaller soup. By the induction hypothesis, we can attain a normal form after finitely many steps.
\end{enumerate}

This completes our proof of the simplified form of the theorem. In order to prove the full theorem, we perform an induction on the structure of the relational soup:
\begin{enumerate}[(i)]
  \item A relational soup that does not contain connections between tensored terms leads to a normal form after finitely many steps. This was already proved in the simplified version of this theorem.
  \item A connection between two tensors $\{a \otimes b : c \otimes d\}$ can be reduced into two distinct connections $\{a:c, b:d\}$. The resulting pairs consist of subterms of the original pair. Hence, by the induction hypothesis, we can attain a normal form after finitely many steps.
\end{enumerate}
\end{proof}

\subsection{Confluence}
\label{Subsection:Confluence}
Another very important property for our language is the Church-Rosser property. It ensures that we can end up with the same sequent regardless of the reduction path we choose to follow. A careful observation of our rewrite rules will reveal that the rules are all left-linear.
\begin{lemma}[Left-linearity]
All of our soup rewrite rules are left-linear.
\end{lemma}
\begin{proof}
In accordance with the linearity constraints of our language, no variable appears more than twice on the left hand side of any of our soup reduction rules:
\begin{align*}
    S \cup \{a \otimes b : c \otimes d\}    \longrightarrow & \; S \cup \{a:c, b:d\} \\
    S \cup \{x :_A x\}                      \longrightarrow & \; S \cup \{D_A : 1\} \\
    S \cup \{1 : 1\}                        \longrightarrow & \; S \\
    S \cup \{t : u\}                        \longrightarrow & \; S \;\;\; \text{(if $t$ or $u$ is a bound variable)}
\end{align*}
\end{proof}

One should note, at this point, that our soup rules do exhibit a form of "harmless" overlap. More specifically, the consumption rule ($S \cup \{t:u\} \longrightarrow S$) forms a critical pair with itself in cases where $t$ and $u$ are both bound. Fortunately, as we will see in the next lemma, these pairs can be proved to be \textit{trivial} as they correspond to sequents that are equivalent up to $\alpha$-renaming.

\begin{lemma}[Symmetry of substitution]
Let $J$ be a typing judgement of the form $J := \Gamma \vdash_{S \cup \{t:u\}} a:A$, where $t$ and $u$ are both bound. The connection $\{t:u\}$ can be consumed in either of two ways; one substitutes $t$ for $u$ and the other substitutes $u$ for $t$ in the typing judgement. Let's call these $J_1$ and $J_2$ respectively. $J_1$ will then be $\alpha$-equivalent to $J_2$.
\end{lemma}
\begin{proof}
We know that
\[ J \longrightarrow J_1 := \bigg(\Gamma \vdash_S a:A\bigg) \bigg[ t / u \bigg] \]
\[ J \longrightarrow J_2 := \bigg(\Gamma \vdash_S a:A\bigg) \bigg[ u / t \bigg] \]

Since $t$ and $u$ are both bound, by linearity, we know that they appear exactly once in $\Gamma \vdash_S a:A$. After substitution is performed, $J_1$ will have two occurrences of $t$ where $t$ and $u$ used to be, so $t$ will be a bound term in that judgement. Similarly, $J_2$ will have two occurrences of $u$ where $t$ and $u$ used to be, so $u$ will be a bound term in that judgement. These bound terms occur in the exact same spots, so we can $alpha$-rename $J_1$ to $J_2$ and vice versa.
\end{proof}

\begin{corollary}[No overlap]
The rewrite rules have no overlap up to $\alpha$-equality of typing judgements.
\end{corollary}

\begin{theorem}[Confluence]
Our reduction rules have the Church-Rosser property.
\end{theorem}
\begin{proof}
Our set of rewrite rules is \textit{left-linear} and has no significant overlap, since it only gives rise to critical pairs that are \textit{trivial} up to $\alpha$-equivalence. Therefore, our rewrite rules constitute a \textit{weakly orthogonal} rewrite system, which is \textit{weakly confluent} according to \cite{Klo92}. Since the rewrite system is both strongly normalising and weakly confluent, we can use Newman's lemma to conclude that it also possesses the Church-Rosser property. See \cite{Klo92} for a more detailed explanation of the properties of orthogonal rewriting systems.
\end{proof}

\subsection{Consistency}
\label{Subsection:Consistency}
In order to show that our type theory is consistent, we have to show that our soup dynamics do not collapse all equityped terms to the same element.

\begin{theorem}[Consistency]
There exist two terms of the same type, henceforth referred to as $t_1$ and $t_2$, such that $\Gamma \vdash_{S_1} t_1:A$ and $\Gamma \vdash_{S_2} t_2:A$ could never reduce to the same typing judgement.
\end{theorem}
\begin{proof}
Consider two combinators of the same type, $t_1 = id_{A \otimes A}$ and $t_2 = \bar{s}_{A \otimes A}$. Both terms are closed, containing no free variables or constants. The sequents $\vdash id_{A \otimes A} : (A \otimes A) \multimap (A \otimes A)$ and $\vdash \bar{s}_{A \otimes A} : (A \otimes A) \multimap (A \otimes A)$ are distinct normal forms: They are clearly distinct from one another and cannot be further reduced using any of our rules, thereby proving that they could never reduce to the same typing judgement.
\end{proof}

\section{Correspondence to dagger compact categories}
\label{Section:Correspondence to dagger compact categories}
The purpose of this section is to provide a full Curry-Howard-Lambek correspondence between the dagger lambda calculus and dagger compact categories. We start by defining a directed graph $\mathcal{G}$, representing a signature for dagger compact categories. We then show how that graph can be interpreted to define the free dagger compact category $\mathcal{C}_{Free}$ and the dagger lambda calculus $\dag\lambda$. An appropriate Cut-elimination procedure is defined to partition the sequents of the dagger lambda calculus into equivalence classes up to soup equivalence. The resulting equivalence classes are modular proof invariants represented by denotations. We show that the types and denotations can be used to form a syntactic category, $\mathcal{C}_{Synt}$, and prove that the category is dagger compact. The diagram below, fashioned to resemble the diagram at the bottom of page 49 in \cite{Mac98}, is provided to help visualise the Curry-Howard-Lambek correspondence. In this diagram, $U\mathcal{C}_{Free}$ and $U\mathcal{C}_{Synt}$ are the underlying graphs of their respective categories, where identities, composition, natural isomorphisms and other structural elements of the parent categories have been "forgotten" by applying the forgetful functor $U$. $F$ is the unique functor between the free and the syntactic category, that satisfies the rest of the conditions in the diagram.
\[
% Generated with LaTeXDraw 2.0.8
% Thu Nov 03 23:39:58 GMT 2011
% \usepackage[usenames,dvipsnames]{pstricks}
% \usepackage{epsfig}
% \usepackage{pst-grad} % For gradients
% \usepackage{pst-plot} % For axes
\scalebox{1} % Change this value to rescale the drawing.
{
\begin{pspicture}(0,-1.7217188)(12.894688,1.7217188)
\usefont{T1}{ptm}{m}{n}
\rput(4.5242186,1.5007813){\large $\mathcal{C}_{Free}$}
\usefont{T1}{ptm}{m}{n}
\rput(4.604219,-1.3192188){\large $\mathcal{C}_{Synt}$}
\usefont{T1}{ptm}{m}{n}
\rput(1.2442187,-1.3192188){\large $\dag\lambda$}
\psline[linewidth=0.04cm,arrowsize=0.05291667cm 2.0,arrowlength=1.4,arrowinset=0.4]{->}(1.576875,-1.2542187)(3.976875,-1.2542187)
\usefont{T1}{ptm}{m}{n}
\rput(2.6482813,-1.5442188){$\ell$}
\psline[linewidth=0.04cm,arrowsize=0.05291667cm 2.0,arrowlength=1.4,arrowinset=0.4]{->}(4.476875,1.2457813)(4.476875,-1.0542188)
\usefont{T1}{ptm}{m}{n}
\rput(4.798281,0.05578125){$!F$}
\usefont{T1}{ptm}{m}{n}
\rput(8.174219,1.5007813){\large $U\mathcal{C}_{Free}$}
\usefont{T1}{ptm}{m}{n}
\rput(8.254219,-1.3192188){\large $U\mathcal{C}_{Synt}$}
\psline[linewidth=0.04cm,arrowsize=0.05291667cm 2.0,arrowlength=1.4,arrowinset=0.4]{->}(7.976875,1.2457813)(7.976875,-1.0542188)
\usefont{T1}{ptm}{m}{n}
\rput(8.368281,0.05578125){$UF$}
\usefont{T1}{ptm}{m}{n}
\rput(11.514218,1.5007813){\large $\mathcal{G}$}
\psline[linewidth=0.04cm,arrowsize=0.05291667cm 2.0,arrowlength=1.4,arrowinset=0.4]{->}(11.276875,1.5457813)(8.776875,1.5457813)
\psline[linewidth=0.04cm,arrowsize=0.05291667cm 2.0,arrowlength=1.4,arrowinset=0.4]{->}(11.276875,1.2457813)(8.276875,-1.0542188)
\end{pspicture}
}
\]

We will prove an equivalence between the free category and the syntactic category. We should note at this point that our typing conventions of an involutive negation ($A \equiv (A^*)^*$) and negation invariance of the tensor unit ($I \equiv I^*$) implicitly introduce equivalence classes on types. Our proof of equivalence will be achieved by fully exhibiting the correspondence in objects and arrows between the two categories, up to the equivalence classes that are induced by our typing conventions.

\subsection{A signature for dagger compact categories}
\label{Subsection:A signature for dagger compact categories}
The notion of signature we will be using combines that of the algebraic signature of \cite{Sel10} with the directed graph used by \cite{Mac98}. Consider a set of object variables $\Sigma_0$. Using the tensor operation, an associated tensor identity and the duality operator star, we can construct the free $(\otimes, I, \Box^*)$-algebra over $\Sigma_0$. This corresponds to the set of all object terms or vertices in a compact closed category and will be denoted by $Dagger(\Sigma_0)$. Now consider a set $\Sigma_1$ of morphism variables or edges between those vertices. Let $dom, cod$ be a pair of functions such that $dom,cod: \Sigma_1 \longrightarrow Dagger(\Sigma_0)$. Throughout the rest of this section, we will be referring to the graph $\mathcal{G}$ as the directed graph whose vertices and edges are defined by $Dagger(\Sigma_0)$ and $\Sigma_1$. This graph forms the signature upon which we will base both the dagger lambda calculus and our description of the free dagger compact category; it includes all of the symbols but none of the logic of the languages that we want to describe.

\subsection{The free dagger compact category}
\label{Subsection:The free dagger compact category}
We will now show how to define the free dagger compact category $\mathcal{C}_{Free}$ as an interpretation of the graph $\mathcal{G}$. A highly intuitive introduction to free categories and how they can be generated from directed graphs can be found in \cite{Mac98}. Furthermore, a more extensive presentation of the process of constructing of various kinds of free categories can be found in \cite{Sel10}. A more detailed presentation of the incremental buildup to the construction of free dagger compact categories can also be found in \cite{Abr05}.

The set of objects for the free category in this section will be the same as the set of vertices $Dagger(\Sigma_0)$ in the graph $\mathcal{G}$. The set of edges $\Sigma_1$ in the graph is used to generate morphisms for the free category. Thus, an edge of the form $f : A \rightarrow B$ generates an arrow in $\mathcal{C}_{Free}$ which we will denote as $\langle A,f,B \rangle$. The free category over a directed graph, also referred to as a path category, includes morphisms that correspond to the paths generated by combining adjoining edges in $\mathcal{G}$. These morphisms are formed using the free category's composition operation. Given two morphisms $\langle A,f,B \rangle$ and $\langle B,g,C \rangle$, we write their composition in $\mathcal{C}_{Free}$ as $\langle A,f,B,g,C \rangle$.

Since the free category is a monoidal category, it allows us to consider two of the graph's edges concurrently by bringing together their corresponding categorical morphisms using a monoidal tensor product. Given two morphisms $\langle A,f,B \rangle$ and $\langle C,h,D \rangle$, we write their tensor product as $\langle A \otimes C, f \otimes h, B \otimes D \rangle$.

The free category generated by the graph $\mathcal{G}$ also includes a number of morphisms that are part of the dagger compact logical structure. The identities are represented by:
\[ \langle A \rangle, \langle B \rangle, \langle C \rangle, \ldots \]

The monoidal natural isomorphisms are written as:
\[ \langle A \otimes (B \otimes C), \alpha_{A,B,C}, (A \otimes B) \otimes C \rangle \]
\[ \langle I \otimes A, \lambda_A, A \rangle \]
\[ \langle A \otimes I, \rho_A, A \rangle \]

The symmetry isomorphism is written as:
\[ \langle A \otimes B, \sigma_{A,B}, B \otimes A \rangle \]

And the units and counits are written as:
\[ \langle I, \eta_A, A^* \otimes A \rangle \]
\[ \langle A \otimes A^*, \varepsilon_A, I \rangle \]

For every map $\langle A,f,B \rangle$ in the free category, the dagger compact logical structure contains maps $f_*$ and $f^\dag$, represented by $\langle A^*,f_*,B^* \rangle$ and $\langle B,f^\dag,A \rangle$ respectively. When acting on compositions of paths, such as $\langle A,f,B,g,C,\ldots,X,h,Y,t,Z \rangle$, the dagger operator reverses the order of operations, yielding:
\[ \langle Z,t^\dag,Y,h^\dag,X,\ldots,C,g^\dag,B,f^\dag,A \rangle \]

\subsection{The dagger lambda calculus}
\label{Subsection:The dagger lambda calculus}
This section demonstrates how the graph signature $\mathcal{G}$ can be interpreted to derive the dagger lambda calculus. The set of types used by $\dag\lambda$ is precisely the set of vertices $Dagger(\Sigma_0)$ used in graph $\mathcal{G}$. Every edge
\[ f : A \rightarrow B \]
\noindent in $\Sigma_1$ is interpreted as a sequent
\[ a:A \vdash_{\{f:a_* \otimes b\}} b:B \]
\noindent up to alpha-equivalence. These interpretations essentially introduce constants, in our case $f:A^* \otimes B$, written as sequents that are reminiscent of $\eta$-expanded forms. The rest of the rules of the dagger lambda calculus can be used to process and combine sequents, yielding a richer logical structure.

\subsection{The syntactic category}
\label{Subsection:The syntactic category}
Following a method that is similar to \cite{Mel09}, we will define a process of Cut-elimination by using the soup reduction relation to partition the sequents of the dagger lambda calculus into equivalence classes. The resulting equivalence classes are modular proof invariants called \textit{denotations}. This section demonstrates how these denotations give rise to a dagger compact categorical structure $\mathcal{C}_{Synt}$ called the \textit{syntactic category}.

\begin{definition}[Denotations]
We will use the term \textit{denotations} to refer to the equivalence classes that are formed by partitioning the sequents of the lambda calculus according to soup equivalence. Hence, two sequents will correspond to the same denotation if and only if they are equivalent up to soup reduction.
\end{definition}

\begin{theorem}[The syntactic category]
The types of the lambda calculus and the denotations generated by soup equivalence form a category whose objects are types and whose arrows are denotations.
\end{theorem}
\begin{proof}
As we noticed during the proof of the subject reduction property, soup reduction rules do not affect our language's type assignments. Consequently, the type of the premises used by a sequent will be the same across all sequents in a given denotation. Similarly, the type of the conclusion produced by a sequent will be the same across all sequents in a given denotation. For any sequent $\Gamma \vdash_S b:B$, corresponding to a denotation $[\pi_1]$, we will say that its \textit{domain} is $\Gamma$ and its \textit{codomain} is $B$, writing this as $[\pi_1] : \Gamma \rightarrow B$.

Let $[f]:A \rightarrow B$ and $[g]:B \rightarrow C$ be denotations representing the soup equivalent forms of some sequents $a:A \vdash_{S_1} b:B$ and $b^\prime:B \vdash_{S_2} c:C$ respectively. For any two such denotations, where the codomain of the first matches the domain of the second, we will define a \textit{composition} operator $\circ$ that can combine them into $[g] \circ [f] : A \rightarrow C$. The new denotation will represent all the soup equivalent forms of the sequent that is generated by combining the two sequents using the Cut rule:
\begin{prooftree}
    \AxiomC{$a:A \vdash_{S_1} b:B$}
    \AxiomC{$b^\prime:B \vdash_{S_2} c:C$}
    \RightLabel{\footnotesize Cut}
    \BinaryInfC{$a:A \vdash_{S_1 \cup S_2 \cup \{b : b^\prime\}} c:C$}
\end{prooftree}

The composition operation we just defined inherits associativity from the Cut rule; the order in which Cuts are performed does not matter since the connected terms are allowed to "float" freely within the soup. Therefore, $[h] \circ ([g] \circ [f]) = ([h] \circ [g]) \circ [f]$. Moreover, for every type $A$, there is a denotation $[id_A]$ that represents the sequent generated by the Identity axiom (Id): $x:A \vdash x:A$.

Composing a denotation $[f]:A \rightarrow B$ with an identity yields $[f] \circ [id_A]$ or $[id_B] \circ [f]$ depending on whether we compose with an identity on the right or on the left. The two resulting denotations represent

\begin{center}
\AxiomC{$x:A \vdash x:A$}
\AxiomC{$a:A \vdash_S b:B$}
\BinaryInfC{$x:A \vdash_{S \cup \{x : a\}} b:B$}
\DisplayProof
\hskip 0.3cm and \hskip 0.3cm
\AxiomC{$a:A \vdash_S b:B$}
\AxiomC{$x:B \vdash x:B$}
\BinaryInfC{$a:A \vdash_{S \cup \{b : x\}} x:B$}
\DisplayProof
\end{center}

\noindent both of which are soup equivalent to $a:A \vdash_S b:B$ and the rest of the sequents represented by $[f]$. Hence $[id_B] \circ [f] = [f] = [f] \circ [id_A]$
\end{proof}

We call this category the \textit{syntactic category}. We will now incrementally check that it satisfies the criteria needed to be a dagger compact category.

\begin{definition}[Syntactic category notational conventions]
For notational convenience, we define the following combinators:
\[ \alpha_{A,B,C} := \lambda\left(a \otimes (b \otimes c)\right).\left((a \otimes b) \otimes c\right) : \left(A \otimes (B \otimes C)\right) \multimap \left((A \otimes B) \otimes C\right) \]
\[ \alpha_{A,B,C}^{-1} := \lambda\left((a \otimes b) \otimes c\right).\left(a \otimes (b \otimes c)\right) : \left((A \otimes B) \otimes C\right) \multimap \left(A \otimes (B \otimes C)\right) \]
\[ \lambda_A := \lambda(1 \otimes a).a : (I \otimes A) \multimap A \]
\[ \lambda_A^{-1} := \lambda a.(1 \otimes a) : A \multimap (I \otimes A) \]
\[ \rho_A := \lambda(a \otimes 1).a : (A \otimes I) \multimap A \]
\[ \rho_A^{-1} := \lambda a.(a \otimes 1) : A \multimap (A \otimes I) \]
\[ \sigma_{A,B} := \lambda(a \otimes b).(b \otimes a) : (A \otimes B) \multimap (B \otimes A) \]
\[ \sigma_{A,B}^{-1} := \sigma_{B,A} = \lambda(b \otimes a).(a \otimes b) : (B \otimes A) \multimap (A \otimes B) \]
\[ \eta_A := \lambda 1.(x_* \otimes x) : I \multimap (A^* \otimes A) \]
\[ \varepsilon_A := \lambda(x \otimes x_*).1 : (A \otimes A^*) \multimap I \]
\end{definition}

\begin{theorem}[Monoidal category]
The syntactic category is a monoidal category
\end{theorem}
\begin{proof}
Let $[f]:A \rightarrow B$ and $[g]:C \rightarrow D$ be denotations representing the soup equivalent forms of some sequents $a:A \vdash_{S_1} b:B$ and $c:C \vdash_{S_2} d:D$ respectively. We define a monoidal product $\otimes$ that can combine them into $[f] \otimes [g] : A \otimes B \rightarrow C \otimes D$. The new denotation will represent all the soup equivalent forms of the sequent that is generated by combining the two sequents using the right tensor rule:
\begin{prooftree}
    \AxiomC{$a:A \vdash_{S_1} b:B$}
    \AxiomC{$c:C \vdash_{S_2} d:D$}
    \RightLabel{\footnotesize $\bigotimes R$}
    \BinaryInfC{$a:A, c:C \vdash_{S_1 \cup S_2} b \otimes d:B \otimes D$}
\end{prooftree}

Let $[f]:A \rightarrow B$, $[g]:B \rightarrow P$, $[h]:C \rightarrow D$ and $[t]:D \rightarrow Q$ be denotations in the syntactic category. Using composition and tensor, we can combine these denotations to form $([g] \circ [f]) \otimes ([t] \circ [h])$, which represents the soup equivalent sequents of:
\begin{prooftree}
    \AxiomC{$a:A \vdash_{S_1} b:B$}
    \AxiomC{$b^\prime:B \vdash_{S_3} p:P$}
    \RightLabel{\footnotesize Cut}
    \BinaryInfC{$a:A \vdash_{S_1 \cup S_3 \cup \{b:b^\prime\}} p:P$}
    \AxiomC{$c:C \vdash_{S_2} d:D$}
    \AxiomC{$d^\prime:D \vdash_{S_4} q:Q$}
    \RightLabel{\footnotesize Cut}
    \BinaryInfC{$c:C \vdash_{S_2 \cup S_4 \cup \{d:d^\prime\}} q:Q$}
    \RightLabel{\footnotesize $\bigotimes R$}
    \BinaryInfC{$a:A, c:C \vdash_{S_1 \cup S_2 \cup S_3 \cup S_4 \cup \{b:b^\prime, d:d^\prime\}} p \otimes q:P \otimes Q$}
\end{prooftree}

\noindent We can also combine the same denotations to form $([g] \otimes [t]) \circ ([f] \otimes [h])$, which represents the soup equivalent sequents of:
\begin{prooftree}
    \AxiomC{$a:A \vdash_{S_1} b:B$}
    \AxiomC{$c:C \vdash_{S_2} d:D$}
    \RightLabel{\footnotesize $\bigotimes R$}
    \BinaryInfC{$a:A, c:C \vdash_{S_1 \cup S_2} b \otimes d:B \otimes D$}
    \AxiomC{$b^\prime:B \vdash_{S_3} p:P$}
    \AxiomC{$d^\prime:D \vdash_{S_4} q:Q$}
    \RightLabel{\footnotesize $\bigotimes R$}
    \BinaryInfC{$b^\prime:B, d^\prime:D \vdash_{S_3 \cup S_4} p \otimes q:P \otimes Q$}
    \RightLabel{\footnotesize $\bigotimes L$}
    \UnaryInfC{$b^\prime \otimes d^\prime:B \otimes D \vdash_{S_3 \cup S_4} p \otimes q:P \otimes Q$}
    \RightLabel{\footnotesize Cut}
    \BinaryInfC{$a:A, c:C \vdash_{S_1 \cup S_2 \cup S_3 \cup S_4 \cup \{b \otimes d:b^\prime \otimes d^\prime\}} p \otimes q:P \otimes Q$}
\end{prooftree}

\noindent Applying our soup's bifunctoriality rule, $\{b \otimes d : b^\prime \otimes d^\prime\} \rightarrow \{b:b^\prime, d:d^\prime\}$, reduces one of these sequents to the other, which means that they belong to the same equivalence class. Hence, the two sequents are represented by the same denotation:
\[ ([g] \circ [f]) \otimes ([t] \circ [h]) = ([g] \otimes [t]) \circ ([f] \otimes [h]) \]

\noindent The tensor product also preserves identities since $a \otimes b : A \otimes B \vdash a \otimes b : A \otimes B$ is $\alpha$-equivalent to $x : A \otimes B \vdash x : A \otimes B$.

\vskip 0.1cm \noindent The denotation $[\alpha_{A \otimes B, C, D}] \circ [\alpha_{A,B,C \otimes D}]$ represents the soup equivalent sequents of:
\[ a:A, b \otimes (c \otimes d) : B \otimes (C \otimes D) \vdash_{S_1} \left((a_3 \otimes b_3) \otimes c_3\right) \otimes d_3 : \left((A \otimes B) \otimes C\right) \otimes D \]

\noindent where the soup $S_1$ is:
\[
\left\{
\begin{array}{l l}
    \alpha_{A,B,C \otimes D}    & : \lambda\left(a \otimes \left(b \otimes (c \otimes d)\right)\right).\left((a_2 \otimes b_2) \otimes (c_2 \otimes d_2)\right),\\
    \alpha_{A \otimes B,C,D}    & : \lambda\left((a_2 \otimes b_2) \otimes (c_2 \otimes d_2)\right).\left(\left((a_3 \otimes b_3) \otimes c_3\right) \otimes d_3\right)
\end{array}
\right\}
\]

\noindent Similarly, the denotation $([\alpha_{A,B,C}] \otimes [id_D]) \circ [\alpha_{A,B \otimes C,D}] \circ ([id_A] \otimes [\alpha_{B,C,D}])$ represents the soup equivalent sequents of:
\[ a:A, b \otimes (c \otimes d) : B \otimes (C \otimes D) \vdash_{S_2} \left((a_3 \otimes b_3) \otimes c_3\right) \otimes d_3 : \left((A \otimes B) \otimes C\right) \otimes D \]

\noindent where the soup $S_2$ is:
\[
\left\{
\begin{array}{l l}
    \alpha_{B,C,D}              & : \lambda\left(b \otimes (c \otimes d)\right).\left((b_1 \otimes c_1) \otimes d_1\right),\\
    \alpha_{A, B \otimes C,D}   & : \lambda\left(a \otimes \left((b_1 \otimes c_1) \otimes d_1\right)\right).\left(\left(a_2 \otimes (b_2 \otimes c_2)\right) \otimes d_3\right),\\
    \alpha_{A,B,C}              & : \lambda\left(a_2 \otimes (b_2 \otimes c_2)\right).\left((a_3 \otimes b_3) \otimes c_3\right)
\end{array}
\right\}
\]

\noindent We can use our soup's bifunctoriality and substitution rules to show that $S_1$ and $S_2$ are equivalent. Hence, the two denotations we constructed are equal, which corresponds to the associativity pentagon for monoidal categories:
\[ [\alpha_{A \otimes B, C, D}] \circ [\alpha_{A,B,C \otimes D}] = ([\alpha_{A,B,C}] \otimes [id_D]) \circ [\alpha_{A,B \otimes C,D}] \circ ([id_A] \otimes [\alpha_{B,C,D}]) \]
\[
% Generated with LaTeXDraw 2.0.2
% Thu Jan 15 12:57:30 GMT 2009
% \usepackage[usenames,dvipsnames]{pstricks}
% \usepackage{epsfig}
% \usepackage{pst-grad} % For gradients
% \usepackage{pst-plot} % For axes
\scalebox{1} % Change this value to rescale the drawing.
{
\begin{pspicture}(4,0)(10.582812,4.382969)
\usefont{T1}{ptm}{m}{n}
\rput(7.3723435,4.1945314){$(A \otimes B) \otimes (C \otimes D)$}
\usefont{T1}{ptm}{m}{n}
\rput(2.9723437,1.9945313){$A \otimes (B \otimes (C \otimes D))$}
\usefont{T1}{ptm}{m}{n}
\rput(11.772344,1.9945313){$((A \otimes B) \otimes C) \otimes D$}
\usefont{T1}{ptm}{m}{n}
\rput(2.9723437,-0.80546874){$A \otimes ((B \otimes C) \otimes D)$}
\usefont{T1}{ptm}{m}{n}
\rput(11.772344,-0.80546874){$(A \otimes (B \otimes C)) \otimes D$}
\psline[linewidth=0.04cm,arrowsize=0.05291667cm 2.0,arrowlength=1.4,arrowinset=0.4]{->}(3.0209374,1.6845312)(3.0209374,-0.5154688)
\psline[linewidth=0.04cm,arrowsize=0.05291667cm 2.0,arrowlength=1.4,arrowinset=0.4]{<-}(11.820937,1.6845312)(11.820937,-0.5154688)
\psline[linewidth=0.04cm,arrowsize=0.05291667cm 2.0,arrowlength=1.4,arrowinset=0.4]{->}(4.8209376,-0.8154687)(10.020938,-0.8154687)
\usefont{T1}{ptm}{m}{n}
\rput(4.4623437,3.1945312){$a$}
\usefont{T1}{ptm}{m}{n}
\rput(10.262343,3.1945312){$a$}
\usefont{T1}{ptm}{m}{n}
\rput(7.262344,-0.60546875){$a$}
\usefont{T1}{ptm}{m}{n}
\rput{-270.0}(3.2314062,-2.0314062){\rput(2.6123438,0.59453124){$id \otimes a$}}
\usefont{T1}{ptm}{m}{n}
\rput{-90.0}(11.631406,12.831407){\rput(12.212344,0.59453124){$a \otimes id$}}
\psline[linewidth=0.04cm,arrowsize=0.05291667cm 2.0,arrowlength=1.4,arrowinset=0.4]{->}(3.0209374,2.2845314)(6.6209373,3.8845313)
\psline[linewidth=0.04cm,arrowsize=0.05291667cm 2.0,arrowlength=1.4,arrowinset=0.4]{<-}(11.820937,2.2845314)(8.220938,3.8845313)
\end{pspicture}
}
\]

\noindent Now consider the denotation $([\rho_A] \otimes [id_B]) \circ [\alpha_{A,I,B}]$, which represents the soup equivalent sequents of:
\[ a:A, i \otimes b : I \otimes B \vdash_{S_3} a_3 \otimes b_2 : A \otimes B \]

\noindent where the soup $S_3$ is:
\[
\left\{
\begin{array}{l l}
    \alpha_{A,I,B}  & : \lambda(a \otimes (i \otimes b)).((a_2 \otimes i_2) \otimes b_2),\\
    \rho_A          & : \lambda(a_2 \otimes i_2).a_3
\end{array}
\right\}
\]

\noindent and the denotation $[id_A] \otimes [\lambda_B]$, which corresponds to the soup equivalent sequents of:
\[ a:A, i \otimes b : I \otimes B \vdash_{S_4} a_3 \otimes b_2 : A \otimes B \]

\noindent where the soup $S_4$ is:
\[
\left\{
\begin{array}{l l}
    \lambda_B   & : \lambda(i \otimes b).b_2,\\
    a           & : a_3
\end{array}
\right\}
\]

\noindent The two soups, $S_3$ and $S_4$, are equivalent up to soup reduction, which means that the two denotations representing them are equal:
\[ ([\rho_A] \otimes [id_B]) \circ [\alpha_{A,I,B}] = [id_A] \otimes [\lambda_B] \]

The syntactic category, therefore, satisfies all of the requirements and coherence conditions of a monoidal category.
\end{proof}

\begin{theorem}[Symmetric monoidal category]
The syntactic category is a symmetric monoidal category
\end{theorem}
\begin{proof}
Consider the denotation $[\sigma_{B,A}] \circ [\sigma_{A,B}]$, representing the soup equivalent sequents of:
\[ a:A, b:B \vdash_{\{ \sigma_{A,B} : \lambda(a \otimes b).(b_2 \otimes a_2),\; \sigma_{B,A} : \lambda(b_2 \otimes a_2).(a_3 \otimes b_3)\}} a_3 \otimes b_3 : A \otimes B \]

\noindent These sequents are soup equivalent to $x:A \otimes B \vdash x:A \otimes B$, which means that the denotation used as a symmetry isomorphism is involutive:
\[ [\sigma_{B,A}] \circ [\sigma_{A,B}] = [id_{A \otimes B}] \]

\noindent Moreover, if we consider $[\lambda_A] \circ [\sigma_{A,I}]$, the denotation corresponding to the sequent:
\[ a:A, i:I \vdash_{\{ \sigma_{A,I} : \lambda(a \otimes i).(i_2 \otimes a_2), \lambda_A : \lambda(i_2 \otimes a_2).a_3 \}} a_3:A \]

\noindent and $[\rho_A]$, the denotation corresponding to:
\[ a:A, i:I \vdash_{\{ \rho_A : \lambda(a \otimes i).a_3 \}} a_3:A \]

\noindent Since the two sequents are soup equivalent, we can say that:
\[ [\rho_A] = [\lambda_A] \circ [\sigma_{A,I}] \]

\noindent Finally, consider the denotation $[\alpha_{C,A,B}] \circ [\sigma_{A \otimes B,C}] \circ [\alpha_{A,B,C}]$, which represents the soup equivalent sequents of:
\[ a:A, b \otimes c : B \otimes C \vdash_{S_5} (c_3 \otimes a_3) \otimes b_3 : (C \otimes A) \otimes B \]

\noindent where the soup $S_5$ is:
\[
\left\{
\begin{array}{l l}
    \alpha_{A,B,C}          & : \lambda\left(a \otimes (b \otimes c)\right).\left((a_1 \otimes b_1) \otimes c_1\right),\\
    \sigma_{A \otimes B, C} & : \lambda\left((a_1 \otimes b_1) \otimes c_1\right).\left(c_2 \otimes (a_2 \otimes b_2)\right),\\
    \alpha_{C,A,B}          & : \lambda\left(c_2 \otimes (a_2 \otimes b_2)\right).\left((c_3 \otimes a_3) \otimes b_3\right)
\end{array}
\right\}
\]

\noindent and the denotation $([\sigma_{A,C}] \otimes [id_B]) \circ [\alpha_{A,C,B}] \circ ([id_A] \otimes [\sigma_{B,C}])$, which represents the soup equivalent sequents of:
\[ a:A, b \otimes c : B \otimes C \vdash_{S_6} (c_3 \otimes a_3) \otimes b_3 : (C \otimes A) \otimes B \]

\noindent where the soup $S_6$ is:
\[
\left\{
\begin{array}{l l}
    \sigma_{B, C}   & : \lambda(b \otimes c).(c_1 \otimes b_1),\\
    \alpha_{A,C,B}  & : \lambda\left(a \otimes (c_1 \otimes b_1)\right).\left((a_2 \otimes c_2) \otimes b_3\right),\\
    \sigma_{A, C}   & : \lambda(a_2 \otimes c_2).(c_3 \otimes a_3)
\end{array}
\right\}
\]

\noindent The two sequents are soup equivalent, so the denotations representing them are equal:
\[ [\alpha_{C,A,B}] \circ [\sigma_{A \otimes B,C}] \circ [\alpha_{A,B,C}] = ([\sigma_{A,C}] \otimes [id_B]) \circ [\alpha_{A,C,B}] \circ ([id_A] \otimes [\sigma_{B,C}]) \]

The syntactic category, therefore, satisfies all of the requirements and coherence conditions of a symmetric monoidal category.
\end{proof}

\begin{theorem}[Compact closure]
The syntactic category is a compact closed category
\end{theorem}
\begin{proof}
Consider the denotation $[\lambda_A] \circ ([\varepsilon_A] \otimes [id_A]) \circ [\alpha_{A,A^*,A}] \circ ([id_A] \otimes [\eta_A]) \circ [\rho_A]^{-1}$, which represents the soup equivalent sequents of $x_1:A \vdash_{S_7} x_8:A$ where the soup $S_7$ is:
\[
\left\{
\begin{array}{l l}
    \rho_A^{-1}         & : \lambda x_1.(x_2 \otimes i_1),\\
    \eta_A              & : \lambda i_1.(x_{3*} \otimes x_4),\\
    \alpha_{A,A^*,A}    & : \lambda\left(x_2 \otimes (x_3 \otimes x_4)\right).\left((x_5 \otimes x_{6*}) \otimes x_7\right),\\
    \varepsilon_A       & : \lambda(x_5 \otimes x_{6*}).i_2,\\
    \lambda_A           & : \lambda(i_2 \otimes x_7).x_8
\end{array}
\right\}
\]

Using our soup reduction rules, we can reduce the above sequent to $x:A \vdash x:A$, hence showing that:
\[ [\lambda_A] \circ ([\varepsilon_A] \otimes [id_A]) \circ [\alpha_{A,A^*,A}] \circ ([id_A] \otimes [\eta_A]) \circ [\rho_A]^{-1} = [id_A] \]

Now, let us consider $[\rho_{A^*}] \circ ([id_{A^*}] \otimes [\varepsilon_A]) \circ [\alpha_{A^*,A,A^*}]^{-1} \circ ([\eta_A] \otimes [id_{A^*}]) \circ [\lambda_{A^*}]^{-1}$; the denotation that corresponds to the soup equivalent sequents of $x_{1*}:A^* \vdash_{S_8} x_{8*}:A^*$, where the soup $S_8$ is:
\[
\left\{
\begin{array}{l l}
    \lambda_{A^*}^{-1}      & : \lambda x_{1*}.(i_1 \otimes x_{2*}),\\
    \eta_A                  & : \lambda i_1.(x_{3*} \otimes x_4),\\
    \alpha_{A^*,A,A^*}^{-1} & : \lambda\left((x_{3*} \otimes x_4) \otimes x_{2*}\right).\left(x_{5*} \otimes (x_6 \otimes x_{7})\right),\\
    \varepsilon_A           & : \lambda(x_6 \otimes x_{7*}).i_2,\\
    \rho_{A*}               & : \lambda(x_{5*} \otimes i_2).x_{8*}
\end{array}
\right\}
\]

Using our soup reduction rules, we can reduce the above sequent to $x_*:A^* \vdash x_*:A^*$, hence showing that:
\[ [\rho_{A^*}] \circ ([id_{A^*}] \otimes [\varepsilon_A]) \circ [\alpha_{A^*,A,A^*}]^{-1} \circ ([\eta_A] \otimes [id_{A^*}]) \circ [\lambda_{A^*}]^{-1} = [id_{A^*}] \]

The syntactic category thus satisfies both of the yanking conditions that are required of a compact closed category.
\end{proof}

\begin{theorem}[Dagger compact closure]
The syntactic category is a dagger compact category
\end{theorem}
\begin{proof}
Let $[f] : A \rightarrow B$ be a denotation representing the soup equivalent sequents of $a:A \vdash_{S_9} b:B$. For every such denotation $[f]$, we define its dagger $[f]^\dag : B \rightarrow A$ such that it represents the soup equivalent sequents of the $\dag$-flipped version of the original sequent: $b:B \vdash_{S_{9*}} a:A$. The $\dag$-flip rule, however, is involutive since $(S_{9*})_* = S_9$, hence $([f]^\dag)^\dag = f$.

Now consider the denotation $[\sigma_{A,A^*}] \circ [\varepsilon_A]^\dag$, which represents the soup equivalent sequents of:
\[ i:I \vdash_{S_{10}} x_{3*} \otimes x_4 : A^* \otimes A \]

\noindent where the soup $S_{10}$ is:
\[
\left\{
\begin{array}{l l}
    \varepsilon_{A*}    & : \left(\lambda(x_1 \otimes x_{2*}).1\right)_*,\\
    \sigma_{A,A^*}      & : \lambda(x_1 \otimes x_{2*}).(x_{3*} \otimes x_4)
\end{array}
\right\}
\]

\noindent By using soup reduction, we get $S_{10} \rightarrow \{ \eta_A : \lambda i.(x_{3*} \otimes x_4) \}$. But the sequent $i:I \vdash_{\{ \eta_A : \lambda i.(x_{3*} \otimes x_4) \}} x_{3*} \otimes x_4 : A^* \otimes A$ is represented by the denotation $[\eta_A]$, which means that
\[ [\sigma_{A,A^*}] \circ [\varepsilon_A]^\dag = [\eta_A] \]

The syntactic category, therefore, satisfies all of the requirements of a dagger compact category.
\end{proof}

\subsection{Proof of equivalence}
\label{Subsection:Proof of equivalence}
We will now prove that the free dagger compact category $\mathcal{C}_{Free}$ is equivalent to the syntactic category $\mathcal{C}_{Synt}$.

\begin{lemma}[Essentially surjective on objects]
The set of objects in the free category and the set of objects in the syntactic category are surjective, up to isomorphism.
\end{lemma}
\begin{proof}
Recall $Dagger(\Sigma_0)$; the free $(\otimes, I, \Box^*)$-algebra over the set of object variables $\Sigma_0$. The sets of objects in $\mathcal{C}_{Free}$ and $\mathcal{C}_{Synt}$ both correspond to $Dagger(\Sigma_0)$, up to the equivalence classes induced by $(A^*)^* \equiv A$ and $I^* \equiv I$.
\end{proof}

\begin{lemma}[Equal arrows correspond to equal denotations]
If two arrows, $\langle A,f,B \rangle$ and $\langle A,f^\prime,B \rangle$ are equal in the free category, then they will also be equal in the syntactic category: $[f] = [f^\prime] : A \rightarrow B$.
\end{lemma}
\begin{proof}
The structure of the free category $\mathcal{C}_{Free}$ imposes the minimum number of equalities for a category to be dagger compact. Moreover, both the free category and the syntactic category derive their symbols from the same signature graph $\mathcal{G}$. Since we have already shown that $\mathcal{C}_{Synt}$ is dagger compact, the same steps can be used to show that any arrows $\langle A,f,B \rangle$ and $\langle A,f^\prime,B \rangle$ that are equal in the free category, correspond to equal denotations $[f] = [g]$ in the syntactic category.
\end{proof}

\begin{lemma}[Equal denotations correspond to equal arrows]
Any denotations that are equal in the syntactic category, correspond to equal arrows in the free category.
\end{lemma}
\begin{proof}
Let $[f] : \Gamma \rightarrow B$ and $[g] : \Gamma \rightarrow B$ be denotations in the syntactic category such that $[f] = [g]$. Since the two denotations are equal, the sequents they represent in the dagger lambda calculus must be equivalent up to soup reduction. Without loss of generality, let's assume that $[f]$ represents a sequent $J_1$ and that $[g]$ represents a sequent $J_2$, where $J_1 \rightarrow J_2$. The soup reduction relation consists of four soup rules: \textit{bifunctoriality}, \textit{trace}, \textit{cancellation} and \textit{consumption}. We prove this theorem by induction on the structure of the soup reduction linking $J_1$ and $J_2$:
\begin{itemize}
  \item If  we use a \textit{bifunctoriality} rule, then we will be reducing a sequent of the form $\Gamma \vdash_{S \cup \{ x_1 \otimes x_2 : x_3 \otimes x_4 \}} b:B$ to one of the form $\Gamma \vdash_{S \cup \{ x_1:x_3, x_2:x_4 \}} b:B$. By carefully separating the appropriate connections, we can break down $[g]$ into $([\pi_2] \circ [\pi_1]) \otimes ([\pi_4] \circ [\pi_3])$ where: $[\pi_1]$ represents $\Gamma_1 \vdash_{S_1} x_1:T_1$; $[\pi_2]$ represents $x_3:T_1 \vdash_{S_2} b_1:B_1$; $[\pi_3]$ represents $\Gamma_2 \vdash_{S_3} x_2:T_2$; $[\pi_4]$ represents $x_4:T_2 \vdash_{S_4} b_2:B_2$; and $\Gamma = \Gamma_1, \Gamma_2$, $S = S_1 \cup S_2 \cup S_3 \cup S_4$, $b = b_1 \otimes b_2$ and $B = B_1 \otimes B_2$. The individual $[\pi_i]$ denotations can be reconstructed in a different way to form $([\pi_2] \otimes [\pi_4]) \circ ([\pi_1] \otimes [\pi_3])$, which actually forms $[f]$. But $\langle \Gamma_1, \pi_1, T_1, \pi_2, B_1 \rangle \otimes \langle \Gamma_2, \pi_3, T_2, \pi_4, B_2 \rangle = \langle \Gamma, \pi_1 \otimes \pi_3, T_1 \otimes T_2, \pi_2 \otimes \pi_4, B \rangle$ because $\mathcal{C}_{Free}$ is a \textit{dagger compact category}, so  $\langle \Gamma, f, B \rangle = \langle \Gamma, g, B \rangle$.
  \item If we use a \textit{trace} rule on $J_1$, we do not in any way affect the information that is contained in the soup connection, we are simply rewriting it using different notation, so we are in no way affecting the derivation of $J_2$ from sequents represented by smaller denotations.
  \item If we use a \textit{cancellation} rule on $J_1$, the information contained on the connection we are striking out is nil, so we are in no way affecting the derivation of $J_2$ from sequents represented by smaller denotations. In this case, $[f]$ represents $\langle \Gamma, f, B \rangle = \langle \Gamma, g, B \rangle \otimes 1$ in the free category, which is the same as $\langle \Gamma, g, B \rangle$.
  \item If  we use a \textit{consumption} rule on $J_1$, then the bound variable we are substituting for will either appear to the left or to the right of the turnstile, or it will appear elsewhere in the soup.
      \begin{itemize}
        \item If the bound variable appears to the \textit{left} of the turnstile then the general form of the sequent $J_1$ can be written as $t_1:T_1, x:T_2, t_3:T_3 \vdash_{S \cup \{x:t_2\}} b:B$. Let $\Gamma = T_1 \otimes T_2 \otimes T_3$. Since $J_2$ is the result of consuming a soup connection that was created by performing a Cut with the identity, the sequent can be written as $t_1:T_1, t_2:T_2, t_3:T_3 \vdash_S b:B$. This means that $[f] = [g] \circ [id_\Gamma]$, which causes $\langle \Gamma, f, B \rangle = \langle \Gamma, g, B \rangle \circ \langle \Gamma \rangle = \langle \Gamma, g, B \rangle$.
        \item If the bound variable appears to the \textit{right} of the turnstile then the general form of the sequent $J_1$ can be written as $\Gamma \vdash_{S \cup \{t_2:x\}} t_1 \otimes x \otimes t_3:B$. Since $J_2$ is the result of consuming a soup connection that was created by performing a Cut with the identity, the sequent can be written as $\Gamma \vdash_S t_1 \otimes t_2 \otimes t_3:B$. This means that $[f] = [id_B] \circ [g]$, which causes $\langle \Gamma, f, B \rangle = \langle B \rangle \circ \langle \Gamma, g, B \rangle = \langle \Gamma, g, B \rangle$.
        \item If the bound variable appears \textit{elsewhere in the soup}, then $J_1$ will be the result of applying the Cut rule twice on $J_2$, each time with an instance of the identity. The denotation representing $J_1$ will either be $[f] = ([id_B] \circ [id_B]) \circ [g]$ or $[f] = [g] \circ ([id_\Gamma] \circ [id_\Gamma])$, both of which cause $\langle \Gamma, f, B \rangle = \langle \Gamma, g, B \rangle$.
      \end{itemize}
\end{itemize}
We have, therefore, shown that in all cases, $\langle \Gamma, f, B \rangle = \langle \Gamma, g, B \rangle$.
\end{proof}

\begin{theorem}[Equivalence between the free category and the syntactic category]
The free dagger compact category $\mathcal{C}_{Free}$ and the syntactic category $\mathcal{C}_{Synt}$ are equivalent.
\end{theorem}
\begin{proof}
The two categories derive their symbols from a common signature graph $\mathcal{G}$. As we have already shown, bearing in mind the equivalence classes that we have induced on types, the categories are essentially surjective on objects. Moreover, arrows that are equal in the free category are equal in the syntactic category and vice versa. This means that the functor $F$ is \textit{full} and \textit{faithful}, causing the notions of equality between arrows overlap in these two categories. Consequently, the categories are equivalent.
\end{proof}

\begin{corollary}[Internal language]
The dagger lambda calculus is an internal language for dagger compact categories.
\end{corollary}

\clearpage
\chapter{Classical control in the $\dag\lambda$-calculus}
\label{Chapter:Classical control in the dagger lambda calculus}
In this section we will see how the dagger lambda calculus can be imbued with classical control structures, similar to the ones in \cite{CD11}. We will begin by defining classical structures within the dagger lambda calculus, hence allowing it to support the Frobenius algebras of \cite{CP06}, \cite{CP07} and \cite{CPP10}. We will then demonstrate how our extension of the lambda calculus can be used to represent the notion of dualiser that was introduced in \cite{CPP08}, as well as how this notion now makes the Currying rule admissible in our language. Finally, we will further extend the dagger lambda calculus with complementary classical structures, by defining a notion of complementarity that extends that of \cite{CD11}, with the dualisers of \cite{CPP08}, to allow for non self-dual $\dag$-compact structures.

Throughout this entire section, we will be using a version of the dagger lambda calculus that has been modified to simplify our notation. In order to do this, we will be restricting our attention to \textit{strict dagger compact categories}; categories whose monoidal natural isomorphisms, $\alpha_{A,B,C}$, $\lambda_\Gamma$ and $\rho_\Gamma$, are identities. As a result, instances of the $\otimes L$, $\lambda_\Gamma$ and $\rho_\Gamma$ rules will be equated with the identity, allowing us to be more relaxed when it comes to tensor identities and parenthesising terms.

\section{Classical structures}
\label{Section:Classical structures}
As we have seen in chapter \ref{Chapter:Quantum computing} and section \ref{Section:Representation of classical structures}, of part \ref{Part:Background material} of this dissertation, the classical world is a lot less "restrictive" than the quantum world in that it allows us to freely copy and delete data. Classical states can be thought of as a basis that spans a vector space of quantum states and it is in this subset of the quantum world that the linearity restrictions can be relaxed.

Linear logic \cite{Gir87} achieves its resource sensitivity by dropping the rules of weakening and contraction. In place of those rules, Girard introduced the exponential connective $!A$, denoting an infinite supply of the type $A$, as a more tightly controlled way of breaching linearity. In designing the dagger lambda calculus, we have also dropped the rules of weakening and contraction. Instead of replacing them, however, with an exponential connective, we will extend our language by providing a stricter and yet more controlled way of relaxing resource sensitivity. This will be done by supporting the Frobenius algebras of \cite{CP06}, \cite{CP07} and \cite{CPP10} with copying and deleting maps and will allow us to model the behaviour of a classical basis.

For every type $A$ in the dagger lambda calculus, we will be introducing two constants; $G_1^2:A \multimap A \otimes A$ and $G_1^0:A \multimap I$. This allows us to represent the copying and deleting operations as sequents in our language:

\[
% Generated with LaTeXDraw 2.0.2
% Fri Jan 16 17:41:15 GMT 2009
% \usepackage[usenames,dvipsnames]{pstricks}
% \usepackage{epsfig}
% \usepackage{pst-grad} % For gradients
% \usepackage{pst-plot} % For axes
\scalebox{1} % Change this value to rescale the drawing.
{
\begin{pspicture}(0,-0.72)(7.52,0.72)
\psline[linewidth=0.04cm](7.22,-0.1)(7.22,-0.7)
\pscircle[linewidth=0.04,dimen=outer,fillstyle=solid,fillcolor=green](7.22,0.0){0.3}
\psline[linewidth=0.04cm](0.62,-0.1)(0.62,-0.7)
\psarc[linewidth=0.04](0.62,0.7){0.6}{-180.0}{0.0}
\pscircle[linewidth=0.04,dimen=outer,fillstyle=solid,fillcolor=green](0.62,0.0){0.3}
\end{pspicture}
}
\]

\begin{center}
$x_1:A \vdash_{\{ G_1^2 : x_{1*} \otimes x_2 \otimes x_3 \}} x_2 \otimes x_3 : A \otimes A$
\hskip 1cm
$\vdash G_1^0 : A^*$
\end{center}

\begin{definition}[Notational conventions for classical structures in the dagger lambda calculus]
For notational convenience, we will use $A^{\otimes n}$ as shorthand for $\underbrace{A \otimes \ldots \otimes A}_{n \text{ times}}$. We also define $G_1^1$, $G_2^1$ and $G_0^1$ as follows:
\begin{center}
$G_1^1 := id_A$, \hskip 1.2cm $G_2^1 := (G_1^2)_*$, \hskip 1.2cm $G_0^1 := (G_1^0)_*$
\end{center}
\end{definition}

We will define the copying and deleting maps by describing all of the conditions that we require them to satisfy in the dagger lambda calculus. The first such condition is the comonoidal coassociativity condition:
\[
% Generated with LaTeXDraw 2.0.8
% Fri Apr 06 13:20:13 BST 2012
% \usepackage[usenames,dvipsnames]{pstricks}
% \usepackage{epsfig}
% \usepackage{pst-grad} % For gradients
% \usepackage{pst-plot} % For axes
\scalebox{1} % Change this value to rescale the drawing.
{
\begin{pspicture}(0,-1.39)(6.64,1.379375)
\usefont{T1}{ptm}{m}{n}
\rput(2.3045313,0.055){\huge =}
\psline[linewidth=0.04cm](0.62,0.55)(0.62,-0.05)
\psarc[linewidth=0.04](0.62,1.35){0.6}{180.0}{0.0}
\pscircle[linewidth=0.04,dimen=outer,fillstyle=solid,fillcolor=green](0.62,0.65){0.3}
\psline[linewidth=0.04cm](1.82,1.35)(1.82,-0.05)
\psline[linewidth=0.04cm](1.22,-0.75)(1.22,-1.35)
\psarc[linewidth=0.04](1.22,0.05){0.6}{-180.0}{0.0}
\pscircle[linewidth=0.04,dimen=outer,fillstyle=solid,fillcolor=green](1.22,-0.65){0.3}
\psline[linewidth=0.04cm](2.72,1.35)(2.72,0.65)
\psline[linewidth=0.04cm](3.92,1.35)(3.92,0.65)
\psline[linewidth=0.04cm](4.82,1.35)(4.82,-0.05)
\usefont{T1}{ptm}{m}{n}
\rput(4.4045315,0.055){\huge =}
\psline[linewidth=0.04cm](3.32,1.35)(3.32,-1.35)
\psarc[linewidth=0.04](3.32,0.75){0.6}{-180.0}{0.0}
\pscircle[linewidth=0.04,dimen=outer,fillstyle=solid,fillcolor=green](3.32,0.05){0.3}
\psline[linewidth=0.04cm](6.02,0.55)(6.02,-0.05)
\psarc[linewidth=0.04](6.02,1.35){0.6}{180.0}{0.0}
\pscircle[linewidth=0.04,dimen=outer,fillstyle=solid,fillcolor=green](6.02,0.65){0.3}
\psline[linewidth=0.04cm](5.42,-0.75)(5.42,-1.35)
\psarc[linewidth=0.04](5.42,0.05){0.6}{-180.0}{0.0}
\pscircle[linewidth=0.04,dimen=outer,fillstyle=solid,fillcolor=green](5.42,-0.65){0.3}
\usefont{T1}{ptm}{m}{n}
\rput(1.0565625,-1.265){\scriptsize 1}
\usefont{T1}{ptm}{m}{n}
\rput(3.1565626,-1.265){\scriptsize 1}
\usefont{T1}{ptm}{m}{n}
\rput(5.5565624,-1.265){\scriptsize 1}
\usefont{T1}{ptm}{m}{n}
\rput(0.78125,0.035){\scriptsize 2}
\usefont{T1}{ptm}{m}{n}
\rput(5.88125,0.035){\scriptsize 2}
\usefont{T1}{ptm}{m}{n}
\rput(1.6728125,0.535){\scriptsize 3}
\usefont{T1}{ptm}{m}{n}
\rput(4.0728126,1.235){\scriptsize 3}
\usefont{T1}{ptm}{m}{n}
\rput(6.4728127,1.235){\scriptsize 3}
\usefont{T1}{ptm}{m}{n}
\rput(0.18296875,1.235){\scriptsize 4}
\usefont{T1}{ptm}{m}{n}
\rput(2.8829687,1.235){\scriptsize 4}
\usefont{T1}{ptm}{m}{n}
\rput(4.982969,0.535){\scriptsize 4}
\usefont{T1}{ptm}{m}{n}
\rput(0.9740625,1.235){\scriptsize 5}
\usefont{T1}{ptm}{m}{n}
\rput(3.4740624,1.235){\scriptsize 5}
\usefont{T1}{ptm}{m}{n}
\rput(5.5740623,1.235){\scriptsize 5}
\end{pspicture}
}
\]

\noindent which requires the following two soups, $S_1$ and $S_2$, to be equal in the lambda calculus:
\[ S_1 = \{ G_1^2 : x_{1*} \otimes x_2 \otimes x_3, G_1^2 : x_{2*} \otimes x_4 \otimes x_5 \} \]
\[ S_2 = \{ G_1^2 : x_{1*} \otimes x_4 \otimes x_2, G_1^2 : x_{2*} \otimes x_5 \otimes x_3 \} \]

\noindent This allows us to relate the sequents for $x_1:A \vdash_{S_1} x_4 \otimes x_5 \otimes x_3 : A^{\otimes 3}$ and $x_1:A \vdash_{S_2} x_4 \otimes x_5 \otimes x_3 : A^{\otimes 3}$, matching the equality of arrows that is required in the categorical setting. Since it does not matter whether we copy one or the other part of a copied pair, we can write both cases as:
\[ x_1:A \vdash_{\{ G_1^3 : x_{1*} \otimes x_4 \otimes x_5 \otimes x_3 \}} x_4 \otimes x_5 \otimes x_3 : A^{\otimes 3} \]

The next condition we require of the dagger lambda calculus is the comonoidal identity condition:
\[
% Generated with LaTeXDraw 2.0.8
% Fri Apr 06 14:05:45 BST 2012
% \usepackage[usenames,dvipsnames]{pstricks}
% \usepackage{epsfig}
% \usepackage{pst-grad} % For gradients
% \usepackage{pst-plot} % For axes
\scalebox{1} % Change this value to rescale the drawing.
{
\begin{pspicture}(0,-1.19)(5.0,1.179375)
\psline[linewidth=0.04cm](2.5,-1.15)(2.5,1.15)
\psline[linewidth=0.04cm](4.7,0.75)(4.7,0.15)
\psline[linewidth=0.04cm](3.5,0.15)(3.5,1.15)
\psline[linewidth=0.04cm](0.3,0.75)(0.3,0.15)
\psline[linewidth=0.04cm](0.9,-0.55)(0.9,-1.15)
\psarc[linewidth=0.04](0.9,0.25){0.6}{-180.0}{0.0}
\pscircle[linewidth=0.04,dimen=outer,fillstyle=solid,fillcolor=green](0.9,-0.45){0.3}
\psline[linewidth=0.04cm](1.5,0.15)(1.5,1.15)
\pscircle[linewidth=0.04,dimen=outer,fillstyle=solid,fillcolor=green](0.3,0.85){0.3}
\usefont{T1}{ptm}{m}{n}
\rput(1.9845313,0.055){\huge =}
\usefont{T1}{ptm}{m}{n}
\rput(2.9845312,0.055){\huge =}
\psline[linewidth=0.04cm](4.1,-0.55)(4.1,-1.15)
\psarc[linewidth=0.04](4.1,0.25){0.6}{-180.0}{0.0}
\pscircle[linewidth=0.04,dimen=outer,fillstyle=solid,fillcolor=green](4.1,-0.45){0.3}
\pscircle[linewidth=0.04,dimen=outer,fillstyle=solid,fillcolor=green](4.7,0.85){0.3}
\usefont{T1}{ptm}{m}{n}
\rput(1.0365624,-1.065){\scriptsize 1}
\usefont{T1}{ptm}{m}{n}
\rput(3.9365625,-1.065){\scriptsize 1}
\usefont{T1}{ptm}{m}{n}
\rput(1.36125,1.035){\scriptsize 2}
\usefont{T1}{ptm}{m}{n}
\rput(3.66125,1.035){\scriptsize 2}
\end{pspicture}
}
\]

\noindent which requires the following soups, $S_3$, $S_4$ and $S_5$, to be equal to each other:
\begin{center}
\begin{tabular}{ l l }
  $S_3 =$   & $\{ G_1^2 : x_{1*} \otimes G_0^1 \otimes x_2 \}$ \\
            & \\
  $S_4 =$   & $\{ G_1^1 : x_{1*} \otimes x_2 \} \longrightarrow \{ x_1 : x_2 \}$ \\
            & \\
  $S_5 =$   & $\{ G_1^2 : x_{1*} \otimes x_2 \otimes G_0^1 \}$ \\
\end{tabular}
\end{center}

\noindent This allows us to relate the sequents for $x_1:A \vdash_{S_3} x_2 : A$, $x_1:A \vdash_{S_4} x_2 : A$ and $x_1:A \vdash_{S_5} x_2 : A$, matching the equality of arrows that is required in the categorical setting, to get an identity sequent $x:A \vdash x:A$.

\vskip 0.3cm

We will also require the copy map to satisfy a cocommutativity requirement:
\[
% Generated with LaTeXDraw 2.0.8
% Fri Apr 06 15:34:32 BST 2012
% \usepackage[usenames,dvipsnames]{pstricks}
% \usepackage{epsfig}
% \usepackage{pst-grad} % For gradients
% \usepackage{pst-plot} % For axes
\scalebox{1} % Change this value to rescale the drawing.
{
\begin{pspicture}(0,-1.34)(3.24,2.529375)
\psline[linewidth=0.04cm](0.02,1.3)(0.02,0.0)
\psline[linewidth=0.04cm](1.22,1.3)(1.22,0.0)
\usefont{T1}{ptm}{m}{n}
\rput(1.6045313,-0.095){\huge =}
\psline[linewidth=0.04cm](0.62,-0.7)(0.62,-1.3)
\psarc[linewidth=0.04](0.62,0.1){0.6}{-180.0}{0.0}
\pscircle[linewidth=0.04,dimen=outer,fillstyle=solid,fillcolor=green](0.62,-0.6){0.3}
\psline[linewidth=0.04cm](2.62,-0.7)(2.62,-1.3)
\psarc[linewidth=0.04](2.62,0.1){0.6}{-180.0}{0.0}
\pscircle[linewidth=0.04,dimen=outer,fillstyle=solid,fillcolor=green](2.62,-0.6){0.3}
\psline[linewidth=0.04](2.02,1.3)(2.02,0.9)(3.22,0.3)(3.22,0.0)
\psline[linewidth=0.04](3.22,1.3)(3.22,0.9)(2.02,0.3)(2.02,0.0)
\usefont{T1}{ptm}{m}{n}
\rput(0.7565625,-1.215){\scriptsize 1}
\usefont{T1}{ptm}{m}{n}
\rput(2.7565625,-1.215){\scriptsize 1}
\usefont{T1}{ptm}{m}{n}
\rput(0.18125,1.185){\scriptsize 2}
\usefont{T1}{ptm}{m}{n}
\rput(2.18125,1.185){\scriptsize 2}
\usefont{T1}{ptm}{m}{n}
\rput(3.08125,0.085){\scriptsize 2}
\usefont{T1}{ptm}{m}{n}
\rput(1.0728126,1.185){\scriptsize 3}
\usefont{T1}{ptm}{m}{n}
\rput(3.0728126,1.185){\scriptsize 3}
\usefont{T1}{ptm}{m}{n}
\rput(2.1728125,0.085){\scriptsize 3}
\end{pspicture}
}
\]

\noindent which requires the following soups, $S_6$ and $S_7$, to be equal in the dagger lambda calculus:
\[ S_6 = \{ G_1^2 : x_{1*} \otimes x_2 \otimes x_3 \} \]
\[ S_7 = \{ G_1^2 : x_{1*} \otimes x_3 \otimes x_2 \} \]

\noindent This allows us to relate the sequents for $x_1:A \vdash_{S_6} x_2 \otimes x_3 : A \otimes A$ and $x_1:A \vdash_{S_7} x_2 \otimes x_3 : A \otimes A$, matching the equality of arrows that is required in the categorical setting.

\vskip 0.3cm

Since all spider monoids have to be \textit{special}, we require the copying map to satisfy the isometry condition:
\[
% Generated with LaTeXDraw 2.0.8
% Fri Apr 06 19:38:05 BST 2012
% \usepackage[usenames,dvipsnames]{pstricks}
% \usepackage{epsfig}
% \usepackage{pst-grad} % For gradients
% \usepackage{pst-plot} % For axes
\scalebox{1} % Change this value to rescale the drawing.
{
\begin{pspicture}(0,-1.39)(2.44,1.37)
\psline[linewidth=0.04cm](0.62,-0.75)(0.62,-1.35)
\psarc[linewidth=0.04](0.62,0.05){0.6}{-180.0}{0.0}
\pscircle[linewidth=0.04,dimen=outer,fillstyle=solid,fillcolor=green](0.62,-0.65){0.3}
\psline[linewidth=0.04cm](0.62,1.35)(0.62,0.75)
\psarc[linewidth=0.04](0.62,-0.05){0.6}{0.0}{180.0}
\pscircle[linewidth=0.04,dimen=outer,fillstyle=solid,fillcolor=green](0.62,0.65){0.3}
\psline[linewidth=0.04cm](2.42,-1.37)(2.42,1.33)
\usefont{T1}{ptm}{m}{n}
\rput(1.8045312,-0.065){\huge =}
\usefont{T1}{ptm}{m}{n}
\rput(0.7565625,-1.185){\scriptsize 1}
\usefont{T1}{ptm}{m}{n}
\rput(0.78296876,1.215){\scriptsize 4}
\usefont{T1}{ptm}{m}{n}
\rput(0.18125,0.015){\scriptsize 2}
\usefont{T1}{ptm}{m}{n}
\rput(1.0728126,0.015){\scriptsize 3}
\end{pspicture}
}
\]

\noindent which requires the following soups to be equal:
\begin{center}
\begin{tabular}{ l l }
  $S_8 =$   & $\left\{ G_1^2 : x_{1*} \otimes x_2 \otimes x_3, \; G_2^1 : (x_2 \otimes x_3)_* \otimes x_4 \right\}$ \\
            & \\
  $S_9 =$   & $\{ G_1^1 : x_{1*} \otimes x_4 \} \longrightarrow \{x_1:x_4\}$
\end{tabular}
\end{center}
\noindent and allows us to match $x_1:A \vdash_{S_8} x_4:A$ with $x_1:A \vdash_{S_9} x_4:A$ and the identity sequent $x:A \vdash x:A$.

\vskip 0.3cm

Finally, we require the copying and deleting maps to satisfy the Frobenius condition:
\[
% Generated with LaTeXDraw 2.0.8
% Fri Apr 06 21:52:06 BST 2012
% \usepackage[usenames,dvipsnames]{pstricks}
% \usepackage{epsfig}
% \usepackage{pst-grad} % For gradients
% \usepackage{pst-plot} % For axes
\scalebox{1} % Change this value to rescale the drawing.
{
\begin{pspicture}(0,-1.4)(10.64,1.38)
\psline[linewidth=0.04cm](3.02,-0.74)(3.02,-1.34)
\psarc[linewidth=0.04](3.02,0.06){0.6}{-180.0}{0.0}
\pscircle[linewidth=0.04,dimen=outer,fillstyle=solid,fillcolor=green](3.02,-0.64){0.3}
\psline[linewidth=0.04cm](4.22,1.36)(4.22,0.76)
\psarc[linewidth=0.04](4.22,-0.04){0.6}{0.0}{180.0}
\pscircle[linewidth=0.04,dimen=outer,fillstyle=solid,fillcolor=green](4.22,0.66){0.3}
\psline[linewidth=0.04cm](2.42,-0.06)(2.42,1.34)
\psline[linewidth=0.04cm](4.82,0.04)(4.82,-1.36)
\psline[linewidth=0.04cm](10.02,-0.74)(10.02,-1.34)
\psarc[linewidth=0.04](10.02,0.06){0.6}{-180.0}{0.0}
\pscircle[linewidth=0.04,dimen=outer,fillstyle=solid,fillcolor=green](10.02,-0.64){0.3}
\psline[linewidth=0.04cm](8.82,1.36)(8.82,0.76)
\psarc[linewidth=0.04](8.82,-0.04){0.6}{0.0}{180.0}
\pscircle[linewidth=0.04,dimen=outer,fillstyle=solid,fillcolor=green](8.82,0.66){0.3}
\psline[linewidth=0.04cm](8.22,0.04)(8.22,-1.36)
\psline[linewidth=0.04cm](10.62,-0.06)(10.62,1.34)
\psline[linewidth=0.04cm](6.52,0.56)(6.52,-0.04)
\psarc[linewidth=0.04](6.52,1.36){0.6}{-180.0}{0.0}
\pscircle[linewidth=0.04,dimen=outer,fillstyle=solid,fillcolor=green](6.52,0.66){0.3}
\psline[linewidth=0.04cm](6.52,0.06)(6.52,-0.54)
\psarc[linewidth=0.04](6.52,-1.34){0.6}{0.0}{180.0}
\pscircle[linewidth=0.04,dimen=outer,fillstyle=solid,fillcolor=green](6.52,-0.64){0.3}
\psarc[linewidth=0.04](0.62,0.66){0.6}{-180.0}{0.0}
\psarc[linewidth=0.04](0.62,-0.74){0.6}{0.0}{180.0}
\pscircle[linewidth=0.04,dimen=outer,fillstyle=solid,fillcolor=green](0.62,-0.04){0.3}
\psline[linewidth=0.04cm](0.02,0.54)(0.02,1.34)
\psline[linewidth=0.04cm](1.22,0.54)(1.22,1.34)
\psline[linewidth=0.04cm](1.22,-0.66)(1.22,-1.36)
\psline[linewidth=0.04cm](0.02,-0.66)(0.02,-1.36)
\usefont{T1}{ptm}{m}{n}
\rput(5.4045315,-0.155){\huge =}
\usefont{T1}{ptm}{m}{n}
\rput(7.6045313,-0.155){\huge =}
\usefont{T1}{ptm}{m}{n}
\rput(1.8045312,-0.155){\huge =}
\usefont{T1}{ptm}{m}{n}
\rput(0.1565625,-1.275){\scriptsize 1}
\usefont{T1}{ptm}{m}{n}
\rput(3.1565626,-1.275){\scriptsize 1}
\usefont{T1}{ptm}{m}{n}
\rput(6.0565624,-1.275){\scriptsize 1}
\usefont{T1}{ptm}{m}{n}
\rput(8.356563,-1.275){\scriptsize 1}
\usefont{T1}{ptm}{m}{n}
\rput(1.08125,-1.275){\scriptsize 2}
\usefont{T1}{ptm}{m}{n}
\rput(4.68125,-1.275){\scriptsize 2}
\usefont{T1}{ptm}{m}{n}
\rput(6.98125,-1.275){\scriptsize 2}
\usefont{T1}{ptm}{m}{n}
\rput(10.18125,-1.275){\scriptsize 2}
\usefont{T1}{ptm}{m}{n}
\rput(3.4728124,0.025){\scriptsize 3}
\usefont{T1}{ptm}{m}{n}
\rput(6.3728123,0.025){\scriptsize 3}
\usefont{T1}{ptm}{m}{n}
\rput(9.572812,0.025){\scriptsize 3}
\usefont{T1}{ptm}{m}{n}
\rput(0.18296875,1.225){\scriptsize 4}
\usefont{T1}{ptm}{m}{n}
\rput(2.5829687,1.225){\scriptsize 4}
\usefont{T1}{ptm}{m}{n}
\rput(6.0829687,1.225){\scriptsize 4}
\usefont{T1}{ptm}{m}{n}
\rput(8.982968,1.225){\scriptsize 4}
\usefont{T1}{ptm}{m}{n}
\rput(1.0740625,1.225){\scriptsize 5}
\usefont{T1}{ptm}{m}{n}
\rput(4.0740623,1.225){\scriptsize 5}
\usefont{T1}{ptm}{m}{n}
\rput(6.9740624,1.225){\scriptsize 5}
\usefont{T1}{ptm}{m}{n}
\rput(10.474063,1.225){\scriptsize 5}
\usefont{T1}{ptm}{m}{n}
\rput(1.4582813,-0.055){\huge :}
\end{pspicture}
}
\]

\noindent which requires the following soups to be equal to each other in the dagger lambda calculus:
\[ S_{10} = \{ G_1^2 : x_{1*} \otimes x_4 \otimes x_3, G_2^1 : (x_3 \otimes x_2)_* \otimes x_5 \} \]
\[ S_{11} = \{ G_2^1 : (x_1 \otimes x_2)_* \otimes x_3, G_1^2 : x_{3*} \otimes x_4 \otimes x_5 \} \]
\[ S_{12} = \{ G_1^2 : x_{2*} \otimes x_3 \otimes x_5, G_2^1 : (x_1 \otimes x_3)_* \otimes x_4 \} \]
\noindent and allows us to relate the sequents for
\begin{align*}
                & x_1:A, x_2:A \vdash_{S_{10}} x_4 \otimes x_5 : A \otimes A,   \\
                & x_1:A, x_2:A \vdash_{S_{11}} x_4 \otimes x_5 : A \otimes A    \\
    \text{and } & x_1:A, x_2:A \vdash_{S_{12}} x_4 \otimes x_5 : A \otimes A
\end{align*}
\noindent matching the equality of arrows that is required in the categorical setting.

\vskip 0.3cm

All of these soup equality conditions can be rewritten as equalities between terms to make them more readable. The following table lists the term equalities corresponding to each of the conditions:

\begin{center}
\renewcommand{\arraystretch}{1.3}
\begin{tabular}{|l|l|}
  \hline
  \multirow{2}{*}{Comonoidal associativity} & $\bar{b} (\bar{t} G_1^2 id_A) G_1^2 =$ \\
                                            & $\bar{b} (\bar{t} id_A G_1^2) G_1^2$ \\
  \hline
  \multirow{2}{*}{Comonoidal identity}  & $\bar{b} (\bar{t} G_1^0 id_A) G_1^2 =$ \\
                                        & $G_1^1 =$ \\
                                        & $\bar{b} (\bar{t} id_A G_1^0) G_1^2$ \\
  \hline
  \multirow{2}{*}{Cocommutativity}  & $G_1^2 =$ \\
                                    & $\bar{b} \sigma_{A,A} G_1^2$ \\
  \hline
  \multirow{2}{*}{Isometry} & $\bar{b} G_2^1 G_1^2 =$ \\
                            & $G_1^1$ \\
  \hline
  \multirow{2}{*}{Frobenius}    & $\bar{b} (\bar{t} id_A G_2^1) (\bar{t} G_1^2 id_A) =$ \\
                                & $\bar{b} G_1^2 G_2^1 =$ \\
                                & $\bar{b} (\bar{t} G_2^1 id_A) (\bar{t} id_A G_1^2)$ \\
  \hline
\end{tabular}
\end{center}

\section{Dualisers}
\label{Section:Dualisers}
In this subsection, we will demonstrate how the \textit{dagger lambda calculus} with \textit{classical structures} can be used to present the notion of a dualiser \cite{CPP08}; an explicit witness of the passage from one object to its dual or, in our case, from one type to its linear negation. While \cite{CPP10} and \cite{CD11} do not use this form of dualiser, we consider them very important from a programming language perspective as they allow us to differentiate inputs from outputs and make the flow of information explicit in the language. After defining the dualiser, the rest of the subsection will show how this notion makes the Currying rule admissible in our language.

\begin{definition}[Dualiser]
For every type $A$, we define a constant $d_A : A \multimap A^*$, the \textit{dualiser}, as a shorthand for $G_2^0 : A^* \otimes A^*$. In other words, the dualiser will be $d_A := (G_0^2)_* = G_2^0 : A^* \otimes A^*$. The sequents $[d_A]$ and $[d_A]^\dag$ that represent the $\eta$-expanded form of the dualiser and its dual are:
\begin{center}
$x_1 : A \vdash_{\{ d_A : x_{1*} \otimes x_{2*} \}}  x_{2*} : A^*$ \hskip 1cm $x_{2*} : A^*    \vdash_{\{ d_{A*} : x_2 \otimes x_1 \}}     x_1 : A$
\end{center}
\end{definition}
\[
% Generated with LaTeXDraw 2.0.8
% Mon Apr 23 22:27:15 BST 2012
% \usepackage[usenames,dvipsnames]{pstricks}
% \usepackage{epsfig}
% \usepackage{pst-grad} % For gradients
% \usepackage{pst-plot} % For axes
\scalebox{1} % Change this value to rescale the drawing.
{
\begin{pspicture}(0,-1.62)(9.42,1.62)
\psline[linewidth=0.04cm,arrowsize=0.05291667cm 2.0,arrowlength=1.4,arrowinset=0.4]{->}(1.7,-1.6)(1.7,1.0)
\psarc[linewidth=0.04](2.3,0.9){0.6}{0.0}{180.0}
\psline[linewidth=0.04cm](3.5,-0.7)(3.5,-1.0)
\pscircle[linewidth=0.04,dimen=outer,fillstyle=solid,fillcolor=green](3.5,-1.2){0.3}
\psline[linewidth=0.04cm,arrowsize=0.05291667cm 2.0,arrowlength=1.4,arrowinset=0.4]{->}(3.5,-0.4)(3.5,-0.8)
\psarc[linewidth=0.04](3.5,0.5){0.6}{-180.0}{0.0}
\psline[linewidth=0.04cm,arrowsize=0.05291667cm 2.0,arrowlength=1.4,arrowinset=0.4]{<-}(4.1,0.4)(4.1,0.9)
\psline[linewidth=0.04cm,arrowsize=0.05291667cm 2.0,arrowlength=1.4,arrowinset=0.4]{<-}(2.9,0.4)(2.9,0.9)
\pscircle[linewidth=0.04,dimen=outer,fillstyle=solid,fillcolor=green](3.5,-0.2){0.3}
\psline[linewidth=0.04cm](4.1,0.9)(4.1,1.6)
\psline[linewidth=0.04cm](0.0,-1.6)(0.0,-0.5)
\psline[linewidth=0.04cm](0.0,0.7)(0.0,1.6)
\psline[linewidth=0.04cm,linecolor=green,arrowsize=0.05291667cm 2.0,arrowlength=1.4,arrowinset=0.4]{->}(0.0,-0.5)(0.0,0.1)
\psline[linewidth=0.04cm,linecolor=green,arrowsize=0.05291667cm 2.0,arrowlength=1.4,arrowinset=0.4]{<-}(0.0,0.1)(0.0,0.7)
\usefont{T1}{pcr}{m}{n}
\rput(0.36609375,0.155){\Huge :}
\usefont{T1}{pcr}{m}{n}
\rput(0.92078125,0.055){\Huge =}
\usefont{T1}{ptm}{m}{n}
\rput(0.1365625,-1.415){\scriptsize 1}
\usefont{T1}{ptm}{m}{n}
\rput(0.16125,1.385){\scriptsize 2}
\usefont{T1}{ptm}{m}{n}
\rput(1.8365625,-1.415){\scriptsize 1}
\usefont{T1}{ptm}{m}{n}
\rput(3.96125,1.385){\scriptsize 2}
\psline[linewidth=0.04cm,arrowsize=0.05291667cm 2.0,arrowlength=1.4,arrowinset=0.4]{<-}(7.0,-1.6)(7.0,1.0)
\psarc[linewidth=0.04](7.6,0.9){0.6}{0.0}{180.0}
\psline[linewidth=0.04cm,arrowsize=0.05291667cm 2.0,arrowlength=1.4,arrowinset=0.4]{<-}(8.8,-0.6)(8.8,-1.0)
\pscircle[linewidth=0.04,dimen=outer,fillstyle=solid,fillcolor=green](8.8,-1.2){0.3}
\psline[linewidth=0.04cm](8.8,-0.3)(8.8,-0.7)
\psarc[linewidth=0.04](8.8,0.5){0.6}{180.0}{0.0}
\psline[linewidth=0.04cm,arrowsize=0.05291667cm 2.0,arrowlength=1.4,arrowinset=0.4]{->}(9.4,0.4)(9.4,1.6)
\psline[linewidth=0.04cm,arrowsize=0.05291667cm 2.0,arrowlength=1.4,arrowinset=0.4]{->}(8.2,0.4)(8.2,1.0)
\pscircle[linewidth=0.04,dimen=outer,fillstyle=solid,fillcolor=green](8.8,-0.2){0.3}
\psline[linewidth=0.04cm](5.3,-1.6)(5.3,-0.4)
\psline[linewidth=0.04cm](5.3,0.6)(5.3,1.6)
\psline[linewidth=0.04cm,linecolor=green,arrowsize=0.05291667cm 2.0,arrowlength=1.4,arrowinset=0.4]{<-}(5.3,-0.5)(5.3,0.1)
\psline[linewidth=0.04cm,linecolor=green,arrowsize=0.05291667cm 2.0,arrowlength=1.4,arrowinset=0.4]{->}(5.3,0.1)(5.3,0.7)
\usefont{T1}{pcr}{m}{n}
\rput(5.666094,0.155){\Huge :}
\usefont{T1}{pcr}{m}{n}
\rput(6.2207813,0.055){\Huge =}
\usefont{T1}{ptm}{m}{n}
\rput(5.4365625,-1.415){\scriptsize 1}
\usefont{T1}{ptm}{m}{n}
\rput(5.46125,1.385){\scriptsize 2}
\usefont{T1}{ptm}{m}{n}
\rput(7.1365623,-1.415){\scriptsize 1}
\usefont{T1}{ptm}{m}{n}
\rput(9.26125,1.385){\scriptsize 2}
\end{pspicture}
}
\]

\begin{theorem}[Unitarity of the dualiser]
The sequent representing the $\eta$-expanded form of the dualiser is unitary. In other words, we can compose a dualiser sequent $[d_A]$ with its dagger $[d_A]^\dag$, via a Cut, and the result will be an identity sequent:
\begin{prooftree}
    \AxiomC{$x_1 : A            \vdash_{\{ d_A : x_{1*} \otimes x_{2*} \}}  x_{2*} : A^*$}
    \AxiomC{$x_3 : A            \vdash_{\{ d_A : x_{3*} \otimes x_{2*} \}}  x_{2*} : A^*$}
    \RightLabel{\footnotesize $\dag$-flip}
    \UnaryInfC{$x_{2*} : A^*    \vdash_{\{ d_{A*} : x_2 \otimes x_3 \}}     x_3 : A$}
    \RightLabel{\footnotesize Cut}
    \BinaryInfC{$x_1:A \vdash_{\{ d_A : x_{1*} \otimes x_{2*}, d_{A*} : x_2 \otimes x_3 \}} x_3 : A$}
\end{prooftree}
whose soup is equal to $\{ G_1^1 : x_{1*} \otimes x_3 \}$.
\end{theorem}
\begin{proof}
Consider the following soup reduction steps:
\begin{align*}
\{ d_A : x_{1*} \otimes x_{2*}, d_{A*} : x_2 \otimes x_3 \} & \rightarrow \\
\{ G_2^0 : x_{1*} \otimes x_{2*}, G_0^2 : x_2 \otimes x_3 \} & \rightarrow \\
\{ G_2^1 : (x_1 \otimes x_2)_* \otimes G_0^1, G_1^2 : (G_0^1)_* \otimes x_2 \otimes x_3 \} & \rightarrow \text{(Frobenius)} \\
\{ G_2^1 : (x_1 \otimes G_0^1)_* \otimes x_2, G_1^2 : x_{2*} \otimes G_0^1 \otimes x_3 \} & \rightarrow \text{(Identity)}\\
\{ G_1^1 : x_{1*} \otimes x_2, G_1^1 : x_{2*} \otimes x_3 \} & \rightarrow \\
\{ x_1 : x_2, x_2 : x_3 \} & \rightarrow \\
\{ x_1 : x_3 \}
\end{align*}
\end{proof}

Having properly defined dualisers in the dagger lambda calculus, we proceed to show how classical structures and dualisers can be used to reconstruct the Currying rule.

\begin{theorem}[Admissibility of the Currying rule]
The Currying rule is admissible in a dagger lambda calculus with classical structures and dualisers.
\end{theorem}
\begin{proof}
We can construct a big sequent following the steps outlined in this proof tree:
\vskip .1cm
\hskip -2cm
    \AxiomC{$\vdash_{\{ G_0^2 : x_2 \otimes x_3 \}} x_2 \otimes x_3 : A \otimes A$}
    \UnaryInfC{$\vdash_{\{ G_1^2 : (G_0^1)_* \otimes x_2 \otimes x_3 \}} x_2 \otimes x_3 : A \otimes A$}
    \AxiomC{$x_2 : A \vdash_{\{ d_A : x_{2*} \otimes x_{1*} \}} x_{1*} : A^*$}
    \UnaryInfC{$x_2 : A \vdash_{\{ G_2^0 : x_{2*} \otimes x_{1*} \}} x_{1*} : A^*$}
    \UnaryInfC{$x_2 : A \vdash_{\{ G_2^1 : x_{2*} \otimes x_{1*} \otimes G_0^1 \}} x_{1*} : A^*$}
    \AxiomC{$x_3 : A \vdash x_3 : A$}
    \BinaryInfC{$x_2 : A, x_3 : A \vdash_{\{ G_2^1 : x_{2*} \otimes x_{1*} \otimes G_0^1 \}} x_{1*} \otimes x_3 : A^* \otimes A$}
    \BinaryInfC{$\vdash_{\{ G_1^2 : (G_0^1)_* \otimes x_2 \otimes x_3, G_2^1 : (x_1 \otimes x_2)_* \otimes G_0^1 \}} x_{1*} \otimes x_3 : A^* \otimes A$}
\DisplayProof \\
\vskip .1cm

\noindent We can then use the following reduction steps on the resulting sequent's soup:
\begin{align*}
\{ G_1^2 : (G_0^1)_* \otimes x_2 \otimes x_3, G_2^1 : (x_1 \otimes x_2)_* \otimes G_0^1 \} & \rightarrow (Frobenius) \\
\{ G_2^1 : (x_1 \otimes G_0^1)_* \otimes x_2, G_1^2 : x_{2*} \otimes G_0^1 \otimes x_3 \} & \rightarrow (Identity) \\
\{ G_1^1 : x_{1*} \otimes x_2, G_1^1 : x_{2*} \otimes x_3 \} & \rightarrow \\
\{ x_1 : x_2, x_2 : x_3 \} & \rightarrow \\
\{ x_1 : x_3 \}
\end{align*}
This gives us $\vdash x_* \otimes x : A^* \otimes A$, which is called a \textit{cup}. Once we have this term, we can Cut it with any sequent of the form $a:A \vdash_S b:B$ to reconstruct the Currying rule:
\begin{prooftree}
    \AxiomC{$\vdash x_* \otimes x : A^* \otimes A$}
    \AxiomC{$x_{1*} : A^* \vdash x_{1*} : A^*$}
    \AxiomC{$a:A \vdash_S b:B$}
    \RightLabel{\footnotesize $\bigotimes R$}
    \BinaryInfC{$x_{1*} : A^*, a:A \vdash_S x_{1*} \otimes b : A^* \otimes B$}
    \BinaryInfC{$\vdash_{S \cup \{ x_* \otimes x : x_{1*} \otimes a \}} x_{1*} \otimes b : A^* \otimes B$}
    \UnaryInfC{$\vdash_S a_* \otimes b : A^* \otimes B$}
\end{prooftree}
\end{proof}

\section{Monoidal product of terms and phase shifts}
\label{Section:Monoidal product of terms and phase shifts}
This section shows how the monoidal operation defined by $G_2^1$ can be used to fuse together the terms of the dagger lambda calculus, or lift them into \textit{phase shifts} in a sequent. This usage of the monoidal operation, described in \cite{CD08} and \cite{CD11}, makes our language more expressive by allowing us to introduce rotations on quantum states.

\begin{definition}[Monoidal product]
We define the notation $G[\phi \odot \psi]_0^1 : A$ to represent the product generated by the action of the monoidal operation $G_2^1 : A \otimes A \multimap A$ on two terms $\phi$ and $\psi$. More specifically, we set $G[\phi \odot \psi]_0^1 := G_2^1(\phi \otimes \psi)$.
\end{definition}
\[
% Generated with LaTeXDraw 2.0.8
% Thu May 03 21:52:02 BST 2012
% \usepackage[usenames,dvipsnames]{pstricks}
% \usepackage{epsfig}
% \usepackage{pst-grad} % For gradients
% \usepackage{pst-plot} % For axes
\scalebox{1} % Change this value to rescale the drawing.
{
\begin{pspicture}(0,-1.19)(8.317187,1.19)
\psline[linewidth=0.04cm](2.659375,1.17)(2.659375,0.57)
\psarc[linewidth=0.04](2.659375,-0.23){0.6}{0.0}{180.0}
\pscircle[linewidth=0.04,dimen=outer,fillstyle=solid,fillcolor=green](2.659375,0.47){0.3}
\psline[linewidth=0.04cm](2.059375,-0.19)(2.059375,-0.79)
\pscircle[linewidth=0.04,dimen=outer,fillstyle=solid,fillcolor=black](2.059375,-0.89){0.3}
\psline[linewidth=0.04cm](3.259375,-0.19)(3.259375,-0.79)
\pscircle[linewidth=0.04,dimen=outer,fillstyle=solid,fillcolor=black](3.259375,-0.89){0.3}
\usefont{T1}{ptm}{m}{n}
\rput(3.2679687,-0.9){\color{white}\boldmath$\psi$\unboldmath}
\usefont{T1}{pcr}{m}{n}
\rput(4.180156,0.0050){\Huge =}
\psline[linewidth=0.04cm](5.359375,0.15)(5.359375,0.95)
\psellipse[linewidth=0.04,dimen=outer,fillstyle=solid,fillcolor=black](5.359375,0.05)(0.6,0.3)
\usefont{T1}{ptm}{m}{n}
\rput(5.3679686,0.06){\color{white}\boldmath$\phi\!\odot\!\psi$\unboldmath}
\usefont{T1}{pcr}{m}{n}
\rput(3.6254687,0.105){\Huge :}
\usefont{T1}{ptm}{m}{n}
\rput(2.0679688,-0.9){\color{white}\boldmath$\phi$\unboldmath}
\end{pspicture}
}
\]

\begin{corollary}[Associativity and commutativity of $\odot$]
The $\odot$ operator inherits associativity and commutativity from the monoidal associativity and commutativity conditions of $G_2^1$. This allows us to introduce the following notation for the fusion of multiple terms:
\[ \bigodot \psi_i := \psi_1 \odot \ldots \odot \psi_n \]
\end{corollary}

Given a classical structure, we now define a way of performing rotations against the axis defined by that observable structure. These rotations, also known as \textit{phase shifts}, are more easily visualised as rotations of a qubit's vector in the Bloch sphere \cite{JNN12}:
\begin{center}
    \includegraphics[width=5cm]{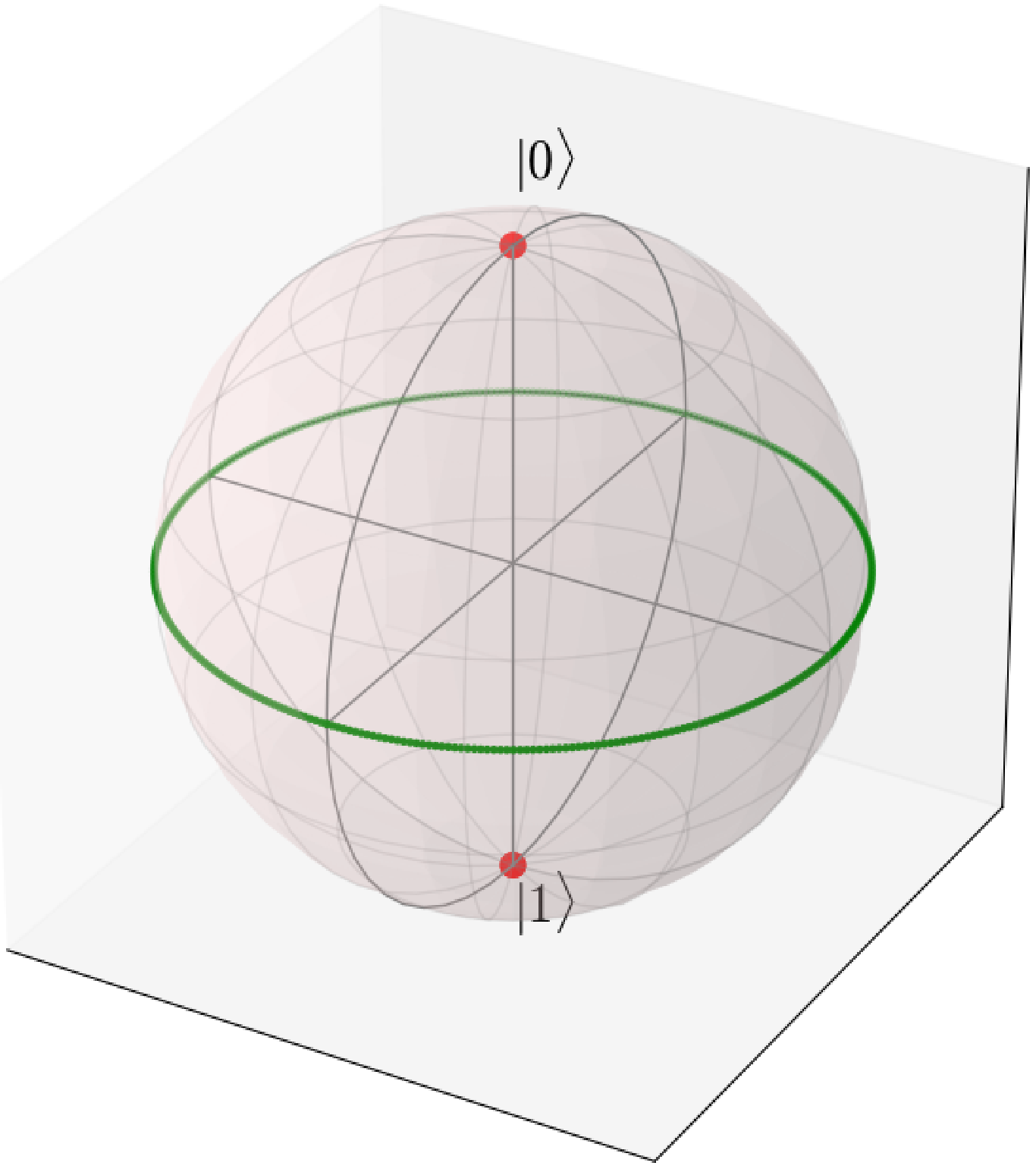}
    \includegraphics[width=5cm]{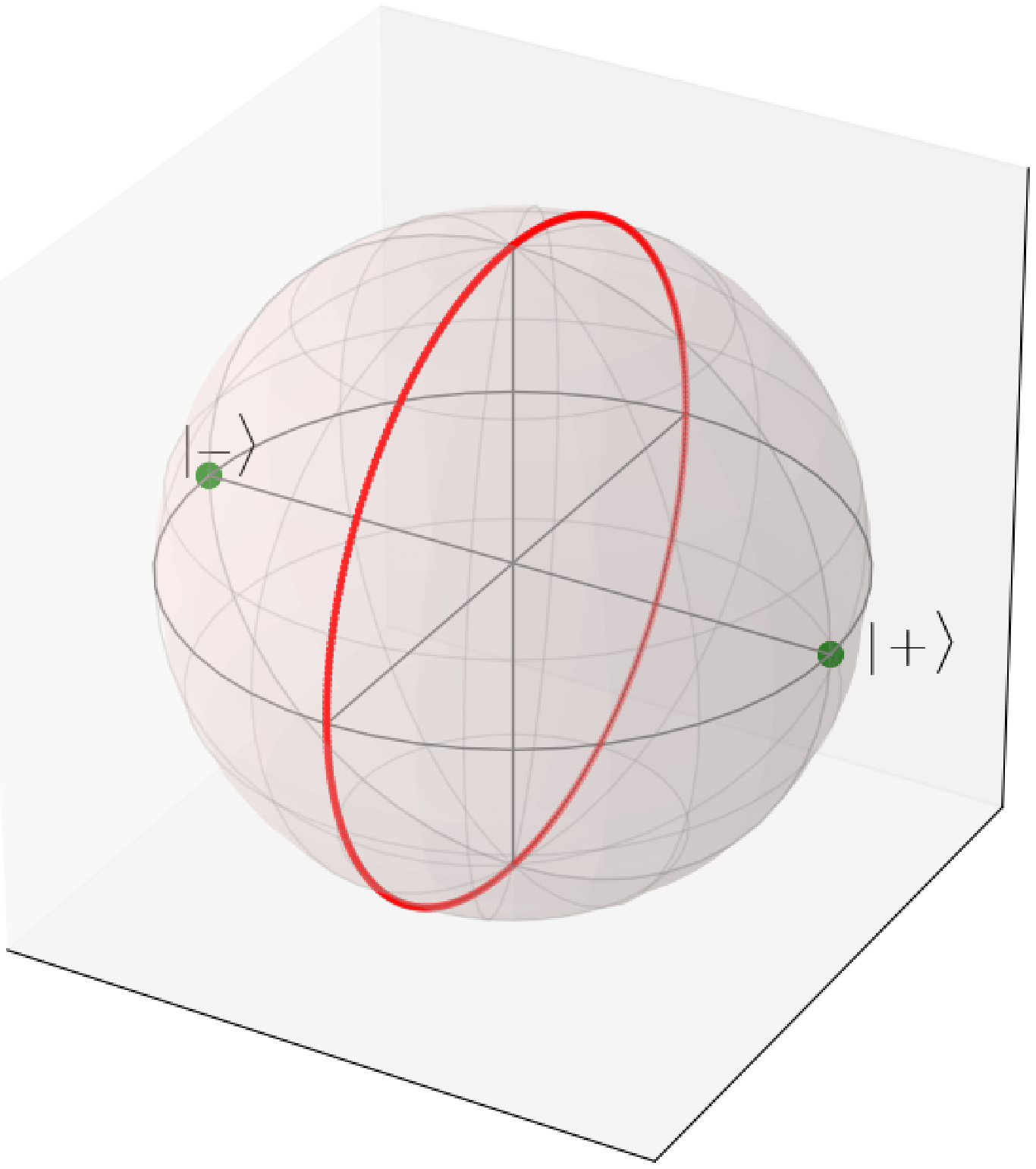}
\end{center}

\begin{definition}[Phase shift]
The \textit{phase shift} generated by the action of $G_2^1$ on a term $\psi$ is represented by the $\eta$-expanded form of the term $G[\psi]_1^1 := \lambda x.G_2^1(\psi \otimes x)$:
\[ x_1:A \vdash_{\{ G[\psi]_1^1 : x_{1*} \otimes x_2 \}} x_2:A \;\;\; \rightarrow x:A \vdash G[\psi \otimes x]_0^1 : A \]
\end{definition}
\[
% Generated with LaTeXDraw 2.0.8
% Thu May 03 22:21:02 BST 2012
% \usepackage[usenames,dvipsnames]{pstricks}
% \usepackage{epsfig}
% \usepackage{pst-grad} % For gradients
% \usepackage{pst-plot} % For axes
\scalebox{1} % Change this value to rescale the drawing.
{
\begin{pspicture}(0,-1.17)(6.1371875,1.17)
\psline[linewidth=0.04cm](4.659375,1.15)(4.659375,0.55)
\psarc[linewidth=0.04](4.659375,-0.27){0.6}{0.0}{180.0}
\pscircle[linewidth=0.04,dimen=outer,fillstyle=solid,fillcolor=green](4.659375,0.43){0.3}
\psline[linewidth=0.04cm](4.059375,-0.21)(4.059375,-0.75)
\pscircle[linewidth=0.04,dimen=outer,fillstyle=solid,fillcolor=black](4.059375,-0.85){0.3}
\usefont{T1}{ptm}{m}{n}
\rput(4.047969,-0.86){\color{white}\boldmath$\psi$\unboldmath}
\psline[linewidth=0.04cm](5.259375,-0.23)(5.259375,-1.15)
\usefont{T1}{pcr}{m}{n}
\rput(3.1801562,0.0050){\Huge =}
\psline[linewidth=0.04cm](2.059375,1.15)(2.059375,-1.15)
\pscircle[linewidth=0.04,dimen=outer,fillstyle=solid,fillcolor=green](2.059375,0.01){0.3}
\usefont{T1}{ptm}{m}{n}
\rput(2.0479689,0.0){\boldmath$\psi$\unboldmath}
\usefont{T1}{pcr}{m}{n}
\rput(2.6254687,0.105){\Huge :}
\end{pspicture}
}
\]

\noindent The composition, under Cut, of the phase shifts generated by two terms $\phi$ and $\psi$ is:
\begin{prooftree}
    \AxiomC{$x_1:A \vdash_{\{ G[\phi]_1^1 : x_{1*} \otimes x_2 \}} x_2:A$}
    \AxiomC{$x_2:A \vdash_{\{ G[\psi]_1^1 : x_{2*} \otimes x_3 \}} x_3:A$}
    \RightLabel{\footnotesize Cut}
    \BinaryInfC{$x_1:A \vdash_{\{ G_2^1 : (\phi \otimes x_1)_* \otimes x_2, G_2^1 : (\psi \otimes x_2)_* \otimes x_3 \}} x_3:A$}
    \RightLabel{\footnotesize Monoidal associativity}
    \UnaryInfC{$x_1:A \vdash_{\{ G_2^1 : (\phi \otimes \psi)_* \otimes x_2, G_2^1 : (x_2 \otimes x_1)_* \otimes x_3 \}} x_3:A$}
\end{prooftree}

\noindent The resulting sequent can be rewritten as
\[ x_1:A \vdash_{\{ G_2^1 : (G[\phi \odot \psi]_0^1 \otimes x_1)_* \otimes x_3 \}} x_3:A \]

\noindent Which is actually a lifting of the monoidal product of those terms
\[ x_1:A \vdash_{\{ G[G[\phi \odot \psi]_0^1]_1^1 : x_{1*} \otimes x_3 \}} x_3:A \]

\begin{definition}[Lifting of the monoidal product of terms]
We define a shorthand for the lifting of the monoidal product of terms:
\[ G[\phi \odot \psi]_1^1 := G[G[\phi \odot \psi]_0^1]_1^1 \]
\end{definition}

\begin{corollary}[Phase shift commutativity]
Phase shifts generated by the action of $G_2^1$ are commutative under Cut:
\[ G[\phi \odot \psi]_1^1 = G[\psi \odot \phi]_1^1 \]
\end{corollary}
\begin{proof}
\noindent The composition of the phase shift corresponding to a term $\phi$, with that of a term $\psi$, produces:
\[ x_1:A \vdash_{\{ G_2^1 : (\phi \otimes \psi)_* \otimes x_2, G_2^1 : (x_2 \otimes x_1)_* \otimes x_3 \}} x_3:A \]

\noindent Composing the phase shifts the other way around results in:
\[ x_1:A \vdash_{\{ G_2^1 : (\psi \otimes \phi)_* \otimes x_2, G_2^1 : (x_2 \otimes x_1)_* \otimes x_3 \}} x_3:A \]

\noindent The monoidal commutativity condition allows us to transform $G_2^1 : (\phi \otimes \psi)_* \otimes x_2$ into $G_2^1 : (\psi \otimes \phi)_* \otimes x_2$, in our soup, which makes the two sequents equal. It therefore follows that the resulting sequents, $G[\phi \odot \psi]_1^1$ and $G[\psi \odot \phi]_1^1$, are also equal.
\[
% Generated with LaTeXDraw 2.0.8
% Sat May 05 13:26:15 BST 2012
% \usepackage[usenames,dvipsnames]{pstricks}
% \usepackage{epsfig}
% \usepackage{pst-grad} % For gradients
% \usepackage{pst-plot} % For axes
\scalebox{1} % Change this value to rescale the drawing.
{
\begin{pspicture}(0,-1.92)(10.377188,1.92)
\psline[linewidth=0.04cm](5.179375,1.9)(5.179375,0.88)
\psarc[linewidth=0.04](5.179375,0.06){0.6}{0.0}{180.0}
\pscircle[linewidth=0.04,dimen=outer,fillstyle=solid,fillcolor=green](5.179375,0.76){0.3}
\psline[linewidth=0.04cm](5.779375,0.14)(5.779375,-1.9)
\psline[linewidth=0.04cm](2.059375,1.88)(2.059375,-0.42)
\pscircle[linewidth=0.04,dimen=outer,fillstyle=solid,fillcolor=green](2.059375,0.74){0.3}
\usefont{T1}{ptm}{m}{n}
\rput(2.0479689,0.73){\boldmath$\psi$\unboldmath}
\psline[linewidth=0.04cm](4.579375,-0.66)(4.579375,0.14)
\psellipse[linewidth=0.04,dimen=outer,fillstyle=solid,fillcolor=black](4.579375,-0.76)(0.6,0.3)
\psline[linewidth=0.04cm](2.059375,0.4)(2.059375,-1.9)
\pscircle[linewidth=0.04,dimen=outer,fillstyle=solid,fillcolor=green](2.059375,-0.74){0.3}
\usefont{T1}{ptm}{m}{n}
\rput(2.0679688,-0.75){\boldmath$\phi$\unboldmath}
\usefont{T1}{pcr}{m}{n}
\rput(3.3001564,0.015){\Huge =}
\usefont{T1}{ptm}{m}{n}
\rput(4.587969,-0.75){\color{white}\boldmath$\phi\!\odot\!\psi$\unboldmath}
\usefont{T1}{pcr}{m}{n}
\rput(6.900156,0.015){\Huge =}
\psline[linewidth=0.04cm](8.259375,1.88)(8.259375,-0.42)
\pscircle[linewidth=0.04,dimen=outer,fillstyle=solid,fillcolor=green](8.259375,0.74){0.3}
\usefont{T1}{ptm}{m}{n}
\rput(8.267969,0.73){\boldmath$\phi$\unboldmath}
\psline[linewidth=0.04cm](8.259375,0.4)(8.259375,-1.9)
\pscircle[linewidth=0.04,dimen=outer,fillstyle=solid,fillcolor=green](8.259375,-0.74){0.3}
\usefont{T1}{ptm}{m}{n}
\rput(8.247969,-0.75){\boldmath$\psi$\unboldmath}
\end{pspicture}
}
\]
\end{proof}

\section{Unbiased and classical constants}
\label{Section:Unbiased and classical constants}
This section establishes the properties of unbiasedness and classicality for constants, with respect to a given observable structure. The notion of unbiasedness differs slightly from the one presented in \cite{CD08} and \cite{CD11}, as it uses the dualisers of \cite{CPP08} to allow for non self-dual structures. The properties of unbiasedness and classicality will later be used in defining the interaction between complementary observables.

\begin{definition}[Unbiasedness (dagger lambda calculus)]
We say that a constant $k:A$ is \textit{unbiased} with respect to an observable structure $(A, G_1^2, G_1^0)$ when there exists a scalar $s:I$ such that we can match the sequents:
\[ \vdash_{\left\{ d_{A*} : k \otimes x_1, G_2^1 : (k \otimes x_1)_* \otimes x_2, s:1 \right\}} x_2:A \text{ and } \vdash_{\left\{G_0^1 : x_2\right\}} x_2:A \]
\noindent by equating their soups:
\[ \left\{ d_{A*} : k \otimes x_1, G_2^1 : (k \otimes x_1)_* \otimes x_2, s:1 \right\} = \left\{G_0^1 : x_2\right\}\]

\[
% Generated with LaTeXDraw 2.0.8
% Sat May 05 21:05:49 BST 2012
% \usepackage[usenames,dvipsnames]{pstricks}
% \usepackage{epsfig}
% \usepackage{pst-grad} % For gradients
% \usepackage{pst-plot} % For axes
\scalebox{1} % Change this value to rescale the drawing.
{
\begin{pspicture}(0,-1.5)(5.9171877,1.52)
\psline[linewidth=0.04cm,arrowsize=0.05291667cm 2.0,arrowlength=1.4,arrowinset=0.4]{<-}(3.199375,1.5)(3.199375,0.8)
\psarc[linewidth=0.04](3.199375,-0.02){0.6}{0.0}{180.0}
\pscircle[linewidth=0.04,dimen=outer,fillstyle=solid,fillcolor=green](3.199375,0.68){0.3}
\psline[linewidth=0.04cm,arrowsize=0.05291667cm 2.0,arrowlength=1.4,arrowinset=0.4]{<-}(2.599375,0.04)(2.599375,-1.0)
\pscircle[linewidth=0.04,dimen=outer,fillstyle=solid,fillcolor=black](2.599375,-1.2){0.3}
\usefont{T1}{ptm}{m}{n}
\rput(2.6279688,-1.21){\color{white}\boldmath$k$\unboldmath}
\psline[linewidth=0.04cm,linecolor=green,arrowsize=0.05291667cm 2.0,arrowlength=1.4,arrowinset=0.4]{->}(3.799375,-0.4)(3.799375,0.0)
\psline[linewidth=0.04cm](3.799375,-0.56)(3.799375,-1.1)
\pscircle[linewidth=0.04,dimen=outer,fillstyle=solid,fillcolor=black](3.799375,-1.2){0.3}
\usefont{T1}{ptm}{m}{n}
\rput(3.8079689,-1.21){\color{white}\boldmath$k_*$\unboldmath}
\psline[linewidth=0.04cm,linecolor=green,arrowsize=0.05291667cm 2.0,arrowlength=1.4,arrowinset=0.4]{<-}(3.799375,-0.7)(3.799375,-0.3)
\usefont{T1}{pcr}{m}{n}
\rput(4.5201564,0.055){\Huge =}
\psline[linewidth=0.04cm,arrowsize=0.05291667cm 2.0,arrowlength=1.4,arrowinset=0.4]{<-}(5.499375,1.5)(5.499375,-0.1)
\pscircle[linewidth=0.04,dimen=outer,fillstyle=solid,fillcolor=green](5.499375,-0.22){0.3}
\usefont{T1}{ptm}{m}{n}
\rput(3.5359375,0.185){\scriptsize 1}
\usefont{T1}{ptm}{m}{n}
\rput(3.360625,1.185){\scriptsize 2}
\usefont{T1}{ptm}{m}{n}
\rput(5.660625,0.885){\scriptsize 2}
\rput{-43.291065}(0.5168585,1.3024106){\psframe[linewidth=0.04,dimen=outer,fillstyle=solid,fillcolor=black](2.099375,0.2)(1.699375,-0.2)}
\usefont{T1}{ptm}{m}{n}
\rput(1.9079688,-0.01){\color{white}\boldmath$s$\unboldmath}
\end{pspicture}
}
\]
\end{definition}

\noindent Intuitively, a constant is said to be \textit{classical}, with respect to a given observable structure, when it is copied and deleted by that structure. More formally, this is defined as:

\begin{definition}[Classicality (dagger lambda calculus)]
We say that a constant $k:A$ is \textit{classical}, with respect to an observable structure $(A, G_1^2, G_1^0)$, when we can match the following pairs of sequents:
\[ \vdash_{\left\{ G_1^2 : k_* \otimes x_1 \otimes x_2 \right\}} x_1 \otimes x_2 : A \otimes A \text{ with } \vdash_{\left\{ k:x_1, k:x_2 \right\}} x_1 \otimes x_2 : A \otimes A \]
\noindent and
\[ \vdash_{\left\{ G_1^0 : k_* \right\}} \text{ with } \vdash \]
\noindent by equating their respective soups:
\[ \left\{ G_1^2 : k_* \otimes x_1 \otimes x_2 \right\} = \left\{ k:x_1, k:x_2 \right\} \]
\noindent and
\[ \left\{ G_1^0 : k_* \right\} = \emptyset \]
\[
% Generated with LaTeXDraw 2.0.8
% Sat May 05 20:52:03 BST 2012
% \usepackage[usenames,dvipsnames]{pstricks}
% \usepackage{epsfig}
% \usepackage{pst-grad} % For gradients
% \usepackage{pst-plot} % For axes
\scalebox{1} % Change this value to rescale the drawing.
{
\begin{pspicture}(0,-1.42)(9.1171875,1.44)
\usefont{T1}{pcr}{m}{n}
\rput(2.8801563,-0.065){\Huge =}
\psline[linewidth=0.04cm](1.859375,0.08)(1.859375,-0.52)
\psarc[linewidth=0.04](1.859375,0.88){0.6}{-180.0}{0.0}
\pscircle[linewidth=0.04,dimen=outer,fillstyle=solid,fillcolor=green](1.859375,0.18){0.3}
\psline[linewidth=0.04cm,arrowsize=0.05291667cm 2.0,arrowlength=1.4,arrowinset=0.4]{<-}(1.859375,-0.32)(1.859375,-0.92)
\pscircle[linewidth=0.04,dimen=outer,fillstyle=solid,fillcolor=black](1.859375,-1.12){0.3}
\usefont{T1}{ptm}{m}{n}
\rput(1.8879688,-1.13){\color{white}\boldmath$k$\unboldmath}
\psline[linewidth=0.04cm,arrowsize=0.05291667cm 2.0,arrowlength=1.4,arrowinset=0.4]{<-}(1.259375,1.38)(1.259375,0.88)
\psline[linewidth=0.04cm,arrowsize=0.05291667cm 2.0,arrowlength=1.4,arrowinset=0.4]{<-}(2.459375,1.42)(2.459375,0.88)
\psline[linewidth=0.04cm,arrowsize=0.05291667cm 2.0,arrowlength=1.4,arrowinset=0.4]{<-}(3.759375,0.88)(3.759375,-0.32)
\pscircle[linewidth=0.04,dimen=outer,fillstyle=solid,fillcolor=black](3.759375,-0.52){0.3}
\usefont{T1}{ptm}{m}{n}
\rput(3.7879686,-0.53){\color{white}\boldmath$k$\unboldmath}
\psline[linewidth=0.04cm,arrowsize=0.05291667cm 2.0,arrowlength=1.4,arrowinset=0.4]{<-}(4.859375,0.88)(4.859375,-0.32)
\pscircle[linewidth=0.04,dimen=outer,fillstyle=solid,fillcolor=black](4.859375,-0.52){0.3}
\usefont{T1}{ptm}{m}{n}
\rput(4.8879685,-0.53){\color{white}\boldmath$k$\unboldmath}
\usefont{T1}{ptm}{m}{n}
\rput(1.3959374,0.965){\scriptsize 1}
\usefont{T1}{ptm}{m}{n}
\rput(2.320625,0.965){\scriptsize 2}
\usefont{T1}{ptm}{m}{n}
\rput(3.8959374,0.465){\scriptsize 1}
\usefont{T1}{ptm}{m}{n}
\rput(4.720625,0.465){\scriptsize 2}
\psline[linewidth=0.04cm](7.159375,0.68)(7.159375,0.08)
\pscircle[linewidth=0.04,dimen=outer,fillstyle=solid,fillcolor=green](7.159375,0.78){0.3}
\psline[linewidth=0.04cm,arrowsize=0.05291667cm 2.0,arrowlength=1.4,arrowinset=0.4]{<-}(7.159375,0.28)(7.159375,-0.32)
\pscircle[linewidth=0.04,dimen=outer,fillstyle=solid,fillcolor=black](7.159375,-0.52){0.3}
\usefont{T1}{ptm}{m}{n}
\rput(7.1879687,-0.53){\color{white}\boldmath$k$\unboldmath}
\usefont{T1}{pcr}{m}{n}
\rput(8.080156,-0.065){\Huge =}
\end{pspicture}
}
\]
\end{definition}

\section{Complementary observables}
\label{Section:Complementary observables}
This section introduces two interacting observable structures to our lambda calculus, a Green one and a Red one, as the final step towards axiomatising basis structures in the dagger lambda calculus. The two observable structures are similar to the ones described in the later parts of \cite{CD11}, with the added support for non self-dual $\dag$-compact structures and the dualisers of \cite{CPP08}. The interaction between the two observable structures is defined by requiring that they satisfy a \textit{Bi-algebra} and a \textit{Hopf law} condition. We then show that the interaction between Green and Red makes them \textit{complementary}, as they possess certain properties with regards to classical and unbiased constants. Finally, we define a function that can be used to transform the Green sequents into Red ones and vice versa.

\begin{definition}[Complementary observable structures (dagger lambda calculus)]
For every type $A$ in the dagger lambda calculus, we define two observable structures: A \textit{Green} structure $(A, G_1^2, G_1^0)$ and a \textit{Red} structure $(A, R_1^2, R_1^0)$. Each of these observable structures comes with a dualiser; we will use $d_A^Z : A \multimap A^*$ to refer to the dualiser generated by the Green observable structure and $d_A^X : A \multimap A^*$ to refer to the one generated by Red.
\end{definition}

We define the interaction between these two observable structures by describing the conditions that we require them to satisfy in the dagger lambda calculus. The first such condition is the \textit{Bi-algebra} condition:
\[
% Generated with LaTeXDraw 2.0.8
% Sat May 05 23:02:17 BST 2012
% \usepackage[usenames,dvipsnames]{pstricks}
% \usepackage{epsfig}
% \usepackage{pst-grad} % For gradients
% \usepackage{pst-plot} % For axes
\scalebox{1} % Change this value to rescale the drawing.
{
\begin{pspicture}(0,-2.42)(8.908125,2.42)
\usefont{T1}{ptm}{m}{n}
\rput(3.29625,-2.275){\scriptsize 1}
\usefont{T1}{ptm}{m}{n}
\rput(5.7009373,-2.275){\scriptsize 2}
\psline[linewidth=0.04cm](3.1596875,-1.62)(3.1596875,-2.22)
\psarc[linewidth=0.04](3.1596875,-0.82){0.6}{-180.0}{0.0}
\pscircle[linewidth=0.04,dimen=outer,fillstyle=solid,fillcolor=green](3.1596875,-1.52){0.3}
\psline[linewidth=0.04cm,arrowsize=0.05291667cm 2.0,arrowlength=1.4,arrowinset=0.4]{<-}(3.1596875,2.4)(3.1596875,1.6)
\psarc[linewidth=0.04](3.1596875,0.78){0.6}{0.0}{180.0}
\pscircle[linewidth=0.04,dimen=outer,fillstyle=solid,fillcolor=red](3.1596875,1.48){0.3}
\psline[linewidth=0.04cm,arrowsize=0.05291667cm 2.0,arrowlength=1.4,arrowinset=0.4]{<-}(2.5596876,0.02)(2.5596876,-0.84)
\psline[linewidth=0.04cm,arrowsize=0.05291667cm 2.0,arrowlength=1.4,arrowinset=0.4]{<-}(5.5396876,2.4)(5.5396876,1.58)
\psarc[linewidth=0.04](5.5596876,0.78){0.6}{0.0}{180.0}
\pscircle[linewidth=0.04,dimen=outer,fillstyle=solid,fillcolor=red](5.5596876,1.48){0.3}
\psline[linewidth=0.04cm](5.5596876,-1.62)(5.5596876,-2.22)
\psarc[linewidth=0.04](5.5596876,-0.82){0.6}{-180.0}{0.0}
\pscircle[linewidth=0.04,dimen=outer,fillstyle=solid,fillcolor=green](5.5596876,-1.52){0.3}
\psarc[linewidth=0.04](4.3596873,0.58){0.6}{-180.0}{-90.0}
\psarc[linewidth=0.04](4.3596873,-0.62){0.6}{0.0}{90.0}
\psarc[linewidth=0.04](4.3596873,-0.62){0.6}{111.03751}{180.0}
\psarc[linewidth=0.04](4.3596873,0.58){0.6}{-71.56505}{0.0}
\psline[linewidth=0.04cm,arrowsize=0.05291667cm 2.0,arrowlength=1.4,arrowinset=0.4]{<-}(5.5596876,-2.0)(5.5596876,-2.4)
\psline[linewidth=0.04cm,arrowsize=0.05291667cm 2.0,arrowlength=1.4,arrowinset=0.4]{<-}(3.1596875,-2.0)(3.1596875,-2.4)
\psline[linewidth=0.04cm](2.5596876,0.8)(2.5596876,-0.06)
\psline[linewidth=0.04cm,arrowsize=0.05291667cm 2.0,arrowlength=1.4,arrowinset=0.4]{<-}(6.1596875,0.02)(6.1596875,-0.84)
\psline[linewidth=0.04cm](6.1596875,0.82)(6.1596875,-0.06)
\psline[linewidth=0.04cm,arrowsize=0.05291667cm 2.0,arrowlength=1.4,arrowinset=0.4]{<-}(3.7596874,0.84)(3.7596874,0.56)
\psline[linewidth=0.04cm,arrowsize=0.05291667cm 2.0,arrowlength=1.4,arrowinset=0.4]{<-}(4.9596877,0.84)(4.9596877,0.56)
\psline[linewidth=0.04cm,arrowsize=0.05291667cm 2.0,arrowlength=1.4,arrowinset=0.4]{<-}(3.7596874,-0.56)(3.7596874,-0.84)
\psline[linewidth=0.04cm,arrowsize=0.05291667cm 2.0,arrowlength=1.4,arrowinset=0.4]{<-}(4.9596877,-0.56)(4.9596877,-0.84)
\usefont{T1}{ptm}{m}{n}
\rput(2.6925,-0.875){\scriptsize 3}
\usefont{T1}{ptm}{m}{n}
\rput(3.6026564,-0.875){\scriptsize 4}
\usefont{T1}{ptm}{m}{n}
\rput(5.09375,-0.875){\scriptsize 5}
\usefont{T1}{ptm}{m}{n}
\rput(5.998125,-0.875){\scriptsize 6}
\usefont{T1}{ptm}{m}{n}
\rput(6.298125,1.025){\scriptsize 6}
\usefont{T1}{ptm}{m}{n}
\rput(3.89375,1.025){\scriptsize 5}
\usefont{T1}{ptm}{m}{n}
\rput(4.802656,1.025){\scriptsize 4}
\usefont{T1}{ptm}{m}{n}
\rput(2.3925,1.025){\scriptsize 3}
\usefont{T1}{ptm}{m}{n}
\rput(3.2982812,1.925){\scriptsize 7}
\usefont{T1}{ptm}{m}{n}
\rput(5.693125,1.925){\scriptsize 8}
\usefont{T1}{pcr}{m}{n}
\rput(6.760469,-0.0050){\Huge =}
\psline[linewidth=0.04cm](8.059688,0.48)(8.059688,-0.12)
\psarc[linewidth=0.04](8.059688,1.28){0.6}{-180.0}{0.0}
\pscircle[linewidth=0.04,dimen=outer,fillstyle=solid,fillcolor=green](8.059688,0.58){0.3}
\psline[linewidth=0.04cm,arrowsize=0.05291667cm 2.0,arrowlength=1.4,arrowinset=0.4]{<-}(8.059688,0.0)(8.059688,-0.8)
\psarc[linewidth=0.04](8.059688,-1.62){0.6}{0.0}{180.0}
\pscircle[linewidth=0.04,dimen=outer,fillstyle=solid,fillcolor=red](8.059688,-0.92){0.3}
\usefont{T1}{ptm}{m}{n}
\rput(8.2925,-0.175){\scriptsize 3}
\usefont{T1}{ptm}{m}{n}
\rput(7.5982814,1.625){\scriptsize 7}
\psline[linewidth=0.04cm,arrowsize=0.05291667cm 2.0,arrowlength=1.4,arrowinset=0.4]{<-}(8.659687,1.96)(8.659687,1.26)
\usefont{T1}{ptm}{m}{n}
\rput(8.793125,1.625){\scriptsize 8}
\usefont{T1}{ptm}{m}{n}
\rput(7.59625,-1.975){\scriptsize 1}
\psline[linewidth=0.04cm,arrowsize=0.05291667cm 2.0,arrowlength=1.4,arrowinset=0.4]{<-}(7.4596877,-1.56)(7.4596877,-2.16)
\usefont{T1}{ptm}{m}{n}
\rput(8.800938,-1.975){\scriptsize 2}
\psline[linewidth=0.04cm,arrowsize=0.05291667cm 2.0,arrowlength=1.4,arrowinset=0.4]{<-}(8.659687,-1.56)(8.659687,-2.16)
\psline[linewidth=0.04cm,arrowsize=0.05291667cm 2.0,arrowlength=1.4,arrowinset=0.4]{<-}(7.4596877,1.96)(7.4596877,1.26)
\rput{-44.62365}(0.42312706,1.1310842){\psframe[linewidth=0.04,dimen=outer,fillstyle=solid,fillcolor=black](1.8396875,0.3)(1.3396875,-0.2)}
\usefont{T1}{ptm}{m}{n}
\rput(1.5639062,0.045){\tiny \color{white}\boldmath$\sqrt{D}$\unboldmath}
\end{pspicture}
}
\]

\noindent which requires the following two soups, $S_{13}$ and $S_{14}$, to be equal in the lambda calculus:

\begin{center}
\begin{tabular}{ l l }
  $S_{13} =$    & $\;   \left\{
                        \begin{array}{l l}
                            G_1^2 : x_{1*} \otimes x_3 \otimes x_4,     & G_1^2 : x_{2*} \otimes x_5 \otimes x_6, \\
                            R_2^1 : (x_3 \otimes x_5)_* \otimes x_7,    & R_2^1 : (x_4 \otimes x_6)_* \otimes x_8, \\
                            \sqrt{D} : 1                                &
                        \end{array}
                        \right\}$ \\
                & \\
  $S_{14} =$    & $\; \Big\{ R_2^1 : (x_1 \otimes x_2)_* \otimes x_3, G_1^2 : x_{3*} \otimes x_7 \otimes x_8 \Big\}$
\end{tabular}
\end{center}

\noindent This allows us to relate the sequents for $x_1:A, x_2:A \vdash_{S_{13}} x_7 \otimes x_8 : A \otimes A$ and $x_1:A, x_2:A \vdash_{S_{14}} x_7 \otimes x_8 : A \otimes A$, matching the equality of arrows that is required in the categorical setting.

The next condition that we require of the interaction between the Green and Red observable structures in the dagger lambda calculus is the \textit{Hopf law} condition:
\[
% Generated with LaTeXDraw 2.0.8
% Tue May 08 23:19:17 BST 2012
% \usepackage[usenames,dvipsnames]{pstricks}
% \usepackage{epsfig}
% \usepackage{pst-grad} % For gradients
% \usepackage{pst-plot} % For axes
\scalebox{1} % Change this value to rescale the drawing.
{
\begin{pspicture}(0,-2.84)(4.22,2.84)
\psline[linewidth=0.04cm,linecolor=green,arrowsize=0.05291667cm 2.0,arrowlength=1.4,arrowinset=0.4]{->}(1.12,-0.9)(1.12,-0.46)
\usefont{T1}{ptm}{m}{n}
\rput(1.8565625,-2.695){\scriptsize 1}
\psline[linewidth=0.04cm](1.72,-2.04)(1.72,-2.64)
\psarc[linewidth=0.04](1.72,-1.24){0.6}{-180.0}{0.0}
\pscircle[linewidth=0.04,dimen=outer,fillstyle=solid,fillcolor=green](1.72,-1.94){0.3}
\psline[linewidth=0.04cm,arrowsize=0.05291667cm 2.0,arrowlength=1.4,arrowinset=0.4]{<-}(1.72,-2.42)(1.72,-2.82)
\psline[linewidth=0.04cm,arrowsize=0.05291667cm 2.0,arrowlength=1.4,arrowinset=0.4]{<-}(2.32,-0.86)(2.32,-1.26)
\usefont{T1}{ptm}{m}{n}
\rput(1.2528125,-1.295){\scriptsize 3}
\usefont{T1}{ptm}{m}{n}
\rput(2.16125,0.105){\scriptsize 2}
\usefont{T1}{pcr}{m}{n}
\rput(2.9207811,-0.025){\Huge =}
\psline[linewidth=0.04cm,arrowsize=0.05291667cm 2.0,arrowlength=1.4,arrowinset=0.4]{<-}(1.72,2.82)(1.72,2.02)
\psarc[linewidth=0.04](1.72,1.2){0.6}{0.0}{180.0}
\pscircle[linewidth=0.04,dimen=outer,fillstyle=solid,fillcolor=red](1.72,1.9){0.3}
\psline[linewidth=0.04cm,arrowsize=0.05291667cm 2.0,arrowlength=1.4,arrowinset=0.4]{<-}(1.12,1.26)(1.12,0.66)
\psline[linewidth=0.04cm,arrowsize=0.05291667cm 2.0,arrowlength=1.4,arrowinset=0.4]{<-}(2.32,1.26)(2.32,-0.98)
\psline[linewidth=0.04cm,arrowsize=0.05291667cm 2.0,arrowlength=1.4,arrowinset=0.4]{<-}(1.12,-0.86)(1.12,-1.26)
\psline[linewidth=0.04cm,linecolor=green,arrowsize=0.05291667cm 2.0,arrowlength=1.4,arrowinset=0.4]{<-}(1.12,-0.48)(1.12,-0.08)
\psline[linewidth=0.04cm,linecolor=red,arrowsize=0.05291667cm 2.0,arrowlength=1.4,arrowinset=0.4]{<->}(1.12,-0.14)(1.12,0.7)
\usefont{T1}{ptm}{m}{n}
\rput(1.2629688,-0.195){\scriptsize 4}
\usefont{T1}{ptm}{m}{n}
\rput(1.2540625,0.905){\scriptsize 5}
\usefont{T1}{ptm}{m}{n}
\rput(1.8584375,2.405){\scriptsize 6}
\psellipse[linewidth=0.04,dimen=outer](0.35,-0.02)(0.35,0.6)
\usefont{T1}{ptm}{m}{n}
\rput(4.0565624,-1.295){\scriptsize 1}
\psline[linewidth=0.04cm](3.92,-0.64)(3.92,-1.24)
\pscircle[linewidth=0.04,dimen=outer,fillstyle=solid,fillcolor=green](3.92,-0.54){0.3}
\psline[linewidth=0.04cm,arrowsize=0.05291667cm 2.0,arrowlength=1.4,arrowinset=0.4]{<-}(3.92,-1.02)(3.92,-1.42)
\psline[linewidth=0.04cm,arrowsize=0.05291667cm 2.0,arrowlength=1.4,arrowinset=0.4]{<-}(3.92,1.42)(3.92,0.62)
\pscircle[linewidth=0.04,dimen=outer,fillstyle=solid,fillcolor=red](3.92,0.5){0.3}
\usefont{T1}{ptm}{m}{n}
\rput(4.0584373,1.005){\scriptsize 6}
\end{pspicture}
}
\]

\noindent which requires the following two soups, $S_{15}$ and $S_{16}$, to be equal in the lambda calculus:

\begin{center}
\begin{tabular}{ l l }
  $S_{15} =$    & $\;   \left\{
                        \begin{array}{l l}
                            G_1^2 : x_{1*} \otimes x_3 \otimes x_2, & R_2^1 : (x_5 \otimes x_2)_* \otimes x_6, \\
                            d_A^Z : x_{3*} \otimes x_{4*},          & d_{A^*}^X : x_4 \otimes x_5, \\
                            D : 1                                   &
                        \end{array}
                        \right\}$ \\
                & \\
  $S_{16} =$    & $\; \Big\{ G_1^0 : x_{1*}, R_0^1 : x_6 \Big\}$
\end{tabular}
\end{center}

\noindent This allows us to relate the sequents for $x_1 : A \vdash_{S_{15}} x_6 : A$ and $x_1 : A \vdash_{S_{16}} x_6 : A$, or in other words $x_1 : A \vdash_{S_{15}} x_6 : A$ with $G_0^1 : A \vdash R_0^1 : A$, matching the equality of arrows that is required in the categorical setting.

\vskip .2cm

We will now show an interesting property about the interaction that we have defined against the two observable structures. When a constant is classical against the Green structure, it is unbiased against the Red one. Conversely, when a constant is classical against the Red structure, it is unbiased against the Green one.

\begin{theorem}[Complementarity]
Let $k:A$ be a constant such that $\{ G_1^2 : k_* \otimes x_1 \otimes x_2 \} = \{ k:x_1, k:x_2 \}$ and $\{ G_1^0 : k_* \}$. (1) The interaction between the two observable structures causes $\{ d_{A^*}^X : k \otimes x_5, R_2^1 : (x_5 \otimes k)_* \otimes x_6, D:1 \}$ to be equal to $\{R_0^1 : x_6\}$. Similarly, let $\ell:A$ be a constant such that $\{ R_1^2 : \ell_* \otimes x_1 \otimes x_2 \} = \{ \ell:x_1, \ell:x_2 \}$ and $\{ R_1^0 : \ell_* \}$. (2) The interaction causes $\{ d_{A^*}^Z : \ell \otimes x_5, G_2^1 : (x_5 \otimes \ell)_* \otimes x_6, D:1 \}$ to be equal to $\{G_0^1 : x_6\}$.
\end{theorem}
\begin{proof}
(1) We begin by trying to \textit{fuse} together $k_*$ and $k$, in a manner that is similar to the definition of unbiasedness. We, therefore, start with the following set of connections in our soup:
\[\{ d_{A^*}^X : k  \otimes x_5, R_2^1 : (x_5 \otimes k)_* \otimes x_6, D : 1 \}\]
By the definition of dualisers and since $k$ is classical under the Green observable structure, we can use the explicit witness of the passage between $k$ and $k_*$ to produce:
\[\{ d_{A}^Z : k_* \otimes x_{4*}, d_{A^*}^X : x_4  \otimes x_5, R_2^1 : (x_5 \otimes k)_* \otimes x_6, D : 1 \}\]
Since $k$ is classical under the Green observable structure, we can replace the two $k$'s with a Green copy map, operating on a single $k$:
\[\{ G_1^2 : k_* \otimes x_3 \otimes x_2, d_A^Z : x_{3*} \otimes x_{4*}, d_{A^*}^X : x_4 \otimes x_5, R_2^1 : (x_5 \otimes x_2)_* \otimes x_6, D : 1 \}\]
We can now use the \textit{Hopf law} to separate some of the connections so that we get:
\[\{ G_1^0 : k_*, R_0^1 : x_6 \}\]
Since $k$ is classical under the Green observable structure, we can safely remove the first soup connection as it deletes a $k$, which simplifies our soup to what we have been trying to prove:
\[\{ R_0^1 : x_6 \}\]

\[
% Generated with LaTeXDraw 2.0.8
% Mon May 14 00:02:53 BST 2012
% \usepackage[usenames,dvipsnames]{pstricks}
% \usepackage{epsfig}
% \usepackage{pst-grad} % For gradients
% \usepackage{pst-plot} % For axes
\scalebox{1} % Change this value to rescale the drawing.
{
\begin{pspicture}(0,-3.13)(14.357187,3.15)
\psline[linewidth=0.04cm,linecolor=green,arrowsize=0.05291667cm 2.0,arrowlength=1.4,arrowinset=0.4]{->}(9.619375,-0.59)(9.619375,-0.15)
\usefont{T1}{ptm}{m}{n}
\rput(10.355938,-2.385){\scriptsize 1}
\psline[linewidth=0.04cm](10.219375,-1.73)(10.219375,-2.33)
\psarc[linewidth=0.04](10.219375,-0.93){0.6}{-180.0}{0.0}
\pscircle[linewidth=0.04,dimen=outer,fillstyle=solid,fillcolor=green](10.219375,-1.63){0.3}
\psline[linewidth=0.04cm,arrowsize=0.05291667cm 2.0,arrowlength=1.4,arrowinset=0.4]{<-}(10.819375,-0.55)(10.819375,-0.95)
\usefont{T1}{ptm}{m}{n}
\rput(9.752188,-0.985){\scriptsize 3}
\usefont{T1}{ptm}{m}{n}
\rput(10.660625,0.415){\scriptsize 2}
\usefont{T1}{pcr}{m}{n}
\rput(11.4201565,0.685){\Huge =}
\psline[linewidth=0.04cm,arrowsize=0.05291667cm 2.0,arrowlength=1.4,arrowinset=0.4]{<-}(10.219375,3.13)(10.219375,2.33)
\psarc[linewidth=0.04](10.219375,1.51){0.6}{0.0}{180.0}
\pscircle[linewidth=0.04,dimen=outer,fillstyle=solid,fillcolor=red](10.219375,2.21){0.3}
\psline[linewidth=0.04cm,arrowsize=0.05291667cm 2.0,arrowlength=1.4,arrowinset=0.4]{<-}(9.619375,1.57)(9.619375,0.97)
\psline[linewidth=0.04cm,arrowsize=0.05291667cm 2.0,arrowlength=1.4,arrowinset=0.4]{<-}(10.819375,1.57)(10.819375,-0.67)
\psline[linewidth=0.04cm,arrowsize=0.05291667cm 2.0,arrowlength=1.4,arrowinset=0.4]{<-}(9.619375,-0.55)(9.619375,-0.95)
\psline[linewidth=0.04cm,linecolor=green,arrowsize=0.05291667cm 2.0,arrowlength=1.4,arrowinset=0.4]{<-}(9.619375,-0.17)(9.619375,0.23)
\psline[linewidth=0.04cm,linecolor=red,arrowsize=0.05291667cm 2.0,arrowlength=1.4,arrowinset=0.4]{<->}(9.619375,0.17)(9.619375,1.01)
\usefont{T1}{ptm}{m}{n}
\rput(9.762343,0.115){\scriptsize 4}
\usefont{T1}{ptm}{m}{n}
\rput(9.753437,1.215){\scriptsize 5}
\usefont{T1}{ptm}{m}{n}
\rput(10.357813,2.715){\scriptsize 6}
\psellipse[linewidth=0.04,dimen=outer](1.349375,0.89)(0.35,0.6)
\usefont{T1}{ptm}{m}{n}
\rput(12.555938,-0.585){\scriptsize 1}
\psline[linewidth=0.04cm](12.419375,0.07)(12.419375,-0.53)
\pscircle[linewidth=0.04,dimen=outer,fillstyle=solid,fillcolor=green](12.419375,0.17){0.3}
\psline[linewidth=0.04cm,arrowsize=0.05291667cm 2.0,arrowlength=1.4,arrowinset=0.4]{<-}(12.419375,-0.37)(12.419375,-1.07)
\psline[linewidth=0.04cm,arrowsize=0.05291667cm 2.0,arrowlength=1.4,arrowinset=0.4]{<-}(12.419375,2.13)(12.419375,1.33)
\pscircle[linewidth=0.04,dimen=outer,fillstyle=solid,fillcolor=red](12.419375,1.21){0.3}
\usefont{T1}{ptm}{m}{n}
\rput(12.557813,1.715){\scriptsize 6}
\psline[linewidth=0.04cm,arrowsize=0.05291667cm 2.0,arrowlength=1.4,arrowinset=0.4]{<-}(2.719375,3.13)(2.719375,2.33)
\psarc[linewidth=0.04](2.719375,1.51){0.6}{0.0}{180.0}
\pscircle[linewidth=0.04,dimen=outer,fillstyle=solid,fillcolor=red](2.719375,2.21){0.3}
\psline[linewidth=0.04cm,arrowsize=0.05291667cm 2.0,arrowlength=1.4,arrowinset=0.4]{<-}(2.119375,1.57)(2.119375,0.97)
\psline[linewidth=0.04cm,arrowsize=0.05291667cm 2.0,arrowlength=1.4,arrowinset=0.4]{<-}(3.319375,1.57)(3.319375,-0.67)
\psline[linewidth=0.04cm,arrowsize=0.05291667cm 2.0,arrowlength=1.4,arrowinset=0.4]{->}(2.119375,0.37)(2.119375,-0.23)
\psline[linewidth=0.04cm,linecolor=red,arrowsize=0.05291667cm 2.0,arrowlength=1.4,arrowinset=0.4]{<->}(2.119375,0.17)(2.119375,1.01)
\psline[linewidth=0.04cm](2.119375,-0.13)(2.119375,-0.73)
\pscircle[linewidth=0.06,linecolor=green,dimen=outer,fillstyle=solid,fillcolor=black](2.119375,-0.83){0.3}
\usefont{T1}{ptm}{m}{n}
\rput(2.1079688,-0.82){\color{white}\boldmath$k_*$\unboldmath}
\psline[linewidth=0.04cm,arrowsize=0.05291667cm 2.0,arrowlength=1.4,arrowinset=0.4]{<-}(3.319375,-0.13)(3.319375,-0.73)
\pscircle[linewidth=0.06,linecolor=green,dimen=outer,fillstyle=solid,fillcolor=black](3.319375,-0.83){0.3}
\usefont{T1}{ptm}{m}{n}
\rput(3.3279688,-0.82){\color{white}\boldmath$k$\unboldmath}
\psellipse[linewidth=0.04,dimen=outer](5.049375,0.89)(0.35,0.6)
\psline[linewidth=0.04cm,arrowsize=0.05291667cm 2.0,arrowlength=1.4,arrowinset=0.4]{<-}(6.419375,3.13)(6.419375,2.33)
\psarc[linewidth=0.04](6.419375,1.51){0.6}{0.0}{180.0}
\pscircle[linewidth=0.04,dimen=outer,fillstyle=solid,fillcolor=red](6.419375,2.21){0.3}
\psline[linewidth=0.04cm,arrowsize=0.05291667cm 2.0,arrowlength=1.4,arrowinset=0.4]{<-}(5.819375,1.57)(5.819375,0.97)
\psline[linewidth=0.04cm,arrowsize=0.05291667cm 2.0,arrowlength=1.4,arrowinset=0.4]{<-}(7.019375,1.57)(7.019375,-0.67)
\psline[linewidth=0.04cm,arrowsize=0.05291667cm 2.0,arrowlength=1.4,arrowinset=0.4]{<-}(7.019375,-0.53)(7.019375,-1.13)
\pscircle[linewidth=0.06,linecolor=green,dimen=outer,fillstyle=solid,fillcolor=black](7.019375,-1.23){0.3}
\usefont{T1}{ptm}{m}{n}
\rput(7.027969,-1.22){\color{white}\boldmath$k$\unboldmath}
\psline[linewidth=0.04cm,linecolor=green,arrowsize=0.05291667cm 2.0,arrowlength=1.4,arrowinset=0.4]{->}(5.819375,-0.71)(5.819375,-0.11)
\psline[linewidth=0.04cm,linecolor=green,arrowsize=0.05291667cm 2.0,arrowlength=1.4,arrowinset=0.4]{<-}(5.819375,-0.11)(5.819375,0.49)
\psline[linewidth=0.04cm,linecolor=red,arrowsize=0.05291667cm 2.0,arrowlength=1.4,arrowinset=0.4]{<->}(5.819375,0.29)(5.819375,1.09)
\psline[linewidth=0.04cm,arrowsize=0.05291667cm 2.0,arrowlength=1.4,arrowinset=0.4]{<-}(5.819375,-0.53)(5.819375,-1.13)
\pscircle[linewidth=0.06,linecolor=green,dimen=outer,fillstyle=solid,fillcolor=black](5.819375,-1.23){0.3}
\usefont{T1}{ptm}{m}{n}
\rput(5.8279686,-1.22){\color{white}\boldmath$k$\unboldmath}
\usefont{T1}{pcr}{m}{n}
\rput(3.9201562,0.685){\Huge =}
\usefont{T1}{pcr}{m}{n}
\rput(7.720156,0.685){\Huge =}
\pscircle[linewidth=0.06,linecolor=green,dimen=outer,fillstyle=solid,fillcolor=black](12.419375,-1.23){0.3}
\usefont{T1}{ptm}{m}{n}
\rput(12.427969,-1.22){\color{white}\boldmath$k$\unboldmath}
\psellipse[linewidth=0.04,dimen=outer](8.849375,0.89)(0.35,0.6)
\psline[linewidth=0.04cm,arrowsize=0.05291667cm 2.0,arrowlength=1.4,arrowinset=0.4]{<-}(10.219375,-2.17)(10.219375,-2.87)
\pscircle[linewidth=0.06,linecolor=green,dimen=outer,fillstyle=solid,fillcolor=black](10.219375,-2.83){0.3}
\usefont{T1}{ptm}{m}{n}
\rput(10.227969,-2.82){\color{white}\boldmath$k$\unboldmath}
\usefont{T1}{ptm}{m}{n}
\rput(6.5578127,2.715){\scriptsize 6}
\usefont{T1}{ptm}{m}{n}
\rput(2.9578125,2.715){\scriptsize 6}
\usefont{T1}{ptm}{m}{n}
\rput(5.9521875,-0.785){\scriptsize 3}
\usefont{T1}{ptm}{m}{n}
\rput(6.860625,0.415){\scriptsize 2}
\usefont{T1}{ptm}{m}{n}
\rput(5.9623437,0.315){\scriptsize 4}
\usefont{T1}{ptm}{m}{n}
\rput(5.9534373,1.215){\scriptsize 5}
\usefont{T1}{ptm}{m}{n}
\rput(3.060625,0.415){\scriptsize 2}
\usefont{T1}{ptm}{m}{n}
\rput(2.3534374,1.215){\scriptsize 5}
\usefont{T1}{ptm}{m}{n}
\rput(2.3623438,-0.085){\scriptsize 4}
\end{pspicture}
}
\]

(2) We begin by trying to \textit{fuse} together $\ell_*$ and $\ell$, in a manner that is similar to the definition of unbiasedness. We, therefore, start with the following set of connections in our soup:
\[\{ d_{A^*}^Z : \ell  \otimes x_5, G_2^1 : (x_5 \otimes \ell)_* \otimes x_6, D : 1 \}\]
By the definition of dualisers and since $\ell$ is classical under the Red observable structure, we can use the explicit witness of the passage between $\ell$ and $\ell_*$ to produce:
\[\{ d_{A}^X : \ell_* \otimes x_{4*}, d_{A^*}^Z : x_4  \otimes x_5, G_2^1 : (x_5 \otimes \ell)_* \otimes x_6, D : 1 \}\]
Since $\ell$ is classical under the Red observable structure, we can replace the two $\ell$'s with a Red copy map, operating on a single $\ell$:
\[\{ R_1^2 : \ell_* \otimes x_3 \otimes x_2, d_A^X : x_{3*} \otimes x_{4*}, d_{A^*}^Z : x_4 \otimes x_5, G_2^1 : (x_5 \otimes x_2)_* \otimes x_6, D : 1 \}\]
We can now use the \textit{Hopf law} to separate some of the connections so that we get:
\[\{ R_1^0 : \ell_*, G_0^1 : x_6 \}\]
Since $\ell$ is classical under the Red observable structure, we can safely remove the first soup connection as it deletes a $\ell$, which simplifies our soup to what we have been trying to prove:
\[\{ G_0^1 : x_6 \}\]

\[
% Generated with LaTeXDraw 2.0.8
% Mon May 14 07:21:24 BST 2012
% \usepackage[usenames,dvipsnames]{pstricks}
% \usepackage{epsfig}
% \usepackage{pst-grad} % For gradients
% \usepackage{pst-plot} % For axes
\scalebox{1} % Change this value to rescale the drawing.
{
\begin{pspicture}(0,-3.13)(14.317187,3.15)
\psline[linewidth=0.04cm,linecolor=red,arrowsize=0.05291667cm 2.0,arrowlength=1.4,arrowinset=0.4]{->}(9.579375,-0.59)(9.579375,-0.15)
\usefont{T1}{ptm}{m}{n}
\rput(10.315937,-2.385){\scriptsize 1}
\psline[linewidth=0.04cm](10.179375,-1.73)(10.179375,-2.33)
\psarc[linewidth=0.04](10.179375,-0.93){0.6}{-180.0}{0.0}
\pscircle[linewidth=0.04,dimen=outer,fillstyle=solid,fillcolor=red](10.179375,-1.63){0.3}
\psline[linewidth=0.04cm,arrowsize=0.05291667cm 2.0,arrowlength=1.4,arrowinset=0.4]{<-}(10.779375,-0.55)(10.779375,-0.95)
\usefont{T1}{ptm}{m}{n}
\rput(9.712188,-0.985){\scriptsize 3}
\usefont{T1}{ptm}{m}{n}
\rput(10.620625,0.415){\scriptsize 2}
\usefont{T1}{pcr}{m}{n}
\rput(11.3801565,0.685){\Huge =}
\psline[linewidth=0.04cm,arrowsize=0.05291667cm 2.0,arrowlength=1.4,arrowinset=0.4]{<-}(10.179375,3.13)(10.179375,2.33)
\psarc[linewidth=0.04](10.179375,1.51){0.6}{0.0}{180.0}
\pscircle[linewidth=0.04,dimen=outer,fillstyle=solid,fillcolor=green](10.179375,2.21){0.3}
\psline[linewidth=0.04cm,arrowsize=0.05291667cm 2.0,arrowlength=1.4,arrowinset=0.4]{<-}(9.579375,1.57)(9.579375,0.97)
\psline[linewidth=0.04cm,arrowsize=0.05291667cm 2.0,arrowlength=1.4,arrowinset=0.4]{<-}(10.779375,1.57)(10.779375,-0.67)
\psline[linewidth=0.04cm,arrowsize=0.05291667cm 2.0,arrowlength=1.4,arrowinset=0.4]{<-}(9.579375,-0.55)(9.579375,-0.95)
\psline[linewidth=0.04cm,linecolor=red,arrowsize=0.05291667cm 2.0,arrowlength=1.4,arrowinset=0.4]{<-}(9.579375,-0.17)(9.579375,0.23)
\psline[linewidth=0.04cm,linecolor=green,arrowsize=0.05291667cm 2.0,arrowlength=1.4,arrowinset=0.4]{<->}(9.579375,0.17)(9.579375,1.01)
\usefont{T1}{ptm}{m}{n}
\rput(9.722343,0.115){\scriptsize 4}
\usefont{T1}{ptm}{m}{n}
\rput(9.713437,1.215){\scriptsize 5}
\usefont{T1}{ptm}{m}{n}
\rput(10.317813,2.715){\scriptsize 6}
\psellipse[linewidth=0.04,dimen=outer](1.309375,0.89)(0.35,0.6)
\usefont{T1}{ptm}{m}{n}
\rput(12.515938,-0.585){\scriptsize 1}
\psline[linewidth=0.04cm](12.379375,0.07)(12.379375,-0.53)
\pscircle[linewidth=0.04,dimen=outer,fillstyle=solid,fillcolor=red](12.379375,0.17){0.3}
\psline[linewidth=0.04cm,arrowsize=0.05291667cm 2.0,arrowlength=1.4,arrowinset=0.4]{<-}(12.379375,-0.37)(12.379375,-1.07)
\psline[linewidth=0.04cm,arrowsize=0.05291667cm 2.0,arrowlength=1.4,arrowinset=0.4]{<-}(12.379375,2.13)(12.379375,1.33)
\pscircle[linewidth=0.04,dimen=outer,fillstyle=solid,fillcolor=green](12.379375,1.21){0.3}
\usefont{T1}{ptm}{m}{n}
\rput(12.517813,1.715){\scriptsize 6}
\psline[linewidth=0.04cm,arrowsize=0.05291667cm 2.0,arrowlength=1.4,arrowinset=0.4]{<-}(2.679375,3.13)(2.679375,2.33)
\psarc[linewidth=0.04](2.679375,1.51){0.6}{0.0}{180.0}
\pscircle[linewidth=0.04,dimen=outer,fillstyle=solid,fillcolor=green](2.679375,2.21){0.3}
\psline[linewidth=0.04cm,arrowsize=0.05291667cm 2.0,arrowlength=1.4,arrowinset=0.4]{<-}(2.079375,1.57)(2.079375,0.97)
\psline[linewidth=0.04cm,arrowsize=0.05291667cm 2.0,arrowlength=1.4,arrowinset=0.4]{<-}(3.279375,1.57)(3.279375,-0.67)
\psline[linewidth=0.04cm,arrowsize=0.05291667cm 2.0,arrowlength=1.4,arrowinset=0.4]{->}(2.079375,0.37)(2.079375,-0.23)
\psline[linewidth=0.04cm,linecolor=green,arrowsize=0.05291667cm 2.0,arrowlength=1.4,arrowinset=0.4]{<->}(2.079375,0.17)(2.079375,1.01)
\psline[linewidth=0.04cm](2.079375,-0.13)(2.079375,-0.73)
\pscircle[linewidth=0.06,linecolor=red,dimen=outer,fillstyle=solid,fillcolor=black](2.079375,-0.83){0.3}
\usefont{T1}{ptm}{m}{n}
\rput(2.0679688,-0.82){\color{white}\boldmath$k_*$\unboldmath}
\psline[linewidth=0.04cm,arrowsize=0.05291667cm 2.0,arrowlength=1.4,arrowinset=0.4]{<-}(3.279375,-0.13)(3.279375,-0.73)
\pscircle[linewidth=0.06,linecolor=red,dimen=outer,fillstyle=solid,fillcolor=black](3.279375,-0.83){0.3}
\usefont{T1}{ptm}{m}{n}
\rput(3.2879686,-0.82){\color{white}\boldmath$k$\unboldmath}
\psellipse[linewidth=0.04,dimen=outer](5.009375,0.89)(0.35,0.6)
\psline[linewidth=0.04cm,arrowsize=0.05291667cm 2.0,arrowlength=1.4,arrowinset=0.4]{<-}(6.379375,3.13)(6.379375,2.33)
\psarc[linewidth=0.04](6.379375,1.51){0.6}{0.0}{180.0}
\pscircle[linewidth=0.04,dimen=outer,fillstyle=solid,fillcolor=green](6.379375,2.21){0.3}
\psline[linewidth=0.04cm,arrowsize=0.05291667cm 2.0,arrowlength=1.4,arrowinset=0.4]{<-}(5.779375,1.57)(5.779375,0.97)
\psline[linewidth=0.04cm,arrowsize=0.05291667cm 2.0,arrowlength=1.4,arrowinset=0.4]{<-}(6.979375,1.57)(6.979375,-0.67)
\psline[linewidth=0.04cm,arrowsize=0.05291667cm 2.0,arrowlength=1.4,arrowinset=0.4]{<-}(6.979375,-0.53)(6.979375,-1.13)
\pscircle[linewidth=0.06,linecolor=red,dimen=outer,fillstyle=solid,fillcolor=black](6.979375,-1.23){0.3}
\usefont{T1}{ptm}{m}{n}
\rput(6.987969,-1.22){\color{white}\boldmath$k$\unboldmath}
\psline[linewidth=0.04cm,linecolor=red,arrowsize=0.05291667cm 2.0,arrowlength=1.4,arrowinset=0.4]{->}(5.779375,-0.71)(5.779375,-0.11)
\psline[linewidth=0.04cm,linecolor=red,arrowsize=0.05291667cm 2.0,arrowlength=1.4,arrowinset=0.4]{<-}(5.779375,-0.11)(5.779375,0.49)
\psline[linewidth=0.04cm,linecolor=green,arrowsize=0.05291667cm 2.0,arrowlength=1.4,arrowinset=0.4]{<->}(5.779375,0.29)(5.779375,1.09)
\psline[linewidth=0.04cm,arrowsize=0.05291667cm 2.0,arrowlength=1.4,arrowinset=0.4]{<-}(5.779375,-0.53)(5.779375,-1.13)
\pscircle[linewidth=0.06,linecolor=red,dimen=outer,fillstyle=solid,fillcolor=black](5.779375,-1.23){0.3}
\usefont{T1}{ptm}{m}{n}
\rput(5.7879686,-1.22){\color{white}\boldmath$k$\unboldmath}
\usefont{T1}{pcr}{m}{n}
\rput(3.8801563,0.685){\Huge =}
\usefont{T1}{pcr}{m}{n}
\rput(7.680156,0.685){\Huge =}
\pscircle[linewidth=0.06,linecolor=red,dimen=outer,fillstyle=solid,fillcolor=black](12.379375,-1.23){0.3}
\usefont{T1}{ptm}{m}{n}
\rput(12.387969,-1.22){\color{white}\boldmath$k$\unboldmath}
\psellipse[linewidth=0.04,dimen=outer](8.809375,0.89)(0.35,0.6)
\psline[linewidth=0.04cm,arrowsize=0.05291667cm 2.0,arrowlength=1.4,arrowinset=0.4]{<-}(10.179375,-2.17)(10.179375,-2.87)
\pscircle[linewidth=0.06,linecolor=red,dimen=outer,fillstyle=solid,fillcolor=black](10.179375,-2.83){0.3}
\usefont{T1}{ptm}{m}{n}
\rput(10.187969,-2.82){\color{white}\boldmath$k$\unboldmath}
\usefont{T1}{ptm}{m}{n}
\rput(6.5178127,2.715){\scriptsize 6}
\usefont{T1}{ptm}{m}{n}
\rput(2.9178126,2.715){\scriptsize 6}
\usefont{T1}{ptm}{m}{n}
\rput(5.9121876,-0.785){\scriptsize 3}
\usefont{T1}{ptm}{m}{n}
\rput(6.820625,0.415){\scriptsize 2}
\usefont{T1}{ptm}{m}{n}
\rput(5.9223437,0.315){\scriptsize 4}
\usefont{T1}{ptm}{m}{n}
\rput(5.9134374,1.215){\scriptsize 5}
\usefont{T1}{ptm}{m}{n}
\rput(3.020625,0.415){\scriptsize 2}
\usefont{T1}{ptm}{m}{n}
\rput(2.3134375,1.215){\scriptsize 5}
\usefont{T1}{ptm}{m}{n}
\rput(2.3223438,-0.085){\scriptsize 4}
\end{pspicture}
}
\]
\end{proof}

We can now define the notion of \textit{complementarity} to describe the property of the interaction that was proved in the previous theorem.

\begin{definition}[Complementarity (dagger lambda calculus)]
We say that two observable structures are \textit{complementary} when their interaction causes the constants that are \textit{classical} under one structure to be \textit{unbiased} against the other.
\end{definition}

Having defined the interaction between complementary observables, we will now define a function that can be used to transform Green sequents into Red ones and vice versa.

\begin{definition}[Hadamard (dagger lambda calculus)]
For every type $A$ in the dagger lambda calculus, we define a constant $H : A \multimap A$, called the \textit{Hadamard}. The sequent that represents the $\eta$-expansion of the \textit{Hadamard}, also known as the \textit{Hadamard gate}, is:
\[ x_1 : A \vdash_{\{ H : x_{1*} \otimes x_2 \}} x_2 : A \]
\noindent The \textit{Hadamard} must satisfy certain conditions; namely, the \textit{Hadamard gate} must be involutive under Cut, and the \textit{Hadamard} must allow us to transform Green sequents into Red ones.
\end{definition}

\noindent The first condition, \textit{involution} under Cut, means that we must be able to relate the following sequent:
\[ x_1 : A \vdash_{\{ H : x_{1*} \otimes x_2, H : x_{2*} \otimes x_3 \}} x_3 : A \]
\noindent to the identity sequent $x_1 : A \vdash_{\{ x_1 : x_3 \}} x_3 : A$, by equating their soups:
\[ \{ H : x_{1*} \otimes x_2, H : x_{2*} \otimes x_3 \} = \{ x_1 : x_3 \} \]
\[
% Generated with LaTeXDraw 2.0.8
% Sun May 06 20:11:53 BST 2012
% \usepackage[usenames,dvipsnames]{pstricks}
% \usepackage{epsfig}
% \usepackage{pst-grad} % For gradients
% \usepackage{pst-plot} % For axes
\scalebox{1} % Change this value to rescale the drawing.
{
\begin{pspicture}(0,-1.32)(1.83,1.32)
\psline[linewidth=0.04cm](0.21,0.1)(0.21,0.7)
\psline[linewidth=0.04cm,arrowsize=0.05291667cm 2.0,arrowlength=1.4,arrowinset=0.4]{->}(0.21,-0.3)(0.21,0.3)
\psline[linewidth=0.04cm](0.21,-0.9)(0.21,-0.3)
\psline[linewidth=0.04cm,arrowsize=0.05291667cm 2.0,arrowlength=1.4,arrowinset=0.4]{->}(0.21,-1.3)(0.21,-0.7)
\psframe[linewidth=0.02,dimen=outer,fillstyle=solid,fillcolor=yellow](0.38,-0.09)(0.0,-0.47)
\usefont{T1}{ptm}{m}{n}
\rput(0.213125,-0.31){H}
\psline[linewidth=0.04cm,arrowsize=0.05291667cm 2.0,arrowlength=1.4,arrowinset=0.4]{->}(0.21,0.7)(0.21,1.3)
\psframe[linewidth=0.02,dimen=outer,fillstyle=solid,fillcolor=yellow](0.38,0.91)(0.0,0.53)
\usefont{T1}{ptm}{m}{n}
\rput(0.213125,0.69){H}
\psline[linewidth=0.04cm,arrowsize=0.05291667cm 2.0,arrowlength=1.4,arrowinset=0.4]{->}(1.81,-1.3)(1.81,1.3)
\usefont{T1}{pcr}{m}{n}
\rput(1.0307813,-0.0050){\Huge =}
\usefont{T1}{ptm}{m}{n}
\rput(0.3665625,-1.075){\scriptsize 1}
\usefont{T1}{ptm}{m}{n}
\rput(0.39125,0.125){\scriptsize 2}
\usefont{T1}{ptm}{m}{n}
\rput(0.3828125,1.125){\scriptsize 3}
\end{pspicture}
}
\]

\noindent Moreover, the \textit{Hadamard} can be capable of transforming a Green copying sequent to a Red one. This is achieved by requiring that the following two soups, $S_{17}$ and $S_{18}$, be equal in our lambda calculus:

\begin{center}
\begin{tabular}{ l l l}
  $S_{17} =$    & $\;   \left\{
                        \begin{array}{l l}
                            H : x_{1*} \otimes x_2, & G_1^2 : x_{2*} \otimes x_3 \otimes x_4, \\
                            H : x_{3*} \otimes x_5, & H : x_{4*} \otimes x_6
                        \end{array}
                        \right\}$ \\
                & \\
  $S_{18} =$    & $\; \Big\{ R_1^2 : x_{1*} \otimes x_5 \otimes x_6 \Big\}$
\end{tabular}
\end{center}

\[
% Generated with LaTeXDraw 2.0.8
% Sun May 06 20:16:03 BST 2012
% \usepackage[usenames,dvipsnames]{pstricks}
% \usepackage{epsfig}
% \usepackage{pst-grad} % For gradients
% \usepackage{pst-plot} % For axes
\scalebox{1} % Change this value to rescale the drawing.
{
\begin{pspicture}(0,-2.22)(3.83,2.22)
\usefont{T1}{pcr}{m}{n}
\rput(2.0107813,-0.105){\Huge =}
\psline[linewidth=0.04cm](0.81,-0.12)(0.81,-0.72)
\psarc[linewidth=0.04](0.81,0.68){0.6}{-180.0}{0.0}
\pscircle[linewidth=0.04,dimen=outer,fillstyle=solid,fillcolor=green](0.81,-0.02){0.3}
\psline[linewidth=0.04cm,arrowsize=0.05291667cm 2.0,arrowlength=1.4,arrowinset=0.4]{->}(0.21,1.6)(0.21,2.2)
\psline[linewidth=0.04cm](0.21,1.0)(0.21,1.6)
\psline[linewidth=0.04cm,arrowsize=0.05291667cm 2.0,arrowlength=1.4,arrowinset=0.4]{->}(0.21,0.6)(0.21,1.2)
\psframe[linewidth=0.02,dimen=outer,fillstyle=solid,fillcolor=yellow](0.38,1.81)(0.0,1.43)
\usefont{T1}{ptm}{m}{n}
\rput(0.213125,1.59){H}
\psline[linewidth=0.04cm,arrowsize=0.05291667cm 2.0,arrowlength=1.4,arrowinset=0.4]{->}(1.41,1.6)(1.41,2.2)
\psline[linewidth=0.04cm](1.41,1.0)(1.41,1.6)
\psline[linewidth=0.04cm,arrowsize=0.05291667cm 2.0,arrowlength=1.4,arrowinset=0.4]{->}(1.41,0.6)(1.41,1.2)
\psframe[linewidth=0.02,dimen=outer,fillstyle=solid,fillcolor=yellow](1.58,1.81)(1.2,1.43)
\usefont{T1}{ptm}{m}{n}
\rput(1.413125,1.59){H}
\psline[linewidth=0.04cm,arrowsize=0.05291667cm 2.0,arrowlength=1.4,arrowinset=0.4]{->}(0.81,-1.2)(0.81,-0.6)
\psline[linewidth=0.04cm](0.81,-1.8)(0.81,-1.2)
\psline[linewidth=0.04cm,arrowsize=0.05291667cm 2.0,arrowlength=1.4,arrowinset=0.4]{->}(0.81,-2.2)(0.81,-1.6)
\psframe[linewidth=0.02,dimen=outer,fillstyle=solid,fillcolor=yellow](0.98,-0.99)(0.6,-1.37)
\usefont{T1}{ptm}{m}{n}
\rput(0.813125,-1.21){H}
\usefont{T1}{ptm}{m}{n}
\rput(3.3465624,-0.775){\scriptsize 1}
\psline[linewidth=0.04cm](3.21,-0.12)(3.21,-0.72)
\psarc[linewidth=0.04](3.21,0.68){0.6}{-180.0}{0.0}
\pscircle[linewidth=0.04,dimen=outer,fillstyle=solid,fillcolor=red](3.21,-0.02){0.3}
\psline[linewidth=0.04cm,arrowsize=0.05291667cm 2.0,arrowlength=1.4,arrowinset=0.4]{<-}(3.21,-0.5)(3.21,-0.9)
\psline[linewidth=0.04cm,arrowsize=0.05291667cm 2.0,arrowlength=1.4,arrowinset=0.4]{<-}(3.81,1.06)(3.81,0.66)
\psline[linewidth=0.04cm,arrowsize=0.05291667cm 2.0,arrowlength=1.4,arrowinset=0.4]{<-}(2.61,1.06)(2.61,0.66)
\usefont{T1}{ptm}{m}{n}
\rput(0.9465625,-1.975){\scriptsize 1}
\usefont{T1}{ptm}{m}{n}
\rput(0.97125,-0.575){\scriptsize 2}
\usefont{T1}{ptm}{m}{n}
\rput(0.3628125,0.825){\scriptsize 3}
\usefont{T1}{ptm}{m}{n}
\rput(1.1729687,0.825){\scriptsize 4}
\usefont{T1}{ptm}{m}{n}
\rput(0.3640625,2.025){\scriptsize 5}
\usefont{T1}{ptm}{m}{n}
\rput(1.1684375,2.025){\scriptsize 6}
\usefont{T1}{ptm}{m}{n}
\rput(2.7640624,0.825){\scriptsize 5}
\usefont{T1}{ptm}{m}{n}
\rput(3.5684376,0.825){\scriptsize 6}
\end{pspicture}
}
\]

\noindent Finally, the \textit{Hadamard} must be capable of transforming a Green deleting sequent to a Red one. This is achieved by requiring that the following two soups, $S_{19}$ and $S_{20}$, be equal in our lambda calculus:

\begin{center}
\begin{tabular}{ l l }
  $S_{19} =$    & $\; \Big\{ H : x_{1*} \otimes x_2, \; G_1^0 : x_{2*} \Big\}$ \\
                & \\
  $S_{20} =$    & $\; \Big\{ R_1^0 : x_{1*} \Big\}$
\end{tabular}
\end{center}

\[
% Generated with LaTeXDraw 2.0.8
% Sun May 06 20:16:22 BST 2012
% \usepackage[usenames,dvipsnames]{pstricks}
% \usepackage{epsfig}
% \usepackage{pst-grad} % For gradients
% \usepackage{pst-plot} % For axes
\scalebox{1} % Change this value to rescale the drawing.
{
\begin{pspicture}(0,-1.26)(2.0,1.24)
\usefont{T1}{ptm}{m}{n}
\rput(0.4365625,-0.815){\scriptsize 1}
\psline[linewidth=0.04cm](0.3,0.84)(0.3,0.24)
\pscircle[linewidth=0.04,dimen=outer,fillstyle=solid,fillcolor=green](0.3,0.94){0.3}
\psline[linewidth=0.04cm,arrowsize=0.05291667cm 2.0,arrowlength=1.4,arrowinset=0.4]{->}(0.3,-0.24)(0.3,0.36)
\psline[linewidth=0.04cm](0.3,-0.84)(0.3,-0.24)
\psline[linewidth=0.04cm,arrowsize=0.05291667cm 2.0,arrowlength=1.4,arrowinset=0.4]{->}(0.3,-1.24)(0.3,-0.64)
\psframe[linewidth=0.02,dimen=outer,fillstyle=solid,fillcolor=yellow](0.48,-0.02)(0.08,-0.42)
\usefont{T1}{ptm}{m}{n}
\rput(0.303125,-0.25){H}
\usefont{T1}{ptm}{m}{n}
\rput(0.46125,0.385){\scriptsize 2}
\psline[linewidth=0.04cm](1.7,0.44)(1.7,-0.16)
\pscircle[linewidth=0.04,dimen=outer,fillstyle=solid,fillcolor=red](1.7,0.54){0.3}
\psline[linewidth=0.04cm,arrowsize=0.05291667cm 2.0,arrowlength=1.4,arrowinset=0.4]{->}(1.7,-0.64)(1.7,-0.04)
\usefont{T1}{ptm}{m}{n}
\rput(1.8365625,-0.015){\scriptsize 1}
\usefont{T1}{pcr}{m}{n}
\rput(0.9007813,0.055){\Huge =}
\end{pspicture}
}
\]

\clearpage
\part{Applications}
\label{Part:Applications}
\chapter{Quantum Key Distribution}
\label{Chapter:Quantum Key Distribution}
Key distribution has traditionally been a very important area of cryptography. While the need for securely distributing a symmetric key has waned with the advent of public key cryptography, key distribution is still required in cases where processing speed is paramount as well as in cases where the level of secrecy does not allow us to rely on complexity assumptions. Because of the nature of the algorithms involved, encrypting messages with one-time pads is inherently faster than encrypting with a public key. Moreover, public key encryption relies on complexity assumptions that, though conjectured correct, have never been proven to be so. Elaborate schemes have been designed for publicising public keys, yet these are usually vulnerable to man-in-the-middle attacks or rely on a pre-existing secure communication channel. The big advantage of Quantum Key Distribution, as it was presented in \cite{BB84} and \cite{Eke91}, is that it allows us to securely distribute a symmetric key in a tamper-proof manner. This chapter explains how the dagger lambda calculus can be used to perform Quantum Key Distribution, by expressing the formalism that was used in \cite{CWWWZ11}.

\noindent We will begin by presenting a high level description of the steps involved in the protocol:
\begin{enumerate}
  \item Alice chooses two random strings of bits; $a = a_1, a_2, \ldots, a_{4n}$ and $b = b_1, b_2, \ldots, b_{4n}$.
  \item She then uses those strings to generate a string of qubits $\ket{\psi} = \ket{\psi_1}, \ket{\psi_2}, \ldots \ket{\psi_{4n}}$ by using a function $m : A \otimes B \rightarrow B$ such that $\ket{\psi_i} = m \circ (a_i \otimes b_i)$, where $m$ acts as an identity on $b_i$ if $a_i$ is $\ket{0}$ and as a Hadamard on $b_i$ if $a_i$ is $\ket{1}$.
  \item Alice transfers the string of quantum bits $\psi$ via a quantum channel to Bob.
  \item Bob chooses a random string of bits $c = c_1, c_2, \ldots, c_{4n}$ and measures each qubit $\ket{\psi_i}$ in the Z-basis if $c_i = 0$ and in the X-basis if $c_i = 1$, yielding $d = d_1, d_2, \ldots, d_{4n}$.
  \item Bob sends $c$ to Alice via a classical channel.
  \item Alice sends $a \oplus c = a_1 \oplus c_1, a_2 \oplus c_2, \ldots, a_{4n} \oplus c_{4n}$ to Bob via a classical channel.
  \item Alice and Bob check for which $i$, $a_i \oplus c_i = 0$. They maintain the corresponding $b_i$ and $d_i$ respectively and they throw away the rest.
  \item Alice and Bob should on average each be left with $2n$ bits, which should coincide if there has been no attack.
  \item Alice and Bob agree on a subset of roughly half of the remaining bits and compare them to ensure that they haven't been tampered with.
  \item If they coincide, they should on average be left with $n$ bits, with which they can engage in one-time-pad cryptographic communication.
\end{enumerate}

The first step in expressing Quantum Key Distribution, is defining the function $m : A \otimes B \rightarrow B$. When considering this from a linear algebraic perspective, $m$ is defined by its action on the standard basis, whereby $m (\ket{0} \otimes id_B) = id_B$ and $m (\ket{1} \otimes id_B) = H_B$. Consequently, in order to represent the function in the dagger lambda calculus, we will use a term $m : (A \otimes B) \multimap B$ such that the following equalities hold in the soup:
\[ \{ m : (R_0^1 \otimes b)_* \otimes b^\prime \} = \{ id_B : b_* \otimes b^\prime \} \]
\[ \{ m : (R[\pi]_0^1 \otimes b)_* \otimes b^\prime \} = \{ H : b_* \otimes b^\prime \} \]

\begin{lemma}[Controlled unitary 1]
If copies of the same classical input, in the dagger lambda calculus, are used as control terms for both $m$ and $(m_*)^*$, then the resulting maps will cancel each other out of the soup when composed:
\[
\left\{
\begin{array}{l l}
    m : (a_2 \otimes b_1)_* \otimes b_2,    & (m_*)^* : (a_4 \otimes b_1) \otimes b_{3*}, \\
    G_1^2 : a_{1*} \otimes a_2 \otimes a_3, & d_A : a_{3*} \otimes a_{4*} \\
\end{array}
\right\} = \left\{ G_1^0 : a_{1*}, b_3 : b_2 \right\}
\]
\end{lemma}
\begin{proof}
If $a_1$ is $R_0^1$, then it is copied by $G_1^2$ into $a_2$ and $a_3$:
\[
\left\{
\begin{array}{l l}
    m : (R_0^1 \otimes b_1)_* \otimes b_2,  & (m_*)^* : (R_0^1 \otimes b_1) \otimes b_{3*}
\end{array}
\right\}
\]
\noindent which, by the definition of $m$, results in two identities:
\[
\left\{
\begin{array}{l l}
    id_B : b_{1*} \otimes b_2,  & id_{B^*} : b_1 \otimes b_{3*}
\end{array}
\right\} =
\left\{
\begin{array}{l l}
    b_1 : b_2,  & b_{1*} : b_{3*}
\end{array}
\right\} = \left\{ b_3 : b_2 \right\}
\]

\noindent If, on the other hand, $a_1$ is $R[\pi]_1^2$, then it will again be copied by $G_1^2$ into $a_2$ and $a_3$:
\[
\left\{
\begin{array}{l l}
    m : (R[\pi]_0^1 \otimes b_1)_* \otimes b_2,  & (m_*)^* : (R[\pi]_0^1 \otimes b_1) \otimes b_{3*}
\end{array}
\right\}
\]
\noindent which, by the definition of $m$, results in two applications of the Hadamard gate:
\[
\left\{
\begin{array}{l l}
    H : b_{1*} \otimes b_2,  & (H_*)^* : b_1 \otimes b_{3*}
\end{array}
\right\} =
\left\{
\begin{array}{l l}
    H : b_{1*} \otimes b_2,  & H : b_{3*} \otimes b_1
\end{array}
\right\} = \left\{ b_3 : b_2 \right\}
\]
\end{proof}

\begin{lemma}[Controlled unitary 2]
If the classical values that are used in the dagger lambda calculus as control terms for $m$ and $(m_*)^*$ do not coincide, then the resulting maps will compose into a Hadamard transform in the soup:
\[
\left\{
\begin{array}{l l}
    m : (a_1 \otimes b_1)_* \otimes b_2,    & (m_*)^* : (a_4 \otimes b_1) \otimes b_{3*}, \\
                                            & d_A : a_{3*} \otimes a_{4*}, \\
    G_1^2 : a_{0*} \otimes a_1 \otimes a_2, & R[\pi]_1^1 : a_{2*} \otimes a_3*
\end{array}
\right\} = \left\{ G_1^0 : a_{0*}, H : b_{3*} \otimes b_2 \right\}
\]
\end{lemma}
\begin{proof}
If $a_0$ is $R_0^1$, then it will be copied by $G_1^2$ into $a_1$ and $a_2$:
\[
\left\{
\begin{array}{l l}
    m : (R_0^1 \otimes b_1)_* \otimes b_2,  & (m_*)^* : (a_4 \otimes b_1) \otimes b_{3*}, \\
                                            & d_A : (R[\pi]_0^1)^* \otimes a_{4*}
\end{array}
\right\}
\]
\noindent By the definition of $m$, the soup reduces to:
\[
\left\{
\begin{array}{l l}
    id_B : b_{1*} \otimes b_2,  & (H_*)^* : b_1 \otimes b_{3*} \\
\end{array}
\right\} =
\left\{
\begin{array}{l l}
    id_B : b_{1*} \otimes b_2,  & H : b_{3*} \otimes b_1
\end{array}
\right\} = \left\{ H : b_{3*} \otimes b_2 \right\}
\]

\noindent If, on the other hand, $a_0$ is $R[\pi]_0^1$, then it will again be copied by $G_1^2$ into $a_1$ and $a_2$:
\[
\left\{
\begin{array}{l l}
    m : (R[\pi]_0^1 \otimes b_1)_* \otimes b_2, & (m_*)^* : (a_4 \otimes b_1) \otimes b_{3*}, \\
                                                & d_A : (R_0^1)^* \otimes a_{4*}
\end{array}
\right\}
\]

Which, by the definition of $m$, reduces to:
\[
\left\{
\begin{array}{l l}
    H : b_{1*} \otimes b_2,  & id_{B^*} : b_1 \otimes b_{3*} \\
\end{array}
\right\} = \left\{ H : b_{3*} \otimes b_2 \right\}
\]
\end{proof}

\noindent Switching back to category-theoretic notation, we will now present an adaptation of the diagram that was derived in \cite{CWWWZ11}, which will be used to represent the protocol's procedures:
\[
% Generated with LaTeXDraw 2.0.8
% Mon Aug 27 18:54:04 BST 2012
% \usepackage[usenames,dvipsnames]{pstricks}
% \usepackage{epsfig}
% \usepackage{pst-grad} % For gradients
% \usepackage{pst-plot} % For axes
\scalebox{1} % Change this value to rescale the drawing.
{
\begin{pspicture}(0,-4.9603124)(9.346875,4.9984374)
\psline[linewidth=0.04cm,arrowsize=0.05291667cm 2.0,arrowlength=1.4,arrowinset=0.4]{<-}(5.7809377,0.5)(5.7809377,1.1)
\psline[linewidth=0.04cm,arrowsize=0.05291667cm 2.0,arrowlength=1.4,arrowinset=0.4]{<-}(5.3809376,0.5)(5.3809376,1.1)
\psline[linewidth=0.04cm,arrowsize=0.05291667cm 2.0,arrowlength=1.4,arrowinset=0.4]{<-}(4.1809373,0.5)(4.1809373,1.1)
\psline[linewidth=0.04cm,arrowsize=0.05291667cm 2.0,arrowlength=1.4,arrowinset=0.4]{<-}(3.7809374,0.5)(3.7809374,1.1)
\psline[linewidth=0.04cm](0.5809375,0.7)(0.5809375,1.1)
\psline[linewidth=0.04cm,arrowsize=0.05291667cm 2.0,arrowlength=1.4,arrowinset=0.4]{->}(0.5809375,-2.2)(0.5809375,0.8)
\psline[linewidth=0.04cm](0.9809375,0.7)(0.9809375,1.1)
\psline[linewidth=0.04cm,arrowsize=0.05291667cm 2.0,arrowlength=1.4,arrowinset=0.4]{->}(0.9809375,0.3)(0.9809375,0.8)
\psline[linewidth=0.04cm,arrowsize=0.05291667cm 2.0,arrowlength=1.4,arrowinset=0.4]{->}(0.7809375,1.5)(0.7809375,2.3)
\psframe[linewidth=0.04,dimen=outer,fillstyle=solid](1.1409374,1.7)(0.3809375,0.94)
\usefont{T1}{ptm}{m}{n}
\rput(1.0623437,2.01){$B$}
\usefont{T1}{ptm}{m}{n}
\rput(0.26234376,0.51){$A$}
\usefont{T1}{ptm}{m}{n}
\rput(0.75828123,1.295){\large $m$}
\usefont{T1}{ptm}{m}{n}
\rput(1.2623438,0.51){$B$}
\psline[linewidth=0.04cm](2.1809375,0.7)(2.1809375,1.1)
\psline[linewidth=0.04cm,arrowsize=0.05291667cm 2.0,arrowlength=1.4,arrowinset=0.4]{->}(2.1809375,0.3)(2.1809375,0.8)
\psline[linewidth=0.04cm](2.5809374,0.7)(2.5809374,1.1)
\psline[linewidth=0.04cm,arrowsize=0.05291667cm 2.0,arrowlength=1.4,arrowinset=0.4]{->}(2.5809374,0.3)(2.5809374,0.8)
\psline[linewidth=0.04cm,arrowsize=0.05291667cm 2.0,arrowlength=1.4,arrowinset=0.4]{->}(2.3809376,1.5)(2.3809376,2.3)
\psframe[linewidth=0.04,dimen=outer,fillstyle=solid](2.7409375,1.7)(1.9809375,0.94)
\usefont{T1}{ptm}{m}{n}
\rput(2.6623437,2.01){$B$}
\usefont{T1}{ptm}{m}{n}
\rput(1.8623438,0.51){$A$}
\usefont{T1}{ptm}{m}{n}
\rput(2.3582811,1.295){\large $m$}
\usefont{T1}{ptm}{m}{n}
\rput(2.8623438,0.51){$B$}
\psline[linewidth=0.04cm,arrowsize=0.05291667cm 2.0,arrowlength=1.4,arrowinset=0.4]{<-}(1.5809375,4.2)(1.5809375,3.32)
\psarc[linewidth=0.04](1.5809375,2.2){0.8}{0.0}{180.0}
\usefont{T1}{ptm}{m}{n}
\rput(1.8623438,4.61){$B$}
\pscircle[linewidth=0.04,dimen=outer,fillstyle=solid,fillcolor=green](1.5809375,3.12){0.3}
\psline[linewidth=0.04cm](3.9809375,1.5)(3.9809375,1.9)
\usefont{T1}{ptm}{m}{n}
\rput(4.3323436,2.01){$B^*$}
\usefont{T1}{ptm}{m}{n}
\rput(4.532344,0.51){$A^*$}
\psline[linewidth=0.04cm](5.5809374,1.5)(5.5809374,1.9)
\usefont{T1}{ptm}{m}{n}
\rput(5.932344,2.01){$B^*$}
\usefont{T1}{ptm}{m}{n}
\rput(6.132344,0.51){$A^*$}
\usefont{T1}{ptm}{m}{n}
\rput(5.032344,0.51){$B^*$}
\psline[linewidth=0.04cm](4.7809377,3.92)(4.7809377,3.32)
\psarc[linewidth=0.04](4.7809377,2.2){0.8}{0.0}{180.0}
\usefont{T1}{ptm}{m}{n}
\rput(5.132344,3.71){$B^*$}
\pscircle[linewidth=0.04,dimen=outer,fillstyle=solid,fillcolor=green](4.7809377,3.12){0.3}
\psframe[linewidth=0.04,dimen=outer,fillstyle=solid](4.3409376,1.7)(3.5809374,0.94)
\usefont{T1}{ptm}{m}{n}
\rput(3.9682813,1.335){\large $\overline{m}$}
\psframe[linewidth=0.04,dimen=outer,fillstyle=solid](5.9409375,1.7)(5.1809373,0.94)
\usefont{T1}{ptm}{m}{n}
\rput(5.568281,1.355){\large $\overline{m}$}
\psline[linewidth=0.04cm,arrowsize=0.05291667cm 2.0,arrowlength=1.4,arrowinset=0.4]{<-}(5.5809374,1.8)(5.5809374,2.2)
\psline[linewidth=0.04cm,arrowsize=0.05291667cm 2.0,arrowlength=1.4,arrowinset=0.4]{<-}(3.9809375,1.8)(3.9809375,2.2)
\psarc[linewidth=0.04](3.1809375,0.3){0.6}{-180.0}{0.0}
\psline[linewidth=0.04cm](3.7809374,0.3)(3.7809374,0.6)
\psarc[linewidth=0.04](3.1809375,0.3){2.2}{-180.0}{0.0}
\psline[linewidth=0.04cm](5.3809376,0.3)(5.3809376,0.6)
\psline[linewidth=0.04cm](2.1809375,-1.5)(2.1809375,0.3)
\psline[linewidth=0.04cm](2.1809375,-2.2)(2.1809375,-1.8)
\psarc[linewidth=0.04](1.3809375,-2.2){0.8}{-180.0}{0.0}
\psline[linewidth=0.04cm](4.1809373,-1.5)(4.1809373,0.6)
\psline[linewidth=0.04cm](4.1809373,-2.2)(4.1809373,-1.8)
\psline[linewidth=0.04cm](5.7809377,-2.2)(5.7809377,0.6)
\psarc[linewidth=0.04](4.9809375,-2.2){0.8}{-180.0}{0.0}
\psline[linewidth=0.04cm](4.9809375,-4.9)(4.9809375,-4.4)
\psline[linewidth=0.04cm](4.9809375,-3.8)(4.9809375,-3.2)
\psline[linewidth=0.04cm,linecolor=green,arrowsize=0.05291667cm 2.0,arrowlength=1.4,arrowinset=0.4]{->}(4.9809375,-4.5)(4.9809375,-4.1)
\psline[linewidth=0.04cm,linecolor=green,arrowsize=0.05291667cm 2.0,arrowlength=1.4,arrowinset=0.4]{<-}(4.9809375,-4.1)(4.9809375,-3.7)
\psline[linewidth=0.04cm](4.7809377,4.3)(4.7809377,4.9)
\psline[linewidth=0.04cm,linecolor=green,arrowsize=0.05291667cm 2.0,arrowlength=1.4,arrowinset=0.4]{<->}(4.7809377,3.8)(4.7809377,4.4)
\pscircle[linewidth=0.04,dimen=outer,fillstyle=solid,fillcolor=green](4.9809375,-3.08){0.3}
\psline[linewidth=0.04cm](1.3809375,-3.8)(1.3809375,-3.2)
\psline[linewidth=0.04cm,arrowsize=0.05291667cm 2.0,arrowlength=1.4,arrowinset=0.4]{->}(1.3809375,-4.9)(1.3809375,-3.7)
\pscircle[linewidth=0.04,dimen=outer,fillstyle=solid,fillcolor=green](1.3809375,-3.08){0.3}
\psline[linewidth=0.04cm](1.5809375,4.9)(1.5809375,4.12)
\usefont{T1}{ptm}{m}{n}
\rput(3.4323437,0.51){$B^*$}
\usefont{T1}{ptm}{m}{n}
\rput(5.3323436,-3.59){$A^*$}
\usefont{T1}{ptm}{m}{n}
\rput(4.662344,-4.69){$A$}
\usefont{T1}{ptm}{m}{n}
\rput(1.0623437,-4.69){$A$}
\usefont{T1}{ptm}{m}{n}
\rput(5.0623436,4.61){$B$}
\usefont{T1}{ptm}{m}{n}
\rput(1.6265625,-3.575){\scriptsize $a_1$}
\usefont{T1}{ptm}{m}{n}
\rput(0.8265625,-2.175){\scriptsize $a_2$}
\usefont{T1}{ptm}{m}{n}
\rput(1.9265625,-2.175){\scriptsize $a_3$}
\usefont{T1}{ptm}{m}{n}
\rput(3.1365626,-1.675){\scriptsize $b_1$}
\usefont{T1}{ptm}{m}{n}
\rput(0.5365625,2.025){\scriptsize $b_2$}
\usefont{T1}{ptm}{m}{n}
\rput(3.1365626,-0.075){\scriptsize $b_3$}
\usefont{T1}{ptm}{m}{n}
\rput(2.1365626,2.025){\scriptsize $b_4$}
\usefont{T1}{ptm}{m}{n}
\rput(1.8365625,3.725){\scriptsize $b_5$}
\usefont{T1}{ptm}{m}{n}
\rput(5.2265625,-4.675){\scriptsize $a_4$}
\usefont{T1}{ptm}{m}{n}
\rput(4.6365623,-3.575){\scriptsize $a_{5*}$}
\usefont{T1}{ptm}{m}{n}
\rput(4.5365624,-2.175){\scriptsize $a_{6*}$}
\usefont{T1}{ptm}{m}{n}
\rput(5.5365624,-2.175){\scriptsize $a_{7*}$}
\usefont{T1}{ptm}{m}{n}
\rput(3.7465625,2.025){\scriptsize $b_{6*}$}
\usefont{T1}{ptm}{m}{n}
\rput(5.3465624,2.025){\scriptsize $b_{7*}$}
\usefont{T1}{ptm}{m}{n}
\rput(4.5465627,3.725){\scriptsize $b_{8*}$}
\usefont{T1}{ptm}{m}{n}
\rput(4.5365624,4.625){\scriptsize $b_9$}
\end{pspicture}
}
\]

This can translated to the dagger lambda calculus, where the corresponding sequent would be represented by $a_1 : A, a_4 : A \vdash_{S_1} b_5 \otimes b_9 : B \otimes B$ and where the soup $S_1$ would be of the form:
\[
S_1 = \left\{
\begin{array}{l l}
                                                & d_{B^*} : b_8 \otimes b_9, \\
    G_2^1 : (b_2 \otimes b_4)_* \otimes b_5,    & G_2^1 : (b_6 \otimes b_7) \otimes b_{8*}, \\
    (m_*)^* : (b_3 \otimes a_6) \otimes b_{6*}, & (m_*)^* : (b_1 \otimes a_7) \otimes b_{7*}, \\
    m : (a_2 \otimes b_1)_* \otimes b_2,        & m : (a_3 \otimes b_3)_* \otimes b_4, \\
    G_1^2 : a_{1*} \otimes a_2 \otimes a_3,     & G_1^2 : a_5 \otimes a_{6*} \otimes a_{7*}, \\
                                                & d_A : a_{4*} \otimes a_{5*}
\end{array}
\right\}
\]

In both the categorical diagram and the sequent in the dagger lambda calculus, Alice's control input is denoted by $a_1$ and Bob's control input is denoted by $a_4$. We will use the soup derivations of the dagger lambda calculus to verify the protocol's behaviour, in a way that resembles the verification what was performed in the categorical setting by \cite{CWWWZ11}. In order to verify the protocol, we will first examine the case where $a_1$ and $a_4$ coincide by ``hardwiring'' them to two different copies of the same original state. This will be done by performing a Cut on $a_0 : A \vdash_{\{G_1^2 : a_{0*} \otimes a_1 \otimes a_4\}} a_1 \otimes a_4 : A \otimes A$ with $a_1 : A, a_4 : A \vdash_{S_1} b_5 \otimes b_9 : B \otimes B$. The resulting sequent can be written as $a_0 : A \vdash_{S_2} b_5 \otimes b_9 : B \otimes B$, where the soup $S_2$ is:
\[
S_2 = \left\{
\begin{array}{l l}
                                                                    & d_{B^*} : b_8 \otimes b_9, \\
    G_2^1 : (b_2 \otimes b_4)_* \otimes b_5,                        & G_2^1 : (b_6 \otimes b_7) \otimes b_{8*}, \\
    (m_*)^* : (a_6 \otimes b_3) \otimes b_{6*},                     & (m_*)^* : (a_7 \otimes b_1) \otimes b_{7*}, \\
    m : (a_2 \otimes b_1)_* \otimes b_2,                            & m : (a_3 \otimes b_3)_* \otimes b_4, \\
                                                                    & d_A : a_{5*} \otimes a_{7*}, \\
                                                                    & d_A : a_{4*} \otimes a_{6*}, \\
    G_1^4 \; : \; a_{0*} \;\; \otimes \;\; a_2 \; \otimes \; a_3    & \!\! \otimes \; a_4 \; \otimes \; a_5
\end{array}
\right\}
\]

\noindent At this point, we can simplify the soup by using Lemma 12.0.2 twice:
\[
\left\{
\begin{array}{l l}
                                                & d_{B^*} : b_8 \otimes b_9, \\
    G_2^1 : (b_2 \otimes b_4)_* \otimes b_5,    & G_2^1 : (b_6 \otimes b_7) \otimes b_{8*}, \\
    b_7 : b_2,                                  & b_6 : b_4, \\
                                                & G_1^0 : a_{0*}
\end{array}
\right\}
\]

\noindent We can use soup reduction to rewrite this as:
\[
\left\{
\begin{array}{l l}
                                                & d_{B^*} : b_8 \otimes b_9, \\
    G_2^1 : (b_2 \otimes b_4)_* \otimes b_5,    & G_2^1 : (b_4 \otimes b_2) \otimes b_{8*}, \\
                                                & G_1^0 : a_{0*}
\end{array}
\right\}
\]

\noindent Which is essentially:
\[
\left\{
\begin{array}{l l}
    G_1^0 : a_{0*}, & G_0^2 : b_5 \otimes b_9
\end{array}
\right\}
\]

\noindent The resulting sequent reduces to:
\[ a_0 : A \vdash_{\{ G_1^0 : a_{0*} \}} b \otimes b : B \otimes B \]

Similarly, to examine the case where the control inputs do not coincide, we can perform a Cut on $a_0 : A \vdash_{\{G_1^2 : a_{0*} \otimes a_1 \otimes a, R_1^1[\pi] : a_* \otimes a_4 \}} a_1 \otimes a_4 : A \otimes A$ with $a_1 : A, a_4 : A \vdash_S b_5 \otimes b_9 : B \otimes B$. The resulting sequent can be written as $a_0 : A \vdash_{S_3} b_5 \otimes b_9 : B \otimes B$, where the soup $S_3$ is:
\[
S_3 = \left\{
\begin{array}{l l}
                                                                    & d_{B^*} : b_8 \otimes b_9, \\
    G_2^1 : (b_2 \otimes b_4)_* \otimes b_5,                        & G_2^1 : (b_6 \otimes b_7) \otimes b_{8*}, \\
    (m_*)^* : (a_6 \otimes b_3) \otimes b_{6*},                     & (m_*)^* : (a_7 \otimes b_1) \otimes b_{7*}, \\
    m : (a_2 \otimes b_1)_* \otimes b_2,                            & m : (a_3 \otimes b_3)_* \otimes b_4, \\
                                                                    & d_A : a_{5*} \otimes a_{7*}, \\
                                                                    & d_A : a_{4*} \otimes a_{6*}, \\
    G_1^4 \; : \; a_{0*} \;\; \otimes \;\; a_2 \; \otimes \; a_3    & \!\! \otimes \; a_4 \; \otimes \; a_5
\end{array}
\right\}
\]

\noindent The resulting sequent reduces to:
\[ a_0 : A \vdash_{\{ G_1^0 : a_{0*} \}} G_0^1 \otimes G_0^1 : B \otimes B \]

This completes our functional verification of the protocol from within the sequents of our dagger lambda calculus. We have proved that in both the case when the control inputs over lap and when they don't, the protocol's behaviour is as expected.

\chapter{Quantum Fourier Transform}
\label{Chapter:Quantum Fourier Transform}
The quantum Fourier transform is a calculation that lies at the heart of many quantum algorithms, including Shor's factoring algorithm \cite{Sho97}. Functionally, when considering its action on basis states,the transformation is exactly the same transformation as the discrete Fourier transform. Because the transform is a linear operator, its action on arbitrary quantum states is completely determined by its discreet counterpart.

In this chapter we will explain how the QFT is constructed, present its corresponding diagrammatic representation in the picture calculus, and demonstrate that the dagger lambda calculus is expressive enough by using it to represent the transform. We will then proceed by ``running'' a calculation through, while observing how connections propagate in the soup. In order to keep our exposition simple, we will limit our description to the behaviour of the two-qubit QFT, though it should not be hard to generalise our QFT to a larger numbers of qubits.

As we can see from \cite{NC00} and \cite{CD11}, the only gates that are required to construct the quantum Fourier transform are the Hadamard gate and a gate called $\wedge Z_\alpha$ that performs a controlled $Z$ rotation on an arbitrary angle $\alpha$. In Hilbert spaces, the two-qubit QFT can be expressed as $(H \otimes id_A) \circ \wedge Z_{\pi / 2} \circ (id_A \otimes H)$. We can verify its behaviour on a given input (for example $\ket{10}$) by applying it to that state. This would give us:
\begin{align*}
    (H \otimes id_A) \circ \wedge Z_{\pi / 2} \circ (\ket{1} \otimes \ket{0})   & = (H \otimes id_A) \circ (\ket{1} \otimes \ket{+_{\pi/2}}) \\
                                                                                & = \ket{-} \otimes \ket{+_{\pi/2}}
\end{align*}

The dagger lambda calculus is best used alongside the diagrammatic calculus, so that one can complement the other. We will proceed by explaining how the quantum Fourier transform can be represented in the calculus for complementary observables, as well as how these diagrams translate to sequents and soup connections in the dagger lambda calculus. The two notations will be presented side by side: a set of soup connections on the left and a picture on the right. For a more complete exposition of the diagrammatic representation, the reader is referred to \cite{CD11}.

The first step towards constructing a QFT in the dagger lambda calculus, consists of representing the controlled phase gate $\wedge Z_\alpha$. The sequent for this gate can be written as $a_1 : A, a_2 : A \vdash_S a_5 \otimes a_9 : A \otimes A$ where the soup $S$ is:
\[
S = \left\{
\begin{array}{l l}
                                                            & G[\alpha/2]_1^1 : a_{8*} \otimes a_9, \\
                                                            & R_2^1 : (a_6 \otimes a_7)_* \otimes a_8, \\
    G_2^2 : (a_1 \otimes a_3)_* \otimes a_5 \otimes a_6,    & G[-\alpha/2]_1^1 : a_{4*} \otimes a_7, \\
                                                            & R_1^2 : a_{2*} \otimes a_3 \otimes a_4
\end{array}
\right\}
\]

We will use $S_\alpha$ in the dagger lambda calculus to denote the soup $S$ of the controlled phase gate, along with its associated angle of rotation $\alpha$. In the context of Hilbert spaces, we presented a way of generating the two-qubit quantum Fourier transform by using $\wedge Z_{\pi/2}$, the controlled $Z$ rotation with angle $\pi/2$, composed with some Hadamard transforms. Therefore, in the dagger lambda calculus, it can be written as:
\[ a_1 : A, a_0 : A \vdash_{S_{\pi/2} \cup \{ H : a_{0*} \otimes a_2, H : a_{5*} \otimes a_{10} \}} a_{10} \otimes a_9 : A \otimes A \]

Similarly to the diagrammatic ``execution'' that was presented in \cite{CD11}, we can now ``run'' the two-qubit quantum Fourier transform in the dagger lambda calculus by plugging in an input state in place of $a_1$ and $a_0$ in our sequent. We will be using $\ket{10} = R[\pi]_0^1 \otimes R_0^1$, by plugging in $R[\pi]_0^1$ for $a_1$ and $R_0^1$ for $a_0$. The soup thus becomes:

\[
\left\{
\begin{array}{l l}
    H : a_{5*} \otimes a_{10},                                  & G[\pi/4]_1^1 : a_{8*} \otimes a_9, \\
                                                                & R_2^1 : (a_6 \otimes a_7)_* \otimes a_8, \\
    G_2^2 : (R[\pi]_0^1 \otimes a_3)_* \otimes a_5 \otimes a_6, & G[-\pi/4]_1^1 : a_{4*} \otimes a_7, \\
                                                                & R_1^2 : a_{2*} \otimes a_3 \otimes a_4, \\
                                                                & H : (R_0^1)_* \otimes a_2
\end{array}
\right\}\!\!\!\!\!\!\!\!\!
% \begin{pspicture}(0,-3.45)(3.9971876,3.47)
% \begin{pspicture}(0,0)(3.9971876,3)
% ------------------------------------------
% Generated with LaTeXDraw 2.0.8
% Sat Sep 08 17:14:47 BST 2012
% \usepackage[usenames,dvipsnames]{pstricks}
% \usepackage{epsfig}
% \usepackage{pst-grad} % For gradients
% \usepackage{pst-plot} % For axes
\scalebox{1} % Change this value to rescale the drawing.
{
\begin{pspicture}(0,0)(3.9971876,3)
\psline[linewidth=0.04cm,arrowsize=0.05291667cm 2.0,arrowlength=1.4,arrowinset=0.4]{->}(3.579375,2.65)(3.579375,3.45)
\psline[linewidth=0.04cm,arrowsize=0.05291667cm 2.0,arrowlength=1.4,arrowinset=0.4]{->}(1.979375,2.65)(1.979375,3.45)
\psline[linewidth=0.04cm](1.979375,-1.75)(1.979375,0.05)
\psline[linewidth=0.04cm,arrowsize=0.05291667cm 2.0,arrowlength=1.4,arrowinset=0.4]{->}(1.979375,-3.05)(1.979375,-1.55)
\psline[linewidth=0.04cm](1.979375,1.25)(1.979375,2.45)
\psline[linewidth=0.04cm,arrowsize=0.05291667cm 2.0,arrowlength=1.4,arrowinset=0.4]{->}(1.979375,0.25)(1.979375,1.45)
\psline[linewidth=0.04cm,arrowsize=0.05291667cm 2.0,arrowlength=1.4,arrowinset=0.4]{->}(3.579375,-1.05)(2.779375,-0.45)
\psline[linewidth=0.04cm](3.579375,-1.05)(1.979375,0.15)
\psline[linewidth=0.04cm,arrowsize=0.05291667cm 2.0,arrowlength=1.4,arrowinset=0.4]{->}(1.979375,0.15)(2.779375,0.75)
\psline[linewidth=0.04cm](1.979375,0.15)(3.579375,1.35)
\psline[linewidth=0.04cm](3.579375,1.85)(3.579375,2.45)
\psline[linewidth=0.04cm,arrowsize=0.05291667cm 2.0,arrowlength=1.4,arrowinset=0.4]{->}(3.579375,1.45)(3.579375,2.05)
\psline[linewidth=0.04cm](3.579375,0.65)(3.579375,1.25)
\psline[linewidth=0.04cm,arrowsize=0.05291667cm 2.0,arrowlength=1.4,arrowinset=0.4]{->}(3.579375,0.25)(3.579375,0.85)
\psline[linewidth=0.04cm](3.579375,-1.75)(3.579375,-1.15)
\psline[linewidth=0.04cm,arrowsize=0.05291667cm 2.0,arrowlength=1.4,arrowinset=0.4]{->}(3.579375,-2.15)(3.579375,-1.55)
\psline[linewidth=0.04cm](3.579375,-0.55)(3.579375,0.05)
\psline[linewidth=0.04cm,arrowsize=0.05291667cm 2.0,arrowlength=1.4,arrowinset=0.4]{->}(3.579375,-0.95)(3.579375,-0.35)
\psline[linewidth=0.04cm](3.579375,-2.75)(3.579375,-2.15)
\psline[linewidth=0.04cm,arrowsize=0.05291667cm 2.0,arrowlength=1.4,arrowinset=0.4]{->}(3.579375,-3.15)(3.579375,-2.55)
\usefont{T1}{ptm}{m}{n}
\rput(3.820625,-1.625){\scriptsize 2}
\pscircle[linewidth=0.04,dimen=outer,fillstyle=solid,fillcolor=red](1.979375,-3.15){0.3}
\pscircle[linewidth=0.04,dimen=outer,fillstyle=solid,fillcolor=red](3.579375,-3.15){0.3}
\psframe[linewidth=0.02,dimen=outer,fillstyle=solid,fillcolor=yellow](3.779375,-1.89)(3.339375,-2.33)
\usefont{T1}{ptm}{m}{n}
\rput(3.5825,-2.14){H}
\pscircle[linewidth=0.04,dimen=outer,fillstyle=solid,fillcolor=red](3.579375,-1.05){0.3}
\pscircle[linewidth=0.04,dimen=outer,fillstyle=solid,fillcolor=red](3.579375,1.35){0.3}
\pscircle[linewidth=0.04,dimen=outer,fillstyle=solid,fillcolor=green](1.979375,0.15){0.3}
\psframe[linewidth=0.02,dimen=outer,fillstyle=solid,fillcolor=yellow](2.179375,2.81)(1.739375,2.37)
\usefont{T1}{ptm}{m}{n}
\rput(1.9825,2.56){H}
\usefont{T1}{ptm}{m}{n}
\rput(1.9779687,-3.14){\boldmath$\pi$\unboldmath}
\usefont{T1}{ptm}{m}{n}
\rput(3.8223438,-0.425){\scriptsize 4}
\usefont{T1}{ptm}{m}{n}
\rput(2.6121874,-0.625){\scriptsize 3}
\usefont{T1}{ptm}{m}{n}
\rput(2.2134376,1.375){\scriptsize 5}
\usefont{T1}{ptm}{m}{n}
\rput(2.6178124,0.975){\scriptsize 6}
\usefont{T1}{ptm}{m}{n}
\rput(3.8179688,0.775){\scriptsize 7}
\usefont{T1}{ptm}{m}{n}
\rput(3.8128126,1.975){\scriptsize 8}
\usefont{T1}{ptm}{m}{n}
\rput(3.8179688,3.175){\scriptsize 9}
\usefont{T1}{ptm}{m}{n}
\rput(2.2779686,3.175){\scriptsize 10}
\pscircle[linewidth=0.04,dimen=outer,fillstyle=solid,fillcolor=green](3.579375,0.15){0.3}
\pscircle[linewidth=0.04,dimen=outer,fillstyle=solid,fillcolor=green](3.579375,2.55){0.3}
\usefont{T1}{ptm}{m}{n}
\rput(3.578125,2.545){\footnotesize \boldmath$\pi\!/\!4$\unboldmath}
\usefont{T1}{ppl}{m}{n}
\rput(3.59,0.145){\tiny -\boldmath$\pi\!/\!4$\unboldmath}
\end{pspicture}
}
\]
\vskip 3cm

\noindent We can then use $G_2^2$ to copy $R[\pi]_0^1$ across three terms, $a_3$, $a_5$ and $a_6$, giving us:
\[
\left\{
\begin{array}{l l}
    H : (R[\pi]_0^1)_* \otimes a_{10},  & G[\pi/4]_1^1 : a_{8*} \otimes a_9, \\
                                        & R_2^1 : (R[\pi]_0^1 \otimes a_7)_* \otimes a_8, \\
                                        & G[-\pi/4]_1^1 : a_{4*} \otimes a_7, \\
                                        & R_1^2 : a_{2*} \otimes R[\pi]_0^1 \otimes a_4, \\
                                        & H : (R_0^1)_* \otimes a_2
\end{array}
\right\}
% \begin{pspicture}(0,-3.45)(4.1971874,3.47)
% \begin{pspicture}(0,0)(4.1971874,3)
% ------------------------------------------
% Generated with LaTeXDraw 2.0.8
% Sat Sep 08 17:12:26 BST 2012
% \usepackage[usenames,dvipsnames]{pstricks}
% \usepackage{epsfig}
% \usepackage{pst-grad} % For gradients
% \usepackage{pst-plot} % For axes
\scalebox{1} % Change this value to rescale the drawing.
{
\begin{pspicture}(0,0)(4.1971874,3)
\psline[linewidth=0.04cm](3.579375,1.25)(2.379375,0.65)
\psline[linewidth=0.04cm,arrowsize=0.05291667cm 2.0,arrowlength=1.4,arrowinset=0.4]{<-}(3.179375,1.05)(2.379375,0.65)
\psline[linewidth=0.04cm,arrowsize=0.05291667cm 2.0,arrowlength=1.4,arrowinset=0.4]{->}(3.579375,-0.95)(2.779375,-0.55)
\psline[linewidth=0.04cm](3.579375,-0.95)(2.379375,-0.35)
\psline[linewidth=0.04cm,arrowsize=0.05291667cm 2.0,arrowlength=1.4,arrowinset=0.4]{->}(3.579375,2.65)(3.579375,3.45)
\psline[linewidth=0.04cm,arrowsize=0.05291667cm 2.0,arrowlength=1.4,arrowinset=0.4]{->}(1.979375,2.65)(1.979375,3.45)
\psline[linewidth=0.04cm](1.979375,1.85)(1.979375,2.45)
\psline[linewidth=0.04cm,arrowsize=0.05291667cm 2.0,arrowlength=1.4,arrowinset=0.4]{->}(1.979375,1.45)(1.979375,2.05)
\psline[linewidth=0.04cm](3.579375,1.85)(3.579375,2.45)
\psline[linewidth=0.04cm,arrowsize=0.05291667cm 2.0,arrowlength=1.4,arrowinset=0.4]{->}(3.579375,1.45)(3.579375,2.05)
\psline[linewidth=0.04cm](3.579375,0.65)(3.579375,1.25)
\psline[linewidth=0.04cm,arrowsize=0.05291667cm 2.0,arrowlength=1.4,arrowinset=0.4]{->}(3.579375,0.25)(3.579375,0.85)
\psline[linewidth=0.04cm](3.579375,-1.75)(3.579375,-1.15)
\psline[linewidth=0.04cm,arrowsize=0.05291667cm 2.0,arrowlength=1.4,arrowinset=0.4]{->}(3.579375,-2.15)(3.579375,-1.55)
\psline[linewidth=0.04cm](3.579375,-0.55)(3.579375,0.05)
\psline[linewidth=0.04cm,arrowsize=0.05291667cm 2.0,arrowlength=1.4,arrowinset=0.4]{->}(3.579375,-0.95)(3.579375,-0.35)
\psline[linewidth=0.04cm](3.579375,-2.75)(3.579375,-2.15)
\psline[linewidth=0.04cm,arrowsize=0.05291667cm 2.0,arrowlength=1.4,arrowinset=0.4]{->}(3.579375,-3.15)(3.579375,-2.55)
\usefont{T1}{ptm}{m}{n}
\rput(3.820625,-1.625){\scriptsize 2}
\pscircle[linewidth=0.04,dimen=outer,fillstyle=solid,fillcolor=red](3.579375,-3.15){0.3}
\psframe[linewidth=0.02,dimen=outer,fillstyle=solid,fillcolor=yellow](3.789375,-1.88)(3.329375,-2.34)
\usefont{T1}{ptm}{m}{n}
\rput(3.5825,-2.14){H}
\pscircle[linewidth=0.04,dimen=outer,fillstyle=solid,fillcolor=red](3.579375,-1.05){0.3}
\pscircle[linewidth=0.04,dimen=outer,fillstyle=solid,fillcolor=red](3.579375,1.35){0.3}
\psframe[linewidth=0.02,dimen=outer,fillstyle=solid,fillcolor=yellow](2.189375,2.82)(1.729375,2.36)
\usefont{T1}{ptm}{m}{n}
\rput(1.9825,2.56){H}
\usefont{T1}{ptm}{m}{n}
\rput(3.8223438,-0.425){\scriptsize 4}
\usefont{T1}{ptm}{m}{n}
\rput(3.8179688,0.775){\scriptsize 7}
\usefont{T1}{ptm}{m}{n}
\rput(3.8128126,1.975){\scriptsize 8}
\usefont{T1}{ptm}{m}{n}
\rput(3.8179688,3.175){\scriptsize 9}
\usefont{T1}{ptm}{m}{n}
\rput(2.2779686,3.175){\scriptsize 10}
\pscircle[linewidth=0.04,dimen=outer,fillstyle=solid,fillcolor=red](2.179375,0.55){0.3}
\usefont{T1}{ptm}{m}{n}
\rput(2.1779687,0.56){\boldmath$\pi$\unboldmath}
\pscircle[linewidth=0.04,dimen=outer,fillstyle=solid,fillcolor=red](2.179375,-0.25){0.3}
\usefont{T1}{ptm}{m}{n}
\rput(2.1779687,-0.24){\boldmath$\pi$\unboldmath}
\pscircle[linewidth=0.04,dimen=outer,fillstyle=solid,fillcolor=red](1.979375,1.35){0.3}
\usefont{T1}{ptm}{m}{n}
\rput(1.9779687,1.36){\boldmath$\pi$\unboldmath}
\pscircle[linewidth=0.04,dimen=outer,fillstyle=solid,fillcolor=green](3.579375,0.15){0.3}
\pscircle[linewidth=0.04,dimen=outer,fillstyle=solid,fillcolor=green](3.579375,2.55){0.3}
\usefont{T1}{ptm}{m}{n}
\rput(3.578125,2.545){\footnotesize \boldmath$\pi\!/\!4$\unboldmath}
\usefont{T1}{ppl}{m}{n}
\rput(3.59,0.145){\tiny -\boldmath$\pi\!/\!4$\unboldmath}
\end{pspicture}
}
\]
\vskip 3cm

\noindent The red phase terms, $R[\pi]_1^1$, can be written in a simpler way:
\[
\left\{
\begin{array}{l l}
    H : (R[\pi]_0^1)_* \otimes a_{10},  & G[\pi/4]_1^1 : a_{8*} \otimes a_9, \\
                                        & R[\pi]_1^1 : a_{7*} \otimes a_8, \\
                                        & G[-\pi/4]_1^1 : a_{4*} \otimes a_7, \\
                                        & R[\pi]_1^1 : a_{2*} \otimes a_4, \\
                                        & H : (R_0^1)_* \otimes a_2
\end{array}
\right\}
% \begin{pspicture}(0,-3.45)(5.5971875,3.47)
% \begin{pspicture}(0,0)(5.5971875,3.5)
% ------------------------------------------
% Generated with LaTeXDraw 2.0.8
% Sat Sep 08 17:05:28 BST 2012
% \usepackage[usenames,dvipsnames]{pstricks}
% \usepackage{epsfig}
% \usepackage{pst-grad} % For gradients
% \usepackage{pst-plot} % For axes
\scalebox{1} % Change this value to rescale the drawing.
{
\begin{pspicture}(0,0)(5.5971875,3.5)
\psline[linewidth=0.04cm,arrowsize=0.05291667cm 2.0,arrowlength=1.4,arrowinset=0.4]{->}(3.579375,2.65)(3.579375,3.45)
\psline[linewidth=0.04cm,arrowsize=0.05291667cm 2.0,arrowlength=1.4,arrowinset=0.4]{->}(1.979375,2.65)(1.979375,3.45)
\psline[linewidth=0.04cm](1.979375,1.85)(1.979375,2.45)
\psline[linewidth=0.04cm,arrowsize=0.05291667cm 2.0,arrowlength=1.4,arrowinset=0.4]{->}(1.979375,1.45)(1.979375,2.05)
\psline[linewidth=0.04cm](3.579375,1.85)(3.579375,2.45)
\psline[linewidth=0.04cm,arrowsize=0.05291667cm 2.0,arrowlength=1.4,arrowinset=0.4]{->}(3.579375,1.45)(3.579375,2.05)
\psline[linewidth=0.04cm](3.579375,0.65)(3.579375,1.25)
\psline[linewidth=0.04cm,arrowsize=0.05291667cm 2.0,arrowlength=1.4,arrowinset=0.4]{->}(3.579375,0.25)(3.579375,0.85)
\psline[linewidth=0.04cm](3.579375,-1.75)(3.579375,-1.15)
\psline[linewidth=0.04cm,arrowsize=0.05291667cm 2.0,arrowlength=1.4,arrowinset=0.4]{->}(3.579375,-2.15)(3.579375,-1.55)
\psline[linewidth=0.04cm](3.579375,-0.55)(3.579375,0.05)
\psline[linewidth=0.04cm,arrowsize=0.05291667cm 2.0,arrowlength=1.4,arrowinset=0.4]{->}(3.579375,-0.95)(3.579375,-0.35)
\psline[linewidth=0.04cm](3.579375,-2.75)(3.579375,-2.15)
\psline[linewidth=0.04cm,arrowsize=0.05291667cm 2.0,arrowlength=1.4,arrowinset=0.4]{->}(3.579375,-3.15)(3.579375,-2.55)
\usefont{T1}{ptm}{m}{n}
\rput(3.820625,-1.625){\scriptsize 2}
\pscircle[linewidth=0.04,dimen=outer,fillstyle=solid,fillcolor=red](3.579375,-3.15){0.3}
\pscircle[linewidth=0.04,dimen=outer,fillstyle=solid,fillcolor=green](3.579375,0.15){0.3}
\psframe[linewidth=0.02,dimen=outer,fillstyle=solid,fillcolor=yellow](3.779375,-1.89)(3.339375,-2.33)
\usefont{T1}{ptm}{m}{n}
\rput(3.5825,-2.14){H}
\pscircle[linewidth=0.04,dimen=outer,fillstyle=solid,fillcolor=green](3.579375,2.55){0.3}
\psframe[linewidth=0.02,dimen=outer,fillstyle=solid,fillcolor=yellow](2.179375,2.81)(1.739375,2.37)
\usefont{T1}{ptm}{m}{n}
\rput(1.9825,2.56){H}
\usefont{T1}{ptm}{m}{n}
\rput(3.8223438,-0.425){\scriptsize 4}
\usefont{T1}{ptm}{m}{n}
\rput(3.8179688,0.775){\scriptsize 7}
\usefont{T1}{ptm}{m}{n}
\rput(3.8128126,1.975){\scriptsize 8}
\usefont{T1}{ptm}{m}{n}
\rput(3.8179688,3.175){\scriptsize 9}
\usefont{T1}{ptm}{m}{n}
\rput(2.2779686,3.175){\scriptsize 10}
\pscircle[linewidth=0.04,dimen=outer,fillstyle=solid,fillcolor=red](3.579375,1.35){0.3}
\usefont{T1}{ptm}{m}{n}
\rput(3.5779688,1.36){\boldmath$\pi$\unboldmath}
\pscircle[linewidth=0.04,dimen=outer,fillstyle=solid,fillcolor=red](3.579375,-1.05){0.3}
\usefont{T1}{ptm}{m}{n}
\rput(3.5779688,-1.04){\boldmath$\pi$\unboldmath}
\pscircle[linewidth=0.04,dimen=outer,fillstyle=solid,fillcolor=red](1.979375,1.35){0.3}
\usefont{T1}{ptm}{m}{n}
\rput(1.9779687,1.36){\boldmath$\pi$\unboldmath}
\usefont{T1}{ptm}{m}{n}
\rput(3.578125,2.545){\footnotesize \boldmath$\pi\!/\!4$\unboldmath}
\usefont{T1}{ppl}{m}{n}
\rput(3.59,0.145){\tiny -\boldmath$\pi\!/\!4$\unboldmath}
\end{pspicture}
}
\]
\vskip 3cm

\noindent We can then commute red and green phases as follows:
\[
\left\{
\begin{array}{l l}
    H : (R[\pi]_0^1)_* \otimes a_{10},  & G[\pi/4]_1^1 : a_{8*} \otimes a_9, \\
                                        & R[\pi]_1^1 : a_{7*} \otimes a_8, \\
                                        & R[\pi]_1^1 : a_{4*} \otimes a_7, \\
                                        & G[\pi/4]_1^1 : a_{2*} \otimes a_4, \\
                                        & H : (R_0^1)_* \otimes a_2
\end{array}
\right\}
% \begin{pspicture}(0,-3.45)(5.5971875,3.47)
% \begin{pspicture}(0,0)(5.5971875,3.5)
% ------------------------------------------
% Generated with LaTeXDraw 2.0.8
% Sat Sep 08 17:18:46 BST 2012
% \usepackage[usenames,dvipsnames]{pstricks}
% \usepackage{epsfig}
% \usepackage{pst-grad} % For gradients
% \usepackage{pst-plot} % For axes
\scalebox{1} % Change this value to rescale the drawing.
{
\begin{pspicture}(0,0)(5.5971875,3.5)
\psline[linewidth=0.04cm,arrowsize=0.05291667cm 2.0,arrowlength=1.4,arrowinset=0.4]{->}(3.579375,2.65)(3.579375,3.45)
\psline[linewidth=0.04cm,arrowsize=0.05291667cm 2.0,arrowlength=1.4,arrowinset=0.4]{->}(1.979375,2.65)(1.979375,3.45)
\psline[linewidth=0.04cm](1.979375,1.85)(1.979375,2.45)
\psline[linewidth=0.04cm,arrowsize=0.05291667cm 2.0,arrowlength=1.4,arrowinset=0.4]{->}(1.979375,1.45)(1.979375,2.05)
\psline[linewidth=0.04cm](3.579375,1.85)(3.579375,2.45)
\psline[linewidth=0.04cm,arrowsize=0.05291667cm 2.0,arrowlength=1.4,arrowinset=0.4]{->}(3.579375,1.45)(3.579375,2.05)
\psline[linewidth=0.04cm](3.579375,0.65)(3.579375,1.25)
\psline[linewidth=0.04cm,arrowsize=0.05291667cm 2.0,arrowlength=1.4,arrowinset=0.4]{->}(3.579375,0.25)(3.579375,0.85)
\psline[linewidth=0.04cm](3.579375,-1.75)(3.579375,-1.15)
\psline[linewidth=0.04cm,arrowsize=0.05291667cm 2.0,arrowlength=1.4,arrowinset=0.4]{->}(3.579375,-2.15)(3.579375,-1.55)
\psline[linewidth=0.04cm](3.579375,-0.55)(3.579375,0.05)
\psline[linewidth=0.04cm,arrowsize=0.05291667cm 2.0,arrowlength=1.4,arrowinset=0.4]{->}(3.579375,-0.95)(3.579375,-0.35)
\psline[linewidth=0.04cm](3.579375,-2.75)(3.579375,-2.15)
\psline[linewidth=0.04cm,arrowsize=0.05291667cm 2.0,arrowlength=1.4,arrowinset=0.4]{->}(3.579375,-3.15)(3.579375,-2.55)
\usefont{T1}{ptm}{m}{n}
\rput(3.820625,-1.625){\scriptsize 2}
\pscircle[linewidth=0.04,dimen=outer,fillstyle=solid,fillcolor=red](3.579375,-3.15){0.3}
\psframe[linewidth=0.02,dimen=outer,fillstyle=solid,fillcolor=yellow](3.789375,-1.88)(3.329375,-2.34)
\usefont{T1}{ptm}{m}{n}
\rput(3.5825,-2.14){H}
\pscircle[linewidth=0.04,dimen=outer,fillstyle=solid,fillcolor=green](3.579375,2.55){0.3}
\psframe[linewidth=0.02,dimen=outer,fillstyle=solid,fillcolor=yellow](2.189375,2.82)(1.729375,2.36)
\usefont{T1}{ptm}{m}{n}
\rput(1.9825,2.56){H}
\usefont{T1}{ptm}{m}{n}
\rput(3.8223438,-0.425){\scriptsize 4}
\usefont{T1}{ptm}{m}{n}
\rput(3.8179688,0.775){\scriptsize 7}
\usefont{T1}{ptm}{m}{n}
\rput(3.8128126,1.975){\scriptsize 8}
\usefont{T1}{ptm}{m}{n}
\rput(3.8179688,3.175){\scriptsize 9}
\usefont{T1}{ptm}{m}{n}
\rput(2.2779686,3.175){\scriptsize 10}
\pscircle[linewidth=0.04,dimen=outer,fillstyle=solid,fillcolor=red](3.579375,1.35){0.3}
\usefont{T1}{ptm}{m}{n}
\rput(3.5779688,1.36){\boldmath$\pi$\unboldmath}
\pscircle[linewidth=0.04,dimen=outer,fillstyle=solid,fillcolor=red](3.579375,0.15){0.3}
\usefont{T1}{ptm}{m}{n}
\rput(3.5779688,0.16){\boldmath$\pi$\unboldmath}
\pscircle[linewidth=0.04,dimen=outer,fillstyle=solid,fillcolor=red](1.979375,1.35){0.3}
\usefont{T1}{ptm}{m}{n}
\rput(1.9779687,1.36){\boldmath$\pi$\unboldmath}
\usefont{T1}{ptm}{m}{n}
\rput(3.578125,2.545){\footnotesize \boldmath$\pi\!/\!4$\unboldmath}
\pscircle[linewidth=0.04,dimen=outer,fillstyle=solid,fillcolor=green](3.579375,-1.05){0.3}
\usefont{T1}{ptm}{m}{n}
\rput(3.578125,-1.055){\footnotesize \boldmath$\pi\!/\!4$\unboldmath}
\end{pspicture}
}
\]
\vskip 3cm

\noindent The red phases cancel each other out, since two consecutive $\pi$ rotations bring us back where we started:
\[
\left\{
\begin{array}{l l}
    H : (R[\pi]_0^1)_* \otimes a_{10},  & G[\pi/4]_1^1 : a_{8*} \otimes a_9, \\
                                        & G[\pi/4]_1^1 : a_{2*} \otimes a_8, \\
                                        & H : (R_0^1)_* \otimes a_2
\end{array}
\right\}
% \begin{pspicture}(0,-2.25)(5.5775,2.27)
% \begin{pspicture}(0,0)(5.5775,2.27)
% ------------------------------------------
% Generated with LaTeXDraw 2.0.8
% Sat Sep 08 17:23:33 BST 2012
% \usepackage[usenames,dvipsnames]{pstricks}
% \usepackage{epsfig}
% \usepackage{pst-grad} % For gradients
% \usepackage{pst-plot} % For axes
\scalebox{1} % Change this value to rescale the drawing.
{
\begin{pspicture}(0,0)(5.5775,2.27)
\psline[linewidth=0.04cm,arrowsize=0.05291667cm 2.0,arrowlength=1.4,arrowinset=0.4]{->}(3.579375,1.45)(3.579375,2.25)
\psline[linewidth=0.04cm,arrowsize=0.05291667cm 2.0,arrowlength=1.4,arrowinset=0.4]{->}(1.979375,1.45)(1.979375,2.25)
\psline[linewidth=0.04cm](1.979375,0.65)(1.979375,1.25)
\psline[linewidth=0.04cm,arrowsize=0.05291667cm 2.0,arrowlength=1.4,arrowinset=0.4]{->}(1.979375,0.25)(1.979375,0.85)
\psline[linewidth=0.04cm](3.579375,0.65)(3.579375,1.25)
\psline[linewidth=0.04cm,arrowsize=0.05291667cm 2.0,arrowlength=1.4,arrowinset=0.4]{->}(3.579375,0.25)(3.579375,0.85)
\psline[linewidth=0.04cm](3.579375,-0.55)(3.579375,0.05)
\psline[linewidth=0.04cm,arrowsize=0.05291667cm 2.0,arrowlength=1.4,arrowinset=0.4]{->}(3.579375,-0.95)(3.579375,-0.35)
\psline[linewidth=0.04cm](3.579375,-1.55)(3.579375,-0.95)
\psline[linewidth=0.04cm,arrowsize=0.05291667cm 2.0,arrowlength=1.4,arrowinset=0.4]{->}(3.579375,-1.95)(3.579375,-1.35)
\usefont{T1}{ptm}{m}{n}
\rput(3.820625,-0.425){\scriptsize 2}
\pscircle[linewidth=0.04,dimen=outer,fillstyle=solid,fillcolor=red](3.579375,-1.95){0.3}
\psframe[linewidth=0.02,dimen=outer,fillstyle=solid,fillcolor=yellow](3.799375,-0.67)(3.319375,-1.15)
\usefont{T1}{ptm}{m}{n}
\rput(3.5825,-0.94){H}
\pscircle[linewidth=0.04,dimen=outer,fillstyle=solid,fillcolor=green](3.579375,1.35){0.3}
\psframe[linewidth=0.02,dimen=outer,fillstyle=solid,fillcolor=yellow](2.199375,1.63)(1.719375,1.15)
\usefont{T1}{ptm}{m}{n}
\rput(1.9825,1.36){H}
\usefont{T1}{ptm}{m}{n}
\rput(3.8128126,0.775){\scriptsize 8}
\usefont{T1}{ptm}{m}{n}
\rput(3.8179688,1.975){\scriptsize 9}
\usefont{T1}{ptm}{m}{n}
\rput(2.2779686,1.975){\scriptsize 10}
\pscircle[linewidth=0.04,dimen=outer,fillstyle=solid,fillcolor=red](1.979375,0.15){0.3}
\usefont{T1}{ptm}{m}{n}
\rput(1.9779687,0.16){\boldmath$\pi$\unboldmath}
\usefont{T1}{ptm}{m}{n}
\rput(3.578125,1.345){\footnotesize \boldmath$\pi\!/\!4$\unboldmath}
\pscircle[linewidth=0.04,dimen=outer,fillstyle=solid,fillcolor=green](3.579375,0.15){0.3}
\usefont{T1}{ptm}{m}{n}
\rput(3.578125,0.145){\footnotesize \boldmath$\pi\!/\!4$\unboldmath}
\end{pspicture}
}
\]
\vskip 2cm

\noindent We can now use the Hadamards to transform the remaining red terms into green ones:
\[
\left\{
\begin{array}{l l}
    G[\pi]_0^1 : a_{10},    & G[\pi/4]_1^1 : a_{8*} \otimes a_9, \\
                            & G[\pi/4]_1^1 : a_{2*} \otimes a_8, \\
                            & G_0^1 : a_2
\end{array}
\right\}
% \begin{pspicture}(0,-1.8)(5.5775,1.82)
% \begin{pspicture}(0,0)(5.5775,1.82)
% ------------------------------------------
% Generated with LaTeXDraw 2.0.8
% Sat Sep 08 17:31:35 BST 2012
% \usepackage[usenames,dvipsnames]{pstricks}
% \usepackage{epsfig}
% \usepackage{pst-grad} % For gradients
% \usepackage{pst-plot} % For axes
\scalebox{1} % Change this value to rescale the drawing.
{
\begin{pspicture}(0,0)(5.5775,1.82)
\psline[linewidth=0.04cm,arrowsize=0.05291667cm 2.0,arrowlength=1.4,arrowinset=0.4]{->}(3.579375,1.0)(3.579375,1.8)
\psline[linewidth=0.04cm,arrowsize=0.05291667cm 2.0,arrowlength=1.4,arrowinset=0.4]{->}(1.979375,1.0)(1.979375,1.8)
\psline[linewidth=0.04cm](3.579375,0.2)(3.579375,0.8)
\psline[linewidth=0.04cm,arrowsize=0.05291667cm 2.0,arrowlength=1.4,arrowinset=0.4]{->}(3.579375,-0.2)(3.579375,0.4)
\psline[linewidth=0.04cm](3.579375,-1.0)(3.579375,-0.4)
\psline[linewidth=0.04cm,arrowsize=0.05291667cm 2.0,arrowlength=1.4,arrowinset=0.4]{->}(3.579375,-1.4)(3.579375,-0.8)
\usefont{T1}{ptm}{m}{n}
\rput(3.820625,-0.875){\scriptsize 2}
\pscircle[linewidth=0.04,dimen=outer,fillstyle=solid,fillcolor=green](3.579375,-1.5){0.3}
\pscircle[linewidth=0.04,dimen=outer,fillstyle=solid,fillcolor=green](3.579375,0.9){0.3}
\usefont{T1}{ptm}{m}{n}
\rput(3.8128126,0.325){\scriptsize 8}
\usefont{T1}{ptm}{m}{n}
\rput(3.8179688,1.525){\scriptsize 9}
\usefont{T1}{ptm}{m}{n}
\rput(2.2779686,1.525){\scriptsize 10}
\usefont{T1}{ptm}{m}{n}
\rput(3.578125,0.895){\footnotesize \boldmath$\pi\!/\!4$\unboldmath}
\pscircle[linewidth=0.04,dimen=outer,fillstyle=solid,fillcolor=green](3.579375,-0.3){0.3}
\usefont{T1}{ptm}{m}{n}
\rput(3.578125,-0.305){\footnotesize \boldmath$\pi\!/\!4$\unboldmath}
\pscircle[linewidth=0.04,dimen=outer,fillstyle=solid,fillcolor=green](1.979375,0.9){0.3}
\usefont{T1}{ptm}{m}{n}
\rput(1.9779687,0.91){\boldmath$\pi$\unboldmath}
\end{pspicture}
}
\]
\vskip 2cm

\noindent Which in the end reduces to the two-qubit quantum Fourier transform's known output for our choice of input:
\[
\Big\{ G[\pi]_0^1 : a_{10}, G[\pi/2]_0^1 : a_9 \Big\}
% \begin{pspicture}(0,-0.6)(5.5775,0.62)
% \begin{pspicture}(0,0)(5.5775,0.62)
% ------------------------------------------
% Generated with LaTeXDraw 2.0.8
% Sat Sep 08 17:36:48 BST 2012
% \usepackage[usenames,dvipsnames]{pstricks}
% \usepackage{epsfig}
% \usepackage{pst-grad} % For gradients
% \usepackage{pst-plot} % For axes
\scalebox{1} % Change this value to rescale the drawing.
{
\begin{pspicture}(0,0)(5.5775,0.62)
\psline[linewidth=0.04cm,arrowsize=0.05291667cm 2.0,arrowlength=1.4,arrowinset=0.4]{->}(3.579375,-0.2)(3.579375,0.6)
\psline[linewidth=0.04cm,arrowsize=0.05291667cm 2.0,arrowlength=1.4,arrowinset=0.4]{->}(1.979375,-0.2)(1.979375,0.6)
\pscircle[linewidth=0.04,dimen=outer,fillstyle=solid,fillcolor=green](3.579375,-0.3){0.3}
\usefont{T1}{ptm}{m}{n}
\rput(3.8179688,0.325){\scriptsize 9}
\usefont{T1}{ptm}{m}{n}
\rput(2.2779686,0.325){\scriptsize 10}
\usefont{T1}{ptm}{m}{n}
\rput(3.578125,-0.305){\footnotesize \boldmath$\pi\!/\!2$\unboldmath}
\pscircle[linewidth=0.04,dimen=outer,fillstyle=solid,fillcolor=green](1.979375,-0.3){0.3}
\usefont{T1}{ptm}{m}{n}
\rput(1.9779687,-0.29){\boldmath$\pi$\unboldmath}
\end{pspicture}
}
\]
\vskip 1cm

The quantum Fourier Transform's behaviour on other possible inputs can be verified by following a similar process for the rest of the basis states. The fact that the dagger lambda calculus can represent the quantum Fourier transform and, by extension, Shor's factoring algorithm speaks volumes about the language's expressive power.

\chapter{Teleportation Protocol}
\label{Chapter:Teleportation Protocol}
The teleportation protocol, one of the most famous quantum protocols, uses entanglement to teleport an arbitrary quantum state. In the most common presentation of the protocol, Alice is in possession of a qubit with a quantum state that she wants to teleport to Bob. Alice and Bob share a classical two-bit communication channel but, since the qubit has not been measured, its state carries much more information than two classical bits. In order for Alice to teleport the data, she makes use of an entangled pair of qubits that she shares with Bob. Alice performs a Bell basis measurement on her qubit and her part of the entangled pair of qubits; she then uses the classical communication channel to communicate the result of that measurement to Bob, who uses it to perform a unitary correction on his qubit.

\begin{remark}
    Similarly to the remark of \cite{CD11} about their ZX-calculus, the diagrammatic forefather of the dagger lambda calculus, our language does not represent the non-deterministic aspect of measurements. Every sequent in this lambda calculus, like every diagram in the ZX-calculus, corresponds to one run of the experiment, as far as measurements are concerned. For this reason, measurements are replaced by the projections into which they will resolve in this run of the experiment. To avoid having to consider sequents for every possible outcome of an experiment, the dagger lambda calculus could be extended to support the conditional diagrams of \cite{DP10} and \cite{CD11}. This would be a significant departure from the simplified logic of the lambda calculus and thereby one that merits to be studied in its own right. Consequently, we include this later on, in the \textit{Further work} section of this dissertation.
\end{remark}

Back to our description of the teleportation protocol, Alice's Bell basis measurement will resolve to a projection on one of the Bell basis states: $\bra{\Psi_+}$, $\bra{\Psi_-}$, $\bra{\Phi_+}$ or $\bra{\Phi_-}$. We therefore use $\bra{+_\alpha}$ and $\bra{+_\beta}$ to represent the outcome of the measurement, where the four possible pairs of $\alpha, \beta \in \{0, \pi\}$ range over all the outcomes of a measurement against the Bell basis. We will now present a diagram for part of the teleportation protocol, like the one used by \cite{CD11}, but adapted to support dualisers. Note that this diagram is from before the unitary corrections are performed. The full diagram for teleportation, with all the corrections, will be presented shortly afterwards:
\[
% Generated with LaTeXDraw 2.0.8
% Sun Oct 07 23:10:43 BST 2012
% \usepackage[usenames,dvipsnames]{pstricks}
% \usepackage{epsfig}
% \usepackage{pst-grad} % For gradients
% \usepackage{pst-plot} % For axes
\scalebox{1} % Change this value to rescale the drawing.
{
\begin{pspicture}(0,-2.97)(7.519375,2.97)
\psframe[linewidth=0.02,linestyle=dashed,dash=0.16cm 0.16cm,dimen=outer](5.519375,2.85)(1.719375,-2.85)
\psframe[linewidth=0.02,linestyle=dashed,dash=0.16cm 0.16cm,dimen=outer](7.519375,2.85)(5.719375,-2.85)
\psline[linewidth=0.04cm,arrowsize=0.05291667cm 2.0,arrowlength=1.4,arrowinset=0.4]{->}(4.619375,-0.85)(4.619375,-0.45)
\psline[linewidth=0.04cm](2.819375,-1.73)(2.819375,-2.25)
\usefont{T1}{ptm}{m}{n}
\rput(2.9559374,-2.425){\scriptsize 1}
\psline[linewidth=0.04cm,arrowsize=0.05291667cm 2.0,arrowlength=1.4,arrowinset=0.4]{<-}(2.819375,-2.15)(2.819375,-2.95)
\rput{-180.0}(11.23875,-3.1){\psarc[linewidth=0.04](5.619375,-1.55){1.0}{0.0}{180.0}}
\psline[linewidth=0.04cm,arrowsize=0.05291667cm 2.0,arrowlength=1.4,arrowinset=0.4]{<-}(4.619375,-1.65)(4.619375,-1.25)
\psarc[linewidth=0.04](2.819375,-0.81){0.6}{-180.0}{0.0}
\pscircle[linewidth=0.04,dimen=outer,fillstyle=solid,fillcolor=green](2.819375,-1.55){0.3}
\psarc[linewidth=0.04](4.019375,-0.53){0.6}{0.0}{180.0}
\psline[linewidth=0.04cm,linecolor=red,arrowsize=0.05291667cm 2.0,arrowlength=1.4,arrowinset=0.4]{<->}(4.619375,-1.35)(4.619375,-0.75)
\psline[linewidth=0.04cm,arrowsize=0.05291667cm 2.0,arrowlength=1.4,arrowinset=0.4]{->}(3.419375,-0.85)(3.419375,-0.45)
\psline[linewidth=0.04cm,arrowsize=0.05291667cm 2.0,arrowlength=1.4,arrowinset=0.4]{->}(2.219375,-0.85)(2.219375,0.95)
\psline[linewidth=0.04cm](4.019375,1.75)(4.019375,2.15)
\psline[linewidth=0.04cm,arrowsize=0.05291667cm 2.0,arrowlength=1.4,arrowinset=0.4]{->}(4.019375,1.35)(4.019375,1.95)
\psline[linewidth=0.04cm](4.019375,0.75)(4.019375,1.35)
\psline[linewidth=0.04cm,arrowsize=0.05291667cm 2.0,arrowlength=1.4,arrowinset=0.4]{->}(4.019375,0.45)(4.019375,0.95)
\usefont{T1}{ptm}{m}{n}
\rput(4.2528124,1.875){\scriptsize 8}
\psframe[linewidth=0.02,dimen=outer,fillstyle=solid,fillcolor=yellow](4.249375,1.64)(3.749375,1.14)
\usefont{T1}{ptm}{m}{n}
\rput(4.0225,1.36){H}
\pscircle[linewidth=0.04,dimen=outer,fillstyle=solid,fillcolor=red](4.019375,0.25){0.3}
\psline[linewidth=0.04cm](2.219375,0.75)(2.219375,2.15)
\psline[linewidth=0.04cm,arrowsize=0.05291667cm 2.0,arrowlength=1.4,arrowinset=0.4]{->}(6.619375,-1.65)(6.619375,0.95)
\psline[linewidth=0.04cm](6.619375,0.75)(6.619375,2.95)
\usefont{T1}{ppl}{m}{n}
\rput{90.0}(-0.32156253,-4.3378124){\rput(2.0021875,-2.33){\small Alice}}
\usefont{T1}{ppl}{m}{n}
\rput{90.0}(4.8682814,-9.528906){\rput(7.1939063,-2.33){\small Bob}}
\pscircle[linewidth=0.04,dimen=outer,fillstyle=solid,fillcolor=green](2.219375,2.35){0.3}
\usefont{T1}{ptm}{m}{n}
\rput(2.2279687,2.36){\boldmath$\alpha$\unboldmath}
\pscircle[linewidth=0.04,dimen=outer,fillstyle=solid,fillcolor=green](4.019375,2.35){0.3}
\usefont{T1}{ptm}{m}{n}
\rput(4.027969,2.34){\boldmath$\beta$\unboldmath}
\usefont{T1}{ptm}{m}{n}
\rput(4.880625,-1.525){\scriptsize 2}
\usefont{T1}{ptm}{m}{n}
\rput(6.3721876,-1.525){\scriptsize 3}
\usefont{T1}{ptm}{m}{n}
\rput(2.3823438,-0.525){\scriptsize 4}
\usefont{T1}{ptm}{m}{n}
\rput(3.6734376,-0.525){\scriptsize 5}
\usefont{T1}{ptm}{m}{n}
\rput(4.8778124,-0.525){\scriptsize 6}
\usefont{T1}{ptm}{m}{n}
\rput(4.277969,0.775){\scriptsize 7}
\end{pspicture}
}
\]

\noindent This diagram can easily be represented as a sequent in the lambda calculus. The sequent will be $a_1 : A \vdash_S a_3 : A$, where the soup S is:
\[
\left\{
\begin{array}{l l}
    G[\alpha]_1^0 : x_{4*}                  & G[\beta]_1^0 : a_{8*}, \\
                                            & H : a_{7*} \otimes a_8, \\
                                            & R_2^1 : (a_5 \otimes a_6)_* \otimes a_7, \\
    G_1^2 : x_{1*} \otimes x_4 \otimes x_5  & d_{A*} : a_2 \otimes a_6, \\
                                            & x_* \otimes x : a_{2*} \otimes a_3
\end{array}
\right\}
\]

Note that, in this sequent, the terms $a_{2*}$ and $a_3$ represent the entangled pair of qubits that are shared between Alice and Bob. Alice possesses $a_{2*}$ and Bob possesses $a_3$. Bob receives the results of Alice's projection, $\alpha$ and $\beta$, from the classical communication channel. We will now present the full diagram, which includes the unitary corrections that will be performed by Bob; the classical communication channel is not represented in the diagram with any wires but is instead reflected by the fact that Alice's $\alpha$ and $\beta$ measurements are matched by corresponding corrections by Bob. Having received Alice's measurement values classically, Bob can perform a red phase rotation by an angle of $\beta$, followed by a green phase rotation by an angle of $\alpha$, on his term $a_3$. This completes Bob's unitary correction, who should now possess Alice's original input state:
\[
% Generated with LaTeXDraw 2.0.8
% Sun Oct 07 23:11:59 BST 2012
% \usepackage[usenames,dvipsnames]{pstricks}
% \usepackage{epsfig}
% \usepackage{pst-grad} % For gradients
% \usepackage{pst-plot} % For axes
\scalebox{1} % Change this value to rescale the drawing.
{
\begin{pspicture}(0,-2.97)(8.897187,2.97)
\psframe[linewidth=0.02,linestyle=dashed,dash=0.16cm 0.16cm,dimen=outer](5.519375,2.85)(1.719375,-2.85)
\psframe[linewidth=0.02,linestyle=dashed,dash=0.16cm 0.16cm,dimen=outer](7.519375,2.85)(5.719375,-2.85)
\psline[linewidth=0.04cm,arrowsize=0.05291667cm 2.0,arrowlength=1.4,arrowinset=0.4]{->}(4.619375,-0.85)(4.619375,-0.45)
\psline[linewidth=0.04cm](2.819375,-1.73)(2.819375,-2.25)
\usefont{T1}{ptm}{m}{n}
\rput(2.9559374,-2.425){\scriptsize 1}
\psline[linewidth=0.04cm,arrowsize=0.05291667cm 2.0,arrowlength=1.4,arrowinset=0.4]{<-}(2.819375,-2.15)(2.819375,-2.95)
\rput{-180.0}(11.23875,-3.1){\psarc[linewidth=0.04](5.619375,-1.55){1.0}{0.0}{180.0}}
\psline[linewidth=0.04cm,arrowsize=0.05291667cm 2.0,arrowlength=1.4,arrowinset=0.4]{<-}(4.619375,-1.65)(4.619375,-1.25)
\psarc[linewidth=0.04](2.819375,-0.81){0.6}{-180.0}{0.0}
\pscircle[linewidth=0.04,dimen=outer,fillstyle=solid,fillcolor=green](2.819375,-1.55){0.3}
\psarc[linewidth=0.04](4.019375,-0.53){0.6}{0.0}{180.0}
\psline[linewidth=0.04cm,linecolor=red,arrowsize=0.05291667cm 2.0,arrowlength=1.4,arrowinset=0.4]{<->}(4.619375,-1.35)(4.619375,-0.75)
\psline[linewidth=0.04cm,arrowsize=0.05291667cm 2.0,arrowlength=1.4,arrowinset=0.4]{->}(3.419375,-0.85)(3.419375,-0.45)
\psline[linewidth=0.04cm,arrowsize=0.05291667cm 2.0,arrowlength=1.4,arrowinset=0.4]{->}(2.219375,-0.85)(2.219375,0.95)
\psline[linewidth=0.04cm](4.019375,1.75)(4.019375,2.15)
\psline[linewidth=0.04cm,arrowsize=0.05291667cm 2.0,arrowlength=1.4,arrowinset=0.4]{->}(4.019375,1.35)(4.019375,1.95)
\psline[linewidth=0.04cm](4.019375,0.75)(4.019375,1.35)
\psline[linewidth=0.04cm,arrowsize=0.05291667cm 2.0,arrowlength=1.4,arrowinset=0.4]{->}(4.019375,0.45)(4.019375,0.95)
\usefont{T1}{ptm}{m}{n}
\rput(4.2528124,1.875){\scriptsize 8}
\psframe[linewidth=0.02,dimen=outer,fillstyle=solid,fillcolor=yellow](4.259375,1.65)(3.739375,1.13)
\usefont{T1}{ptm}{m}{n}
\rput(4.0225,1.36){H}
\pscircle[linewidth=0.04,dimen=outer,fillstyle=solid,fillcolor=red](4.019375,0.25){0.3}
\psline[linewidth=0.04cm](2.219375,0.75)(2.219375,2.15)
\psline[linewidth=0.04cm,arrowsize=0.05291667cm 2.0,arrowlength=1.4,arrowinset=0.4]{->}(6.619375,0.75)(6.619375,2.05)
\usefont{T1}{ppl}{m}{n}
\rput{90.0}(-0.32156253,-4.3378124){\rput(2.0021875,-2.33){\small Alice}}
\usefont{T1}{ppl}{m}{n}
\rput{90.0}(4.8682814,-9.528906){\rput(7.1939063,-2.33){\small Bob}}
\pscircle[linewidth=0.04,dimen=outer,fillstyle=solid,fillcolor=green](2.219375,2.35){0.3}
\usefont{T1}{ptm}{m}{n}
\rput(2.2279687,2.36){\boldmath$\alpha$\unboldmath}
\pscircle[linewidth=0.04,dimen=outer,fillstyle=solid,fillcolor=green](4.019375,2.35){0.3}
\usefont{T1}{ptm}{m}{n}
\rput(4.027969,2.34){\boldmath$\beta$\unboldmath}
\usefont{T1}{ptm}{m}{n}
\rput(4.880625,-1.525){\scriptsize 2}
\usefont{T1}{ptm}{m}{n}
\rput(6.3721876,-1.525){\scriptsize 3}
\usefont{T1}{ptm}{m}{n}
\rput(2.3823438,-0.525){\scriptsize 4}
\usefont{T1}{ptm}{m}{n}
\rput(3.6734376,-0.525){\scriptsize 5}
\usefont{T1}{ptm}{m}{n}
\rput(4.8778124,-0.525){\scriptsize 6}
\usefont{T1}{ptm}{m}{n}
\rput(4.277969,0.775){\scriptsize 7}
\pscircle[linewidth=0.04,dimen=outer,fillstyle=solid,fillcolor=green](6.619375,1.35){0.3}
\usefont{T1}{ptm}{m}{n}
\rput(6.627969,1.36){\boldmath$\alpha$\unboldmath}
\psline[linewidth=0.04cm](6.619375,1.85)(6.619375,2.95)
\psline[linewidth=0.04cm,arrowsize=0.05291667cm 2.0,arrowlength=1.4,arrowinset=0.4]{->}(6.619375,-1.65)(6.619375,-0.45)
\psline[linewidth=0.04cm,arrowsize=0.05291667cm 2.0,arrowlength=1.4,arrowinset=0.4]{->}(6.619375,-0.65)(6.619375,0.95)
\pscircle[linewidth=0.04,dimen=outer,fillstyle=solid,fillcolor=red](6.619375,0.25){0.3}
\usefont{T1}{ptm}{m}{n}
\rput(6.627969,0.24){\boldmath$\beta$\unboldmath}
\usefont{T1}{ptm}{m}{n}
\rput(6.857969,0.775){\scriptsize 9}
\usefont{T1}{ptm}{m}{n}
\rput(6.9179688,2.275){\scriptsize 10}
\end{pspicture}
}
\]

\noindent Going back to the lambda calculus, the sequent for the full diagram is expressed by performing a Cut of $a_1 : A \vdash_S a_3 : A$ with $a_3 : A \vdash_{\{ R[\beta]_1^1 : a_{3*} \otimes a_9, G[\alpha]_1^1 : a_{9*} \otimes a_{10} \}} a_{10} : A$, which gives us:
\[ a_1 : A \vdash_{S \cup \{ R[\beta]_1^1 : a_{3*} \otimes a_9, G[\alpha]_1^1 : a_{9*} \otimes a_{10} \}} a_{10} : A \]

\noindent We will now start performing operations in our soup, to simplify the existing connections, and prove that the teleportation sequent does actually produce Alice's original input state when Bob runs his corrections at the other end. We will use the soup rules for lifting terms, to lift $G[\alpha]_1^0$ and $G[\beta]_1^0$ into phase shifts. We will also use the rule for the Hadamard gate, to transform the resulting $G[\beta]_1^1$ to a $R[\beta]_1^1$. The resulting soup will be of the form:
\[
\left\{
\begin{array}{l l}
                                        & G[\alpha]_1^1 : a_{9*} \otimes a_{10}, \\
                                        & R[\beta]_1^1 : a_{3*} \otimes a_9, \\
    R[\beta]_1^1 : a_{5*} \otimes a_3,  & \\
    G[\alpha]_1^1 : a_{1*} \otimes a_5  &
\end{array}
\right\}
\]

\noindent But recall that $\alpha, \beta \in \{0, \pi\}$, so $2\alpha = 2\beta = 0$. The consecutive red $\beta$ rotations hence cancel out and so do the remaining green $\alpha$ rotations after them. This leaves us with the following soup, thus proving that Alice's original input state was teleported to Bob:
\[ \big\{ a_1 : a_{10} \big\} \]

\clearpage
\part{Conclusion}
\label{Part:Conclusion}
\chapter{Concluding remarks}
\label{Chapter:Concluding remarks}
We have now completed our study of the dagger lambda calculus; a higher-order language that was initially defined with dagger compact categories in mind, but which was later imbued with classical control by enriching it with complementary control structures. The dagger lambda calculus has proven to be a powerful and expressive language, capable of running quantum protocols and programs. Our main motivation behind this language has been to bridge the gap between two different approaches towards quantum computation, by reconciling the diagrammatic formalism with the type theoretic research of higher-order quantum programming languages. Our study of the language's semantics has revealed a number of insights about the building blocks of quantum computation, such as the symmetry of substitution, or the admissibility of some of the rules, giving us the ability to decompose classical notions of computation into finer primitives.

\section{Future work}
\label{Section:Future work}
As per our remark in chapter \ref{Chapter:Teleportation Protocol}, the dagger lambda calculus in its current form does not have a way of representing the non-deterministic aspect of measurements. Every sequent in the lambda calculus corresponds to one run of the experiment, where measurements are, as a result, replaced by the projections into which they will resolve. A way of binding measurement results, using \textit{conditional diagrams}, was recently proposed by \cite{DP10} as a modification to the diagrammatic calculus. One area of future work is the modification of the dagger lambda calculus, so that it supports \textit{conditional sequents}. This would require a significant departure from the language's original structure, but it would be interesting to examine whether these new connections between measurement results could be incorporated into a richer type of soup.

Another area for potential future work is that of Measurement Based Quantum Computation \cite{RB01, RB02, RBB03}. Since the dagger lambda calculus supports both quantum and classical data, it would be ideally suited for such a computational paradigm. The lambda calculus could be extended to support a language like the \textit{Measurement Calculus} of \cite{DKP07}. It would then be interesting to investigate whether the dagger lambda calculus can exhibit an equational correspondence to the Measurement Calculus, as well as to analyse the language's properties in a way that is similar to the analysis we performed in section \ref{Section:Proofs of properties}.

Once we have defined a programming language that fully supports Measurement Based Quantum Computation and the non-deterministic behaviour of measurements, we will be able to use it to represent all of the elements of quantum algorithms. There are not many efficient quantum algorithms out there. By representing the ones that are available, however, we will be able to spot their common elements and distil the structural elements that are responsible for the quantum mechanical speedup.

\section{Acknowledgements}
\label{Section:Acknowledgements}
I would like to thank my supervisors, Samson Abramsky and Bob Coecke, for their guidance and support throughout the course of my DPhil. Samson is an inexhaustible source of information in the areas of type theory, logic and category theory. In addition to being a supporting, fatherly figure for his students, he could always see many steps ahead in my research and steer me clear of problems. Moreover, his advice has been instrumental in helping me shape the dagger lambda calculus, which serves as the cornerstone of this dissertation. Bob introduced me to categorical quantum computing and his guidance in the early stages of my DPhil inspired my interest in diagrammatic calculi and measurement based quantum computation.

I also want to thank my examiners, Prakash Panangaden and Jonathan Barrett, for their very detailed and helpful comments and the corrections they suggested. I am also thankful towards the anonymous reviewers at a premature conference submission of mine (if you are reading this, you know who you are); their feedback has proved invaluable in shaping the exposition of some of my chapters.

I particularly want to thank my friend and colleague Nikos Tzevelekos for teaching me category theory and for the many in depth discussions and insights he offered during the development of my calculus. I also want to thank Rick Blute, Andrzej Murawski and Andreas Doering for their helpful feedback on early drafts of my research, and Peter Selinger, Beno\^{i}t Valiron, John Baez and Mike Stay for our discussions on quantum programming languages. Thanks are also due to Ross Duncan, Mehrnoosh Sadrzadeh, Simon Perdrix, \'{E}ric Paquette, Chris Heunen, Jaime Vicary, Aleks Kissinger, Bill Edwards, Alejandro D\'{\i}az-Caro and Du\v{s}ko Pavlovi\'{c} for answering questions, sharing their ideas and offering advice on topics ranging from Frobenius algebras and logic, to free categories and quantum programming. I am also thankful to Bob Harper, Benjamin Piece and Frank Pfenning for their lectures at the 2010 Oregon Programming Languages Summer School and for the engaging discussions we had in between sessions.

I am particularly grateful to Joel Spencer, for transmitting to me his love of mathematical logic and theoretical computer science during my undergraduate years. I am also grateful to Erasmia Kiriazi and Elias Kamouzis for lighting up in me a passion for mathematics.

I am grateful to the US Office of Naval Research, FQXi, and the Levelhulme Foundation for supporting my studies financially; this research would not have been possible without them. I also want to thank my college, Wolfson College, and the Department of Computer Science for supporting my trip to attend the 2010 Oregon Programming Languages Summer School.

An acknowledgements section, like this one, would not be complete without mentioning those whose love, friendship, and support have accompanied me throughout my DPhil. A doctoral degree is as much a struggle with one's inner self as it is a struggle with uncovering the secrets of science and, in that sense, personal relationships play a pivotal role in helping the human achieve his true potential. I am thankful to my fianc\'{e}e, Chara Tzanetaki, for her unwavering love and support throughout the best and the worst times in my doctoral research. Chara has always been there to push me forward, lift my spirits in times I could not do so myself, and even proofread drafts of my research. I am also thankful to my parents, my brother, my friends from Oxford: Yiannis Hadjimichael, Konstantinos Stamatis, Thomas Papadopoulos, Illektra Apostolidou, Mary Kopsacheili and Kyriaki Michailidou; as well as my friends from Greece: Tasos Katechis, Nicholas Tsiroyiannis, Konstantinos Nikolaras, George Pallis. Finally, a big thanks to my friends in Cambridge: Napoleon Katsos, Nausica Smith, and little George Jason Katsos who are now a lot more than just friends and a lot more than just family.

\clearpage
\bibliographystyle{alpha}
\bibliography{thesis}

\newcommand{\etalchar}[1]{$^{#1}$}
\begin{thebibliography}{CWW{\etalchar{+}}11}

\bibitem[ABP99]{ABP99}
Samson Abramsky, Rick Blute, and Prakash Panangaden.
\newblock Nuclear and trace ideals in tensored $*$-categories.
\newblock {\em Journal of Pure and Applied Algebra}, 143:3--47, 1999.

\bibitem[Abr93]{Abr93}
Samson Abramsky.
\newblock Computational interpretations of linear logic.
\newblock {\em Theoretical Computer Science}, 111:3--57, 1993.
\newblock (DOI:10.1.1.16.2984).

\bibitem[Abr05]{Abr05}
Samson Abramsky.
\newblock Abstract scalars, loops, and free traced and strongly compact closed
  categories.
\newblock In {\em In Proceedings of the First Conference on Algebra and
  Coalgebra in Computer Science (CALCO 2005)}, volume 3629, pages 1--31.
  Springer Lecture Notes in Computer Science, 2005.
\newblock (arXiv:0910.2931v1 [quant-ph]).

\bibitem[Abr10]{Abr10}
Samson Abramsky.
\newblock No-cloning in categorical quantum mechanics.
\newblock In S.~Gay and I.~Mackie, editors, {\em Semantic Techniques in Quantum
  Computation}. Cambridge University Press, 2010.
\newblock (arXiv:0910.2401v2 [quant-ph]).

\bibitem[AC04]{AC04}
Samson Abramsky and Bob Coecke.
\newblock A categorical semantics of quantum protocols.
\newblock In {\em Proceedings of the 19th IEEE conference on Logic in Computer
  Science (LiCS'04)}. IEEE Computer Science Press, 2004.
\newblock (arXiv:quant-ph/0402130v5).

\bibitem[AD06]{AD06}
Samson Abramsky and Ross Duncan.
\newblock A categorical quantum logic.
\newblock {\em Mathematical Structures in Computer Science}, 16:469--489, 2006.
\newblock (arXiv:quant-ph/0512114v1).

\bibitem[AT10]{AT10}
Samson Abramsky and Nikos Tzevelekos.
\newblock Introduction to categories and categorical logic.
\newblock In Bob Coecke, editor, {\em New Structures for Physics}. Springer
  Lecture Notes in Physics, 2010.
\newblock (arXiv:1102.1313v1 [math.CT]).

\bibitem[BB84]{BB84}
{Charles H.} Bennett and Gilles Brassard.
\newblock Quantum cryptography: Public key distribution and coin tossing.
\newblock In {\em Proceedings of the IEEE International Conference on
  Computers, Systems, and Signal Processing}, pages 175--179. IEEE Press, 1984.

\bibitem[BS10]{BS10}
John Baez and Michael Stay.
\newblock Physics, topology, logic and computation: A rosetta stone.
\newblock In Bob Coecke, editor, {\em New Structures for Physics}. Springer
  Lecture Notes in Physics, 2010.
\newblock (arXiv:0903.0340v3 [quant-ph]).

\bibitem[CD08]{CD08}
Bob Coecke and Ross Duncan.
\newblock Interacting quantum observables.
\newblock In {\em Proceedings of the 35th International Colloquium on Automata,
  Languages and Programming}, pages 298–--310. Lecture Notes in Computer
  Science 5126, Springer-Verlag, 2008.
\newblock (arXiv:0906.4725v1 [quant-ph]).

\bibitem[CD11]{CD11}
Bob Coecke and Ross Duncan.
\newblock Interacting quantum observables: Categorical algebra and
  diagrammatics.
\newblock {\em New Journal of Physics}, 13:043016, 2011.
\newblock (arXiv:0906.4725v3 [quant-ph]).

\bibitem[Che07]{Che07}
Eugenia Cheng.
\newblock Adjunctions 1.
\newblock TheCatsters Channel, YouTube, September 13 2007.
\newblock (\url{http://www.youtube.com/watch?v=loOJxIOmShE}).

\bibitem[Coe06]{Coe06}
Bob Coecke.
\newblock Introducing categories to the practicing physicist.
\newblock {\em Advanced Studies in Mathematics and Logic, Polimetrica
  Publishing}, 30:45--74, 2006.
\newblock (arXiv:0808.1032v1 [quant-ph]).

\bibitem[CP06]{CP06}
Bob Coecke and {\'{E}ric Oliver}~Paquette.
\newblock {POVMs and Naimark's theorem without sums}.
\newblock {\em Electronic Notes in Theoretical Computer Science}, 2006.
\newblock (arXiv:quant-ph/0608072).

\bibitem[CP07]{CP07}
Bob Coecke and Du\v{s}ko Pavlovi\'{c}.
\newblock Quantum measurements without sums.
\newblock In G.~Chen, L.~Kauffman, and S.~Lamonaco, editors, {\em Mathematics
  of Quantum Computing and Technology}, pages 567–--604. Taylor and Francis,
  2007.
\newblock (arXiv:quant-ph/0608035).

\bibitem[CPP08]{CPP08}
Bob Coecke, {\'{E}ric Oliver}~Paquette, and Simon Perdrix.
\newblock Bases in diagrammatic quantum protocols.
\newblock {\em Electronic Notes in Theoretical Computer Science}, 218:131--152,
  2008.
\newblock (arXiv:0808.1029v1 [quant-ph]).

\bibitem[CPP10]{CPP10}
Bob Coecke, {\'{E}ric Oliver}~Paquette, and Du\v{s}ko Pavlovi\'{c}.
\newblock Classical and quantum structuralism.
\newblock In S.~Gay and I.~Mackie, editors, {\em Semantic Techniques in Quantum
  Computation}. Cambridge University Press, 2010.
\newblock (arXiv:0904.1997v2 [quant-ph]).

\bibitem[CPV08]{CPV08}
Bob Coecke, Du\v{s}ko Pavlovi\'{c}, and Jamie Vicary.
\newblock A new description of orthogonal bases.
\newblock {\em Mathematical Structures in Computer Science}, page~13, 2008.
\newblock (arXiv:0810.0812v1 [quant-ph]).

\bibitem[CWW{\etalchar{+}}11]{CWWWZ11}
Bob Coecke, Quanlong Wang, Baoshan Wang, Yongjun Wang, and Qiye Zhang.
\newblock Graphical calculus for quantum key distribution (extended abstract).
\newblock {\em Electronic Notes in Theoretical Computer Science},
  270(2):231--249, 2011.
\newblock <ce:title>Proceedings of the 6th International Workshop on Quantum
  Physics and Logic (QPL 2009)</ce:title>.

\bibitem[DKP07]{DKP07}
Vincent Danos, Elham Kashefi, and Prakash Panangaden.
\newblock The measurement calculus.
\newblock {\em Journal of the ACM (JACM)}, 54(2), 2007.
\newblock (arXiv:0704.1263v1 [quant-ph]).

\bibitem[DP10]{DP10}
Ross Duncan and Simon Perdrix.
\newblock Rewriting measurement-based quantum computations with generalised
  flow.
\newblock In {\em Proceedings of the 37th international colloquium conference
  on Automata, languages and programming: Part II}, ICALP'10, pages 285--296,
  Berlin, Heidelberg, 2010. Springer-Verlag.

\bibitem[DR89]{DR89}
Sergio Doplicher and {John E.} Roberts.
\newblock A new duality theory for compact groups.
\newblock {\em Inventiones mathematicae}, 98(1):157--218, 1989.

\bibitem[Eke91]{Eke91}
{Arthur K.} Ekert.
\newblock Quantum cryptography based on bell's theorem.
\newblock {\em Physical review letters}, 67(6):661--663, 1991.

\bibitem[Gir87]{Gir87}
{Jean-Yves} Girard.
\newblock Linear logic.
\newblock {\em Theoretical Computer Science}, 50(1):1--102, 1987.

\bibitem[JNN12]{JNN12}
{J. R.} Johansson, {P. D.} Nation, and Franco Nori.
\newblock Qutip: An open-source python framework for the dynamics of open
  quantum systems.
\newblock {\em Computer Physics Communications}, 183(8):1760 -- 1772, 2012.
\newblock (arXiv:1110.0573 [quant-ph]).

\bibitem[JS91]{JS91}
Andr\'{e} Joyal and Ross Street.
\newblock An introduction to tannaka duality and quantum groups.
\newblock In Aurelio Carboni, Maria Pedicchio, and Guiseppe Rosolini, editors,
  {\em Category Theory}, volume 1488 of {\em Lecture Notes in Mathematics},
  pages 411--492. Springer Berlin / Heidelberg, 1991.

\bibitem[Kis11]{Kis11}
Aleks Kissinger.
\newblock {\em Pictures of processes: Automated graph rewriting for monoidal
  categories and applications to quantum computing}.
\newblock PhD thesis, Department of Computer Science, University of Oxford,
  2011.
\newblock (arXiv:1203.0202 [math.CT]).

\bibitem[KL01]{KL01}
Thomas Kerler and {Volodymyr Vasyliovych} Lyubashenko.
\newblock {\em Non-Semisimple Topological Quantum Field Theories for
  3-Manifolds with Corners}.
\newblock Springer, 2001.

\bibitem[Klo92]{Klo92}
{Jan Willem} Klop.
\newblock Term rewriting systems.
\newblock In S.~Abramsky, {D.M.} Gabbay, and {T.S.E.} Maibaum, editors, {\em
  Handbook of Logic in Computer Science}, volume~2, pages 1–--116. Oxford
  University Press, 1992.
\newblock (DOI:10.1.1.35.425).

\bibitem[Koc03]{Koc03}
Joachim Kock.
\newblock {\em Frobenius Algebras and 2-D Topological Quantum Field Theories},
  volume~59 of {\em London Mathematical Society Student Texts}.
\newblock Cambridge University Press, 2003.

\bibitem[{Mac}98]{Mac98}
Saunders {Mac Lane}.
\newblock {\em Categories for the Working Mathematician}.
\newblock Springer, second edition, 1998.

\bibitem[Mel09]{Mel09}
{Paul-Andr\'{e}} Melli\`{e}s.
\newblock Categorical semantics of linear logic.
\newblock {\em Panoramas et synth\`{e}ses - Soci\'{e}t\'{e} math\'{e}matique de
  France}, (27):1--196, 2009.
\newblock (DOI:10.1.1.62.5117).

\bibitem[Mer07]{Mer07}
{N. David} Mermin.
\newblock {\em Quantum Computer Science: An Introduction}.
\newblock Cambridge University Press, first edition, 2007.

\bibitem[NC00]{NC00}
{Michael A.} Nielsen and {Isaac L.} Chuang.
\newblock {\em Quantum Computation and Quantum Information}.
\newblock Cambridge University Press, 2000.

\bibitem[PB00]{PB00}
{Arun Kumar} Pati and {Samuel L.} Braunstein.
\newblock Impossibility of deleting an unknown quantum state.
\newblock {\em Nature}, 404:164--–165, 2000.
\newblock (arXiv:quant-ph/9911090v2).

\bibitem[RB01]{RB01}
Robert Raussendorf and {Hans J.} Briegel.
\newblock A one-way quantum computer.
\newblock {\em Physical Review Letters}, 86(22):5188, 2001.
\newblock DOI:10.1103/PhysRevLett.86.5188.

\bibitem[RB02]{RB02}
Robert Raussendorf and {Hans J.} Briegel.
\newblock Computational model underlying the one-way quantum computer.
\newblock {\em Quantum Information and Computation}, 2, 2002.
\newblock (arXiv:quant-ph/0108067v2).

\bibitem[RBB03]{RBB03}
Robert Raussendorf, {Dan E.} Browne, and {Hans J.} Briegel.
\newblock Measurement-based quantum computation on cluster states.
\newblock {\em Physical Review A}, 68, 2003.
\newblock (arXiv:quant-ph/0301052v2, DOI:10.1103/PhysRevA.68.022312).

\bibitem[Sel04a]{Sel04a}
Peter Selinger.
\newblock A brief survey of quantum programming languages.
\newblock In {\em Proceedings of the 7th International Symposium on Functional
  and Logic Programming}, volume 2998, pages 1--6, Nara, Japan, 2004. Springer
  Lecture Notes in Computer Science.
\newblock (DOI:10.1.1.94.463).

\bibitem[Sel04b]{Sel04b}
Peter Selinger.
\newblock Towards a quantum programming language.
\newblock {\em Mathematical Structures in Computer Science}, 14(4):527--586,
  2004.
\newblock (DOI:10.1.1.144.6380).

\bibitem[Sel07]{Sel05}
Peter Selinger.
\newblock Dagger compact closed categories and completely positive maps.
\newblock In {\em Proceedings of the 3rd International Workshop on Quantum
  Programming Languages (QPL 2005)}, volume 170, pages 139--163, Chicago, 2007.
  Electronic Notes in Theoretical Computer Science.
\newblock (DOI:10.1.1.134.2476).

\bibitem[Sel10]{Sel10}
Peter Selinger.
\newblock A survey of graphical languages for monoidal categories.
\newblock In Bob Coecke, editor, {\em New Structures for Physics}. Springer
  Lecture Notes in Physics, 2010.
\newblock (arXiv:0908.3347v1 [math.CT]).

\bibitem[Sho97]{Sho97}
{Peter W.} Shor.
\newblock Polynomial-time algorithms for prime factorization and discrete
  logarithms on a quantum computer.
\newblock In {\em SIAM Journal on Scientific and Statistical Computing},
  volume~26, page 1484–1509, 1997.
\newblock arXiv:quant-ph/9508027v2.

\bibitem[Str04]{Str04}
Ross Street.
\newblock Frobenius monads and pseudomonoids.
\newblock {\em Journal of Mathematical Physics}, 45(10.III):3930--3948, 2004.
\newblock (DOI:10.1063/1.1788852).

\bibitem[SV06]{SV06}
Peter Selinger and Beno\^{i}t Valiron.
\newblock A lambda calculus for quantum computation with classical control.
\newblock {\em Mathematical Structures in Computer Science}, 16(3):527--552,
  2006.
\newblock (arXiv:cs/0404056v2 [cs.LO]).

\bibitem[SV08]{SV08}
Peter Selinger and Beno\^{i}t Valiron.
\newblock A linear-non-linear model for a computational call-by-value lambda
  calculus (extended abstract).
\newblock In {\em Proceedings of the Eleventh International Conference on
  Foundations of Software Science and Computation Structures (FOSSACS 2008)},
  volume 4962, pages 81--96, Budapest, 2008. Springer Lecture Notes in Computer
  Science.
\newblock (arXiv:0801.0813v1 [cs.LO]).

\bibitem[SV10]{SV10}
Peter Selinger and Beno\^{i}t Valiron.
\newblock Quantum lambda calculus.
\newblock In S.~Gay and I.~Mackie, editors, {\em Semantic Techniques in Quantum
  Computation}. Cambridge University Press, 2010.
\newblock (\url{http://www.mscs.dal.ca/~selinger/papers.html#qlambdabook}).

\bibitem[vT04]{vT04}
Andr\'{e} van Tonder.
\newblock A lambda calculus for quantum computation.
\newblock {\em SIAM Journal on Computing}, 33(5):1109--1135, 2004.
\newblock (arXiv:quant-ph/0307150v5).

\bibitem[vTD03]{vTD03}
Andr\'{e} van Tonder and Miquel Dorca.
\newblock Quantum computation, categorical semantics and linear logic.
\newblock Archive, 2003.
\newblock (arXiv:quant-ph/0312174v4).

\bibitem[WZ82]{WZ82}
{William K.} Wootters and {Wojciech H.} Zurek.
\newblock A single quantum cannot be cloned.
\newblock {\em Nature}, 299:802–--803, 1982.
\newblock (DOI:10.1038/299802a0).

\end{thebibliography}

\clearpage
\renewcommand{\listtheoremname}{List of Definitions}
\listoftheorems
[
    ignoreall,
    show={definition}
]

\clearpage
\renewcommand{\listtheoremname}{List of Theorems, Lemmas and Corollaries}
\listoftheorems
[
    ignoreall,
    show={theorem,lemma,corollary} % MUST NOT HAVE SPACES!
]
\end{document}